\documentclass[11pt,english,ruled]{article}
\usepackage[latin9]{inputenc}
\usepackage{geometry}
\usepackage{color}
\usepackage{babel}
\usepackage{booktabs}
\usepackage{calc}
\usepackage{url}
\usepackage{multirow}
\usepackage{algorithm2e}
\usepackage{amsmath}
\usepackage{amsthm}
\usepackage{amssymb}
\usepackage{graphicx}
\usepackage{setspace}
\usepackage[authoryear]{natbib}
\PassOptionsToPackage{normalem}{ulem}
\usepackage{ulem}
\usepackage[bookmarks=true,bookmarksnumbered=false,bookmarksopen=false,
 breaklinks=false,pdfborder={0 0 0},pdfborderstyle={},backref=false,colorlinks=true]
 {hyperref}
\hypersetup{
 linkcolor=red,citecolor=blue}

\makeatletter

\providecommand{\tabularnewline}{\\}

\theoremstyle{remark}
\newtheorem{rem}{\protect\remarkname}
\theoremstyle{plain}
\newtheorem{lyxalgorithm}{\protect\algorithmname}
\theoremstyle{plain}
\newtheorem{assumption}{\protect\assumptionname}
\theoremstyle{plain}
\newtheorem{lem}{\protect\lemmaname}
\theoremstyle{plain}
\newtheorem{prop}{\protect\propositionname}
\theoremstyle{plain}
\newtheorem{thm}{\protect\theoremname}

\usepackage{cite}

\numberwithin{equation}{section}
\usepackage{amsfonts}
\usepackage{amsthm}
\usepackage{mathrsfs}
\usepackage{bbm}
\usepackage{appendix}

\usepackage{multirow}
\usepackage{array}
\usepackage{rotating}
\usepackage{longtable}
\usepackage{float}
\usepackage{booktabs}

\usepackage{caption}
\usepackage{subcaption}

\providecommand{\tabularnewline}{\\}

\usepackage{color}
\usepackage{xcolor}

\usepackage{url}

\usepackage{chngcntr}
\counterwithin*{section}{part}

\usepackage{todonotes}
\usepackage{rotfloat}
\usepackage{pdfpages}
\usepackage{indentfirst}

\newcommand{\lep}{\stackrel{\mathrm{p}}{\preccurlyeq}}
\newcommand{\lec}{\preccurlyeq}
\newcommand{\nep}{\stackrel{\mathrm{p}}{\asymp}}
\newcommand{\convp}{\stackrel{\mathrm{p}}{\to}}
\newcommand{\convd}{\stackrel{\mathrm{d}}{\to}}
\renewcommand{\hat}{\widehat}
\renewcommand{\tilde}{\widetilde}

\date{}

\makeatother

\providecommand{\algorithmname}{Algorithm}
\providecommand{\assumptionname}{Assumption}
\providecommand{\lemmaname}{Lemma}
\providecommand{\propositionname}{Proposition}
\providecommand{\remarkname}{Remark}
\providecommand{\theoremname}{Theorem}

\begin{document}
\title{LASSO Inference for High Dimensional Predictive Regressions\thanks{Zhan Gao: \protect\url{zhangao@smu.edu}; Ji Hyung Lee: \protect\url{jihyung@illinois.edu};
Ziwei Mei (Corresponding author): \protect\url{ziweimei@um.edu.mo},
Faculty of Business Administration, University of Macau, Taipa, Macao
SAR, China; Zhentao Shi: \protect\url{zhentao.shi@cuhk.edu.hk}. 
}}
\author{Zhan Gao$^{a}$, Ji Hyung Lee$^{b}$, Ziwei Mei$^{c}$, Zhentao Shi$^{d}$
\\
\\
$^{a}$Department of Economics, Southern Methodist University \\
$^{b}$Department of Economics, University of Illinois at Urbana-Champaign\\
$^{c}$Faculty of Business Administration, University of Macau\\
$^{d}$Department of Economics, The Chinese University of Hong Kong}
\maketitle
\begin{abstract}
LASSO inflicts shrinkage bias on estimated coefficients, which undermines
asymptotic normality and invalidates standard inferential procedures
based on the \emph{t}-statistic. Given cross sectional data, the desparsified
LASSO has emerged as a well-known remedy for correcting the shrinkage
bias. In the context of high dimensional predictive regression, the
desparsified LASSO faces an additional challenge: the Stambaugh bias
arising from nonstationary regressors modeled as local unit roots.
To restore standard inference, we propose a novel estimator called
IVX-desparsified LASSO (XDlasso). XDlasso simultaneously eliminates
both shrinkage bias and Stambaugh bias and does not require prior
knowledge about the identities of nonstationary and stationary regressors.
We establish the asymptotic properties of XDlasso for hypothesis testing,
and our theoretical findings are supported by Monte Carlo simulations.
Applying our method to real-world applications from the FRED-MD database,
we investigate two important empirical questions: (i) the predictability
of the U.S. stock returns based on the earnings-price ratio, and (ii)
the predictability of the U.S. inflation using the unemployment.
\end{abstract}
\thispagestyle{empty}

\vspace{0.8cm}



\vspace{0.8cm}

\small \noindent 

\newpage{}

\normalsize

\section{Introduction}

The evaluation of economic and financial predictability has attracted
widespread interest from theoretical and applied researchers for many
decades. In today's era of big data, we have unprecedented access
to a vast amount of digital information about the economy. Recent
advancements in inference with high dimensional data have uncovered
new empirical patterns in predictive practices using large-scale datasets
with temporal features.

 This paper aims at a plain quest: in a high dimensional linear predictive
regression model, where the number of potential regressors is larger
than the sample size, how can one conduct valid inference for a regressor
of primary interest? No research has solved this question before.
The challenges are twofold. First, predictive regressions were mainly
studied in the low-dimensional context. A defining feature of predictive
regression theory lies in persistent regressors \citep{stambaugh1999predictive},
which upend standard inference based on the standard $t$-statistics.
Second, one must estimate the coefficient via some regularization
methods to cope with high dimensionality. For example, when the underlying
true regression model is sparse, LASSO \citep{tibshirani1996regression}
is the off-the-shelf method. It is well known that the LASSO estimator
is biased toward zero due to absolute-value shrinkage and exhibits
a nonstandard asymptotic distribution that is distinct from the normal
distribution. If we intend to provide an asymptotically normally distributed
estimator to facilitate standard statistical inference, we must simultaneously
combat two evils: the \emph{Stambaugh bias} due to persistent regressors
and the \emph{shrinkage bias} caused by the LASSO penalty. 

The above diagnosis hints at a plausible solution path. In low-dimensional
predictive regressions, \citet{phillips2009econometric}'s IVX method
leverages a self-generated instrument to alleviate the regressor persistence,
thereby overcoming the Stambaugh bias. In high dimensional cross sectional
regressions, \citet{zhang2014confidence}'s desparsified LASSO (Dlasso)
constructs a score vector for the parameter of interest and removes
the shrinkage bias via an auxiliary regression. Each method provides
an asymptotic normal estimator in its respective environment. 

Can we combine these two methods into a single procedure to address
inference in high dimensional predictive regressions? We find that
the answer is both \emph{no} and \emph{yes}. ``No'' is in the sense
that a naive combination of the two does not lead to desirable results.
``Yes'', on the other hand, is established upon a deep understanding
of the mechanisms of both components and their adaptation to the context.
This research culminates in a new \emph{IVX-Desparsified LASSO} (XDlasso)
estimator that is free from both biases and has an asymptotic normal
distribution. 

With a predictor of interest in mind, the construction of XDlasso,
detailed in Algorithm \ref{alg:XDlasso} in Section \ref{subsec:IVX-Desparsified-LASSO},
is summarized as follows. First, a workhorse estimator is needed to
lay the groundwork for a high dimensional predictive regression. When
both nonstationary and stationary predictors are present in the regression,
\citet[MS24 hereafter]{mei2022lasso} has recently established the
consistency of the standardized LASSO (Slasso), making it a natural
candidate for the workhorse estimator. Beyond the pure unit roots
considered in MS24, we extend the characterization of nonstationary
regressors by generalizing our framework to allow for local unit root
(LUR) processes. The convergence rates of Slasso for LURs align with
those for pure unit roots. However, the technical proofs for LURs
are more involved than those for pure unit roots. We address the complexity
arising from high dimensional predictors with both LURs and stationary
regressors, and derive the convergence rates of the initial Slasso
estimator. 

Second, the common practice of generating the instrument in IVX is
insufficient --- the IV must be scale-standardized to have the stochastic
order aligned with all other predictors in Slasso. The standardized
IV serves as the target variable for the auxiliary Slasso regression
in Dlasso to estimate the shrinkage bias. The XDlasso estimator of
the parameter of interest is defined as the initial Slasso estimator
plus the bias-correction term, and the companion $t$-statistic is
employed for statistical inference by comparing it to critical values
from the standard normal distribution. 

We further establish the asymptotic normality of our proposed XDlasso
estimator and the convergence rate of its standard error. Specifically,
the XDlasso estimator is $\sqrt{n}$-consistent for a stationary regressor
while its convergence accelerates for an LUR regressor. Moreover,
to conduct simultaneous inference for multiple parameters of interest,
we develop a Wald statistic with an asymptotic $\chi^{2}$ distribution
based on XDlasso. This Wald test is valid even when the parameters
involve both stationary and nonstationary regressors.

To tackle persistent regressors, the self-generated IVX instrument
in the second step is the key ingredient\textbf{.} The generated IV
is less persistent than the nonstationary regressors modeled as LURs.
This important feature enables us to decorrelate the IV from other
covariates in the auxiliary LASSO regression, so that the resulting
XDlasso estimator possesses these two properties: (i) It is free from
the Stambaugh bias and thus enjoys asymptotic normality; (ii) It reduces
the order of shrinkage bias to make it correctable. In contrast, the
ordinary Dlasso encounters a \emph{spurious} auxiliary regression,
failing to correct the bias arising from persistent regressors (see
Section \ref{subsec:Necessity-of-IVX} for details). More importantly,
XDlasso inference does not require \emph{a priori} knowledge of the
persistence of the regressor of interest and is thus immune to pretesting
bias. To the best of our knowledge, this is the first methodology
to handle the inferential problem in high dimensional predictive regressions
($p\gg n$) with nonstationary predictors. This is also the first
paper that extends the IVX technique into the high dimensional framework. 

Monte Carlo simulations show that XDlasso successfully removes the
bias for inference on the coefficient of a nonstationary regressor,
but the ordinary\emph{ }Dlasso fails to do so. Our procedure is applied
to the high dimensional macroeconomic FRED-MD dataset \citep{mccracken2016fred}
with both stationary and persistent variables, to study two important
macro-finance problems: financial market return predictability and
the Phillips curve in macroeconomics.

\textbf{Literature review}. With the advent of big data, machine learning
methods have spread to time series topics such as nonstationarity
\citep{phillips2021boosting,smeekes2021automated,mei2024boosted},
cointegration testing \citep{onatski2018alternative,zhang2019identifying,bykhovskaya2022cointegration},
and structural breaks \citep{deshpande2023online,tu2023penetrating}.
This paper builds on several strands of literature. First, LASSO is
one of the most studied methods in recent years, with well-developed
theory in high dimension \citep{bickel2009simultaneous}. It is well
received and used for economic applications; see \citet{belloni2012sparse},
\citet{shi2016estimation}, \citet{caner2018asymptotically}, and
\citet{babii2022machine}, to name a few. In recent years, the properties
of LASSO are studied in various topics in high dimensional time series,
including nonstationary time series models \citep{koo2020high,lee2022lasso}
and inference based on the heteroskedasticity and autocorrelation
consistent (HAC) estimation \citep{babii2021highdimensionalgrangercausalitytests,babii2024high}.
None of these works has considered hypothesis testing problem for
high dimensional predictive regressions with both LURs and stationary
regressors.

Hypothesis testing after LASSO is challenging because of the shrinkage
bias. To validate hypothesis testing in high dimensions, \citet{zhang2014confidence},
\citet{van2014asymptotically}, and \citet{javanmard2014confidence}
have developed the desparsified (debiased) LASSO estimators under
the independently and identically distributed (i.i.d.) setting. \citet{adamek2022lasso}
generalize the Dlasso inference\emph{ }to high dimensional stationary
time series. We follow this line of desparsified LASSO literature
thanks to its convenience, which requires a baseline regression and
an auxiliary regression only. On the other hand, \citet{chernozhukov2018}'s
\emph{double machine learning }(DML) is a more general theoretical
framework of debiased inference, widely used in cross sectional data
where sample-splitting is readily implementable. However, none of
the aforementioned works has devised any inferential procedure for
high dimensional nonstationary time series. \citet{hecq2023inference}
apply a post-double selection procedure to test the Granger causality
in high dimensional nonstationary vector autoregressive models with
cointegrated data. In contrast, our procedure relies on desparsified
LASSO without variable selection. 

The other strand is the vast literature on predictive regressions.
As highlighted by \citet{campbell2006efficient} and \citet{jansson2006optimal},
non-standard distortion in the asymptotic distribution arises from
persistent regressors. The peculiar asymptotic distributions invalidate
the standard inferential procedures. There have been multiple proposals
for valid inference, for example, the Bonferroni method \citep{campbell2006efficient},
the conditional likelihood method \citep{jansson2006optimal}, the
linear projection method \citep{cai2014testing}, the weighted empirical
likelihood approach \citep{zhu2014predictive,liu2019unified,yang2021unified},
and the implication-based inference \citep{xu2020testing}. Some of
these methods are designed for univariate predictive regressions;
it would be difficult to extend them to the high dimensional case,
where regularization is required to handle many parameters. On the
other hand, \citet{phillips2009econometric}'s IVX estimator gained
its popularity by recovering asymptotic normality, enabling valid
inference for mean regressions \citep{kostakis2015robust,kostakis2018taking,phillips2013predictive,phillips2016robust,yang2020testing,demetrescu2023extensions}
and quantile regressions \citep{lee2016predictive,fan2019predictive,cai2023new,liu2023unified}
with low dimensional regressors. IVX recovers asymptotic normality
by projecting the persistent regressor onto a self-generated IV.

\textbf{Layout}. The rest of the paper is organized as follows. Section
\ref{sec:Model-and-Methodology} introduces the high dimensional predictive
regression model with a mixture of stationary and nonstationary regressors
and proposes XDlasso. Section \ref{sec:Asymptotic-Theory} establishes
the theoretical results, justifying the size and power of the XDlasso
inference procedure. Section \ref{sec:Simulation} carries out simulation
studies that corroborate the theory. Section \ref{sec:Empirical-Application}
applies XDlasso inference to two macro-finance empirical examples.
Technical proofs are relegated to the Online Appendices. 

\textbf{Notations.} We set up the notation before the formal discussion.
We define $\boldsymbol{1}\left\{ \cdot\right\} $ as the indicator
function, and $\Delta$ as the difference operator so that $\Delta x_{t}=x_{t}-x_{t-1}$.
The set of natural numbers, integers, and real numbers are denoted
as $\mathbb{N}$, $\mathbb{Z}$, and $\mathbb{R}$, respectively.
For some $n\in\mathbb{N}$, the integer set $\{1,2,\cdots,n\}$ is
denoted as $[n]$, and the space of $n$-dimensional vectors is denoted
as $\mathbb{R}^{n}.$ For $x=(x_{t})_{t\in[n]}\in\mathbb{R}^{n}$,
the $L_{1}$-norm is $\left\Vert x\right\Vert _{1}=\sum_{t=1}^{n}\left|x_{t}\right|$,
and the sup-norm is $\|x\|_{\infty}=\sup_{t\in[n]}|x_{t}|$. Let $0_{n}$
be an $n\times1$ zero vector, $1_{n}$ be an $n\times1$ vector of
ones, and $I_{n}$ be the $n\times n$ identity matrix. For a generic
matrix $B,$ let $B_{ij}$ be the $(i,j)$-th element, and $B^{\top}$
be its transpose. Let $\|B\|_{\infty}=\max_{i,j}|B_{ij}|$, and $\lambda_{\min}(B)$
and $\lambda_{\max}(B)$ be the minimum and maximum eigenvalues, respectively.
Define $a\wedge b:=\min\left\{ a,b\right\} $, and $a\vee b:=\max\left\{ a,b\right\} $.
An \emph{absolute constant} is a positive, finite constant that is
invariant with the sample size. The abbreviation ``w.p.a.1'' is short
for ``with probability approaching one''. We use $\convp$ and $\convd$
to denote convergence in probability and in distribution, respectively.
For any time series $\{a_{t}\}_{t=1}^{n}$, we use $\bar{a}$ to denote
its sample mean $n^{-1}\sum_{t=1}^{n}a_{t}$. For any time series
$\{a_{t}\}$ and $\{b_{t}\}$, we say they are \emph{asymptotically
uncorrelated} if their sample correlation coefficient $\frac{\sum_{t=1}^{n}(a_{t}-\bar{a})(b_{t}-\bar{b})}{\sqrt{\sum_{t=1}^{n}(a_{t}-\bar{a})^{2}\sum_{t=1}^{n}(b_{t}-\bar{b})^{2}}}\convp0$
as $n\to\infty.$ 

\section{Model and Procedure \protect\label{sec:Model-and-Methodology}}

\subsection{High Dimensional Predictive Regression}

Suppose that a time series of the outcome $y_{t}$ is generated by
the following linear predictive regression:\footnote{For simplicity of exposition, an intercept in \eqref{eq:DGP} is omitted,
without loss of generality. As explained by MS24, the intercept in
LASSO can be handled by the well-known Frisch-Waugh-Lovell theorem.
In practical implementation --- throughout all simulations and empirical
exercises in this paper --- we keep an unpenalized intercept in the
model.}
\begin{align}
y_{t} & =W_{t-1}^{\top}\theta^{*}+u_{t}=X_{t-1}^{\top}\beta^{*}+Z_{t-1}^{\top}\gamma^{*}+u_{t},\label{eq:DGP}
\end{align}
where the error term $u_{t}$ is a stationary martingale difference
sequence (m.d.s.)~with mean zero and conditional variance $\sigma_{u}^{2}$.
We consider two types of regressors with different stochastic properties.
Firstly, the $p_{x}\times1$ vector $X_{t}=(x_{1,t},\cdots,x_{p_{x},t})^{\top}$
collects the LURs:
\begin{equation}
x_{j,t}=\rho_{j}^{*}x_{j,t-1}+e_{j,t}\text{ for }j=1,2,...,p_{x},\label{eq:AR1}
\end{equation}
where $e_{t}=(e_{1,t},\cdots,e_{p_{x},t})^{\top}$ is a $p_{x}$-dimensional
vector of stationary time series. The AR(1) coefficient $\rho_{j}^{*}$
in (\ref{eq:AR1}) is close to 1 when the sample size $n$ is large,
specified as 
\begin{equation}
\rho_{j}^{*}=1+\frac{c_{j}^{*}}{n}\text{ for }j=1,2,...,p_{x},\label{eq:LUR}
\end{equation}
where $c_{j}^{*}\in\mathbb{R}$ is allowed to be negative, positive,
or zero. Therefore, our framework accommodates nonstationary regressors
that are locally integrated ($c_{j}^{*}<0$), unit roots ($c_{j}^{*}=0$),
and locally explosive ($c_{j}^{*}>0$). Secondly, stationary regressors
are stored in the $p_{z}\times1$ vector $Z_{t}=(z_{1,t},\cdots,z_{p_{z},t})^{\top}$.
The two types of regressors are combined into a long vector $W_{t}=(X_{t}^{\top},Z_{t}^{\top})^{\top}=(w_{1,t},\cdots,w_{p,t})^{\top}$
of $p$ ($=p_{x}+p_{z}$) elements, and the associated coefficients
are placed into $\theta^{*}=(\beta^{*\top},\gamma^{*\top})^{\top}\in\mathbb{R}^{p}$.
Following the literature, we refer to $W_{t}$, which has multiple
degrees of persistence, as \emph{mixed root }regressors. For simplicity,
let the initial value $\|W_{t=0}\|_{\infty}=O_{p}(1)$. Define the
sparsity index $s=\sum_{j=1}^{p}\boldsymbol{1}\{\theta_{j}^{*}\neq0\}$
as the number of nonzero components in the coefficient vector $\theta^{*}$.

As in the default \texttt{R} program option \texttt{glmnet::glmnet(x,y)},
it is a common practice in LASSO to scale-standardize each regressor
$w_{j,t}$ by its sample standard deviation (s.d.)~ $\hat{\sigma}_{j}=\sqrt{\frac{1}{n}\sum_{t=1}^{n}(w_{j,t-1}-\bar{w}_{j})^{2}},$
where $\bar{w}_{j}=n^{-1}\sum_{t=1}^{n}w_{j,t-1}$ is the sample mean.
Let the diagonal matrix $D={\rm diag}\left(\hat{\sigma}_{1},\hat{\sigma}_{2},\cdots,\hat{\sigma}_{p}\right)$
store the sample standard deviations. The standardized LASSO (Slasso)
estimator is 
\begin{equation}
\hat{\theta}^{{\rm S}}:=\arg\min_{\theta\in\mathbb{R}^{p}}\dfrac{1}{n}\sum_{t=1}^{n}(y_{t}-W_{t-1}^{\top}\theta)^{2}+\lambda\|D\theta\|_{1}.\label{eq:Slasso}
\end{equation}
The Slasso estimator is scale-invariant: if the regressor $w_{j,t-1}$
is multiplied by a nonzero constant $m$, then the $j$-th coefficient
estimator changes proportionally to $\hat{\theta}_{j}^{{\rm S}}/m.$
The standardization renders the magnitudes of LURs into the same order
as those of stationary regressors, so that the same LASSO tuning parameter
$\lambda$ in (\ref{eq:Slasso}) is valid for both stationary and
persistent regressors. In contrast, the plain LASSO (Plasso) with
the matrix $D$ in (\ref{eq:Slasso}) replaced with the identity matrix
is scale-variant. What is worse, equipped with a single tuning parameter
$\lambda$, Plasso favors the LURs of a larger order, and shrinks
the coefficients of stationary regressors with a smaller order all
the way to zero --- thus becoming inconsistent.
\begin{rem}[Choice of Scaling in Slasso]
We follow the default option of the statistical software to use the
s.d.~for scaling in the Slasso estimator (\ref{eq:Slasso}). The
purpose of scaling is to ensure consistency of Slasso (\ref{eq:Slasso}).
The long-run variance, which requires tuning a bandwidth in its estimation,
provides no additional benefit. We therefore prefer and stick to the
vanilla s.d.~for scaling. 
\end{rem}

\subsection{Two Types of Biases \protect\label{subsec:The-Biases}}

When data are i.i.d.~or stationary, LASSO is subject to shrinkage
bias as well as a nonstandard asymptotic distribution, which cannot
be used for standard inference \citep{fu2000asymptotics}. This motivates
\citet{zhang2014confidence} to bring forth the desparsified LASSO
to correct the bias and recover asymptotic normality. With high dimensional
LURs, $\hat{\theta}^{{\rm S}}$ is subject to not only the shrinkage
bias, but also the Stambaugh bias due to persistence, which further
distorts the standard \emph{t}-statistics inference. In this section,
we examine both the shrinkage bias and the Stambaugh bias, and propose
XDlasso for correcting both biases. 

We are interested in inference on a null hypothesis $\mathbb{H}_{0}:\theta_{j}^{*}=\theta_{0,j}$
for a $j\in[p]$, a prevalent practice in empirical studies. In a
low dimensional linear regression where $p$ is fixed, the Frisch-Waugh-Lovell
theorem yields the following formulation of the ordinary least squares
(OLS) estimator:
\[
\hat{\theta}_{j}^{{\rm OLS}}=\frac{\sum_{t=1}^{n}w_{j,t-1}^{\perp}y_{t}}{\sum_{t=1}^{n}w_{j,t-1}^{\perp}w_{j,t-1}},
\]
where $w_{j}^{\perp}=(w_{j,0}^{\perp},\dots,w_{j,n-1}^{\perp})^{\top}$
is the OLS residual from regressing $w_{j,t}$ on all other regressors
$W_{-j,t}=(w_{k,t})_{k\neq j}$. OLS induces a large variance as $p$
gets large, and becomes infeasible when $p>n$. Now, consider replacing
the OLS residual $w_{j}^{\perp}$ by a generic \emph{score vector}
$r_{j}=(r_{j,0},\dots,r_{j,n-1})^{\top}$ to construct an estimator
of $\theta_{j}^{*}$ that is linear in $y_{t}$ in the form of $\hat{\theta}_{j}^{({\rm lin})}=\frac{\sum_{t=1}^{n}r_{j,t-1}y_{t}}{\sum_{t=1}^{n}r_{j,t-1}w_{j,t-1}}.$
Since 
\[
y_{t}=w_{j,t-1}\theta_{j}^{*}+W_{-j,t-1}^{\top}\theta_{-j}^{*}+u_{t}
\]
 where $\theta_{-j}^{*}=(\theta_{k}^{*})_{k\neq j}$ is the vector
of coefficients excluding the $j$-th entry, the generic estimator
$\hat{\theta}_{j}^{({\rm lin})}$ can be decomposed into
\[
\hat{\theta}_{j}^{({\rm lin})}=\theta_{j}^{*}+\frac{\sum_{t=1}^{n}r_{j,t-1}u_{t}}{\sum_{t=1}^{n}r_{j,t-1}w_{j,t-1}}+\frac{\sum_{t=1}^{n}r_{j,t-1}W_{-j,t-1}^{\top}\theta_{-j}^{*}}{\sum_{t=1}^{n}r_{j,t-1}w_{j,t-1}}=:\theta_{j}^{*}+{\rm N}_{j}+{\rm B}_{j},
\]
where ${\rm N}_{j}$ is the noise component that determines the asymptotic
distribution of $\hat{\theta}_{j}^{({\rm lin})}$, and ${\rm B}_{j}$
is the potential bias due to the choice of $r_{j}$. 

For OLS, the score vector $r_{j}=w_{j}^{\perp}$ is orthogonal to
the column space of $W_{-j,\cdot}:=(W_{-j,0},\dots,W_{-j,n-1})^{\top}$,
under which ${\rm B}_{j}=0$ and no bias is present. The bias term
${\rm B}_{j}$ pops up whenever $r_{j}$ is not orthogonal to $W_{-j,\cdot}$,
which happens if we add a penalty to the OLS objective function. With
the LASSO penalty at place, we call ${\rm B}_{j}$ \emph{shrinkage
bias. }

Following \citet{zhang2014confidence}, we replace the unknown parameter
$\theta_{-j}^{*}$ by the feasible workhorse Slasso estimator $\hat{\theta}_{-j}^{{\rm S}}=(\hat{\theta}_{k}^{{\rm S}})_{k\neq j}$
to obtain
\begin{equation}
\hat{\theta}_{j}=\hat{\theta}_{j}^{({\rm lin})}-\widehat{{\rm B}}_{j}\text{ where }\widehat{{\rm B}}_{j}=\frac{\sum_{t=1}^{n}r_{j,t-1}W_{-j,t-1}^{\top}\hat{\theta}_{-j}^{{\rm S}}}{\sum_{t=1}^{n}r_{j,t-1}w_{j,t-1}}\label{eq:debiased}
\end{equation}
to compensate for ${\rm B}_{j}$. Equivalently, $\hat{\theta}_{j}$
can be written as 
\begin{equation}
\hat{\theta}_{j}=\hat{\theta}_{j}^{{\rm S}}+\frac{\sum_{t=1}^{n}r_{j,t-1}\hat{u}_{t}}{\sum_{t=1}^{n}r_{j,t-1}w_{j,t-1}},\label{eq:debiased-1}
\end{equation}
where $\hat{u}_{t}=y_{t}-W_{t-1}^{\top}\hat{\theta}^{{\rm S}}$ is
the Slasso residual. Though LASSO may shrink $\hat{\theta}_{j}^{{\rm S}}$
all the way to exactly zero, the second term in (\ref{eq:debiased-1})
is continuously distributed and therefore $\hat{\theta}_{j}$ is a
desparsified version of $\hat{\theta}_{j}^{{\rm S}}$. A straightforward
calculation yields 
\begin{equation}
\hat{\theta}_{j}-\theta_{j}^{*}=\frac{\sum_{t=1}^{n}r_{j,t-1}u_{t}}{\sum_{t=1}^{n}r_{j,t-1}w_{j,t-1}}+\frac{\sum_{t=1}^{n}r_{j,t-1}W_{-j,t-1}^{\top}(\theta_{-j}^{*}-\hat{\theta}_{-j}^{{\rm S}})}{\sum_{t=1}^{n}r_{j,t-1}w_{j,t-1}}={\rm N}_{j}-(\widehat{{\rm B}}_{j}-{\rm B}_{j}),\label{eq:N+AE}
\end{equation}
where $\widehat{{\rm B}}_{j}-{\rm B}_{j}$ is the approximation error
of the shrinkage bias. 

Let $\omega_{j}$ denote the standard deviation of ${\rm N}_{j}$.
To secure asymptotic normality for $\hat{\theta}_{j}$, we need an
appropriate score vector $r_{j}$ such that as $n\to\infty$:
\begin{center}
(\textbf{R1}) ${\rm N}_{j}/\omega_{j}\convd\mathcal{N}(0,1),$ and
(\textbf{R2}) $(\widehat{{\rm B}}_{j}-{\rm B}_{j})/\omega_{j}\stackrel{\mathrm{p}}{\to}0$.
\par\end{center}

\noindent These two results will furnish $\hat{\theta}_{j}$ with
asymptotic normality, validating inference based on the standard $t$-statistic. 

The result (R1) is well established by \citet{zhang2014confidence}
under i.i.d.~data when the score $r_{j,t}$ is taken as the LASSO
residual of regressing $w_{j,t}$ on all other regressors. However,
when $w_{j,t}$ is an LUR, this practice leads to a highly persistent
$r_{j,t}$ and ruins the asymptotic normality of ${\rm N}_{j}$. This
is the Stambaugh bias --- a highly persistent regressor produces
an asymptotically non-normally distributed OLS estimator skewed away
from zero. To safeguard (R1), we must seek a score that is less persistent
than an LUR. 

To retain (R2), we again examine the estimation error of the shrinkage
bias 
\begin{align}
{\rm \widehat{B}}_{j}-{\rm B}_{j} & =g_{j}^{\top}(\hat{\theta}_{-j}^{{\rm S}}-\theta_{-j}^{*}),\ \text{where}\ g_{j}=\dfrac{\sum_{t=1}^{n}W_{-j,t-1}r_{j,t-1}}{\sum_{t=1}^{n}w_{j,t-1}r_{j,t-1}}\cdot\label{eq:RB=000020demonstrate}
\end{align}
Note that $|{\rm \widehat{B}}_{j}-{\rm B}_{j}|\leq\|g_{j}\|_{\infty}\|\hat{\theta}_{-j}^{{\rm S}}-\theta_{-j}^{*}\|_{1}.$
Slasso's $L_{1}$ estimation error $\|\hat{\theta}_{-j}^{{\rm S}}-\theta_{-j}^{*}\|_{1}$
is invariant with the choice of score vector $r_{j,t}$, and diminishes
if $\hat{\theta}^{{\rm S}}$ is consistent. Thus, it suffices to control
the order of $\|g_{j}\|_{\infty}$ to achieve (R2). The explicit expression
of $g_{j}$ in (\ref{eq:RB=000020demonstrate}) suggests that a weak
correlation in sup-norm between $W_{-j,t}$ and the score $r_{j,t}$
(relative to the correlation between $w_{j,t}$ and $r_{j,t}$) helps.

With these routes in mind, we devise XDlasso in the following section.
It proceeds with two key steps: (i) conducting an IVX transformation
to get a new variable less persistent than the LUR regressors to eliminate
the Stambaugh bias; (ii) running an auxiliary LASSO regression to
construct an $r_{j,t}$ to remove the shrinkage bias. 

\subsection{IVX-Desparsified LASSO\protect\label{subsec:IVX-Desparsified-LASSO}}

For inference of $\theta_{j}^{*}$ in low dimensions, \citet{phillips2009econometric}'s
IVX method generates an instrument by quasi-differencing $w_{j,t}$.
When $w_{j,t}$ is an LUR, this self-generated instrument is mildly
integrated. The mitigation of persistence will remove the Stambaugh
bias when the sample size passes to infinity, and thus recovers asymptotic
normality of the IVX estimator. Due to the coexistence of Stambaugh
bias and shrinkage bias, the wisdom of IVX cannot be directly transplanted
into the high dimensional case. Instead, we integrate IVX with the
idea of desparsified LASSO after comprehending the mechanism of the
biases illustrated in Section \ref{subsec:The-Biases}.

Specifically, IVX adopts the following instrumental variable:
\begin{equation}
\zeta{}_{j,t}=\sum_{s=1}^{t}\rho_{\zeta}^{t-s}\Delta w_{j,s}\label{eq:IVX_original}
\end{equation}
where $\rho_{\zeta}\in(0,1)$ is a user-determined tuning parameter.
Define the s.d.~of the instrument as $\hat{\varsigma}_{j}=\sqrt{n^{-1}\sum_{t=0}^{n-1}\left(\zeta_{j,t}-\bar{\zeta}_{j}\right)^{2}.}$
We unify the scale for LURs and stationary regressors by standardizing
the instrument with its s.d.:
\begin{equation}
\tilde{\zeta}_{j,t}=\zeta_{j,t}/\hat{\varsigma}_{j}.\label{eq:IVX_std}
\end{equation}
For low dimensional predictive regressions, the IVX literature has
established  asymptotic normality of ${\rm N}_{j}$ in (\ref{eq:N+AE})
taking $r_{j,t}=\tilde{\zeta}_{j,t}$; see \citet{phillips2009econometric}.
When $w_{j,t}$ is an LUR, the IV $\zeta_{j,t}$ is \emph{mildly integrated
}and less persistent than an LUR. Furthermore, recall that the vector
of regressors $W_{-j,t}$ includes either LURs or stationary regressors.
Due to different degrees of persistence, the mildly integrated IV
$\tilde{\zeta}_{j,t}$ and the regressors $W_{-j,t}$ are asymptotically
uncorrelated. Thus, when $w_{j,t}$ is an LUR, we can choose the score
vector as $r_{j,t}=\tilde{\zeta}_{j,t}$ to deliver a small order
of $\|g_{j}\|_{\infty}$. If $w_{j,t}$ is stationary, however, this
score vector fails. When $n$ is large, the instrumental variable
$\zeta_{j,t}$ behaves similarly as the stationary regressor $w_{j,t}$,
and thus its correlation to the high dimensional stationary regressors
in $W_{-j,t}$ is not negligible. 

The above analysis implies that we must decorrelate the score vector
with the other regressors $W_{-j,t}$ to reduce the order of $g_{j}$
to control for the approximation error $\widehat{{\rm B}}_{j}-{\rm B}_{j}$.
This decorrelation will produce a unified testing approach for both
LURs and stationary regressors. To this end, we construct a residual
score vector $\hat{r}_{j}=(\hat{r}_{j,0},\cdots,\hat{r}_{j,n-1})^{\top}$
by the following auxiliary LASSO regression 
\begin{align}
\hat{r}_{j,t} & =\tilde{\zeta}_{j,t}-W_{-j,t}^{\top}\hat{\varphi}^{(j)},\ \text{where}\label{eq:def-rhat}\\
\hat{\varphi}^{(j)} & =\arg\min_{\varphi\in\mathbb{R}^{p-1}}\dfrac{1}{n}\sum_{t=1}^{n}(\tilde{\zeta}_{j,t-1}-W_{-j,t-1}^{\top}\varphi)^{2}+\mu_{j}\|D_{-j}\varphi\|_{1}\label{eq:def-lasso-proj}
\end{align}
with the LASSO tuning parameter $\mu_{j}$ and $D_{-j}={\rm diag}(\{\hat{\sigma}_{k}\}_{k\neq j})$.
In low dimensional multivariate predictive regressions, IVX transforms
\emph{each} regressor into a less persistent instrumental variable
parallel to (\ref{eq:IVX_original}), and constructs a two-stage least
squares estimator using \emph{all} the self-generated instrumental
variables. In contrast, we only transform the variable of interest
$w_{j,t}$ and estimate one auxiliary regression (\ref{eq:def-lasso-proj}).

The score vector $\hat{r}_{j}$ in (\ref{eq:def-rhat}) accommodates
stationary and LUR regressors. Recall that when $w_{j,t}$ is stationary,
the magnitude of the instrument $\zeta{}_{j,t}$ behaves similarly
as $w_{j,t}$. Thus, the score $\hat{r}_{j}$ is asymptotically equivalent
to the standardized residual of the LASSO regression of $w_{j,t}$
on $W_{-j,t}$. The latter is proportional to the score in \citet{zhang2014confidence}
for cross-sectional data, which is also used in \citet{adamek2022lasso}
for stationary time series. When $w_{j,t}$ is an LUR, the instrument
$\zeta_{j,t}$ is mildly integrated and has a different degree of
persistence from $W_{-j,t}$. Therefore, $\zeta_{j,t}$ is asymptotically
uncorrelated to each regressor in $W_{-j,t}$. For notational conciseness,
define $\tilde{W}_{-j,t}=(D_{-j})^{-1}W_{-j,t}$ and $\tilde{\varphi}^{(j)}=D_{-j}\hat{\varphi}^{(j)}$.
The analysis above suggests $\tilde{\varphi}^{(j)}\approx0$ so that
$\hat{r}_{j,t}=\tilde{\zeta}_{j,t}-\tilde{W}_{-j,t}^{\top}\tilde{\varphi}^{(j)}\approx\tilde{\zeta}_{j,t}$.
Recall from the discussion right after (\ref{eq:IVX_std}) that the
standardized instrument $\tilde{\zeta}_{j,t}$ is a valid score process
for an LUR regressor, and thus the score $\hat{r}_{j,t}$ can also
remove the biases asymptotically. As a result, the residual of the
auxiliary LASSO regression provides a unified construction of the
score for either stationary or nonstationary $w_{j,t}$. It allows
practitioners to conduct hypothesis testing on the coefficient in
high dimensional predictive regression regardless of the order of
integration of $w_{j,t}$. We maintain an agnostic attitude about
the persistence of the regressors; in practical implementation, we
need not distinguish these two types.
\begin{rem}[Standardized Instrumental Variable]
\label{rem:std_IV}The dependent variable of the auxiliary LASSO
regression (\ref{eq:def-lasso-proj}) is the standardized instrumental
variable (\ref{eq:IVX_std}). It is possible to use its original form
(\ref{eq:IVX_original}). However, the s.d.~of $\zeta_{j,t}$ passes
to infinity if $w_{j,t}$ is an LUR. Thus, without standardizing the
instrument, the theoretical order of the tuning parameter $\mu_{j}$
in (\ref{eq:def-lasso-proj}) for an LUR $w_{j,t}$ would be much
larger than the order for a stationary regressor. This difference
complicates the theoretical justifications. The standardization in
(\ref{eq:IVX_std}) unifies the convergence rate of $\mu_{j}$ regardless
of $w_{j,t}$ being an LUR or a stationary regressor. 
\end{rem}
\begin{rem}[Relation to DML]
The central idea of debiasing technique via orthogonalization in
high dimensional regressions is shared by the \citet{van2014asymptotically},
\citet{javanmard2014confidence} and \citet{chernozhukov2018}. DML
by \citet{chernozhukov2018} provides a more general framework by
allowing nonlinear and semiparametric models. Given cross sectional
data, DML advocates using sample splitting to eliminate a remainder
term in approximation to the normal distribution (denoted as $c^{*}$
in \citet[p. C4]{chernozhukov2018}), since with sample splitting
this remainder term could vanish asymptotically under much weaker
conditions. In other words, the quality of normal approximation using
the sample splitting technique is much better than without it. Justifying
sample splitting in time series is not as straightforward as that
in cross sectional data \citep[Remark 4]{beutner2021justification,adamek2022lasso}.
Without sample splitting, the theory of empirical processes with nonstationary
data will be challenging, and this topic deserves thorough investigations
in future research. 
\end{rem}
Following (\ref{eq:debiased-1}), XDlasso is constructed as 
\begin{equation}
\hat{\theta}_{j}^{{\rm XD}}=\hat{\theta}_{j}^{{\rm S}}+\frac{\sum_{t=1}^{n}\hat{r}_{j,t-1}\hat{u}_{t}}{\sum_{t=1}^{n}\hat{r}_{j,t-1}w_{j,t-1}}\label{eq:XDlasso}
\end{equation}
with the standard error
\begin{equation}
\hat{\omega}_{j}^{{\rm XD}}=\dfrac{\hat{\sigma}_{u}\sqrt{\sum_{t=1}^{n}\hat{r}_{j,t-1}^{2}}}{\left|\sum_{t=1}^{n}\hat{r}_{j,t-1}w_{j,t-1}\right|}\label{eq:s.e.}
\end{equation}
where $\hat{\sigma}_{u}^{2}=n^{-1}\sum_{t=1}^{n}\hat{u}_{t}^{2}$.
For simplicity, we focus on homoskedastic errors in the main text.
In Section B.3 of the appendix, we provide a formula of the heteroskedasticity-robust
standard error, and show that our method is robust to conditional
heteroskedasticity in both Monte Carlo simulations (Section B.3) and
empirical applications (Section C.2). 

With the point estimator (\ref{eq:XDlasso}) and the associated standard
error in (\ref{eq:s.e.}), we perform the $t$-test for the null hypothesis
$\mathbb{H}_{0}:\theta_{j}^{*}=\theta_{0,j}$. We summarize the testing
procedure using the XDlasso estimator below. 

\singlespacing%
\noindent\fbox{\begin{minipage}[t]{1\columnwidth - 2\fboxsep - 2\fboxrule}%
\begin{lyxalgorithm}[XDlasso Inference for $\mathbb{H}_{0}:\theta_{j}^{*}=\theta_{0,j}$]
$\mbox{}$\label{alg:XDlasso}
\end{lyxalgorithm}
\begin{description}
\item [{Step1}] Obtain $\hat{\theta}^{{\rm S}}$ from the Slasso regression
(\ref{eq:Slasso}). Save the residual $\hat{u}_{t}=y_{t}-W_{t-1}^{\top}\hat{\theta}^{{\rm S}}$
and $\hat{\sigma}_{u}^{2}=n^{-1}\sum_{t=1}^{n}\hat{u}_{t}^{2}$.
\item [{Step2}] Obtain the IV $\zeta_{j,t}$ by the transformation (\ref{eq:IVX_original}),
and standardize it by (\ref{eq:IVX_std}).
\item [{Step3}] Run the auxiliary LASSO regression (\ref{eq:def-lasso-proj}),
and save the residual $\hat{r}_{j,t}$ in (\ref{eq:def-rhat}).
\item [{Step4}] Compute the XDlasso estimator (\ref{eq:XDlasso}) and the
standard error (\ref{eq:s.e.}).
\item [{Step5}] Obtain the $t$-statistic 
\begin{equation}
t_{j}^{{\rm XD}}=(\hat{\theta}_{j}^{{\rm XD}}-\theta_{0,j})\big/\hat{\omega}_{j}^{{\rm XD}}.\label{eq:XD=000020t=000020stat}
\end{equation}
Reject $\mathbb{H}_{0}$ under the significance level $\alpha$ if
$|t_{j}^{{\rm XD}}|>\Phi_{1-\alpha/2}$, where $\Phi_{1-\alpha/2}$
is the $100(1-\alpha/2)$-th percentile of the standard normal distribution.
\end{description}
\end{minipage}}

\bigskip
\onehalfspacing

The testing procedure in Algorithm \ref{alg:XDlasso} refines the
conventional predictive regression inference using modern high dimensional
inference techniques. In the following, we elaborate on the necessity
of the IVX transformation by explaining the drawback of Dlasso under
high dimensional LURs.

\subsection{Necessity of IVX Transformation\protect\label{subsec:Necessity-of-IVX}}

Given i.i.d.~data, Dlasso proceeds with the following residual as
the score vector: 
\begin{align}
r_{j,t} & =\tilde{w}_{j,t}-W_{-j,t}^{\top}\hat{\psi}^{(j)}\text{ where }\text{\ensuremath{\tilde{w}_{j,t}=w_{j,t}/\hat{\sigma}_{j}},}\label{eq:def-rhat-1}\\
\hat{\psi}^{(j)} & =\arg\min_{\psi\in\mathbb{R}^{p-1}}\dfrac{1}{n}\sum_{t=1}^{n}(\tilde{w}_{j,t-1}-W_{-j,t-1}^{\top}\psi)^{2}+\mu_{j}\|D_{-j}\psi\|_{1}.\label{eq:def-lasso-proj-1}
\end{align}
where the dependent variable in (\ref{eq:def-rhat-1}) is also scaled
to be kept in line with (\ref{eq:def-rhat}) and (\ref{eq:def-lasso-proj}).
For simplicity, in this subsection we temporarily restrict all nonstationary
regressors to be pure unit roots ($\rho_{j}^{*}=1$). We show that
the standard Dlasso procedure fails to correct the bias in this special
case of LURs, thereby invalidating the inference in general cases. 

First, (\ref{eq:def-lasso-proj-1}) is a \emph{spurious regression}
as in \citet{granger1974spurious}. In a low dimensional regression
where all regressors have unit roots, the standardized least squares
counterpart of the regression (\ref{eq:def-lasso-proj-1}) follows
a functional central limit theorem (FCLT)
\begin{eqnarray}
\tilde{\psi}^{(j,{\rm OLS})} & = & D_{-j}\left(\sum_{t=1}^{n}W_{-j,t-1}W_{-j,t-1}^{\top}\right)^{-1}\sum_{t=1}^{n}W_{-j,t-1}\tilde{w}_{j,t-1}\nonumber \\
 & \stackrel{\mathrm{d}}{\to} & \tilde{\psi}^{*}:={\rm diag}(\{\sigma_{k}^{*}/\sigma_{j}^{*}\}_{k\neq j})\left(\int_{0}^{1}\mathcal{B}_{-j}\mathcal{B}_{-j}^{\top}\right)^{-1}\int_{0}^{1}\mathcal{B}_{-j}\mathcal{B}_{j},\label{eq:FCLT_spurious}
\end{eqnarray}
where $\mathcal{B}=(\mathcal{B}_{1},\dots,\mathcal{B}_{p})^{\top}$
is a $p$-dimensional Brownian motion, $\sigma_{j}^{*}=\sqrt{\int_{0}^{1}\mathcal{B}_{j}^{2}-(\int_{0}^{1}\mathcal{B}_{j})^{2}}$,
and $\mathcal{B}_{-j}=(\mathcal{B}_{k})_{k\neq j}$. In other words,
the OLS estimator converges \emph{in distribution }to a nondegenerate
random variable. This is the well-known spurious regression\emph{
}phenomenon in unit root regressions. While (\ref{eq:FCLT_spurious})
suggests that the limit target coefficients are random and continuously
distributed, the asymptotics for the LASSO estimator (\ref{eq:def-lasso-proj-1})
in high dimensions depends on the sparsity of the target coefficients.
The randomness of coefficients is contradictory to the sparsity required
in LASSO.

Second, the score vector in (\ref{eq:def-rhat-1}) cannot remove the
Stambaugh bias. Even in low dimensions, the least squares residual
$\hat{r}_{j,t}^{{\rm OLS}}=\tilde{w}_{j,t}-\tilde{W}_{-j,t}^{\top}\tilde{\psi}^{(j,{\rm OLS})}$remains
highly persistent due to spurious regression. Using $\hat{r}_{j,t}^{{\rm OLS}}$
as the score vector in low dimensions would therefore keep the Stambaugh
bias in the noise component ${\rm N}_{j}$. This issue in low dimension
is also present in the score vector $r_{j,t}$ construction in (\ref{eq:def-rhat-1})
by high dimensional LASSO. 

\begin{figure}[t]

\medskip{}

\centering{}\includegraphics[scale=0.5]{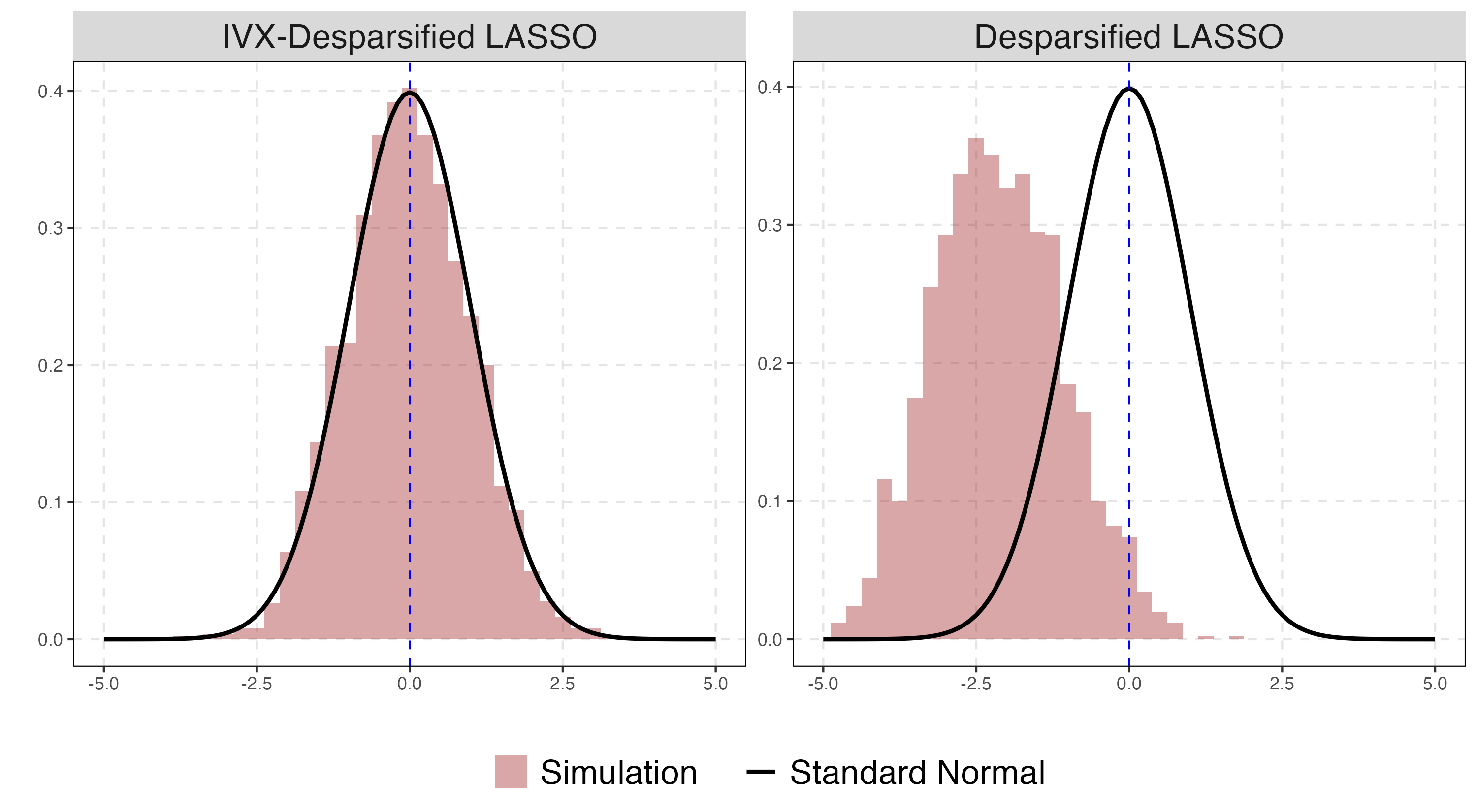}
\caption{Histograms of $t$-statistics from XDlasso and Dlasso\protect\label{fig:Comparison-xdlasso-zz}}
\end{figure}

In contrast, with the help of the IVX transformation, the score vector
in (\ref{eq:def-rhat}) achieves (R1) and (R2). Figure \ref{fig:Comparison-xdlasso-zz}
provides an illustrative simulation to compare Dlasso with XDlasso.
We set $\beta_{1}^{*}$ (the coefficient of the first unit root) as
zero, $n=300$, $(p_{x},p_{z})=(150,300)$, and use i.i.d.~innovations
(See Eq.~(\ref{eq:dgp_iid_inno})). We generate the data following
the DGP (\ref{eq:DGP}), with a mixture of stationary and nonstationary
regressors as in Section \ref{subsec:Setup}. Figure \ref{fig:Comparison-xdlasso-zz}
displays the histograms of the $t$-statistics over 2000 replications.
The density of XDlasso $t$-statistic is well approximated by $\mathcal{N}(0,1)$,
whereas the Dlasso $t$-statistic suffers from a substantial bias. 

\subsection{Joint Inference for Low-Dimensional Coefficients\protect\label{subsec:joint}}

We have devised the XDlasso inference for a scalar coefficient $\theta_{j}^{*}$.
In low dimensional predictive regression, the IVX estimator is applicable
to multiple coefficients, which jointly follows an asymptotic multivariate
normal distribution. Therefore, a Wald statistic by IVX is available
to jointly test the predictability of multiple regressors, provided
there are a finite number of them. This test statistic is shown to
be valid, even when the parameters of interest involve both stationary
and nonstationary regressors. 

The validity of Wald test extends to XDlasso in high dimensional predictive
regression. Specifically, suppose that we are interested in a subset
of regressors indexed by $\mathcal{J}\subset[p]$ with a fixed cardinality
$|\mathcal{J}|$. In this case, the XDlasso estimators $\hat{\theta}_{\mathcal{J}}^{{\rm XD}}=(\hat{\theta}_{j}^{{\rm XD}})_{j\in\mathcal{J}}$
defined in (\ref{eq:XDlasso}) asymptotically follow a multivariate
normal distribution. Let the null hypothesis $\mathbb{H}_{0}:\theta_{\mathcal{J}}^{*}=\theta_{0,\mathcal{J}}$
involve $|\mathcal{J}|$ restrictions, where $\theta_{0,\mathcal{J}}$
is a $|\mathcal{J}|$-dimensional vector. We construct the following
Wald statistic 
\begin{equation}
{\rm Wald}_{\mathcal{J}}^{{\rm XD}}=(\hat{\theta}_{\mathcal{J}}^{{\rm XD}}-\theta_{0,\mathcal{J}})^{\top}[\hat{\Omega}_{\mathcal{J}}^{{\rm XD}}]^{-1}(\hat{\theta}_{\mathcal{J}}^{{\rm XD}}-\theta_{0,\mathcal{J}}),\label{eq:Wald}
\end{equation}
where $\hat{\Omega}_{\mathcal{J}}^{{\rm XD}} = (\hat{\Omega}_{j,k}^{{\rm XD}})_{j,k\in\mathcal{J}}$ estimates the covariance
matrix of $\hat{\theta}_{\mathcal{J}}^{{\rm XD}}$, with
\[
\hat{\Omega}_{j,k}^{{\rm XD}}=\hat{\sigma}_{u}^{2}\dfrac{\sum_{t=1}^{n}\hat{r}_{j,t-1}\hat{r}_{k,t-1}}{\sum_{t=1}^{n}\hat{r}_{j,t-1}w_{j,t-1}\cdot\sum_{t=1}^{n}\hat{r}_{k,t-1}w_{k,t-1}}
\]
measuring the covariance between $\hat{\theta}_{j}^{{\rm XD}}$ and
$\hat{\theta}_{k}^{{\rm XD}}$. Under the null hypothesis, the Wald
statistic in (\ref{eq:Wald}) will follow an asymptotic $\chi^{2}$
distribution with the degree of freedom $|\mathcal{J}|$, enabling
joint inference on $\theta_{\mathcal{J}}^{*}$. 

\section{Asymptotic Theory \protect\label{sec:Asymptotic-Theory}}

This section develops the limit theory to shed light on the asymptotic
behaviors of the XDlasso estimator. Unsurprisingly, this paper's assumptions
share similarities with those in MS24. We state our theoretical assumptions
and then highlight and explain the differences between the assumptions
in these two papers. Regarding the asymptotic framework, we define
the number of regressors $p=p\left(n\right)$ and the sparsity index
$s=s\left(n\right)$ as deterministic functions of the sample size
$n$. In asymptotic statements, we will explicitly send $n\to\infty$,
and it is understood that $p(n)\to\infty$ as $n\to\infty$, while
$s\left(n\right)$ is allowed to be either fixed or divergent. Recall
that the stationary vector $e_{t}\in\mathbb{R}^{p_{x}}$ is the innovation
of the LUR regressors. We assume that the stationary high dimensional
vector $v_{t}=(e_{t}^{\top},Z_{t}^{\top})^{\top}$ is generated by
the innovations $\varepsilon_{t}=(\varepsilon_{k,t})_{k\in[p]}$ via
a linear transformation 
\begin{equation}
v_{t}=\Phi\varepsilon_{t},\label{eq:error_transform}
\end{equation}
where $\Phi$ is a $p\times p$ deterministic matrix. Let $\mathcal{F}_{t}$
denote the $\sigma$-field generated by $\{u_{s},\varepsilon_{s}\}_{s\leq t}$. 
\begin{assumption}
\label{assu:tail} Suppose that $u_{t}$ and $\varepsilon_{t}$ are
strictly stationary. Moreover, $u_{t}$ is a martingale difference
sequence (m.d.s.)~such that $\mathbb{E}\left(u_{t}|\mathcal{F}_{t-1}\right)=0$
and $\mathbb{E}(u_{t}^{2}|\mathcal{F}_{t-1})=\sigma_{u}^{2}>0$. There
exist absolute constants $C_{u}$, $b_{u}$, $C_{\varepsilon}$, and
$b_{\varepsilon}$ such that for all $t\in\mathbb{Z}$ and $a>0$,
\begin{align}
\Pr\left(|u_{t}|>a\right) & \leq C_{u}\exp(-a/b_{u}),\label{eq:innovtail}\\
\Pr\left(|\varepsilon_{k,t}|>a\right) & \leq C_{\varepsilon}\exp(-a/b_{\varepsilon}),\ \ \forall k\in[p].\label{eq:epstail}
\end{align}
 Furthermore, $\{\varepsilon_{k,t}\}_{t\in\mathbb{Z}}$ and $\{\varepsilon_{\ell,t}\}_{t\in\mathbb{Z}}$
are independent for all $k\neq\ell.$
\end{assumption}
\begin{rem}[Sub-Exponential Tails]
\label{rem:heavy-tail} In Assumption \ref{assu:tail}, Eqs.~(\ref{eq:innovtail})
and (\ref{eq:epstail}) impose the \emph{sub-exponential} tails for
the innovations, which includes the familiar \emph{sub-Gaussian} tail
as a special case. The sub-exponential condition is also imposed by
MS24 and it is needed in Section \ref{subsec:Consistency-of-Slasso}
to deduce the \emph{restricted eigenvalue }and \emph{deviation bound}
under high dimensional nonstationary data. The heavy-tail features
of financial data like extreme returns are not covered in the current
paper, and will be an important extension in future studies.
\end{rem}
The following Assumption \ref{assu:alpha} imposes restrictions on
the $\alpha$-mixing coefficients that characterize the time dependence
of the innovations $u_{t}$ and $\varepsilon_{t}$. For any two $\sigma$-fields
$\mathcal{A}$ and $\mathcal{B}$, define $\alpha(\mathcal{A},\mathcal{B})=\sup_{A\in\mathcal{A},B\in\mathcal{B}}|\Pr\left(A\cap B\right)-\Pr(A)\Pr(B)|$
and $\alpha(d)=\sup_{s\in\mathbb{Z}}\alpha(\sigma(\{u_{t},\varepsilon_{t}\}_{t\leq s}),\sigma(\{u_{t},\varepsilon_{t}\}_{t\geq s+d}))$. 
\begin{assumption}
\label{assu:alpha}There exist some absolute constants $C_{\alpha}$,
$c_{\alpha}$, $r$, $c_{\varepsilon}$ such that 
\begin{equation}
\alpha(d)\leq C_{\alpha}\exp\left(-c_{\alpha}d^{r}\right),\ \ \forall d\in\mathbb{Z},\label{eq:alpha_u}
\end{equation}
and the long-run variance $\mathbb{\mathbb{E}}\left[\sum_{d=-\infty}^{\infty}\varepsilon_{k,t}\varepsilon_{k,t-d}\right]\geq c_{\varepsilon}$
for all $k\in[p].$
\end{assumption}
The following Assumption \ref{assu:covMat} depicts the contemporary
correlation of $v_{t}$ defined as (\ref{eq:error_transform}), as
well as the constants in the local-to-unity AR coefficients specified
in (\ref{eq:LUR}). Define $\Omega=\Phi\Phi^{\top}$ where $\Phi$
has appeared in (\ref{eq:error_transform}). 
\begin{assumption}
\label{assu:covMat} There are absolute constants $\underline{c}$
and $\overline{C}$ such that: (a) $\underline{c}\leq\lambda_{\min}(\Omega)\leq\lambda_{\max}(\Omega)\leq\overline{C}$;
(b) $\max_{j\in[p]}\sum_{\ell=1}^{p}\left|\Phi_{j\ell}\right|\leq\overline{C}$;
(c) $\max_{j\in[p_{x}]}|c_{j}^{*}|\leq\overline{C}.$
\end{assumption}
We specify the user-determined parameter in (\ref{eq:IVX_original})
as
\begin{equation}
\rho_{\zeta}=1-C_{\zeta}/n^{\tau}\label{eq:rho=000020zeta}
\end{equation}
with absolute constants $C_{\zeta}>0$ and $\tau\in(0,1)$. The choice
of $\tau$ determines the persistence of the IV $\zeta_{j,t}$, which
will be elaborated in Remark \ref{rem:tau}. The following Assumption
\ref{assu:asym_n} characterizes the number of regressors $p$ and
the sparsity index $s$ relative to the sample size $n$. 
\begin{assumption}
\label{assu:asym_n} (a) $p=O(n^{\nu})$ for an arbitrary $\nu>0$
and (b) $s=O(n^{\frac{1}{4}(\tau\wedge(1-\tau))-\xi}\wedge p^{1-\xi})$
for an arbitrary small $\xi>0$.
\end{assumption}
In the following, we expound the differences between the assumptions
above and those in MS24. First, Assumption \ref{assu:tail} imposes
the m.d.s.~and conditional homoskedasticity conditions for the error
term $u_{t}$. Three remarks are in order to justify these two conditions.
\begin{rem}[m.d.s.~Condition]
\label{rem:mds=000020necessity}  Although the m.d.s.~assumption
is not required in MS24, in this paper it is essential for the asymptotic
normality of XDlasso. Without the m.d.s.~assumption, we would need
to use long-run covariances that not only complicate the procedures
but also rule out stationary regressors in the theory. See \citet[Remark 2.3]{phillips2016robust}
for detailed discussions. 
\end{rem}
\begin{rem}[Empirical Implications of the m.d.s.]
\label{rem:mds=000020reasonable} In empirical finance, the m.d.s.~assumption
is commonly imposed on the error term, especially when testing asset
return predictability \citep{zhu2014predictive,kostakis2015robust}.
It indicates that the dependent variable (financial asset return)
is not predictable if the null hypothesis of zero regression coefficients
is not rejected, which aligns with the Efficient Market Hypothesis.
In macroeconomic applications, the high dimensional covariates with
different degrees of persistence alleviate the concern of variable
omission, which makes the m.d.s.~assumption plausible.\footnote{Mild model misspecification with \emph{approximate sparsity }can be
accommodated by our framework; see \citet{belloni2012sparse} and
\citet[Remark 1]{mei2022lasso}.}
\end{rem}
\begin{rem}[Conditional Heteroskedasticity]
\label{rem:hetero}  It is possible to extend our methodology and
theoretical results to conditional heteroskedastic errors. Under low
dimensional predictive regressions, \citet[Theorem 1]{kostakis2015robust}
show that the homoskedastic-only standard error of the IVX estimator
is robust to conditional heteroskedastic error terms when the regressor
of interest is persistent. We conjecture that this result applies
to XDlasso, and the expression of our standard error (\ref{eq:s.e.})
is robust to conditional heteroskedasticity in our two empirical applications,
where each predictor of interest is persistent. Simulation results
in Tables B.5 and B.6
provide supportive evidence on the conjecture. In Appendix B.3,
we also consider a heteroskedasticity-robust standard error (B.5),
and verify its validity by simulations. A complete theory of conditional
heteroskedasticity in high dimensional predictive regression deserves
a standalone paper for future research. 
\end{rem}
Second, MS24 assume the innovations follow linear processes, and impose
the mixing condition and the lower bounded long run variances through
the coefficients in the linear process. In contrast, our Assumption
\ref{assu:alpha} does not assume any specific form of the linear
process for $u_{t}$ and $\varepsilon_{t}$, but directly imposes
the exponentially decaying rate for the mixing coefficient and the
lower bound of the long run variances. 

Third, Assumption \ref{assu:asym_n}(a) follows MS24 by allowing $p$
to diverge at a polynomial rate of $n$; it can be extended to an
exponential rate of $n$ at the cost of expositional complications.
Assumption \ref{assu:asym_n}(b) imposes a more restrictive condition
for the sparsity index $s$, compared to $s=o(n^{1/4})$ in MS24.
This is understandable as asymptotic normality is more delicate and
demanding than consistency. This condition ensures that the shrinkage
bias is accurately estimated to achieve the result (R2) in Section
\ref{subsec:The-Biases}, so that XDlasso asymptotically follows a
normal distribution centered at the true coefficient; also see the
second term of (\ref{eq:t=000020xd=000020decom}) below. 
\begin{rem}[Choice of $\tau$]
\label{rem:tau} In practice, the sparsity index $s$ is unknown.
We therefore recommend $\tau=1/2$ for practitioners, under which
the quantity $(\tau\wedge(1-\tau))/4$ achieves its maximum $1/8$
and thus permits the weakest sparsity condition. This is different
from the conventional wisdom of IVX \citep{phillips2016robust,kostakis2015robust}
where $\tau$ is recommended to be as large as 0.95 to minimize the
loss of rate efficiency (or local power). Their context is different
from our setting of high dimensional predictive regression, under
which an excessively large $\tau$ will permit a very small sparsity
index $s$ relative to $n$ and thus cause a severe distortion in
the size of the test under finite sample. As discussed in Theorem
\ref{thm:length} below, the hypothesis testing based on XDlasso is
consistent for a wide class of local-to-zero $\theta_{j}^{*}$, despite
a slower convergence rate compared with the case when $\tau$ is close
to 1 \citep{campbell2006efficient}.
\end{rem}
\newpage
\begin{rem}[Cointegration]
Our theoretical results do not cover cointegrated regressors. Theorem
4 in MS24 shows that in the presence of cointegration, Slasso over-penalizes
the coefficients of cointegrated regressors and shrinks them all the
way to zero, regardless of their true values. We are unaware of any
regularization method that achieves consistent estimation in $p>n$
regime with cointegrated variables mixed with LURs and stationary
ones. Without consistency, inferential theory on the cointegrated
variable's true parameter is beyond reach at this moment. Our Appendix
B.2 uses a numerical example to show that
under some specifications XDlasso can remain robust despite the presence
of cointegrated control variables. 
\end{rem}
Assumptions \ref{assu:tail}--\ref{assu:asym_n} will be sufficient
for the consistency of Slasso with high dimensional LURs and stationary
regressors. As the bias correction of XDlasso is mounted on the workhorse
estimator Slasso, the consistency of the latter is a prerequisite
for the ensuing maneuver. 

\subsection{Consistency of Slasso\protect\label{subsec:Consistency-of-Slasso}}

The leading case of persistent regressors is LUR in the low-dimensional
predictive regressions \citep{campbell2006efficient}. LUR includes
the unit root as a special case, and is thus more general in modeling
nonstationary behaviors. \citet{lee2022lasso} study the variable
selection properties by the adaptive LASSO under a finite number of
LUR regressors, and MS24 cover the consistency of Slasso under high
dimensional unit roots. In this paper, the first theoretical result
extends Slasso's consistency in the latter paper to incorporate the
LUR processes in the former one. This generalization calls for sophisticated
arguments, as briefed in Remark \ref{rem:LUR-tech} about the technical
challenges of LURs. 

The consistency of Slasso is founded on two building blocks. The first
one, which is essential and challenging, is the \emph{restricted eigenvalue
}(RE) condition of the Gram matrix of the standardized regressors.
For any $L>1$, the RE of any $p\times p$ matrix $\Sigma$ is defined
as 
\begin{equation}
\kappa_{H}(\Sigma,L,s):=\inf_{\delta\in\mathcal{R}(L,s)}\dfrac{\delta^{\top}H^{-1}\Sigma H^{-1}\delta}{\delta^{\top}\delta},\label{eq:def=000020RE}
\end{equation}
where $\mathcal{R}(L,s)=\{\delta\in\mathbb{R}^{p}\backslash\{0_{p}\}:\|\delta_{\mathcal{M}^{c}}\|_{1}\leq L\|\delta_{\mathcal{M}}\|_{1},\ \text{for all }|\mathcal{M}|\leq s\}.$
The generic matrix $H$ is a placeholder and varies in different contexts.
Let $\hat{\Sigma}=W^{\top}W/n$ be the sample Gram matrix of all regressors.
In the context of Slasso, we consider $\Sigma=\hat{\Sigma}$ and $H=D$
along with the scale standardization in (\ref{eq:Slasso}). The choice
of the constant $L$ is related to the procedures of technical proofs
and does not impact the rate of convergence. Following the common
practice \citep{bickel2009simultaneous}, we set $L=3$ as a convenient
choice, and simplify the notation as $\hat{\kappa}_{D}=\kappa_{D}(\hat{\Sigma},3,s)$.
The quantity $\hat{\kappa}_{D}$ will appear, according to Lemma 1
in MS24, in the denominator of Slasso's convergence rates. Therefore,
a lower bound for $\hat{\kappa}_{D}$ is essential for the consistency
of Slasso.

The second condition for Slasso's consistency is the \emph{deviation
bound} (DB) of the cross-product between the error term $u_{t}$ in
(\ref{eq:DGP}) and the standardized regressors. The theoretical order
of the tuning parameter $\lambda$ must be no smaller than that of
$\|n^{-1}\sum_{t=1}^{n}D^{-1}W_{t-1}u_{t}\|_{\infty}$ to avoid overfitting.
On the other hand, an excessively large $\lambda$ causes over shrinkage
and damages consistency. A tight upper bound of $\|n^{-1}\sum_{t=1}^{n}D^{-1}W_{t-1}u_{t}\|_{\infty}$
is therefore indispensable. 

Next, we establish the RE and DB conditions for high dimensional mixed
roots, and highlight their similarities to and differences from those
in MS24. We then leverage them to derive the convergence rates of
Slasso. 
\begin{lem}
\label{lem:std=000020RE=000020all} Under Assumptions \ref{assu:tail}--\ref{assu:asym_n},
there exists an absolute constant $c_{\kappa}$ such that\textcolor{red}{{}
}
\begin{equation}
\hat{\kappa}_{D}\geq\frac{c_{\kappa}}{s(\log p)^{4}}\cdot\label{eq:RE}
\end{equation}
 w.p.a.1 as $n\to\infty.$ In addition, there exists some absolute
constant $C_{{\rm DB}}$ sch that 
\begin{equation}
4\,\bigg\Vert\dfrac{1}{n}\sum_{t=1}^{n}D^{-1}W_{t-1}u_{t}\bigg\Vert_{\infty}\leq\dfrac{C_{{\rm DB}}(\log p)^{\frac{3}{2}+\frac{1}{2r}}}{\sqrt{n}}\label{eq:DB}
\end{equation}
 w.p.a.1 as $n\to\infty$, where $r$ is defined in Assumption \ref{assu:alpha}.
\end{lem}
Since LURs share similar asymptotic behavior with unit roots, the
orders of the RE (Eq.~(\ref{eq:RE})) and DB (Eq.~(\ref{eq:DB}))
are the same as those in Proposition 3 of MS24. Nevertheless, the
technical proofs for LURs are challenging.
\begin{rem}[Technical Challenges of LURs]
\label{rem:LUR-tech} For illustration, in this remark we suppose
that all regressors are LURs. In low dimensions, the Gram matrix of
LURs, after scaled by $n^{-1}$, converges \emph{in distribution }to
a non-degenerate stochastic integral 
\begin{equation}
n^{-1}\hat{\Sigma}\stackrel{{\rm d}}{\to}\int_{0}^{1}\mathcal{\mathcal{U}}_{{\rm {\bf C}}^{*}}(r)\mathcal{\mathcal{\mathcal{U}}}_{{\rm {\bf C}}^{*}}(r)^{\top}{\rm d}r,\label{eq:lim=000020OU=000020integral}
\end{equation}
where $\mathcal{\mathcal{\mathcal{U}}}_{{\rm {\bf {\rm {\bf C}}^{*}}}}(t):=\int_{0}^{t}{\rm e}^{{\rm {\bf {\rm {\bf C}}^{*}}}(t-r)}{\rm d}\mathcal{B}(r)$
is a vector of Ornstein--Uhlenbeck processes, with ${\rm {\bf C}}^{*}:={\rm diag}(c_{1}^{*},c_{2}^{*},\cdots,c_{p_{x}}^{*})$
storing the constants in the AR(1) coefficients of LURs in (\ref{eq:LUR}),
and $\mathcal{B}(r)$ being a multivariate Brownian motion. The diagonal
entries of the stochastic integral on the right-hand side of (\ref{eq:lim=000020OU=000020integral})
are nonnegative and continuously distributed, with a non-trivial probability
in a neighborhood of zero. Consequently, the minimum diagonal entry
of $n^{-1}\hat{\Sigma}$ diminishes to zero as the dimension of LUR
regressors passes to infinity. Eq.~(\ref{eq:RE}) establishes a lower
bound of RE that shrinks to zero in a sufficiently slow speed, thereby
still ensuring the consistency of Slasso. Under LUR, the linear coefficients
in $x_{j,t}=\sum_{\ell\geq0}(1+c_{j}^{*}/n)^{\ell}e_{j,t-\ell}$ depends
on $n$ and $\ell$, which is much more complicated to deal with than
the special case $c_{j}^{*}=0$ in MS24, where all the linear coefficients
become $1$ and the representation of $x_{j,t}$ becomes a simple
partial sum of stationary components. Interested readers may refer
to Lemma A.4 in Section
A.1.2 of the Appendix for details.
\end{rem}
With RE and DB, we formulate the consistency of the initial workhorse
Slasso estimator in the following lemma. 
\begin{lem}
\label{lem:SlassoError} Under Assumptions \ref{assu:tail}--\ref{assu:asym_n},
there exists some absolute constant $C_{{\rm m}}$ such that when
the tuning parameter in (\ref{eq:Slasso}) for the main regression
satisfies $\lambda=C_{{\rm m}}(\log p)^{\frac{3}{2}+\frac{1}{2r}}/\sqrt{n}$,
as $n\to\infty$, we have
\begin{align}
\|D(\hat{\theta}^{{\rm S}}-\theta^{*})\|_{1} & =O_{p}\left(\dfrac{s^{2}}{\sqrt{n}}(\log p)^{6+\frac{1}{2r}}\right).\label{eq:UnitL1-1}
\end{align}
\end{lem}
Lemma \ref{lem:SlassoError} provides the consistency of the standardized
coefficients, which is necessary to establish the asymptotic normality
of XDlasso. The rates of tuning parameters are technical devices for
proofs. In practice, we recommend cross-validation to select the tuning
parameter $\lambda$.

\subsection{The Auxiliary Regression\protect\label{subsec:The-Auxiliary-Regression} }

While Lemma \ref{lem:SlassoError} generalizes the convergence of
Slasso from unit roots to LUR regressors, brand new theory must be
developed for the auxiliary regression (\ref{eq:def-lasso-proj}),
again by the workhorse Slasso when $w_{j,t}$ is either LUR or stationary.
Importantly, as we maintain an agnostic attitude about the persistence
of the variable of interest, which is the dependent variable of the
auxiliary regression, the theory must be applicable in a unified manner
to accommodate two very different types of time series. Without loss
of generality, let $\mathcal{M}_{x}=[p_{x}]$ and $\mathcal{M}_{z}=[p]\backslash[p_{x}]$
denote the integer sets indexing the locations of LURs and stationary
regressors, respectively. 

We first examine the case when $j\in\mathcal{M}_{x}.$ The standardized
instrument $\tilde{\zeta}_{j,t}$ has a different degree of persistence
from $X_{-j,t}$ and $Z_{t}$, due to the different orders of integration.
In the low dimensional framework, the FCLT in \citet{phillips2016robust}
yields the following rate of convergence in OLS: 
\begin{equation}
D_{-j}\hat{\varphi}^{(j){\rm OLS}}=D_{-j}\left(\sum_{t=1}^{n}W_{-j,t}W_{-j,t}^{\top}\right)^{-1}\sum_{t=1}^{n}W_{-j,t}\tilde{\zeta}_{j,t}=O_{p}\left(1/\sqrt{n^{\tau\wedge(1-\tau)}}\right).\label{eq:LS_psi}
\end{equation}
In high dimensional models, the sample Gram matrix is rank-deficient
and the FCLT no longer works. Thanks to the $L_{1}$ penalization,
Slasso has a comparable local-to-zero order as (\ref{eq:LS_psi}),
which is shown in Proposition \ref{prop:DB-Aux} below. 

We then turn to $j\in\mathcal{M}_{z}$, and slightly abuse the notation
$Z_{-j,t}$ to denote the vector of stationary regressors excluding
$w_{j,t}$. Define 
\begin{equation}
\varphi_{0z}^{(j)*}:=\mathbb{E}\left(Z_{-j,t}Z_{-j,t}^{\top}\right)^{-1}\mathbb{E}\left(Z_{-j,t}w_{j,t}\right)\label{eq:phi=000020j=0000200=000020z}
\end{equation}
 as the linear projection of $w_{j,t}$ on $Z_{-j,t}$, and normalize
it by the standard deviation of the IV as 
\begin{equation}
\varphi_{z}^{(j)*}=\hat{\varsigma}_{j}^{-1}\varphi_{0z}^{(j)*}.\label{eq:phi=000020j=000020z}
\end{equation}
 The ``pseudo-true'' model for the LASSO regression (\ref{eq:def-lasso-proj})
is 
\begin{equation}
\tilde{\zeta}_{j,t}=X_{t}^{\top}0_{p_{x}}+Z_{-j,t}^{\top}\varphi_{z}^{(j)*}+\tilde{\eta}_{j,t},\label{eq:zeta=000020I0=000020pseudo=000020true}
\end{equation}
where $\text{\ensuremath{\tilde{\eta}}}_{j,t}=(\zeta_{j,t}-Z_{-j,t}^{\top}\varphi_{0z}^{(j)*})/\hat{\varsigma}_{j}$.
Note that when $w_{j,t}$ is stationary, the IV $\zeta_{j,t}$ is
close to $w_{j,t}$ under a large sample size, and thus asymptotically
uncorrelated with the LUR regressors $X_{t}$. As a result, the coefficients
associated with $X_{t}$ in the pseudo-true model (\ref{eq:zeta=000020I0=000020pseudo=000020true})
are zero. 

In addition, the error term $\tilde{\eta}_{j,t}$ is close to the
stationary time series $(w_{j,t}-Z_{-j,t}^{\top}\varphi_{0z}^{(j)*})/\hat{\varsigma}_{j}$
and thus asymptotically uncorrelated to the nonstationary regressors
$X_{t}$ due to different persistence. Furthermore, the coefficient
$\varphi_{0z}^{(j)*}$ satisfies $\mathbb{E}\left(Z_{-j,t}(w_{j,t}-Z_{-j,t}^{\top}\varphi_{0z}^{(j)*})\right)=0$.
Therefore, $\tilde{\eta}_{j,t}$ is also asymptotically uncorrelated
to the stationary regressors $Z_{-j,t}$, thereby ensuring the
consistency of the Slasso estimator $\hat{\varphi}^{(j)}$ in (\ref{eq:def-lasso-proj}).
We impose the following assumption on the coefficient $\varphi_{0z}^{(j)*}$. 
\begin{assumption}
\label{assu:phi=000020z}$\|\varphi_{0z}^{(j)*}\|_{0}\leq s$ with
$\varphi_{0z}^{(j)*}$ defined in (\ref{eq:phi=000020j=0000200=000020z})
and $s$ specified in Assumption \ref{assu:asym_n}. Moreover, $\|\varphi_{0z}^{(j)*}\|_{1}\leq C_{1}$
 for some absolute constant $C_{1}.$
\end{assumption}
To bound the LASSO estimation errors, we need sparsity of not only
the main regression (\ref{eq:DGP}), but also the auxiliary regression
(\ref{eq:zeta=000020I0=000020pseudo=000020true}). This is  similar
to the high dimensional sparse instrumental variable regression; see
\citet{zhu2018sparse}, \citet{gold2020inference}, and \citet{fan2025heteroscedasticity}. In Assumption
\ref{assu:phi=000020z}, we slightly abuse the sparsity index $s$
to bound the number of nonzero coefficients in the vector $\varphi_{0z}^{(j)*}$.
For simplicity, we directly impose this high-level sparsity assumption
on $\varphi_{0z}^{(j)*}$, which can be deduced under the commonly
used sparsity restriction on the precision matrix $\left(\mathbb{E}\left(Z_{t}Z_{t}^{\top}\right)\right)^{-1}$;
see the definition of $s_{j}$ and the conditions in Theorem 2.1 of
\citet[p.~759]{zhang2017simultaneous}. Finally, the upper bound of
the $L_{1}$-norm of $\varphi_{0z}^{(j)*}$ controls the variance
of the error term in the pseudo-true model (\ref{eq:zeta=000020I0=000020pseudo=000020true}).
Parallel conditions naturally hold for $j\in\mathcal{M}_{x}$ as the
pseudo-true coefficients are zero in view of the local-to-zero OLS
estimate displayed in (\ref{eq:LS_psi}).

The following proposition formally lays out the convergence rate of
the auxiliary estimator (\ref{eq:def-lasso-proj}). 
\begin{prop}
\label{prop:DB-Aux}Suppose that Assumptions \ref{assu:tail}--\ref{assu:phi=000020z}
hold. Then, there exists some absolute constant $C_{{\rm a},j}>0$
such that when $\mu_{j}=C_{{\rm a},j}(\log p)^{2+\frac{1}{2r}}/\sqrt{n^{(1-\tau)\wedge\tau}}$,
we have 
\begin{equation}
\|D_{-j}(\hat{\varphi}^{(j)}-\varphi^{(j)*})\|_{1}=O_{p}\left(\dfrac{s^{2}(\log p)^{6+\frac{1}{2r}}}{\sqrt{n^{\tau\wedge(1-\tau)}}}\right),\label{eq:local_L1}
\end{equation}
where $\varphi^{(j)*}={\bf 1}\{j\in\mathcal{M}_{x}\}\cdot0_{p-1}+{\bf 1}\{j\in\mathcal{M}_{z}\}\cdot(0_{p_{x}}^{\top},\varphi_{z}^{(j)*\top})^{\top}$. 
\end{prop}
\begin{rem}[Tuning Parameter $\mu_{j}$]
The rate of tuning parameter $\mu_{j}$ specified in Proposition
\ref{prop:DB-Aux} induces the following Karush--Kuhn--Tucker (KKT)
condition
\begin{equation}
\|n^{-1}\sum_{t=1}^{n}(D_{-j})^{-1}W_{-j,t}\hat{r}_{j,t}\|_{\infty}\leq\dfrac{\mu_{j}}{2}\leq\dfrac{C_{{\rm a},j}(\log p)^{2+\frac{1}{2r}}}{2\sqrt{n^{(1-\tau)\wedge\tau}}},\label{eq:KKT}
\end{equation}
which is sufficient to bound the approximation error $(\widehat{{\rm B}}_{j}-{\rm B}_{j})$
in (\ref{eq:N+AE}). Nevertheless, it is still necessary to establish
the consistency of $\hat{\varphi}^{(j)}$. This consistency result
not only helps us show the asymptotic normality of the $t$-statistic
$t_{j}^{{\rm XD}}$ defined in (\ref{eq:XD=000020t=000020stat}) to
guarantee the asymptotic size of the test, but is also critical for
the convergence rate of the standard error that governs the asymptotic
power.
\end{rem}
The consistency in the main equation and the auxiliary regression
has paved the way for statistical inference. We will analyze the asymptotic
size and power of the concerning test statistics in the next section.

\subsection{Asymptotic Distributions \protect\label{subsec:Asymptotic-Normality}}

The desirable rate of the auxiliary regression built in Proposition
\ref{prop:DB-Aux} guarantees the following decomposition of the $t$-statistic: 

\begin{equation}
\frac{\hat{\theta}_{j}^{{\rm XD}}-\theta_{j}^{*}}{\hat{\omega}_{j}^{{\rm XD}}}={\rm sgn}_{j}\cdot\frac{\sum_{t=1}^{n}\hat{r}_{j,t-1}u_{t}}{\sigma_{u}\sqrt{\sum_{t=1}^{n}\hat{r}_{j,t-1}^{2}}}+O_{p}\left(\dfrac{s^{2}(\log p)^{8+\frac{1}{r}}}{\sqrt{n^{(1-\tau)\wedge\tau}}}\right),\label{eq:t=000020xd=000020decom}
\end{equation}
where ${\rm sgn}_{j}=\frac{\left|\sum_{t=1}^{n}\hat{r}_{j,t-1}w_{j,t-1}\right|}{\sum_{t=1}^{n}\hat{r}_{j,t-1}w_{j,t-1}}$
is either 1 or $-1$ with probability one as $\sum_{t=1}^{n}\hat{r}_{j,t-1}w_{j,t-1}$
is continuously distributed. The first term of the above expression
is a counterpart of $\mathrm{N}_{j}/\omega_{j}$ in the discussion
of generic desparsifying argument in (\ref{eq:N+AE}). The second
term is the convergence rate of the approximation error of the shrinkage
bias analogous to $(\widehat{{\rm B}}_{j}-{\rm B}_{j})/\omega_{j}$.
Under Assumption \ref{assu:asym_n}, the second term on the right-hand
side of (\ref{eq:t=000020xd=000020decom}) is asymptotically negligible,
thereby yielding the following asymptotic normality for any $j\in[p]$. 
\begin{thm}
\label{thm:CLT} Suppose Assumptions \ref{assu:tail}-\ref{assu:phi=000020z}
hold. There exist absolute constants $C_{{\rm m}}$ and $C_{{\rm a},j}$
such that when $\lambda=C_{{\rm m}}(\log p)^{\frac{3}{2}+\frac{1}{2r}}/\sqrt{n}$
and $\mu_{j}=C_{{\rm a},j}(\log p)^{2+\frac{1}{2r}}/\sqrt{n^{(1-\tau)\wedge\tau}}$,
as $n\to\infty$ we have 
\begin{equation}
(\hat{\theta}_{j}^{{\rm XD}}-\theta_{j}^{*})\big/\hat{\omega}_{j}^{{\rm XD}}\stackrel{\mathrm{d}}{\to}\mathcal{N}(0,1).\label{eq:CLT}
\end{equation}
\end{thm}
The asymptotic normality in Theorem \ref{thm:CLT} is our main theoretical
result that delivers the valid asymptotic size of the hypothesis testing
for $\mathbb{H}_{0}:\theta_{j}^{*}=\theta_{0,j}$ using the $t$-statistic
$t_{j}^{{\rm XD}}$ in (\ref{eq:XD=000020t=000020stat}). The following
Theorem \ref{thm:length} provides the convergence rate of the estimated
standard error, which characterizes the asymptotic power of the test.
\begin{thm}
\label{thm:length} Conditions in Theorem \ref{thm:CLT} yield 
\[
\hat{\omega}_{j}^{{\rm XD}}=\boldsymbol{1}\left\{ j\in\mathcal{M}_{x}\right\} \cdot O_{p}(1/\sqrt{n^{1+\tau}})+\boldsymbol{1}\left\{ j\in\mathcal{M}_{z}\right\} \cdot O_{p}\left(1/\sqrt{n}\right).
\]
\end{thm}
The convergence rates displayed in Theorem \ref{thm:length} for the
two types of regressors are coherent with the results in low dimensional
predictive regressions. When $j\in\mathcal{M}_{x}$, the standard
error converges faster than the rate $1/n^{\delta_{j}}$ for any $\delta_{j}\in(0,(1+\tau)/2)$.
Thus, for the null hypothesis $\mathbb{H}_{0}:\beta_{j}^{*}=0$, the
hypothesis testing based on the $t$-statistic $t_{j}^{{\rm XD}}$
is consistent under the alternative with $\beta_{j}^{*}=c/n^{\delta_{j}}$
for some $c>0$ over this class of $\delta_{j}$. The range $\delta_{j}\in(0,(1+\tau)/2)$
includes the important $1/\sqrt{n}$ rate of coefficients for the
LUR regressor $X_{j,t}$, under which ${\rm var}(X_{j,t}\beta_{j}^{*})=O(1).$
The $1/\sqrt{n}$ factor thus balances the larger order of LUR regressors
$X_{t}$ and the standard $O_{p}(1)$ stochastic order of $y_{t}$
\citep{phillips2015halbert,lee2022lasso}, and the test by XDlasso
is consistent under this class of alternatives. When $j\in\mathcal{M}_{z}$,
XDlasso achieves the standard $\sqrt{n}$-consistency. 
\begin{rem}[Convergence Rate of XDlasso]
Unlike the convergence rates of Slasso in Lemma \ref{lem:SlassoError},
the convergence rate of XDlasso in Theorem \ref{thm:length} is independent
of the variable dimension $p$ and the sparsity index $s$. In view
of (\ref{eq:N+AE}), in XDlasso the order of the approximation error
of the shrinkage bias $(\widehat{{\rm B}}_{j}-{\rm B}_{j})$ depends
on $p$ and $s$, but it is dominated by the order of the noise component
${\rm N}_{j}$. The convergence rate of XDlasso therefore follows
the order of ${\rm N}_{j}$, which relates only to the sample size
$n$.
\end{rem}
Finally, when we are interested in a joint null hypothesis $\mathbb{H}_{0}:\theta_{\mathcal{J}}^{*}=\theta_{0,\mathcal{J}}$
for the coefficients indexed by $\mathcal{J}\subset[p]$ with a fixed
cardinality, we can use the Wald statistic in (\ref{eq:Wald}) for
this test. 
\begin{thm}
\label{thm:wald} Suppose that $|\mathcal{J}|$ is fixed as $n\to\infty$,
and the conditions in Theorem \ref{thm:CLT} hold for all $j\in\mathcal{J}$.
Under the null hypothesis $\mathbb{H}_{0}:\theta_{\mathcal{J}}^{*}=\theta_{0,\mathcal{J}}$,
we have 
\[
{\rm Wald}_{\mathcal{J}}^{{\rm XD}}\stackrel{\mathrm{d}}{\to}\chi_{|\mathcal{J}|}^{2}.
\]
\end{thm}
Theorem \ref{thm:wald} is a natural extension from the $t$-statistic
for a one-dimensional univariate hypothesis to a simultaneous multivariate
hypothesis. Notice that the Wald statistic is valid even if $\mathcal{J}$
includes a mixture of LUR regressors and stationary ones. It allows
the researcher to maintain the agnostic attitude when conducting the
joint hypothesis testing. \medskip

Let us summarize the insights gained from the theoretical development.
The fundamental principle behind the Dlasso method, as discussed in
\citet{zhang2014confidence}, is based on the Frisch-Waugh-Lovell
theorem. This theorem purges the influence of other control variables
to overcome shrinkage bias. In the context of persistent predictors,
the two-stage least squares approach developed in \citet{magdalinos2009limit}
is used to eliminate Stambaugh bias. Each piece of XDlasso is a machine
learning version of a classical idea. To adapt these existing procedures
to high dimensional predictive regressions, we must rely on the consistency
of Slasso, with its extension to LURs. This consistency is crucial
for both the main regression and the auxiliary regression. Moreover,
as mentioned in Remark \ref{rem:tau}, the construction of the IV
and the choice of $\tau$ must be tailored to balance the two types
of predictors. This approach is original in that such a restriction
is unique to high dimensional models and has not been studied before
even in conventional low dimensional settings.

\section{Monte Carlo Simulations\protect\label{sec:Simulation}}

In this section, we evaluate the performance of the proposed XDlasso
inference procedure by comparing its test size and power with those
of Dlasso. Although Dlasso and other existing LASSO bias-correction
procedures are not designed to handle persistent regressors, we include
this comparison to highlight the value added by the IVX transformation.
Furthermore, to demonstrate the robustness of the XDlasso approach,
we benchmark it using tests based on infeasible estimators, where
an oracle reveals the locations of the nonzero coefficients.

\subsection{Setup\protect\label{subsec:Setup}}

We consider the linear predictive regression model in (\ref{eq:DGP}).
The vector of the stationary components, denoted as $v_{t}=\left(u_{t},e_{t}^{\top},Z_{t}^{\top}\right)^{\top}$,
is generated by the two different processes: 
\begin{align}
\text{Case I (IID Innovations): } & v_{t}\sim\text{i.i.d. }\mathcal{N}\left(0,\Sigma\right),\label{eq:dgp_iid_inno}\\
\text{Case II (AR(1) Innovations): } & v_{t}=R_{n}v_{t-1}+\xi_{t},\ \ \xi_{t}\sim\text{i.i.d. }\mathcal{N}\left(0,\Sigma\right),\label{eq:dgp_ar1_inno}
\end{align}
where $R_{n}=\text{diag}(0,0.3,0.3,\ldots,0.3)$. Under this $R_{n}$
the error term $u_{t}$ in the main regression (\ref{eq:DGP}) remains
i.i.d., satisfying the m.d.s.~condition in Assumption \ref{assu:tail},
while the (local) unit root innovations $e_{t}$ and the stationary
regressors $Z_{t}$ are AR(1) processes. The covariance matrix $\Sigma=(\Sigma_{ij})_{i,j\in[p+1]}$
is specified as 
\[
\Sigma_{ij}=\begin{cases}
0, & \text{if }(i,j)\text{ is associated with }Z_{t}\text{ and }u_{t};\\
0.5^{|i-j|}, & \text{otherwise}.
\end{cases}
\]
The persistent regressors $X_{t}$ are generated by 
\begin{equation}
X_{t}=\text{diag}(\rho^{*})X_{t-1}+e_{t},\label{eq:dgp_lur}
\end{equation}
where $\rho^{*}=(1,1-1/n,1+1/n,1,1-1/n,1+1/n,\cdots)^{\top}\in\mathbb{R}^{p_{x}}$.
Recall $0_{p}$ is a $p$-dimensional zero vector and $1_{p}$ is
a $p$-dimensional vector of ones. The true coefficient vectors are:
\begin{equation}
\beta^{*}=\left(\beta_{1}^{*},\dfrac{0.5}{\sqrt{n}}\times1_{4}^{\top},0_{p_{x}-5}^{\top}\right)^{\top},\ \ \gamma^{*}=(\gamma_{1}^{*},0.5\times1_{2}^{\top},0.25\times1_{2}^{\top},0_{p_{z}-5}^{\top})^{\top}.\label{eq:true_coef}
\end{equation}
The specification involves four active LUR regressors and four active
stationary regressors.\footnote{In Section B.1, we consider a setting
with more nonzero coefficients and find robust performance of XDlasso.} The $1/\sqrt{n}$ scaling balances the regression by normalizing
the coefficients of the LUR regressors. We test the hypothesis
$\mathbb{H}_{0}:\beta_{1}^{*}=0$, $\mathbb{H}_{0}:\gamma_{1}^{*}=0$,
and the joint null hypothesis $\mathbb{H}_{0}:\beta_{1}^{*}=\gamma_{1}^{*}=0$,
respectively, under the sample sizes $n\in\{200,300,400,500,600\}$
and the dimensionality pairs $(p_{x},p_{z})\in\{(50,100),(100,150),(150,300)\}.$
We conduct 2000 replications in each setting.

We compare the finite sample performance of XDlasso, as described
in Algorithm \ref{alg:XDlasso}, and Dlasso with the score vector
in \eqref{eq:def-rhat-1} and \eqref{eq:def-lasso-proj-1}. In addition,
we consider two infeasible testing procedures as benchmarks. Using
the known active set of regressors, we conduct IVX inference (IVX
oracle) and the standard \textbf{$t$}-test based on the OLS estimator
(OLS oracle), employing only the regressor of interest and the active
regressors, which form a low-dimensional predictive regression model.
We set $C_{\zeta}=5$ and $\tau=0.5$ for the parameter $\rho_{\zeta}$
specified as (\ref{eq:rho=000020zeta}). As discussed in Remark \ref{rem:tau},
the choice $\tau=$ 0.5 admits the weakest sparsity condition, and
thus effectively improves the finite sample performance. 

Both XDlasso and Dlasso involve Slasso, where the selection of tuning
parameters $\lambda$ and $\mu$ affects finite sample performance.
In our experiments, we employ the \emph{block} 10-fold cross-validation
(CV), splitting the sample into 10 equally sized chronologically ordered
consecutive blocks for validations. Though the unconditional variances
of nonstationary regressors vary in different chronological blocks,
the standardization in the Slasso estimators like (\ref{eq:Slasso})
and (\ref{eq:def-lasso-proj}) account for such variation. This explains
the robustness of the block CV to nonstationary time series, as shown
by our simulation results.

In Theorem \ref{thm:CLT}, the tuning parameters are specified as
constants multiplied by the appropriate rates of convergence determined
by the sample size $n$, dimensionality $p$ and the mixing condition
constant $r$. As a benchmark, we also calibrate the tuning parameters
following \citet{lee2022lasso} to examine the validity of the theoretical
orders of tuning parameters specified in Theorem \ref{thm:CLT}. Specifically,
we perform $500$ pilot replications for each DGP, with $n_{0}=400$
, $\left(p_{x0},p_{z0}\right)=\left(100,150\right)$, and $p_{0}=p_{x0}+p_{z0}$.
In each replication $q=1,2,\dots,500$, we use the 10-fold cross-validation
to choose the tuning parameters $\lambda^{(q)}$ and $\mu^{(q)}$,
and calibrate the constants as 
\[
C_{{\rm m}}^{\left(q\right)}=\lambda^{(q)}n_{0}^{1/2}/(\log p_{0})^{\frac{3}{2}+\frac{1}{2r}},\ \ C_{{\rm a}}^{\left(q\right)}=\mu^{(q)}n_{0}^{[(1-\tau)\wedge\tau]/2}/(\log p_{0})^{2+\frac{1}{2r}},
\]
where $r=1$ and $\tau=0.5$ are chosen in the simulation. We then
fix $\widehat{C}_{\star}=\mathrm{median}(C_{\star}^{\left(1\right)},\ldots,C_{\star}^{\left(500\right)})$
for $\star\in\left\{ {\rm m},{\rm a}\right\} $ in the full-scale
experiments. The tuning parameters are then set as $\widehat{\lambda}=\widehat{C}_{{\rm m}}(\log p)^{\frac{3}{2}+\frac{1}{2r}}/\sqrt{n}$
and $\widehat{\mu}=\widehat{C}_{{\rm a}}(\log p)^{2+\frac{1}{2r}}/\sqrt{n^{(1-\tau)\wedge\tau}}$
as in Theorem \ref{thm:CLT}.

\subsection{Results}

We first investigate the empirical size of different testing methods
at a 5\% nominal significance level when the true coefficients $\beta_{1}^{*}=0$
and $\gamma_{1}^{*}=0$. Tables \ref{tab:mc_size_iid} and \ref{tab:mc_size_ar1}
report the empirical sizes under the IID and AR(1) innovations, respectively.
Foremost among the findings is that XDlasso effectively controls the
empirical size for both $\beta_{1}^{*}$, associated with a unit root
regressor, and $\gamma_{1}^{*}$, associated with a stationary regressor.
This performance stands in sharp contrast to that of Dlasso and OLS
oracle, which exhibit severe size distortions for $\beta_{1}^{*}$.
Such distortions can be attributed to the failure to account for the
Stambaugh bias arising from nonstationarity. Furthermore, the results
yield noteworthy insights regarding the tuning parameter selection.
The empirical size of XDlasso with both cross-validated and calibrated
tuning parameters is close to the nominal level. This result not only
validates the asymptotic rates of tuning parameters specified in Theorem
\ref{thm:CLT} but also supports CV as a feasible data-driven tuning
parameter selection method in practice. In addition, we investigate
the efficiency and robustness of XDlasso in comparison with alternatives.
When compared to the unbiased but infeasible ``IVX oracle'' estimator,
the confidence intervals produced by XDlasso are only slightly wider.
In contrast to the estimators solely for low-dimensional data, XDlasso
demonstrates robustness by accommodating high dimensional covariates
without compromising much in efficiency. Lastly, the empirical sizes
for testing the joint null hypothesis $\mathbb{H}_{0}:\beta_{1}^{*}=\gamma_{1}^{*}=0$,
as reported in Table \ref{tab:mc_size_joint}, are also well controlled
around the nominal level across setups. The results validate the theoretical
result in Theorem \ref{thm:wald} and is consistent with the findings
for testing $\mathbb{H}_{0}:\beta_{1}^{*}=0$ and $\mathbb{H}_{0}:\gamma_{1}^{*}=0$
individually. 
\begin{table}[hp]
\begin{centering}
\caption{Empirical size and length of confidence interval: IID innovations\protect\label{tab:mc_size_iid}}
{\small\begin{tabular}{c|cccc|cccc|cccc}
\hline 
\multirow{3}{*}{$n$ } & \multicolumn{4}{c|}{Oracle} & \multicolumn{4}{c|}{Calibrated} & \multicolumn{4}{c}{CV}\tabularnewline
 & \multicolumn{2}{c}{IVX Oracle} & \multicolumn{2}{c|}{OLS Oracle} & \multicolumn{2}{c}{XDlasso} & \multicolumn{2}{c|}{Dlasso} & \multicolumn{2}{c}{XDlasso} & \multicolumn{2}{c}{Dlasso}\tabularnewline
\cline{2-13}
 & Size  & Len.  & Size  & Len.  & Size  & Len.  & Size  & Len.  & Size  & Len.  & Size  & Len. \tabularnewline
\hline 
\multicolumn{13}{c}{$\mathbb{H}_{0}:\beta_{1}^{*}=0$ for nonstationary regressor}\tabularnewline
\hline 
\multicolumn{13}{c}{$\left(p_{x},p_{z}\right)=\left(50,100\right)$}\tabularnewline
\hline 
200  & 0.035  & 0.216  & 0.150  & 0.098  & 0.050  & 0.222  & 0.354  & 0.104  & 0.060  & 0.227  & 0.438  & 0.159 \tabularnewline
300  & 0.044  & 0.153  & 0.145  & 0.066  & 0.052  & 0.163  & 0.436  & 0.079  & 0.060  & 0.168  & 0.527  & 0.122 \tabularnewline
400  & 0.055  & 0.122  & 0.148  & 0.050  & 0.049  & 0.132  & 0.474  & 0.064  & 0.061  & 0.136  & 0.572  & 0.096 \tabularnewline
500  & 0.049  & 0.102  & 0.149  & 0.040  & 0.056  & 0.112  & 0.512  & 0.054  & 0.070  & 0.114  & 0.591  & 0.076 \tabularnewline
600  & 0.046  & 0.087  & 0.141  & 0.033  & 0.049  & 0.098  & 0.553  & 0.047  & 0.061  & 0.100  & 0.602  & 0.065 \tabularnewline
\hline 
\multicolumn{13}{c}{$\left(p_{x},p_{z}\right)=\left(100,150\right)$}\tabularnewline
\hline 
200  & 0.041  & 0.217  & 0.158  & 0.099  & 0.052  & 0.220  & 0.382  & 0.102  & 0.060  & 0.225  & 0.484  & 0.160 \tabularnewline
300  & 0.047  & 0.154  & 0.158  & 0.066  & 0.059  & 0.163  & 0.487  & 0.077  & 0.077  & 0.169  & 0.601  & 0.126 \tabularnewline
400  & 0.048  & 0.123  & 0.156  & 0.050  & 0.051  & 0.130  & 0.540  & 0.062  & 0.067  & 0.135  & 0.681  & 0.106 \tabularnewline
500  & 0.047  & 0.103  & 0.159  & 0.040  & 0.045  & 0.109  & 0.599  & 0.053  & 0.065  & 0.115  & 0.714  & 0.088 \tabularnewline
600  & 0.042  & 0.088  & 0.156  & 0.033  & 0.045  & 0.096  & 0.632  & 0.046  & 0.065  & 0.099  & 0.745  & 0.077 \tabularnewline
\hline 
\multicolumn{13}{c}{$\left(p_{x},p_{z}\right)=\left(150,300\right)$}\tabularnewline
\hline 
200  & 0.047  & 0.219  & 0.157  & 0.100  & 0.052  & 0.213  & 0.349  & 0.097  & 0.064  & 0.222  & 0.508  & 0.140 \tabularnewline
300  & 0.046  & 0.154  & 0.138  & 0.065  & 0.046  & 0.157  & 0.421  & 0.073  & 0.055  & 0.166  & 0.585  & 0.115 \tabularnewline
400  & 0.046  & 0.122  & 0.134  & 0.049  & 0.043  & 0.126  & 0.500  & 0.059  & 0.051  & 0.134  & 0.665  & 0.097 \tabularnewline
500  & 0.045  & 0.101  & 0.136  & 0.039  & 0.057  & 0.105  & 0.547  & 0.050  & 0.068  & 0.114  & 0.694  & 0.084 \tabularnewline
600  & 0.045  & 0.087  & 0.144  & 0.033  & 0.052  & 0.093  & 0.593  & 0.044  & 0.066  & 0.099  & 0.742  & 0.075 \tabularnewline
\hline 
\multicolumn{13}{c}{$\mathbb{H}_{0}:\gamma_{1}^{*}=0$ for stationary regressor}\tabularnewline
\hline 
\multicolumn{13}{c}{$\left(p_{x},p_{z}\right)=\left(50,100\right)$}\tabularnewline
\hline 
200  & 0.043  & 0.373  & 0.064  & 0.324  & 0.068  & 0.324  & 0.065  & 0.288  & 0.071  & 0.323  & 0.066  & 0.288 \tabularnewline
300  & 0.042  & 0.294  & 0.050  & 0.264  & 0.062  & 0.265  & 0.065  & 0.240  & 0.061  & 0.265  & 0.065  & 0.240 \tabularnewline
400  & 0.039  & 0.251  & 0.045  & 0.229  & 0.053  & 0.229  & 0.060  & 0.210  & 0.056  & 0.229  & 0.060  & 0.211 \tabularnewline
500  & 0.041  & 0.222  & 0.045  & 0.203  & 0.050  & 0.204  & 0.052  & 0.189  & 0.047  & 0.204  & 0.054  & 0.190 \tabularnewline
600  & 0.043  & 0.200  & 0.045  & 0.185  & 0.049  & 0.186  & 0.049  & 0.174  & 0.049  & 0.187  & 0.050  & 0.174 \tabularnewline
\hline 
\multicolumn{13}{c}{$\left(p_{x},p_{z}\right)=\left(100,150\right)$}\tabularnewline
\hline 
200  & 0.041  & 0.371  & 0.055  & 0.323  & 0.075  & 0.324  & 0.074  & 0.287  & 0.080  & 0.320  & 0.080  & 0.284 \tabularnewline
300  & 0.049  & 0.294  & 0.057  & 0.263  & 0.070  & 0.264  & 0.066  & 0.239  & 0.072  & 0.263  & 0.065  & 0.239 \tabularnewline
400  & 0.058  & 0.251  & 0.057  & 0.228  & 0.074  & 0.228  & 0.073  & 0.209  & 0.074  & 0.228  & 0.075  & 0.209 \tabularnewline
500  & 0.051  & 0.222  & 0.063  & 0.203  & 0.071  & 0.203  & 0.074  & 0.189  & 0.072  & 0.204  & 0.071  & 0.189 \tabularnewline
600  & 0.051  & 0.201  & 0.057  & 0.186  & 0.062  & 0.186  & 0.067  & 0.173  & 0.061  & 0.186  & 0.067  & 0.174 \tabularnewline
\hline 
\multicolumn{13}{c}{$\left(p_{x},p_{z}\right)=\left(150,300\right)$}\tabularnewline
\hline 
200  & 0.040  & 0.375  & 0.056  & 0.324  & 0.060  & 0.329  & 0.056  & 0.291  & 0.066  & 0.317  & 0.065  & 0.282 \tabularnewline
300  & 0.037  & 0.297  & 0.040  & 0.264  & 0.054  & 0.265  & 0.055  & 0.241  & 0.059  & 0.260  & 0.056  & 0.237 \tabularnewline
400  & 0.036  & 0.251  & 0.046  & 0.228  & 0.053  & 0.228  & 0.047  & 0.210  & 0.053  & 0.227  & 0.050  & 0.208 \tabularnewline
500  & 0.030  & 0.222  & 0.048  & 0.204  & 0.051  & 0.203  & 0.052  & 0.188  & 0.050  & 0.203  & 0.049  & 0.188 \tabularnewline
600  & 0.043  & 0.201  & 0.049  & 0.186  & 0.057  & 0.185  & 0.057  & 0.172  & 0.052  & 0.185  & 0.054  & 0.173 \tabularnewline
\hline 
\end{tabular}
}
\par\end{centering}
{

{\footnotesize\textit{Notes}}{\footnotesize : The data generating process
corresponds to \eqref{eq:dgp_iid_inno}. The upper and lower panels
report the empirical size of testing the null hypotheses $\mathbb{H}_{0}:\beta_{1}^{*}=0$
and $\mathbb{H}_{0}:\gamma_{1}^{*}=0$, respectively, at a 5\% nominal
significance level. ``Size'' is calculated as $R^{-1}\sum_{r=1}^{R}\mathbf{1}\left[|t^{\text{XD}(r)}|>\mathrm{\Phi}_{0.975}\right]$
across $R=2,000$ replications, where $t^{\text{XD}(r)}$ is computed
based on \eqref{eq:XD=000020t=000020stat} for the $r$-th replication,
and the critical value $\mathrm{\Phi}_{0.975}$ ($\approx1.96$) is
the 97.5-th percentile of the standard normal distribution. ``Len.''
refers to the median length of the 95\% confidence intervals across
replications. The IVX oracle and OLS oracle are infeasible estimators.
The ``Calibrated'' and ``CV'' columns refer to the methods used
for choosing the tuning parameters through calibration and cross-validation,
respectively.  }}{\footnotesize\par}
\end{table}
\begin{table}[hp]
\begin{centering}
\caption{Empirical size and length of confidence interval: AR(1) innovations\protect\label{tab:mc_size_ar1}}
{\small\begin{tabular}{c|cccc|cccc|cccc}
\hline 
\multirow{3}{*}{$n$ } & \multicolumn{4}{c|}{Oracle} & \multicolumn{4}{c|}{Calibrated } & \multicolumn{4}{c}{CV}\tabularnewline
 & \multicolumn{2}{c}{IVX Oracle} & \multicolumn{2}{c|}{OLS Oracle} & \multicolumn{2}{c}{XDlasso} & \multicolumn{2}{c|}{Dlasso} & \multicolumn{2}{c}{XDlasso} & \multicolumn{2}{c}{Dlasso}\tabularnewline
\cline{2-13}
 & Size  & Len.  & Size  & Len.  & Size  & Len.  & Size  & Len.  & Size  & Len.  & Size  & Len. \tabularnewline
\hline 
\multicolumn{13}{c}{$\mathbb{H}_{0}:\beta_{1}^{*}=0$ for nonstationary regressor}\tabularnewline
\hline 
\multicolumn{13}{c}{$\left(p_{x},p_{z}\right)=\left(50,100\right)$}\tabularnewline
\hline 
200  & 0.052  & 0.162  & 0.162  & 0.072  & 0.064  & 0.167  & 0.391  & 0.077  & 0.089  & 0.172  & 0.489  & 0.134 \tabularnewline
300  & 0.045  & 0.113  & 0.154  & 0.048  & 0.058  & 0.122  & 0.458  & 0.058  & 0.091  & 0.125  & 0.551  & 0.096 \tabularnewline
400  & 0.052  & 0.088  & 0.160  & 0.036  & 0.056  & 0.098  & 0.507  & 0.047  & 0.081  & 0.100  & 0.582  & 0.070 \tabularnewline
500  & 0.050  & 0.073  & 0.145  & 0.029  & 0.056  & 0.081  & 0.547  & 0.040  & 0.082  & 0.084  & 0.617  & 0.056 \tabularnewline
600  & 0.049  & 0.063  & 0.147  & 0.024  & 0.053  & 0.071  & 0.574  & 0.035  & 0.075  & 0.073  & 0.625  & 0.047 \tabularnewline
\hline 
\multicolumn{13}{c}{$\left(p_{x},p_{z}\right)=\left(100,150\right)$}\tabularnewline
\hline 
200  & 0.052  & 0.161  & 0.149  & 0.073  & 0.059  & 0.164  & 0.410  & 0.075  & 0.095  & 0.172  & 0.541  & 0.130 \tabularnewline
300  & 0.053  & 0.113  & 0.149  & 0.048  & 0.059  & 0.120  & 0.511  & 0.057  & 0.098  & 0.126  & 0.646  & 0.105 \tabularnewline
400  & 0.047  & 0.088  & 0.158  & 0.036  & 0.046  & 0.093  & 0.564  & 0.046  & 0.082  & 0.099  & 0.711  & 0.078 \tabularnewline
500  & 0.046  & 0.074  & 0.154  & 0.028  & 0.055  & 0.080  & 0.620  & 0.039  & 0.076  & 0.084  & 0.735  & 0.063 \tabularnewline
600  & 0.046  & 0.063  & 0.157  & 0.023  & 0.056  & 0.068  & 0.641  & 0.034  & 0.082  & 0.072  & 0.762  & 0.056 \tabularnewline
\hline 
\multicolumn{13}{c}{$\left(p_{x},p_{z}\right)=\left(150,300\right)$}\tabularnewline
\hline 
200  & 0.051  & 0.161  & 0.157  & 0.072  & 0.050  & 0.159  & 0.368  & 0.071  & 0.081  & 0.170  & 0.536  & 0.113 \tabularnewline
300  & 0.046  & 0.112  & 0.146  & 0.047  & 0.049  & 0.113  & 0.461  & 0.054  & 0.075  & 0.122  & 0.624  & 0.086 \tabularnewline
400  & 0.052  & 0.088  & 0.142  & 0.035  & 0.057  & 0.091  & 0.532  & 0.043  & 0.076  & 0.099  & 0.662  & 0.072 \tabularnewline
500  & 0.051  & 0.073  & 0.140  & 0.028  & 0.058  & 0.077  & 0.587  & 0.037  & 0.076  & 0.084  & 0.718  & 0.062 \tabularnewline
600  & 0.039  & 0.062  & 0.150  & 0.023  & 0.056  & 0.065  & 0.619  & 0.032  & 0.076  & 0.072  & 0.756  & 0.057 \tabularnewline
\hline 
\multicolumn{13}{c}{$\mathbb{H}_{0}:\gamma_{1}^{*}=0$ for stationary regressor}\tabularnewline
\hline 
\multicolumn{13}{c}{$\left(p_{x},p_{z}\right)=\left(50,100\right)$}\tabularnewline
\hline 
200  & 0.044  & 0.380  & 0.057  & 0.312  & 0.067  & 0.333  & 0.065  & 0.275  & 0.067  & 0.331  & 0.068  & 0.273 \tabularnewline
300  & 0.051  & 0.298  & 0.055  & 0.254  & 0.065  & 0.269  & 0.061  & 0.229  & 0.068  & 0.269  & 0.066  & 0.228 \tabularnewline
400  & 0.048  & 0.252  & 0.053  & 0.219  & 0.059  & 0.231  & 0.060  & 0.200  & 0.061  & 0.231  & 0.064  & 0.200 \tabularnewline
500  & 0.043  & 0.222  & 0.044  & 0.195  & 0.052  & 0.205  & 0.049  & 0.180  & 0.053  & 0.205  & 0.049  & 0.180 \tabularnewline
600  & 0.041  & 0.200  & 0.046  & 0.178  & 0.055  & 0.186  & 0.050  & 0.165  & 0.056  & 0.187  & 0.048  & 0.165 \tabularnewline
\hline 
\multicolumn{13}{c}{$\left(p_{x},p_{z}\right)=\left(100,150\right)$}\tabularnewline
\hline 
200  & 0.049  & 0.381  & 0.054  & 0.313  & 0.080  & 0.334  & 0.078  & 0.274  & 0.084  & 0.330  & 0.080  & 0.272 \tabularnewline
300  & 0.050  & 0.299  & 0.056  & 0.254  & 0.071  & 0.268  & 0.074  & 0.228  & 0.074  & 0.268  & 0.083  & 0.227 \tabularnewline
400  & 0.053  & 0.252  & 0.055  & 0.219  & 0.073  & 0.230  & 0.070  & 0.200  & 0.075  & 0.230  & 0.074  & 0.199 \tabularnewline
500  & 0.050  & 0.221  & 0.058  & 0.195  & 0.069  & 0.204  & 0.067  & 0.179  & 0.069  & 0.205  & 0.066  & 0.179 \tabularnewline
600  & 0.053  & 0.200  & 0.056  & 0.178  & 0.063  & 0.185  & 0.065  & 0.165  & 0.060  & 0.186  & 0.064  & 0.165 \tabularnewline
\hline 
\multicolumn{13}{c}{$\left(p_{x},p_{z}\right)=\left(150,300\right)$}\tabularnewline
\hline 
200  & 0.037  & 0.382  & 0.052  & 0.314  & 0.066  & 0.334  & 0.068  & 0.276  & 0.068  & 0.326  & 0.067  & 0.268 \tabularnewline
300  & 0.036  & 0.300  & 0.050  & 0.254  & 0.058  & 0.268  & 0.061  & 0.228  & 0.061  & 0.265  & 0.060  & 0.225 \tabularnewline
400  & 0.042  & 0.252  & 0.050  & 0.218  & 0.064  & 0.229  & 0.062  & 0.199  & 0.065  & 0.228  & 0.065  & 0.197 \tabularnewline
500  & 0.041  & 0.222  & 0.045  & 0.196  & 0.048  & 0.203  & 0.060  & 0.179  & 0.048  & 0.203  & 0.058  & 0.178 \tabularnewline
600  & 0.047  & 0.200  & 0.053  & 0.178  & 0.058  & 0.184  & 0.061  & 0.164  & 0.056  & 0.185  & 0.061  & 0.164 \tabularnewline
\hline 
\end{tabular}}
\par\end{centering}
{

{\footnotesize\textit{Notes}}{\footnotesize : The data generating process
corresponds to \eqref{eq:dgp_ar1_inno}. The upper and lower panels
report the empirical size of testing the null hypotheses $\mathbb{H}_{0}:\beta_{1}^{*}=0$
and $\mathbb{H}_{0}:\gamma_{1}^{*}=0$ at a 5\% nominal significance
level, respectively. ``Size'' is calculated as $R^{-1}\sum_{r=1}^{R}\mathbf{1}\left[|t^{\text{XD}(r)}|>\mathrm{\Phi}_{0.975}\right]$
across $R=2,000$ replications, where $t^{\text{XD}(r)}$ is computed
based on \eqref{eq:XD=000020t=000020stat} for the $r$-th replication,
and the critical value $\mathrm{\Phi}_{0.975}$ ($\approx$ 1.96)
is the 97.5-th percentile of the standard normal distribution. ``Len.''
refers to the median length of the 95\% confidence intervals across
replications. The IVX oracle and OLS oracle are infeasible estimators.
The ``Calibrated'' and ``CV'' columns refer to the methods used
for choosing the tuning parameters through calibration and cross-validation,
respectively. }}{\footnotesize\par}
\end{table}
\begin{table}[h]
\begin{centering}
\caption{Empirical size: Joint null hypothesis $\mathbb{H}_{0}:\beta_{1}^{*}=\gamma_{1}^{*}=0$\protect\label{tab:mc_size_joint}}
{\small\begin{tabular}{c|cc|cc|cc}
\hline
     \multicolumn{7}{c}{Case I: IID Innovations} \tabularnewline
     \hline
    & \multicolumn{2}{c|}{$\left(p_{x},p_{z}\right)=\left(50,100\right)$} & \multicolumn{2}{c|}{$\left(p_{x},p_{z}\right)=\left(100,150\right)$} & \multicolumn{2}{c}{$\left(p_{x},p_{z}\right)=\left(150,300\right)$} \tabularnewline
$n$ & Calibrated & CV & Calibrated & CV & Calibrated & CV \tabularnewline
\hline
200 & 0.064 & 0.077 & 0.067 & 0.076 & 0.053 & 0.069 \tabularnewline
300 & 0.057 & 0.066 & 0.064 & 0.081 & 0.052 & 0.063 \tabularnewline
400 & 0.058 & 0.069 & 0.067 & 0.080 & 0.047 & 0.056 \tabularnewline
500 & 0.054 & 0.064 & 0.061 & 0.069 & 0.053 & 0.063 \tabularnewline
600 & 0.051 & 0.055 & 0.060 & 0.069 & 0.056 & 0.059 \tabularnewline
\hline
     \multicolumn{7}{c}{Case II: AR(1) Innovations} \tabularnewline
     \hline
    & \multicolumn{2}{c|}{$\left(p_{x},p_{z}\right)=\left(50,100\right)$} & \multicolumn{2}{c|}{$\left(p_{x},p_{z}\right)=\left(100,150\right)$} & \multicolumn{2}{c}{$\left(p_{x},p_{z}\right)=\left(150,300\right)$} \tabularnewline
$n$ & Calibrated & CV & Calibrated & CV & Calibrated & CV \tabularnewline
\hline
200 & 0.070 & 0.090 & 0.064 & 0.098 & 0.061 & 0.084 \tabularnewline
300 & 0.070 & 0.095 & 0.065 & 0.096 & 0.055 & 0.079 \tabularnewline
400 & 0.056 & 0.075 & 0.063 & 0.090 & 0.054 & 0.070 \tabularnewline
500 & 0.059 & 0.079 & 0.064 & 0.086 & 0.055 & 0.067 \tabularnewline
600 & 0.053 & 0.071 & 0.067 & 0.079 & 0.060 & 0.076 \tabularnewline
\hline
\end{tabular}}
\par\end{centering}
\medskip{}

{ {\footnotesize\textit{Notes}}{\footnotesize : The data generating
processes for Case I and Case II correspond to \eqref{eq:dgp_iid_inno}
and \eqref{eq:dgp_ar1_inno}, respectively. The table reports the
empirical size of testing the joint null hypothesis $\mathbb{H}_{0}:\beta_{1}^{*}=\gamma_{1}^{*}=0$
at a 5\% nominal significance level. ``Size'' is calculated as $R^{-1}\sum_{r=1}^{R}\mathbf{1}\left[\text{Wald}_{\mathcal{J}}^{\text{XD}(r)}>\chi_{0.95,2}^{2}\right]$
across $R=2,000$ replications, where $\text{Wald}_{\mathcal{J}}^{\text{XD}(r)}$
is computed based on \eqref{eq:Wald} for the $r$-th replication
with $\mathcal{J}=\left\{ 1,p_{x}+1\right\} $, and the critical value
$\chi_{0.95,2}^{2}(\approx5.99)$ is the 95-th percentile of the chi-squared
distribution with 2 degrees of freedom. The ``Calibrated'' and ``CV''
columns refer to the methods used for choosing the tuning parameters
through calibration and cross-validation, respectively.}}
\end{table}

We now turn to the empirical power of the XDlasso inference. Figure
\ref{fig:power_all} plots power curves for the null hypotheses $\mathbb{H}_{0}:\beta_{1}^{*}=0$
and $\mathbb{H}_{0}:\gamma_{1}^{*}=0$ under varying true coefficient
values. In our analysis, we vary either $\beta_{1}^{*}$ or $\gamma_{1}^{*}$
from $0$ to $0.5$. All remaining coefficients are held fixed as
specified in \eqref{eq:true_coef}. Across various configurations,
XDlasso exhibits increasingly high power against the null hypothesis
as the true coefficient moves away from 0 and the sample size $n$
increases. Furthermore, Figure \ref{fig:power_compare_I0_I1} compares
the power for testing $\beta_{1}^{*}$ and $\gamma_{1}^{*}$, which
reveals that the power associated with a unit root regressor surpasses
that of a stationary regressor. This observation provides finite sample
evidence to support the theoretical results in Theorem \ref{thm:length}.
Specifically, it corroborates the faster convergence of standard error
for unit root regressors compared with stationary regressors, thereby
inducing higher power in the hypothesis testing. 
\begin{figure}[h]
\begin{centering}
\begin{subfigure}[b]{0.48\textwidth} \centering 
\includegraphics[height=0.36\textheight]{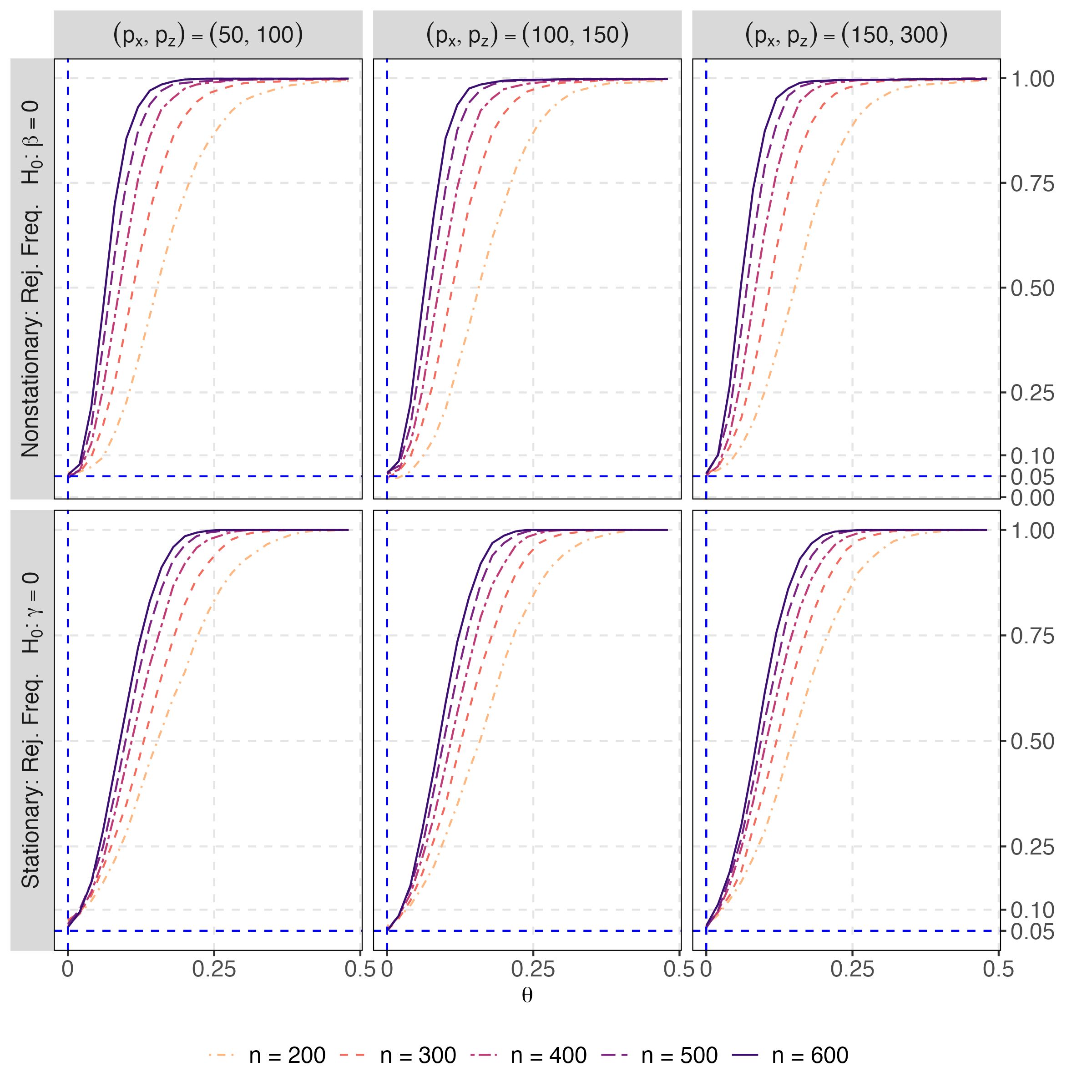} \label{fig:iid_power}
\caption{IID Innovations}
\end{subfigure} \begin{subfigure}[b]{0.48\textwidth} \centering
\includegraphics[height=0.36\textheight]{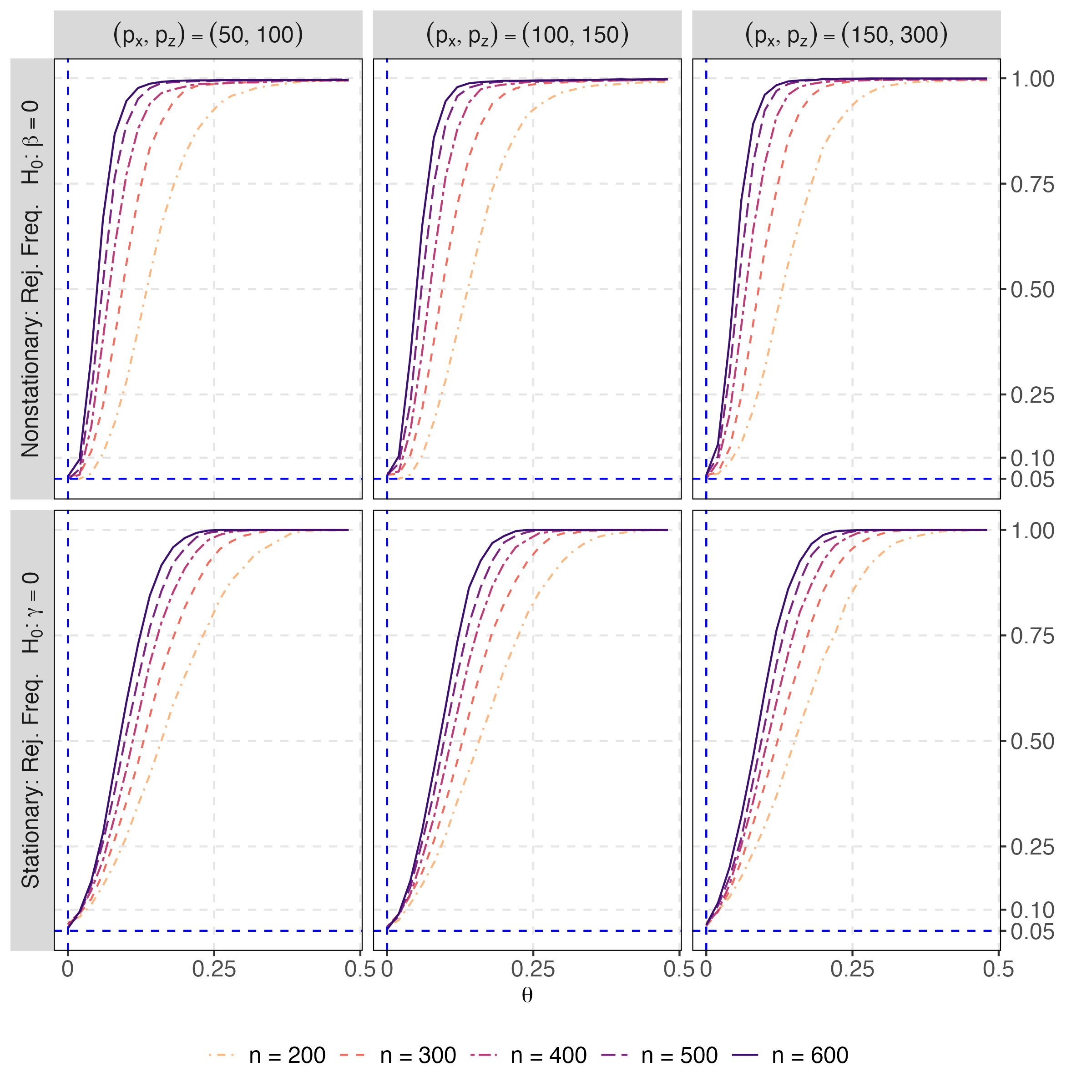} \label{fig:ar1_power}
\caption{AR(1) Innovations}
\end{subfigure}
\caption{Power curves of XDlasso inference\protect\label{fig:power_all}}
\par\end{centering}
\smallskip
{

{\footnotesize\textit{Notes}}{\footnotesize : The left and right panels
correspond to DGPs \eqref{eq:dgp_iid_inno} and \eqref{eq:dgp_ar1_inno},
respectively. In each subplot, the first row depicts the empirical
power function for $\beta_{1}^{*}$, associated with a nonstationary
regressor, across various $\left(p_{x},p_{z}\right)$ configurations,
while the second row pertains to $\gamma_{1}^{*}$, associated with
a stationary regressor. The empirical power is calculated as $R^{-1}\sum_{r=1}^{R}\mathbf{1}\left[|t^{\text{XD}(r)}|>\mathrm{\Phi}_{0.975}\right]$
across $R=2,000$ replications, where $t^{\text{XD}(r)}$ is computed
based on \eqref{eq:XD=000020t=000020stat} for the $r$-th replication,
and the critical value $\mathrm{\Phi}_{0.975}\left(\approx1.96\right)$
is the 97.5-th percentile of the standard normal distribution.}}{\footnotesize\par}
\end{figure}

\begin{figure}[h]
\begin{centering}
\includegraphics[height=0.4\textheight]{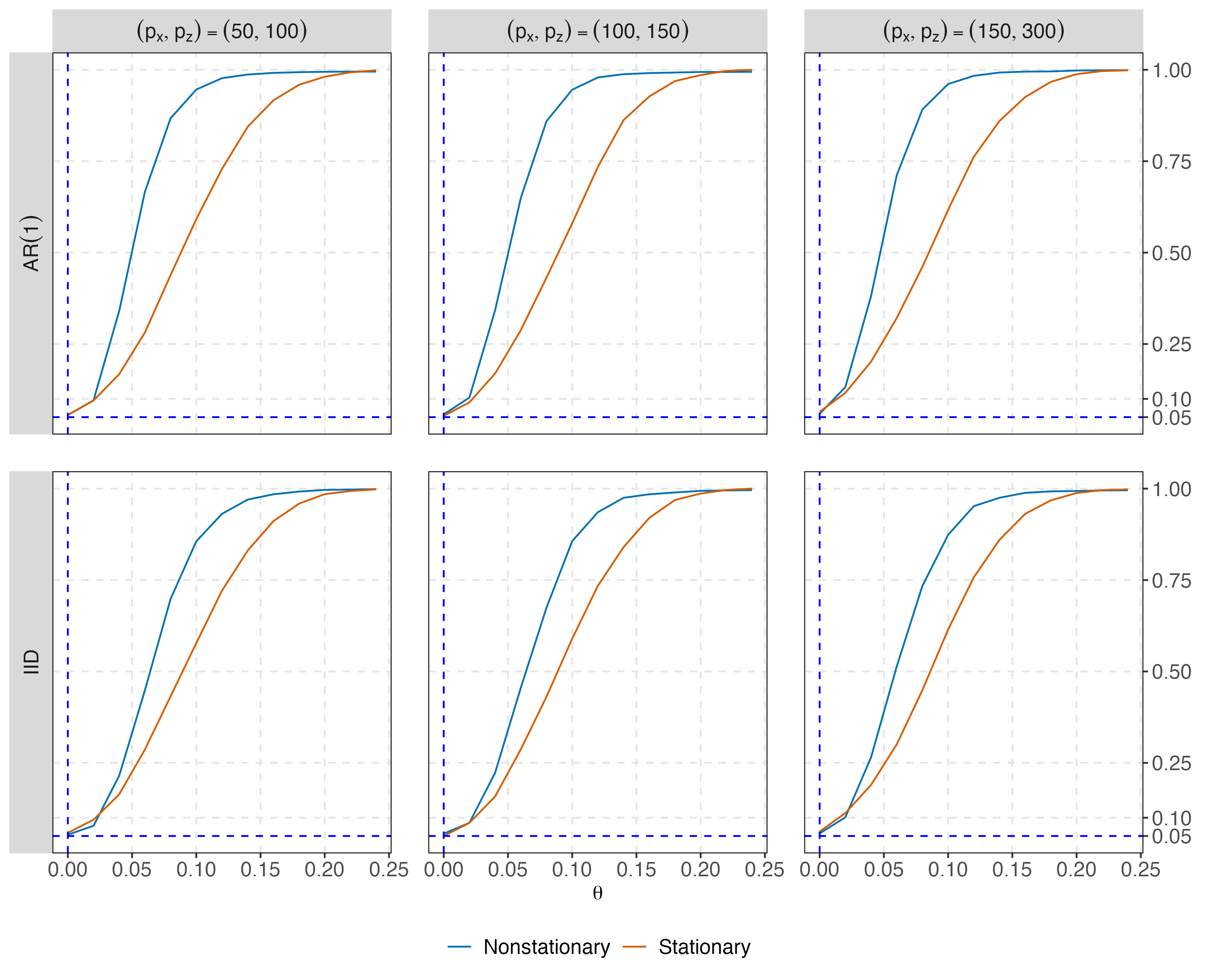}
\caption{Power of XDlasso inference for nonstationary and stationary regressors\protect\label{fig:power_compare_I0_I1}}
\par\end{centering}
\smallskip
{

{\footnotesize\textit{Notes}}{\footnotesize : This figure plots the
power curves under $n=600.$ The first and second rows correspond
to DGPs \eqref{eq:dgp_ar1_inno} and \eqref{eq:dgp_iid_inno}, respectively.
In each subplot, blue lines represent the empirical power function
for $\beta_{1}^{*}$, associated with a nonstationary regressor, while
red lines represent that for $\gamma_{1}^{*}$, associated with a
stationary regressor. Empirical power is calculated as $R^{-1}\sum_{r=1}^{R}\mathbf{1}\left[|t^{\text{XD}(r)}|>\mathrm{\Phi}_{0.975}\right]$
across $R=2,000$ replications, where $t^{\text{XD}(r)}$ is computed
based on \eqref{eq:XD=000020t=000020stat} for the $r$-th replication,
and the critical value $\mathrm{\Phi}_{0.975}\left(\approx1.96\right)$
is the 97.5-th percentile of the standard normal distribution.}}{\footnotesize\par}
\end{figure}

\section{Empirical Applications\protect\label{sec:Empirical-Application}}

In recent years, high dimensional macroeconomic data have been extensively
used to forecast key variables of interest; see \citet{smeekes2018macroeconomic},
\citet{medeiros2021forecasting}, and \citet{giannone2021economic},
for example. Researchers have primarily focused on the point estimation
of forecast. There has been limited empirical exploration of statistical
inference on the predictive power of specific predictors, due to a
lack of suitable toolkits. This section showcases two empirical applications
using our proposed XDlasso inference method. We utilize the monthly
data of 112 U.S. macroeconomic variables spanning from January 1960
to April 2025, sourced from the FRED-MD dataset by \citet{mccracken2016fred}. 

\subsection{Predictability of Stock Return Using Earnings-Price Ratio \protect\label{sec:emp-return-ep}}

In financial economics, testing the predictive power of valuation
ratios, particularly the \textit{earnings-price ratio}, on stock returns
has been subject to widespread discussions. Much of the literature
tests the predictability using univariate predictive regressions (e.g.,
\citet{welch2008comprehensive}, \citet{zhu2014predictive}, and \citet{goyal2024comprehensive}),
but inference can be sensitive to model misspecification arising from
omitted variables. Controlling for high dimensional covariates is
therefore necessary not only to enhance out-of-sample prediction but
also to mitigate the omitted variables for credible and accurate inference.
The literature on predictive regression for stock returns has focused
on identifying which variables possess significant predictive power
for future returns. For better out-of-sample prediction, recent literature
has documented gains in forecasting performance from incorporating
high-dimensional covariates into predictive regressions; see, for
instance, \citet{smeekes2018macroeconomic}, \citet{gu2020empirical}
and \citet{medeiros2021forecasting}. The increasing popularity of
statistical learning methods with high dimensional predictors calls
for suitable inference methods beyond univariate predictive regressions.
The goal is to provide investors with signals to adjust their portfolios
based on changes in these predictive variables. In this section, we
investigate the predictability of stock returns using the \emph{log
earnings-price ratio} in a data-rich environment with high dimensional
mixed-root control variables.

\subsubsection{Data}

Our analysis focuses on predicting the monthly return of the S\&P
500 index, calculated as $\text{Return}_{t}=\log(P_{t})-\log(P_{t-1}),$
where $P_{t}$ refers to \texttt{S\&P 500} (S\&P's Common Stock Price
Index: Composite). The primary predictor of interest is the log earnings-price
ratio. We obtain it from the variable \texttt{S\&P PE ratio} (S\&P's
Composite Common Stock: Price-to-Earnings Ratio), denoted as $\text{PE}_{t}$.
The log \emph{earnings-price ratio} is calculated by inverting the
original price-earnings ratio as $\text{logEP}_{t}=\log(1/\text{PE}_{t}).$

Figure \ref{fig:time_series_return_ep} displays the time series plot
of the monthly return of the S\&P 500 index and the log earnings-price
ratio from January 1960 to April 2025. The log earnings-price ratio
exhibits persistent patterns, while the S\&P 500 return appears stationary.
We further report the AR(1) coefficient estimates and Augmented Dickey--Fuller
(ADF) test $p$-values of both series under different sample periods
in Table \ref{tab:Persistence-of-return-ep}. The S\&P 500 return
is evidenced to be stationary under the full sample with an ADF test
$p$-value below 1\%, rejecting the null hypothesis of nonstationarity.
Conversely, the log earnings-price ratio shows high persistence with
an AR(1) coefficient estimate equal to $0.993$. Given the $p$-value
of $0.074$, nonstationarity is not rejected at the 5\% significance
level. Note that the nonstationary log earnings-price ratio can predict
the stationary monthly return of S\&P 500 in our model, since we allow
for a local-to-zero coefficient to balance the different scales between
a stationary outcome and a nonstationary regressor \citep{phillips2015halbert}.
Theorem \ref{thm:length} and the paragraph that follows illuminate
that inference by XDlasso has the power to detect a wide range of
local-to-zero alternatives. 

\begin{figure}[h!]
\centering{}\includegraphics[width=0.9\textwidth]{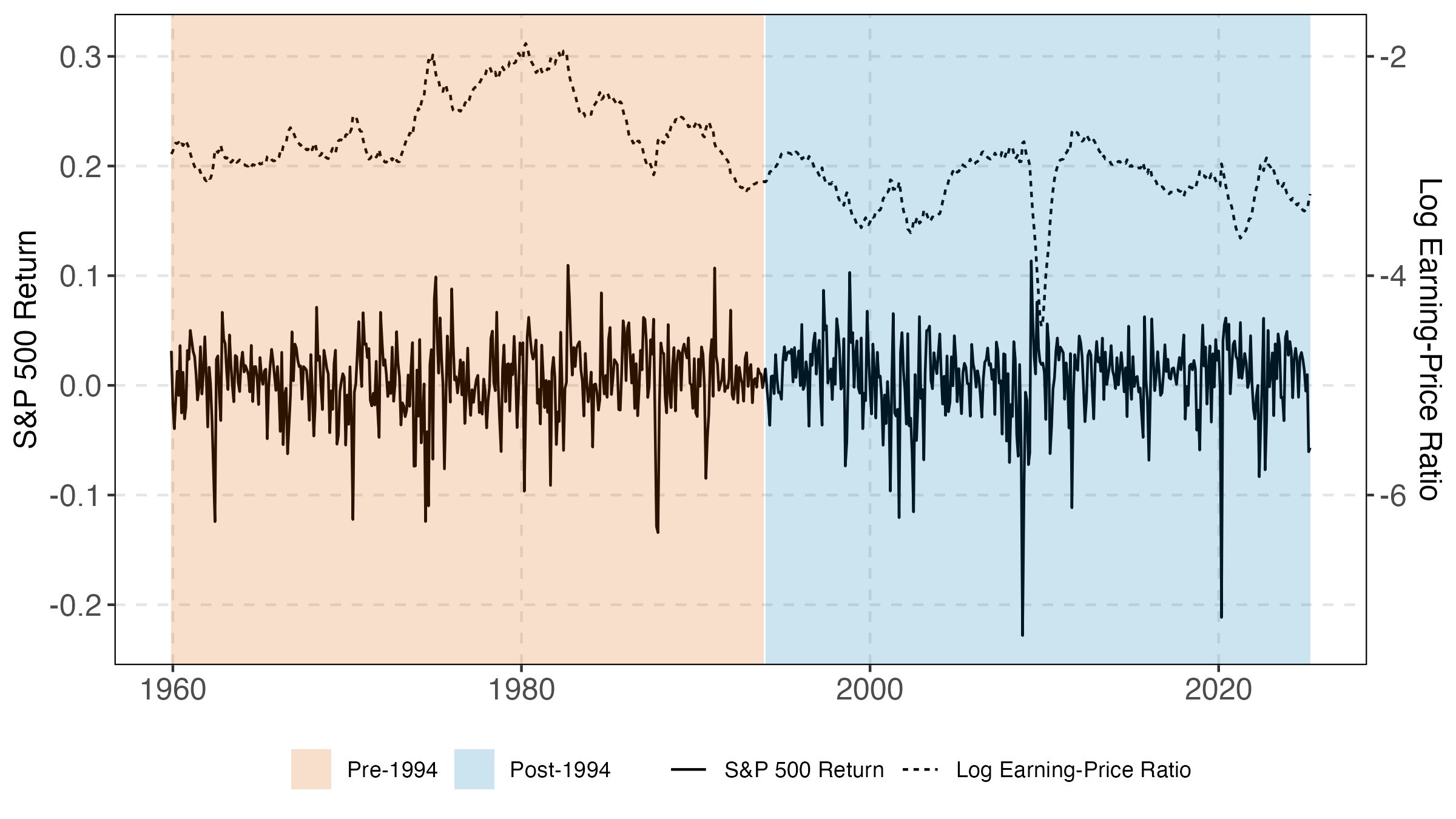}
\caption{S\&P 500 monthly Return and Log Earnings-price Ratio\protect\label{fig:time_series_return_ep}}
\end{figure}
 
\begin{table}[htbp]
\caption{Persistence of S\&P 500 Monthly Return and Log Earnings-Price Ratio\protect\label{tab:Persistence-of-return-ep}}

\begin{centering}
{\small{}%
\begin{tabular}{r|cc|cc}
\hline 
\multirow{1}{*}{{\small Sample Period}} & \multicolumn{2}{c|}{{\small S\&P 500 Monthly Return}} & \multicolumn{2}{c}{{\small Log Earnings-Price Ratio}}\tabularnewline
\hline 
 & {\small AR(1)} & {\small ADF $p$-value} & {\small AR(1)} & {\small ADF $p$-value}\tabularnewline
\hline 
{\small Full Sample }{\small\textit{(Jan. 1960 - Apr. 2025)}} & \multirow{1}{*}{{\small 0.227}} & \multirow{1}{*}{{\small$<$0.01}} & \multirow{1}{*}{{\small 0.993}} & \multirow{1}{*}{{\small 0.074}}\tabularnewline
\hline 
{\small Pre-1994 }{\small\textit{(Jan. 1960 - Dec. 1993)}} & \multirow{1}{*}{{\small 0.255}} & \multirow{1}{*}{{\small$<$0.01}} & \multirow{1}{*}{{\small 0.994}} & \multirow{1}{*}{{\small 0.776}}\tabularnewline
\hline 
{\small Post-1994}{\small\textit{ (Jan. 1994 - Apr. 2025)}} & \multirow{1}{*}{{\small 0.200}} & \multirow{1}{*}{{\small$<$0.01}} & \multirow{1}{*}{{\small 0.979}} & \multirow{1}{*}{{\small$<$0.01}}\tabularnewline
\hline 
\end{tabular}}{\small\par}
\par\end{centering}
{\footnotesize\textit{{}Notes:}}{\footnotesize{} The lag order for
the ADF test is set to $\lfloor n^{1/3}\rfloor$ where $n$ is the
effective sample size. The exact start and end dates of subperiods
are provided in the first column of the table.}{\footnotesize\par}
\end{table}

Return predictability can be time-varying and sporadic \citep{tu2023penetrating}.
\citet{campbell2006efficient} discover that the return predictability
test results can vary depending on the inclusion of the samples after
1994, and the break is revisited by \citet{zhu2014predictive}. Following
this empirical finding, we divide the sample into two periods: pre-1994
(January 1960 - December 1993) and post-1994 (January 1994 - April
2025). The S\&P 500 return demonstrates stationarity across both subperiods
with ADF $p$-values below 1\%. In contrast, the persistence of the
log earnings-price ratio depends on the sample period considered.
Nonstationarity is evident in the Pre-1994 period with a large ADF
test $p$-value $0.776$, while the null hypothesis of nonstationarity
is rejected at the 1\% significance level in the Post-1994 period.
The ambiguity in stationarity of the log earnings-price ratio motivates
the use of the XDlasso procedure, a unified approach for both stationary
and nonstationary regressors without prior knowledge of their persistence.

Besides the monthly return of the S\&P 500 index and the log earnings-price
ratio, our analysis incorporates high dimensional covariates including
all other 110 macroeconomic variables from the FRED-MD dataset. These
variables comprise a mixture of stationary and nonstationary time
series. In practice, it is common for empirical analysts to transform
potentially nonstationary time series into stationary ones to avoid
challenges arising from nonstationarity. To facilitate such stationarization,
FRED assigns a transformation code (TCODE) to each variable denoting
the recommended transformation.

In our analysis, we follow the common practice of using the TCODE
to transform the 110 time series, and subsequently use the transformed
variables as covariates. Nevertheless, we highlight that these elementary
transformations are not a silver bullet in taming nonstationarity.
We perform the ADF test for each of the transformed variables for
different sample periods. A nontrivial proportion (9.1\%) of the transformed
variables still demonstrate nonstationarity based on the ADF test
at the 5\% significance level for both Pre-1994 and Post-1994 subperiods.
The high persistence of the log earnings-price ratio and other covariates
suggests the necessity of XDlasso.

As highlighted by \citet{smeekes2020unit}, the predictive performance
in regressions using FRED-MD data is sensitive to the transformations.
To assess the robustness of our results, we also conduct our analysis
using the original (untransformed) FRED-MD time series as covariates.
To reduce the impact of highly nonstationary series, we exclude I(2)
variables that require second differencing for stationarity according
to their TCODE classification.\footnote{Appendix C.1 further investigates robustness by considering:
(i) excluding nonstationary variables based on integrated orders determined
by the bootstrap sequential testing procedure of \citet{smeekes2015bootstrap},
as reported in \citet{smeekes2020unit}, and (ii) applying only logarithmic
transformations, as indicated by TCODE, without differencing.}

\subsubsection{Results}

We study the one-month-ahead regression $\text{Return}_{t}=\alpha^{*}+\theta_{1}^{*}\times\text{logEP}_{t-1}+W_{-1,t-1}^{\top}\theta_{-1}^{^{*}}+u_{t},$
where $W_{-1,t-1}$ denotes a high dimensional vector that collects
all control variables. We carry out hypothesis testing on the key
parameter of interest $\theta_{1}^{*}$ that measures the predictive
power of the log earnings-price ratio in forecasting the S\&P 500
monthly return. Given the stationary pattern of the dependent variable,
we additionally consider specifications where $W_{-1,t-1}$ includes
the lagged dependent variable $\text{Return}_{t-1}$. For each model,
we apply the wild bootstrapped automatic variance ratio test \citep{kim2009automatic}
to the Slasso residuals as a diagnostic test for the martingale difference
sequence (m.d.s.)~condition in Assumption \ref{assu:tail}.\footnote{This practice serves as a heuristic diagnostic. Demonstrated by our
simulation results in Appendix B.4, the variance
ratio test on the Slasso residuals $\hat{u}_{t}$ tends to severely
over-reject the m.d.s.~condition for $u_{t}$. To the best of our
knowledge, at present the literature has no valid testing procedure
yet for m.d.s.~in high dimensional predictive regressions. It is
an open question for future research.}

In addition to XDlasso and Dlasso, we conduct the IVX inference in
a simple regression setup, omitting all control variables, to demonstrate
the necessity of incorporating high dimensional control variables
in practice. We compare the results of XDlasso to those of Dlasso
for the high dimensional regression and IVX for a simple regression
$\text{Return}_{t}=\alpha^{*}+\theta_{1}^{*}\times\text{logEP}_{t-1}+u_{t}.$

\begin{table}
\begin{centering}
\caption{Test $\mathbb{H}_{0}:\theta_{1}^{*}=0$ across sample periods and
specifications in stock return prediction\protect\label{tab:app_est_return_ep}}
\par\end{centering}
{\small 
\begin{subtable}[t]{1\textwidth}
\centering
\caption{\texttt{TCODE} Transformed Data}

\begin{tabular}{c|c|ccc|ccc}
  \hline
  & & \multicolumn{3}{c|}{Without \(\text{Return} _{t-1}\) } & \multicolumn{3}{c}{Include \(\text{Return} _{t-1}\)}   \tabularnewline 
  \hline
Sample Period & IVX & Dlasso & XDlasso & VR Test  & Dlasso & XDlasso & VR Test \tabularnewline 
  \hline
Full Sample & -0.023** & 0.009 & 0.003 & \multirow{2}{*}{0.074} & 0.009 & 0.005 & \multirow{2}{*}{0.216} \tabularnewline 
 {\footnotesize \textit{(Jan. 1960 - Apr. 2025)}} & \textit{(0.010)} & \textit{(0.006)} & \textit{(0.015)} & \multirow{2}{*}{} & \textit{(0.006)} & \textit{(0.015)} & \multirow{2}{*}{} \tabularnewline 
 \hline
Pre-1994 & -0.059** & 0.025** & 0.059* & \multirow{2}{*}{0.227} & 0.024** & 0.062* & \multirow{2}{*}{0.296} \tabularnewline 
 {\footnotesize \textit{(Jan. 1960 - Dec. 1993)}} & \textit{(0.026)} & \textit{(0.010)} & \textit{(0.033)} & \multirow{2}{*}{} & \textit{(0.010)} & \textit{(0.032)} & \multirow{2}{*}{} \tabularnewline 
Post-1994 & -0.022 & 0.002 & -0.001 & \multirow{2}{*}{0.053} & 0.002 & -0.001 & \multirow{2}{*}{0.049} \tabularnewline 
 {\footnotesize \textit{(Jan. 1994 - Apr. 2025)}} & \textit{(0.015)} & \textit{(0.007)} & \textit{(0.017)} & \multirow{2}{*}{} & \textit{(0.007)} & \textit{(0.017)} & \multirow{2}{*}{} \tabularnewline 
  \hline
\end{tabular}

\end{subtable}
\hfill
\medskip

\begin{subtable}[t]{1\textwidth}
\centering
\caption{Untransformed Data: Excluding I(2) Variables Based on \texttt{TCODE}}

\begin{tabular}{c|c|ccc|ccc}
    \hline
    & & \multicolumn{3}{c|}{Without \(\text{Return} _{t-1}\)} & \multicolumn{3}{c}{Include \(\text{Return} _{t-1}\)}   \tabularnewline 
    \hline
  Sample Period & IVX & Dlasso & XDlasso & VR Test  & Dlasso & XDlasso & VR Test \tabularnewline 
    \hline
  Full Sample & -0.023** & 0.013 & -0.008 & \multirow{2}{*}{0.001} & 0.019 & 0.012 & \multirow{2}{*}{0.811} \tabularnewline 
   {\footnotesize \textit{(Jan. 1960 - Apr. 2025)}} & \textit{(0.010)} & \textit{(0.014)} & \textit{(0.011)} & \textit{} & \textit{(0.014)} & \textit{(0.010)} & \textit{} \tabularnewline 
   \hline
  Pre-1994 & -0.059** & 0.064** & -0.312 & \multirow{2}{*}{0.046} & 0.055* & 0.096 & \multirow{2}{*}{0.467} \tabularnewline 
   {\footnotesize \textit{(Jan. 1960 - Dec. 1993)}} & \textit{(0.026)} & \textit{(0.031)} & \textit{(0.290)} & \textit{} & \textit{(0.033)} & \textit{(0.070)} & \textit{} \tabularnewline 
  Post-1994 & -0.022 & -0.003 & -0.022 & \multirow{2}{*}{0.016} & -0.000 & -0.004 & \multirow{2}{*}{0.280} \tabularnewline 
   {\footnotesize \textit{(Jan. 1994 - Apr. 2025)}} & \textit{(0.015)} & \textit{(0.009)} & \textit{(0.017)} & \textit{} & \textit{(0.008)} & \textit{(0.016)} & \textit{} \tabularnewline 
    \hline
  \end{tabular}

\end{subtable}
}
\medskip

{ {\footnotesize\textit{{}Notes:}}{\footnotesize{} We report estimates
and the standard error (in parentheses below the estimates) across
methods and setups. The upper and lower panels corresponds to the
scenarios where we use TCODE transformed or untransformed time series
as covariates, respectively. The symbols {*}, {*}{*}, and {*}{*}{*}
indicate significance levels at 10\%, 5\%, and 1\%, respectively.
``VR Test'' represents the $p$-value of the variance ratio test
\citep{kim2009automatic} on the LASSO residual. The tuning parameter
for LASSO estimation is selected through 10-fold block cross-validation.
In XDlasso, instruments are generated based on \eqref{eq:IVX_original}
and \eqref{eq:rho=000020zeta} with $C_{\zeta}=5$ and $\tau=0.5$.}}
\end{table}

Table \ref{tab:app_est_return_ep} presents the point estimates and
testing results for $\theta_{1}^{*}$. IVX delivers negative point
estimates across all sample periods and detects significant effects
in the full sample period and Pre-1994 subperiod at the 5\% significance
level. However, the resulting negative relation between stock return
and earning performance contradicts the economic mechanism between
the two variables. A large number of potential confounding variables
are present in the dataset, which may be the culprit in producing
the counterintuitive result from the simple regression.

Dlasso, which accounts for high dimensionality but not nonstationarity,
generally reverses the sign of the estimates compared to the simple
regression. For the full sample, Dlasso consistently yields positive
but statistically insignificant coefficients. Zooming into the Pre-1994
period, Dlasso detects strong evidence of predictability with transformed
data, with coefficients significant at the 5\% level. This result
persists with the untransformed data, though the significance weakens
to the 10\% level when the lagged dependent variable is included.
Accounting for the potential nonstationarity in the predictors, XDlasso
provides notably different estimates and standard errors from Dlasso.
XDlasso reveals no significant predictive power of log earnings-price
ratio for stock return in the full sample and the Post-1994 period.
Moreover, for the Pre-1994 period, the evidence of predictability
is weakened to the 10\% significance level with transformed data,
and becomes entirely insignificant with untransformed data. Our empirical
findings with XDlasso are consistent with the general recognition
that there is little predictability in the financial market, thereby
supporting the efficient market hypothesis. The divergence between
Dlasso and XDlasso underscores our theoretical prediction: XDlasso
mitigates the Stambaugh bias arising from highly persistent regressors
while Dlasso does not. Despite the short confidence intervals, Dlasso
results can be misleading when nonstationary time series is present,
as clearly shown by the illustrative simulation in Section \ref{subsec:Necessity-of-IVX}.

The variance ratio tests on the Slasso residuals provide a diagnostic
check on the m.d.s.~condition required by our asymptotic theory.
The results exhibit stark differences across data specifications:
with untransformed data and without the lagged dependent variable,
the tests strongly reject the m.d.s. condition ($p$-values ranging
from 0.001 to 0.046), suggesting potential model misspecification.
Including the lagged return or using transformed data substantially
improves the diagnostic results, with $p$-values generally exceeding
0.05. This pattern suggests that both data transformation and including
dynamics help satisfy the underlying assumptions, thereby boosting
credibility for the XDlasso results under these specifications.

\subsection{Predictability of Inflation Using Unemployment Rate\protect\label{subsec:emp_inflation_unrate}}

It is essential for monetary policymakers to understand the relationship
between unemployment and inflation. As \citet{engemann2020phillips}
pointed out, ``The Federal Reserve has a dual mandate to promote
maximum sustainable employment and price stability.'' First alluded
to by \citet{fisher1926statistical,fisher1973discovered} though,
\citet{phillips1958relation} popularized the \emph{Phillips curve}
--- a negative relationship between the level of unemployment and
the change rate of money wage rates. There has been a prolonged debate
about whether the unemployment is a credible barometer for inflation
among not only modern economic studies, but also policymakers.\footnote{Mary Daly, San Francisco Fed President, delivered at \citet{daly2019new}
a negative view on the Phillips curve that \textquotedblleft As for
the Phillips curve\ldots{} most arguments today center around whether
it\textquoteright s dead or just gravely ill. Either way, the relationship
between unemployment and inflation has become very difficult to spot.\textquotedblright{}
John Williams, New York Fed President, expressed a different opinion
that \textquotedblleft The Phillips curve is the connective tissue
between the Federal Reserve\textquoteright s dual mandate goals of
maximum employment and price stability. Despite regular declarations
of its demise, the Phillips curve has endured. It is useful, both
as an empirical basis for forecasting and for monetary policy analysis.\textquotedblright{}
See \citet{engemann2020phillips} for more details.} Empirical findings suggest that inflation rate can be either positively
or negatively correlated with unemployment, depending on the shocks
to the economies, the policies, and the lag orders \citep{niskanen2002death,gordon2011history,gordon2013phillips}.
Given the ongoing debate, we revisit the Phillips curve in a predictive
regression framework utilizing the FRED-MD dataset.

\subsubsection{Data}

The inflation rate, as the outcome variable in the predictive regression,
is calculated by $\text{Inflation}_{t}=\left(\log({\rm CPI}_{t})-\log({\rm CPI}_{t-1})\right)\times100,$
where ${\rm CPI}_{t}$ denotes the \textit{Consumer Price Index for
All Urban Consumers: All Items (CPI)}. The unemployment rate, as the
predictor of interest, denoted as $\text{Unrate}_{t}$, is retrieved
in its original form under the name \texttt{UNRATE}.

Similar to our first empirical application in Section \ref{sec:emp-return-ep},
we follow \citet{benati2015long} to delineate three subperiods in
addition to the full sample: \textit{Pre-Volcker} (January 1960 -
July 1979), \textit{Volcker and Greenspan} (August 1979 - January
2006), and \textit{Bernanke, Yellen, and Powell} (February 2006 -
April 2025). These subperiods correspond to different eras in U.S.
monetary policy, each named after the Federal Reserve chairperson
who presided during that time.\footnote{Pre-Volcker (before August 1979): This period was characterized by
high and volatile inflation, with the Federal Reserve lacking a clear
nominal anchor. Volcker and Greenspan (August 1979--January 2006):
Volcker, who served as the Federal Reserve Chairman from August 1979
to August 1987, implemented aggressive anti-inflation measures, notably
raising interest rates to historically high levels. Greenspan succeeded
Volcker and continued to focus on maintaining price stability during
his tenure, which contributed to a period of low and stable inflation
known as the ``Great Moderation''. Bernanke, Yellen, and Powell
(February 2006--December 2019): Bernanke's tenure as Fed Chairman
was marked by the Great Recession and the implementation of unconventional
monetary policies, such as quantitative easing, aimed at stimulating
the economy and preventing deflation. Yellen and Powell carried on
these policies.} This periodization allows us to examine how the relationship between
inflation and unemployment may have evolved across different policy
regimes and economic conditions.

Figure \ref{fig:time_series_unrate_infl} plots the inflation and
unemployment rates over our sample period. Visual inspection suggests
that the unemployment rate is more persistent than the inflation rate.
Table \ref{tab:Persistence-of-Inflation} further reports the AR(1)
coefficient estimates and ADF test $p$-values of the inflation and
the unemployment rate under different sample periods. The inflation
rate appears stationary in most periods, with AR(1) coefficient estimates
ranging from 0.54 to 0.64. However, the inflation rate during the
pre-Volcker period shows a slight upward trend, and nonstationarity
is indicated by the ADF test with a $p$-value of 0.22. In contrast,
the inflation rate is found to be stationary in the other two subperiods.
The unemployment rate, on the other hand, appears highly persistent
with AR(1) coefficient estimates close to 1. The ADF test rejects
the null hypothesis of nonstationarity at a 10\% significance level
for both the full sample and the Volcker-Greenspan period. However,
during the Pre-Volcker and Bernanke-Yellen-Powell periods, there is
strong evidence of nonstationarity in the unemployment rate. The stationarity
of both inflation and the unemployment rate in different periods is
again ambiguous, which prompts the use of XDlasso from an agnostic
perspective.

\begin{figure}[h!]
\centering{}\includegraphics[width=0.9\textwidth]{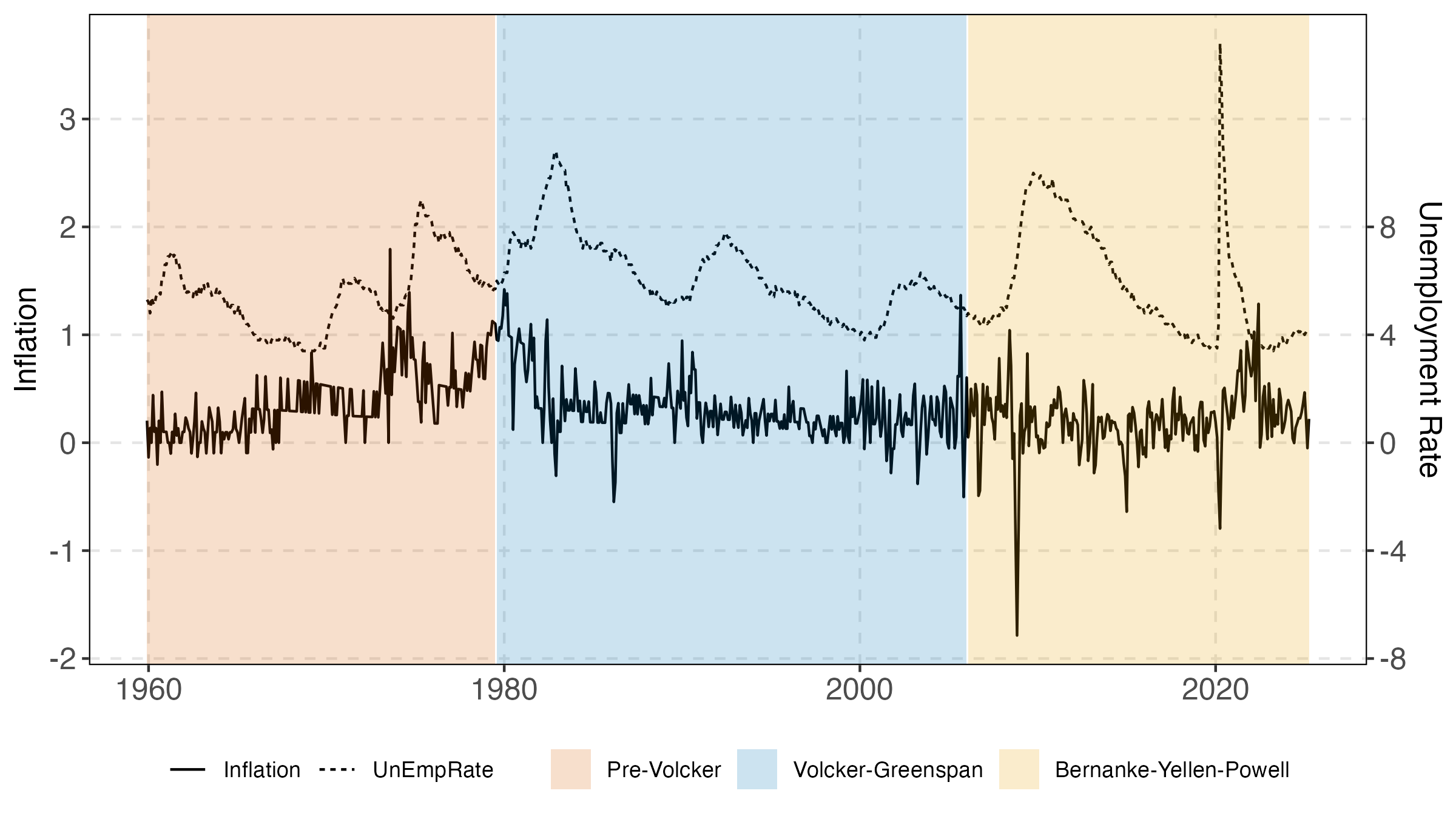}
\caption{Inflation and Unemployment Rate\protect\label{fig:time_series_unrate_infl}}
\end{figure}
 
\begin{table}
\caption{Persistence of Inflation Rate and Unemployment Rate\protect\label{tab:Persistence-of-Inflation}}

\begin{centering}
{\small{}%
\begin{tabular}{r|cc|cc}
\hline 
\multirow{1}{*}{{\small Sample Period}} & \multicolumn{2}{c|}{{\small Inflation Rate}} & \multicolumn{2}{c}{{\small Unemployment Rate}}\tabularnewline
\hline 
 & {\small AR(1)} & {\small ADF $p$-value} & {\small AR(1)} & {\small ADF $p$-value}\tabularnewline
\hline 
{\small Full Sample }{\small\textit{(Jan. 1960 - Apr. 2025)}} & \multirow{1}{*}{{\small 0.620}} & \multirow{1}{*}{{\small 0.0183}} & \multirow{1}{*}{{\small 0.999}} & \multirow{1}{*}{{\small 0.074}}\tabularnewline
\hline 
{\small Pre-Volcker }{\small\textit{(Jan. 1960 - July 1979)}} & \multirow{1}{*}{{\small 0.640}} & \multirow{1}{*}{{\small 0.220}} & \multirow{1}{*}{{\small 0.986}} & \multirow{1}{*}{{\small 0.406}}\tabularnewline
\hline 
{\small Volcker-Greenspan }{\small\textit{(Aug. 1979 - Jan. 2006)}} & \multirow{1}{*}{{\small 0.623}} & \multirow{1}{*}{{\small$<$0.01}} & \multirow{1}{*}{{\small 0.992}} & \multirow{1}{*}{{\small 0.071}}\tabularnewline
\hline 
{\small Bernanke-Yellen-Powell }{\small\textit{(Feb. 2006 - Apr. 2025)}} & \multirow{1}{*}{{\small 0.537}} & \multirow{1}{*}{{\small$<$0.01}} & \multirow{1}{*}{{\small 0.939}} & \multirow{1}{*}{{\small 0.325}}\tabularnewline
\hline 
\end{tabular}}{\small\par}
\par\end{centering}
{\footnotesize\textit{{}Notes:}}{\footnotesize{} The lag order for
the ADF test is set to $\lfloor n^{1/3}\rfloor$ where $n$ is the
effective sample size. The exact start and end dates of subperiods
are provided in the first column of the table.}{\footnotesize\par}
\end{table}

In addition to the unemployment rate, we incorporate all other 110
macroeconomic variables from the FRED-MD dataset as controls after
the TCODE transformation. Still, a significant proportion, from 11\%
to 17\% across subperiods, of the transformed variables exhibit nonstationarity
according to the ADF test at the 5\% significance level. Consistent
with Section \ref{sec:emp-return-ep}, we also perform the analysis
using the original (untransformed) FRED-MD time series as covariates.
To further address potential concerns regarding the high correlation
between control variables and the unemployment rate, we additionally
present results excluding variables from the labor market group, as
classified by FRED-MD.

\subsubsection{Results}

We study, in our predictive regression framework, a one-month-ahead
regression $\text{Inflation}_{t}=\alpha^{*}+\theta_{1}^{*}\times\text{Unrate}_{t-1}+W_{-1,t-1}^{\top}\theta_{-1}^{^{*}}+u_{t},$
where $W_{-1,t-1}$ denotes high-dimensional covariates. Given the
ambiguous persistence of the inflation rate discussed in the previous
section, including the lagged dependent variable in our model may
introduce further technical complications. Therefore, we do not recommend
its inclusion in this analysis, and will focus on the traditional
predictive regression setting. We carry out hypothesis testing on
the key parameter of interest $\theta_{1}^{*}$ that measures the
predictive power of the unemployment rate in forecasting inflation.
The predictive form is of particular interest among policymakers in
leveraging the relationship as a practical tool.

\begin{table}
\begin{centering}
\caption{Test $\mathbb{H}_{0}:\theta_{1}^{*}=0$ across sample periods and
specifications in inflation prediction\protect\label{tab:app_est_pc_result}}
\par\end{centering}
{\small 
\begin{subtable}[t]{1\textwidth}
\centering
\caption{\texttt{TCODE} Transformed Data}

\begin{tabular}{c|c|ccc}
  \hline
Sample Period & IVX & Dlasso & XDlasso & VR Test \tabularnewline 
  \hline
Full Sample & -0.014 & 0.018*** & -0.024 & \multirow{2}{*}{0.000} \tabularnewline 
 {\footnotesize \textit{(Jan. 1960 - Apr. 2025)}} & \textit{(0.018)} & \textit{(0.006)} & \textit{(0.077)} \tabularnewline 
 \hline
Pre-Volcker & 0.080 & 0.074*** & 0.013 & \multirow{2}{*}{0.125} \tabularnewline 
 {\footnotesize \textit{(Jan. 1960 - Jul. 1979)}} & \textit{(0.069)} & \textit{(0.017)} & \textit{(0.224)} \tabularnewline 
Volcker-Greenspan & 0.036 & -0.020 & 0.161 & \multirow{2}{*}{0.025} \tabularnewline 
 {\footnotesize \textit{(Aug. 1979 - Jan. 2006)}} & \textit{(0.063)} & \textit{(0.020)} & \textit{(0.118)} \tabularnewline 
Bernanke/Yellen/Powell & -0.043 & -0.002 & -0.054 & \multirow{2}{*}{0.002} \tabularnewline 
 {\footnotesize \textit{(Feb. 2006 - Apr. 2025)}} & \textit{(0.027)} & \textit{(0.011)} & \textit{(0.120)} \tabularnewline 
  \hline
\end{tabular}

\end{subtable}
\hfill
\medskip

\begin{subtable}[t]{1\textwidth}
\centering
\caption{Untransformed Data: Excluding I(2) Variables Based on \texttt{TCODE}}

\begin{tabular}{c|ccc|ccc}
  \hline
  & \multicolumn{3}{c|}{Include Labor Variables} & \multicolumn{3}{c}{Exclude Labor Variables}   \tabularnewline 
  \hline
Sample Period & Dlasso & XDlasso & VR Test  & Dlasso & XDlasso & VR Test \tabularnewline 
  \hline
Full Sample & -0.077 & 0.068 & \multirow{2}{*}{0.080} & -0.050*** & 0.033 & \multirow{2}{*}{0.065} \tabularnewline 
 {\footnotesize \textit{(Jan. 1960 - Apr. 2025)}} & \textit{(0.069)} & \textit{(0.218)} &  & \textit{(0.015)} & \textit{(0.032)} & \tabularnewline 
 \hline
Pre-Volcker & -0.129 & -0.014 & \multirow{2}{*}{0.522} & 0.007 & 0.113 & \multirow{2}{*}{0.548} \tabularnewline 
 {\footnotesize \textit{(Jan. 1960 - Jul. 1979)}} & \textit{(0.117)} & \textit{(0.276)} & & \textit{(0.051)} & \textit{(0.098)} & \tabularnewline 
Volcker-Greenspan & 0.094 & -0.272 & \multirow{2}{*}{0.741} & -0.092* & -0.259 & \multirow{2}{*}{0.339} \tabularnewline 
 {\footnotesize \textit{(Aug. 1979 - Jan. 2006)}} & \textit{(0.217)} & \textit{(0.287)} &  & \textit{(0.054)} & \textit{(0.228)} &  \tabularnewline 
Bernanke-Yellen-Powell & 0.550 & 0.585 & \multirow{2}{*}{0.283} & 0.001 & 0.040 & \multirow{2}{*}{0.230} \tabularnewline 
 {\footnotesize \textit{(Feb. 2006 - Apr. 2025)}} & \textit{(0.442)} & \textit{(0.975)} & & \textit{(0.055)} & \textit{(0.073)} & \tabularnewline 
  \hline
\end{tabular}

\end{subtable}
}
\medskip

{ {\footnotesize\textit{{}Notes:}}{\footnotesize{} We report estimates
and the standard error (in parentheses below the estimates) across
methods and setups. The upper and lower panels corresponds to the
scenarios where we use TCODE transformed or untransformed time series
as covariates, respectively. The symbols {*}, {*}{*}, and {*}{*}{*}
indicate significance levels at 10\%, 5\%, and 1\%, respectively.
``VR Test'' represents the $p$-value of the variance ratio test
\citep{kim2009automatic} on the LASSO residual. The tuning parameter
for LASSO estimation is selected through 10-fold block cross-validation.
In XDlasso, instruments are generated based on \eqref{eq:IVX_original}
and \eqref{eq:rho=000020zeta} with $C_{\zeta}=5$ and $\tau=0.5$.}}
\end{table}

Table \ref{tab:app_est_pc_result} reports the point estimates and
standard errors for $\theta_{1}^{*}$ using IVX, Dlasso, and XDlasso
across sample periods and specifications. In the benchmark setup with
the transformed data, the diagnostic check rejects the m.d.s.~condition
except for the Pre-Volcker period. Nevertheless, Dlasso, which ignores
nonstationarity, delivers significantly positive coefficients for
both the full sample and the Pre-Volcker period, in striking contrast
to XDlasso and IVX. Using untransformed data significantly alleviates
concerns about model misspecification, with $p$-values of the diagnostic
test greater than 5\% in all cases. Both XDlasso and IVX consistently
find no significant predictive power of unemployment for inflation
across all specifications and time periods. The empirical findings
add new insight to the recent debates on the Phillips curve and echo
\citet{Mankiw2024}'s latest pessimistic remark: ``The large confidence
intervals for the natural rate, together with the apparent futility
of this Holy Grail search, lead me to think that we should not expect
much from the Phillips curve as a guide for forecasting inflation
or for judging the stance of policy.'' On the other hand, with untransformed
data, Dlasso yields significantly negative coefficients at the 1\%
level for the full sample and at the 10\% level in the Volcker-Greenspan
period, without controlling for other labor market variables. The
unstable results of Dlasso across setups highlight the necessity of
accounting for nonstationarity in the inference.

\section{Conclusion\protect\label{sec:Conclusion}}

This paper proposes XDlasso to overcome the difficulties in hypothesis
testing for high dimensional predictive regressions with stationary
and nonstationary regressors. XDlasso fuses the IVX technique from
time series econometrics and the debiasing technique from the high
dimensional statistics, thereby reducing the order of biases to make
them readily correctable. We establish the asymptotic normality and
convergence rate of XDlasso. The validity of our methods is further
evidenced by simulation studies and empirical applications. 

\section*{Acknowledgement}
We thank the editor, the associate editor, and two referees for their constructive comments that led to substantial improvement of this paper. We thank Andrii Babii, Yongmiao Hong, Degui Li, Yiu Lim Lui, Kenwin Maung, Jun Yu, Wenyang Zhang, and Yichong Zhang for helpful comments and suggestions. Gao acknowledges the partial support from Southern Methodist University (SMU) University Research Council (URC) Grants. Mei acknowledges the partial  financial support from the Start-up Research Grant (SRG) from University of Macau (Project No.~SRG2025-00062-FBA). Shi acknowledges the partial financial support from the National Natural Science Foundation of China (Project No.~72425007, No.~72133002).

\bigskip \bigskip

\bibliographystyle{apalikeyear}
\bibliography{ref}

\onehalfspacing
\normalsize

\newpage{}

\setcounter{footnote}{0}
\setcounter{table}{0} 
\setcounter{figure}{0} 
\setcounter{equation}{0} 
\renewcommand{\thefootnote}{\thesection.\arabic{footnote}} 
\renewcommand{\theequation}{\thesection.\arabic{equation}} 
\renewcommand{\thefigure}{\thesection.\arabic{figure}} 
\renewcommand{\thetable}{\thesection.\arabic{table}} 
\setcounter{thm}{0} 
\setcounter{lem}{0} 
\setcounter{rem}{0}
\setcounter{prop}{0} 
\renewcommand{\thethm}{\thesection.\arabic{thm}}
\renewcommand{\thelem}{\thesection.\arabic{lem}} 
\renewcommand{\therem}{\thesection.\arabic{rem}} 
\renewcommand{\theprop}{\thesection.\arabic{prop}} 
\begin{appendices}
\begin{center}
{\LARGE Online Appendices to ``LASSO Inference \\ for High Dimensional Predictive Regressions''}
\par\end{center}

\begin{center}
{\large Zhan Gao$^a$, Ji Hyung Lee$^b$, Ziwei Mei$^c$, Zhentao Shi$^d$ \\
\ \\
$^a$Southern Methodist University\\
$^b$University of Illinois at Urbana-Champaign\\
$^c$University of Macau\\
$^d$The Chinese University of Hong Kong}
\par\end{center}

Section \ref{sec:Technical-Proofs} includes the proofs of all theoretical
statements in the main text. Section \ref{sec:Additional-Simulation-Results}
and \ref{sec:Additional-Empirical-Results} collect additional simulation
and empirical results omitted from the main text. 

\section{Technical Proofs\protect\label{sec:Technical-Proofs}}

In the proofs, we use $c$ and $C$, without superscripts or subscripts,
to denote generic positive constants that may vary place to place.
For any positive sequences $\{a_{n}\}$ and $\{b_{n}\}$, ``$a_{n}\stackrel{\mathrm{p}}{\preccurlyeq}b_{n}$''
means that there is an absolute constant, say $c$, such that the
event $\left\{ a_{n}\leq cb_{n}\right\} $ holds with probability
approaching one (w.p.a.1.). Symmetrically, ``$a_{n}\stackrel{\mathrm{p}}{\succcurlyeq}b_{n}$''
means ``$b_{n}\stackrel{\mathrm{p}}{\preccurlyeq}a_{n}$''. The
integer floor function is denoted as $\left\lfloor \cdot\right\rfloor $.
For an $n$-dimensional vector $x=(x_{t})_{t\in[n]}$, the $L_{2}$-norm
is $\left\Vert x\right\Vert _{2}=\sqrt{\sum_{t=1}^{n}x_{t}^{2}}$.
For notational simplicity, in the proofs we assume $p\geq n^{\nu_{1}}$
for some absolute constant $\nu_{1}$, which is reasonable as we focus
on the high dimensional case with a larger $p$ relative to $n$.\footnote{There is no technical difficulty in allowing $p$ to grow either slowly
at a logarithmic or fast at an exponential rate of $n$, but we have
to compare $\log p$ and $\log n$ in many places, and in many conditions
and rates the term ``$\log p$'' has to be changed into $\log(np)$.}We use $\mathbf{I}_{n}$ to denote the $n$-dimensional identity matrix,
where the subscript may be omitted when there is no ambiguity with
matrix dimensions.

Section \ref{subsec:Proofs-for-Section-1} proves the results in Section
\ref{subsec:Consistency-of-Slasso} for the consistency of Slasso
at the presence of both LUR and stationary regressors. Section \ref{subsec:Proofs-for-Section}
includes the proofs for the results in Section \ref{subsec:The-Auxiliary-Regression}
that constructs consistency of the auxiliary LASSO regression adapted
for bias correction. Section \ref{subsec:Proofs-for-Section-normal}
includes the proofs for the theorems in Section \ref{subsec:Asymptotic-Normality}
about the asymptotic distribution of the XDlasso estimator and the
order of its standard error, which determine the size and power of
the XDlasso inference. 

\subsection{Proofs for Section \ref{subsec:Consistency-of-Slasso}\protect\label{subsec:Proofs-for-Section-1}}

\subsubsection{Technical Lemmas of Gaussian Approximation }

The assumptions imposed in Section \ref{sec:Asymptotic-Theory} slightly
differ from those in \citet{mei2022lasso} (MS24, hereafter). Specifically,
the linear process assumption in MS24 is extended to more general
$\alpha$-mixing and sub-exponential conditions. One of the main modifications
of the proof is the Gaussian approximation error deduced in MS24's
Lemma B.4. The following lemma re-establishes Gaussian approximation
errors for the stationary components $v_{t}$ defined as (\ref{eq:error_transform}).
The results of Gaussian approximation will be useful for the RE condition
required for the Slasso's consistency.
\begin{lem}
\label{lem:Gaussian=000020Approximation} Under Assumptions \ref{assu:tail}--\ref{assu:asym_n},
there exist standard Brownian motions $\{\mathcal{B}_{k}(t)\}$ with
independent increment $\mathcal{B}_{k}(t)-\mathcal{B}_{k}(s)\sim\mathcal{N}(t-s)$
for $t\geq s\geq0$ such that
\[
\sup_{k\in[p],t\in[n]}\left|\frac{1}{\sqrt{n}}\left(\sum_{s=0}^{t-1}\varepsilon_{k,s}-\mathcal{B}_{k}(t\cdot V_{k}^{*})\right)\right|=O_{p}\left(\dfrac{(\log p)^{3/2}}{n^{1/4}}\right)
\]
where $V_{k}^{*}=\mathbb{E}\left[\sum_{d=-\infty}^{\infty}\varepsilon_{k,t}\varepsilon_{k,t-d}\right]$
is the long-run variance of $\{\varepsilon_{k,t}\}.$
\end{lem}
\begin{rem}[Convergence Rate of Gaussian Approximation Error]
The convergence rate in Lemma \ref{lem:Gaussian=000020Approximation}
is less sharp than that in MS24's Lemma B.4, since we work with general
$\alpha$-mixing sequences without specifying linear processes. 
\end{rem}
\begin{proof}[Proof of Lemma \ref{lem:Gaussian=000020Approximation}]
Define $V_{k,t}=\mathbb{E}\left[\left(\sum_{s=0}^{t-1}\varepsilon_{k,s}\right)^{2}\right].$
The proof includes the following two steps: 
\begin{enumerate}
\item The variance of the partial sum $V_{k,t}$ is well approximated by
the long run variance $V_{k}^{*}$ scaled by $t$, in the sense that
\begin{align}
\sup_{k\in[p],t\in[n]}\left|\dfrac{t\cdot V_{k}^{*}-V_{k,t}}{n}\right| & =O\left(\dfrac{1}{n}\right).\label{eq:approx=000020lr=000020var}
\end{align}
\item Strong Gaussian approximation: There exist standard Brownian motions
$\{\mathcal{B}_{k}(t)\}$ with independent increment $\mathcal{B}_{k}(t)-\mathcal{B}_{k}(s)\sim\mathcal{N}(t-s)$
for $t\geq s\geq0$ such that 
\begin{equation}
\sup_{k\in[p],t\in[n]}\left|\frac{1}{\sqrt{n}}\left(\sum_{s=0}^{t-1}\varepsilon_{k,s}-\mathcal{B}_{k}(V_{k,t})\right)\right|=O_{a.s.}\left(\dfrac{(\log n)^{3/2}}{n^{1/4}}\right).\label{eq:strong=000020approx}
\end{equation}
\end{enumerate}
Given the two steps above, Lemma \ref{lem:Gaussian=000020Approximation}
follows by the triangular inequality 
\begin{align*}
 & \sup_{k\in[p],t\in[n]}\left|\frac{1}{\sqrt{n}}\left(\sum_{s=0}^{t-1}\varepsilon_{k,s}-\mathcal{B}_{k}(t\cdot V_{j}^{*})\right)\right|\\
\leq & \sup_{j\in[p],t\in[n]}\left|\dfrac{1}{\sqrt{n}}\left(\mathcal{B}_{k}(V_{k,t})-\mathcal{B}_{k}(t\cdot V_{j}^{*})\right)\right|+\sup_{k\in[p],t\in[n]}\left|\frac{1}{\sqrt{n}}\left(\sum_{s=0}^{t-1}\varepsilon_{k,s}-\mathcal{B}_{k}(V_{k,t})\right)\right|\\
= & O_{p}\left(\sqrt{\log(np)\cdot\sup_{k\in[p],t\in[n]}\left|\dfrac{t\cdot V_{k}^{*}-V_{k,t}}{n}\right|}\right)+\sup_{k\in[p],t\in[n]}\left|\frac{1}{\sqrt{n}}\left(\sum_{s=0}^{t-1}\varepsilon_{k,s}-\mathcal{B}_{k}(V_{k,t})\right)\right|\\
= & O_{p}\left(\sqrt{\dfrac{\log(np)}{n}}\right)+O_{a.s.}\left(\dfrac{(\log n)^{3/2}}{n^{1/4}}\right)\\
= & O_{p}\left(\dfrac{(\log p)^{3/2}}{n^{1/4}}\right)
\end{align*}
where the second row applies the fact that $n^{-1/2}\left(\mathcal{B}_{k}(V_{k,t})-\mathcal{B}_{k}(t\cdot V_{k}^{*})\right)$
for each $(j,t)$ follows a normal distribution of mean zero and variance
$\left|\left(t\cdot V_{k}^{*}-V_{k,t}\right)/n\right|$, the third
row applies (\ref{eq:approx=000020lr=000020var}), (\ref{eq:strong=000020approx}),
and the last row applies the assumption $p\geq n^{\nu_{1}}$. 

\textbf{Step 1. Verifying (\ref{eq:approx=000020lr=000020var}). }Define
${\rm cov}_{k}(d):={\rm cov}(\varepsilon_{k,t},\varepsilon_{k,t-d})$
as the autocovariance function of $\{\varepsilon_{k,t}\}$. Then by
some fundamental calculations, 
\begin{equation}
V_{k,t}=t\cdot{\rm cov}_{k}(0)+2\sum_{d=1}^{t}(t-d)\cdot{\rm cov}_{k}(d).\label{eq:decom=000020Vjt}
\end{equation}
In addition, 
\[
t\cdot V_{k}^{*}=t\cdot{\rm cov}_{k}(0)+2t\cdot\sum_{d=1}^{\infty}{\rm cov}_{k}(d).
\]
Then 
\begin{align}
t\cdot V_{k}^{*}-V_{k,t} & =2t\sum_{d=1}^{\infty}{\rm cov}_{k}(d)-2\sum_{d=1}^{t}(t-d)\cdot{\rm cov}_{k}(d)\nonumber \\
 & =2t\sum_{d=t+1}^{\infty}{\rm cov}_{k}(d)+2\sum_{d=1}^{t}d\cdot{\rm cov}_{k}(d).\label{eq:sum=000020cov}
\end{align}
Recall that Assumption \ref{assu:alpha} imposes an upper bound for
the $\alpha$-mixing coefficient. With $p=q=3$ in Equation (2.2)
in the Corollary of \citet[pp. 692]{davydov1968convergence}, we have
\begin{equation}
\sup_{j\in[p]}|{\rm cov}_{k}(d)|\leq12\left(\mathbb{E}|\varepsilon_{k,t}|^{3}\right)^{2/3}\sqrt{\alpha(d)}\leq C\cdot\exp\left(-c_{\alpha}d^{r}/2\right).\label{eq:sup=000020cov=000020bound}
\end{equation}
where $C=12\left(\mathbb{E}|\varepsilon_{k,t}|^{3}\right)^{2/3}\sqrt{C_{\alpha}}$
and $C_{\alpha}$ is in Assumption \ref{assu:alpha}. By Equation
(B.78) in MS24's supplement, 
\[
\sup_{t\in[n]}\sum_{d=t+1}^{\infty}\exp\left(-c_{\alpha}d^{r}/2\right)\leq\dfrac{2}{c_{\alpha}}\exp\left(-c_{\alpha}t^{r}/4\right),
\]
and thus uniformly for all $t,$ there exists some absolute constant
$C_{1}$ such that
\begin{equation}
\sup_{j\in[p],t\in[n]}t\cdot\sum_{d=t+1}^{\infty}|{\rm cov}_{k}(d)|\leq C\cdot\sup_{t\in[n]}t\sum_{d=t+1}^{\infty}\exp\left(-c_{\alpha}d^{r}/2\right)\leq\sup_{t\in[n]}\dfrac{2Ct}{c_{\alpha}}\exp\left(-c_{\alpha}t^{r}/4\right)<C_{1}.\label{eq:bound=000020t=000020sum=000020cov}
\end{equation}
 In addition, 
\begin{equation}
\sup_{j\in[p],t\in[n]}\sum_{d=1}^{t}d\cdot|{\rm cov}_{k}(d)|<C\sum_{d=1}^{\infty}d\cdot\exp\left(-c_{\alpha}d^{r}/2\right)<C_{2}\label{eq:bound=000020sum=000020d=000020cov}
\end{equation}
for some absolute constant $C_{2}$. Then
\begin{align*}
\sup_{j\in[p],t\in[n]}\left|\dfrac{t\cdot V_{k}^{*}-V_{k,t}}{n}\right| & \leq\dfrac{2}{n}\cdot\left(\sup_{j\in[p],t\in[n]}t\cdot\sum_{d=t+1}^{\infty}|{\rm cov}_{k}(d)|+\sup_{j\in[p],t\in[n]}\sum_{d=1}^{t}d\cdot|{\rm cov}_{k}(d)|\right)\\
 & \leq\dfrac{2(C_{1}+C_{2})}{n},
\end{align*}
which implies (\ref{eq:approx=000020lr=000020var}). 

\textbf{Step 2. Verifying (\ref{eq:strong=000020approx}). }We use
the strong Gaussian approximation from \citet{lin1997limit}'s Theorem
9.3.1. Specifically, define $g(x)=\exp(x)$. By the sub-exponential
tail imposed by Assumption \ref{assu:tail}, the \emph{sub-exponential
norm }of $\varepsilon_{k,t}$, denoted as $\|\varepsilon_{k,t}\|_{g}$
in \citet{lin1997limit}, is uniformly bounded by an absolute constant.
It then suffices to verify the following two conditions required in
the aforementioned theorem:
\begin{itemize}
\item[(i)] $V_{k,t}\geq ct$ for some absolute constant $c$. 
\item[(ii)] $\sum_{d=1}^{\infty}\alpha(d)^{1/4}\cdot\log(1/\alpha(d))<\infty$,
where the parameter $\delta$ in \citet[Theorem 9.3.1]{lin1997limit}
is taken as 2.
\end{itemize}
Then by Theorem 9.3.1 of \citet{lin1997limit}, for any $j\in[p]$
\begin{equation}
\sum_{s=0}^{t-1}\varepsilon_{k,s}-\mathcal{B}_{k}(V_{k,t})=O_{a.s.}\left(V_{k,t}^{1/4}(\log V_{k,t})^{3/2}\right).\label{eq:temp=000020strong=000020approx}
\end{equation}
By (\ref{eq:decom=000020Vjt}) and (\ref{eq:sup=000020cov=000020bound}),
\begin{align*}
V_{k,t} & =t\left[{\rm cov}_{k}(0)+2\sum_{d=1}^{t}(1-d/t)\cdot{\rm cov}_{k}(d)\right]\\
 & \leq t\left[{\rm cov}_{k}(0)+2C\cdot\sum_{d=1}^{\infty}\exp\left(-c_{\alpha}d^{r}/2\right)\right]=O(t)
\end{align*}
uniformly for all $(k,t)$. Then by (\ref{eq:temp=000020strong=000020approx}),
\[
\sup_{j\in[p],t\in[n]}\left|\frac{1}{\sqrt{n}}\left(\sum_{s=0}^{t-1}\varepsilon_{k,s}-\mathcal{B}_{k}(V_{k,t})\right)\right|=O_{a.s.}\left(\sup_{t\in[n]}\dfrac{t^{1/4}(\log t)^{3/2}}{\sqrt{n}}\right)=O_{a.s.}\left(\dfrac{(\log n)^{3/2}}{n^{1/4}}\right),
\]
which leads to (\ref{eq:strong=000020approx}). It then suffices to
show the Conditions (i) and (ii) above. 

\textbf{Proof of (i). }By (\ref{eq:sum=000020cov}) and (\ref{eq:sup=000020cov=000020bound}),
\begin{align*}
\lim_{t\to\infty}\sup_{k\in[p]}\left|\dfrac{V_{k,t}}{t}-V_{k}^{*}\right| & \leq\lim_{t\to\infty}2\sup_{k\in[p]}\left|\sum_{d=t+1}^{\infty}{\rm cov}_{k}(d)+\dfrac{2}{t}\sum_{d=1}^{t}d\cdot{\rm cov}_{k}(d)\right|\\
 & \leq\lim_{t\to\infty}C\left[\sum_{d=t+1}^{\infty}\exp\left(-c_{\alpha}d^{r}/2\right)+\dfrac{2}{t}\sum_{d=1}^{t}d\cdot{\rm cov}_{k}(d)\right]\\
 & =0,
\end{align*}
where $C$ is an absolute constant, and the limit applies (\ref{eq:bound=000020t=000020sum=000020cov})
and (\ref{eq:bound=000020sum=000020d=000020cov}). By Assumption \ref{assu:alpha},
the long-run variance $V_{k}^{*}$ is bounded from below by some absolute
constant. This result implies that $V_{k,t}/t$ is lower bounded by
some absolute constant uniformly for all $(k,t)$.

\textbf{Proof of (ii). }This is a direct corollary of the exponential
decaying mixing coefficient imposed by Assumption \ref{assu:alpha},
in the sense that 
\[
\sum_{d=1}^{\infty}\alpha(d)^{1/4}\cdot\log(1/\alpha(d))\leq\sum_{d=1}^{\infty}C_{\alpha}^{1/4}\exp\left(-c_{\alpha}d^{r}/4\right)\cdot c_{\alpha}d^{r}<\infty.
\]
This completes to proof of Lemma \ref{lem:Gaussian=000020Approximation}. 
\end{proof}
The next lemma establishes that result that the LUR regressors $X_{t}$
with general weakly dependent innovations can be approximated by another
vector LUR processes with normally distributed innovations.
\begin{lem}
\label{lem:LUR=000020Gaussian=000020Approx} Suppose that Assumptions
\ref{assu:tail}--\ref{assu:asym_n} hold. There exists independent
normally distributed variables $\eta_{j,t}$ for all $j\in\mathcal{M}_{x}$,
such that the LUR processes $\xi_{j,t}=\rho_{j}^{*}\xi_{j,t-1}+\sum_{k=1}^{p}\Phi_{j,k}\eta_{k,t}$
satisfy 
\[
\sup_{j\in\mathcal{M}_{x},t\in[n]}\left|x_{j,t}-\xi_{j,t}\right|=O_{p}(n^{1/4}\log p),
\]
where $\Phi_{j,k}$ is the $(j,k)$-th entry of the matrix $\Phi$
in (\ref{eq:error_transform}).
\end{lem}
\begin{proof}[Proof of Lemma \ref{lem:LUR=000020Gaussian=000020Approx}]
Note that $x_{j,t}-x_{j,t-1}=e_{j,t}+\dfrac{c_{j}^{*}}{n}x_{j,t-1}$.
Without loss of generality, assume that $x_{j,0}=0.$ Then
\[
x_{j,t}=\sum_{s_{1}=1}^{t}(x_{j,t}-x_{j,t-1})=\sum_{s_{1}=1}^{t}e_{j,s_{1}}+\dfrac{c_{j}^{*}}{n}\sum_{s_{1}=1}^{t}x_{j,s_{1}-1}.
\]
By induction, we have for any \emph{fixed} integer $M$, whenever
$t>M$
\begin{align*}
x_{j,t} & =\sum_{s_{1}=1}^{t}e_{j,s_{1}}+\dfrac{c_{j}^{*}}{n}\sum_{s_{1}=1}^{t}\left(\sum_{s_{2}=2}^{s_{1}}e_{j,s_{2}-1}+\dfrac{c_{j}^{*}}{n}\sum_{s_{2}=2}^{s_{1}}x_{j,s_{2}-2}\right)\\
 & =\sum_{s_{1}=1}^{t}e_{j,s_{1}}+\dfrac{c_{j}^{*}}{n}\sum_{s_{1}=1}^{t}\sum_{s_{2}=2}^{s_{1}}e_{j,s_{2}-1}+\dfrac{c_{j}^{*2}}{n^{2}}\sum_{s=1}^{t}\sum_{s_{2}=2}^{s_{1}}x_{j,s_{2}-2}\\
 & =...=\sum_{\ell=1}^{M}\left(\dfrac{c_{j}^{*}}{n}\right)^{\ell-1}\sum_{s_{1}=1}^{t}\sum_{s_{2}=2}^{s_{1}}\cdots\sum_{s_{\ell}=\ell}^{s_{(\ell-1)}}e_{j,s_{\ell}-\ell+1}+\left(\dfrac{c_{j}^{*}}{n}\right)^{M}\sum_{s_{1}=1}^{t}\sum_{s_{2}=2}^{s_{1}}\cdots\sum_{s_{M}=M}^{s_{(M-1)}}x_{j,s_{M}-M},
\end{align*}
where we define $s_{0}=t$. By Assumption \ref{assu:covMat}, we have
$e_{j,t}=\sum_{k=1}^{p}\Phi_{j,k}\varepsilon_{k,t}$, and therefore
\begin{equation}
x_{j,t}=\sum_{k=1}^{p}\Phi_{j,k}\sum_{\ell=1}^{M}\left(\dfrac{c_{j}^{*}}{n}\right)^{\ell-1}\sum_{s_{1}=1}^{t}\sum_{s_{2}=2}^{s_{1}}\cdots\sum_{s_{\ell}=\ell}^{s_{(\ell-1)}}\varepsilon_{k,s_{\ell}-\ell+1}+\left(\dfrac{c_{j}^{*}}{n}\right)^{M}\sum_{s_{1}=1}^{t}\sum_{s_{2}=2}^{s_{1}}\cdots\sum_{s_{M}=M}^{s_{(M-1)}}x_{j,s_{M}-M}.\label{eq:xjt=000020approx}
\end{equation}
Let $\mathcal{B}_{j}$ denote the Brownian motion in the Gaussian
approximation of Lemma \ref{lem:Gaussian=000020Approximation}, and
define $\eta_{k,t}=\mathcal{B}_{k}(tV_{k}^{*})-\mathcal{B}_{j}((t-1)V_{k}^{*}).$
Then $\{\eta_{j,t}\}$ are i.i.d.~distributed, and 
\begin{equation}
\sup_{k\in[p],t\in[n]}\dfrac{1}{\sqrt{n}}\left|\sum_{s=0}^{t-1}\varepsilon_{k,s}-\sum_{s=0}^{t-1}\eta_{k,s}\right|\lep\dfrac{\log p}{n^{1/4}}.\label{eq:Gaussian=000020Approx=000020eps}
\end{equation}
Let $\xi_{j,t}$ be an LUR satisfying $\xi_{j,t}=\rho_{j}^{*}\xi_{j,t-1}+\sum_{k=1}^{p}\Phi_{j,k}\eta_{k,t}$,
where $\rho_{j}^{*}=1+c_{j}^{*}/n$ is the same as the AR coefficient
of the LUR regressor $x_{j,t}$. Following the same arguments for
(\ref{eq:xjt=000020approx}), we have 
\[
\xi_{j,t}=\sum_{k=1}^{p}\Phi_{j,k}\sum_{\ell=1}^{M}\left(\dfrac{c_{j}^{*}}{n}\right)^{\ell-1}\sum_{s_{1}=1}^{t}\sum_{s_{2}=2}^{s_{1}}\cdots\sum_{s_{\ell}=\ell}^{s_{(\ell-1)}}\eta_{k,s_{\ell}-\ell+1}+\left(\dfrac{c_{j}^{*}}{n}\right)^{M}\sum_{s_{1}=1}^{t}\sum_{s_{2}=2}^{s_{1}}\cdots\sum_{s_{M}=M}^{s_{(M-1)}}\xi_{j,s_{M}-M}.
\]
Thus,
\begin{equation}
x_{j,t}-\xi_{j,t}=A_{j,t}^{(1)}+A_{j,t}^{(2)},\label{eq:A1=000020+=000020A2}
\end{equation}
\[
A_{j,t}^{(1)}:=\sum_{k=1}^{p}\Phi_{j,k}\sum_{\ell=1}^{M}\left(\dfrac{c_{j}^{*}}{n}\right)^{\ell-1}\sum_{s_{1}=1}^{t}\sum_{s_{2}=2}^{s_{1}}\cdots\sum_{s_{\ell}=\ell}^{s_{(\ell-1)}}(\varepsilon_{k,s_{\ell}-\ell+1}-\eta_{k,s_{\ell}-\ell+1}),
\]
\[
A_{j,t}^{(2)}:=\left(\dfrac{c_{j}^{*}}{n}\right)^{M}\sum_{s_{1}=1}^{t}\sum_{s_{2}=2}^{s_{1}}\cdots\sum_{s_{M}=M}^{s_{(M-1)}}(x_{j,s_{M}-M}-\xi_{j,s_{M}-M}).
\]

\textbf{We first bound $A_{j,t}^{(1)}.$} Recall that $|c_{j}^{*}|\leq\bar{C}$
for all $j$ by Assumption \ref{assu:asym_n}, and thus $|c_{j}^{*}|^{\ell-1}\leq\max\{\bar{C}^{M},1\}.$
Therefore, 
\begin{align*}
\left|A_{j,t}^{(1)}\right|\leq & \left|\sum_{k=1}^{p}\Phi_{j,k}\right|\cdot\max\{\bar{C}^{M},1\}\cdot\sup_{k\in[p]}\sum_{\ell=1}^{M}\left(\dfrac{1}{n}\right)^{\ell-1}\sum_{s_{1}=1}^{t}\sum_{s_{2}=2}^{s_{1}}\cdots\sum_{s_{(\ell-1)}=\ell-1}^{s_{(\ell-2)}}\left|\sum_{s_{\ell}=\ell}^{s_{(\ell-1)}}(\varepsilon_{k,s_{\ell}}-\eta_{k,s_{\ell}})\right|.
\end{align*}
The summation ``$\sum_{s_{1}=1}^{t}\sum_{s_{2}=2}^{s_{1}}\cdots\sum_{s_{(\ell-1)}=\ell-1}^{s_{(\ell-2)}}$''
for each $\ell$ involves no more than $n^{\ell-1}$ terms. Therefore,
\[
\left(\dfrac{1}{n}\right)^{\ell-1}\sum_{s_{1}=1}^{t}\sum_{s_{2}=2}^{s_{1}}\cdots\sum_{s_{(\ell-1)}=\ell-1}^{s_{(\ell-2)}}\left|\sum_{s_{\ell}=\ell}^{s_{(\ell-1)}}(\varepsilon_{k,s_{\ell}-\ell+1}-\eta_{k,s_{\ell}-\ell+1})\right|\leq\max_{t\in[n]}\left|\sum_{s=1}^{t}(\varepsilon_{k,s}-\eta_{k,s})\right|,
\]
which implies 
\begin{align}
\sup_{j\in\mathcal{M}_{x},t\in[n]}\left|A_{j,t}^{(1)}\right|\leq & \sup_{j\in[p]}\left|\sum_{k=1}^{p}\Phi_{j,k}\right|\cdot\max\{\bar{c}^{M},1\}\cdot\sup_{k\in[p]}\sum_{\ell=1}^{M}\sup_{t\in[n]}\left|\sum_{s=1}^{t}(\varepsilon_{k,s}-\eta_{k,s})\right|\nonumber \\
\leq & C_{M}\sup_{k\in[p],t\in[n]}\left|\sum_{s=1}^{t}(\varepsilon_{k,s}-\eta_{k,s})\right|,\label{eq:A1=000020bound}
\end{align}
where $C_{M}$ is an absolute constant dependent on the integer $M$
only. 

\textbf{We then bound $A_{j,t}^{(2)}.$} Note that 
\[
\left|A_{j,t}^{(2)}\right|\leq\dfrac{\max\{\bar{C}^{M},1\}}{n^{M}}\sum_{s_{1}=1}^{t}\sum_{s_{2}=2}^{s_{1}}\cdots\sum_{s_{M}=M}^{s_{(M-1)}}\left|x_{j,s_{M}-1}-\xi_{j,s_{M}-1}\right|.
\]
The summations ``$\sum_{s_{1}=1}^{t}\sum_{s_{2}=2}^{s_{1}}\cdots\sum_{s_{M}=M}^{s_{(M-1)}}$''
involve no more than $\binom{t}{M}\leq\binom{n}{M}$ terms, and 
\[
\lim_{n\to\infty}\dfrac{1}{n^{M}}\binom{n}{M}=\dfrac{1}{M!}.
\]
Therefore, when $n$ is large enough, 
\[
\sup_{j\in\mathcal{M}_{x},t\in[n]}\left|A_{j,t}^{(2)}\right|\leq\dfrac{2\max\{\bar{C}^{M},1\}}{M!}\sup_{j\in\mathcal{M}_{x},t\in[n]}|x_{j,t}-\xi_{j,t}|.
\]
Note that the upper bound holds for any fixed $M$. Let $M$ be sufficiently
large so that $\dfrac{2\max\{\bar{C}^{M},1\}}{M!}<0.5$. Then 
\begin{equation}
\sup_{j\in\mathcal{M}_{x},t\in[n]}\left|A_{j,t}^{(2)}\right|\leq0.5\sup_{j\in\mathcal{M}_{x},t\in[n]}|x_{j,t}-\xi_{j,t}|.\label{eq:A2=000020bound}
\end{equation}
By (\ref{eq:A1=000020+=000020A2}), (\ref{eq:A1=000020bound}), (\ref{eq:A2=000020bound}),
\[
\sup_{j\in\mathcal{M}_{x},t\in[n]}|x_{j,t}-\xi_{j,t}|\leq C_{M}\sup_{k\in[p],t\in[n]}\left|\sum_{s=1}^{t}(\varepsilon_{k,s}-\eta_{k,s})\right|+0.5\sup_{j\in\mathcal{M}_{x},t\in[n]}|x_{j,t}-\xi_{j,t}|,
\]
which implies 
\[
\sup_{j\in\mathcal{M}_{x},t\in[n]}|x_{j,t}-\xi_{j,t}|\leq2C_{M}\sup_{k\in[p],t\in[n]}\left|\sum_{s=1}^{t}(\varepsilon_{k,s}-\eta_{k,s})\right|.
\]
Then Lemma \ref{lem:LUR=000020Gaussian=000020Approx} is implied by
Lemma \ref{lem:Gaussian=000020Approximation}.
\end{proof}

\subsubsection{Technical Lemmas for DB and RE \protect\label{subsec:Prop} }

Define the sample Gram matrices of the LURs and stationary regressors
as $\hat{\Sigma}^{(x)}=\frac{1}{n}\sum_{t=1}^{n}X_{t-1}X_{t-1}^{\top}$
and $\hat{\Sigma}^{(z)}=\frac{1}{n}\sum_{t=1}^{n}Z_{t-1}Z_{t-1}^{\top}$,
respectively. The following lemma shows that after standardization,
the Gram matrix of all regressors $\hat{\Sigma}$ can be approximated
by the block-diagonal matrix 
\[
\hat{\Delta}={\rm diag}\left(\hat{\Sigma}^{(z)},\hat{\Sigma}^{(x)}\right).
\]

\begin{lem}
\label{lem:diagonal} Under Assumptions \ref{assu:tail}-\ref{assu:asym_n},
\begin{equation}
\|D^{-1}\left(\hat{\Sigma}-\hat{\Delta}\right)D^{-1}\|_{\infty}=O_{p}\left(\dfrac{(\log p)^{\frac{3}{2}+\frac{1}{2r}}}{\sqrt{n}}\right)\label{eq:diag=000020approx}
\end{equation}
as $n\to\infty,$ where $r$ is specified in Assumption \ref{assu:alpha}.
\end{lem}
\begin{proof}[Proof of Lemma \ref{lem:diagonal}]
By Lemma \ref{lem:sd=000020bound}, 
\begin{equation}
\dfrac{1}{\min_{j\in\mathcal{M}_{x}}\hat{\sigma}_{j}\min_{\ell\in\mathcal{M}_{z}}\hat{\sigma}_{\ell}}\lep\dfrac{1}{\sqrt{n/\log p}}.\label{eq:min=000020sigma=000020j=000020ell}
\end{equation}
Therefore,
\begin{align*}
\|D^{-1}\left(\hat{\Sigma}-\hat{\Delta}\right)D^{-1}\|_{\infty} & \leq\dfrac{\|n^{-1}\sum_{t=1}^{n}x_{t-1}z_{t-1}^{\top}\|_{\infty}}{\min_{j\in\mathcal{M}_{x}}\hat{\sigma}_{j}\min_{\ell\in\mathcal{M}_{z}}\hat{\sigma}_{\ell}}\\
 & \lep\dfrac{(\log p)^{1+\frac{1}{2r}}}{\sqrt{n/\log p}}\\
 & \lep\dfrac{(\log p)^{\frac{3}{2}+\frac{1}{2r}}}{\sqrt{n}}=O_{p}\left(\dfrac{(\log p)^{\frac{3}{2}+\frac{1}{2r}}}{\sqrt{n}}\right),
\end{align*}
where the second row applies (\ref{eq:min=000020sigma=000020j=000020ell})
and (\ref{eq:DB=000020e}).
\end{proof}
The following Lemma establishes the RE condition for LURs without
standardization. We focus on the demeaned regressors $X_{t-1}-\bar{X}$
where $\bar{X}=n^{-1}\sum_{t=1}^{n}X_{t-1}$, since the results will
be helpful to bound the standard deviations used in Slasso. Following
(B.33) and (B.34) of MS, define 
\begin{equation}
C_{m}(L)=\lceil4L^{2}\tilde{C}/\tilde{c}\rceil\text{ and }m=C_{m}s,\label{eq:m=000020=00003D=000020Cms}
\end{equation}
 where $\tilde{C}>\tilde{c}>0$ are absolute constants following the
definitions between (B.33) and (B.34) of MS.
\begin{lem}
\label{lem:RE=000020LUR=000020generic=000020L}Suppose that $(1+C_{m}(L))s=o(n\wedge p)$
as $n\to\infty$, and $k=1.$ Define $\ddot{\Sigma}^{(x)}=n^{-1}\sum_{t=1}^{n}(X_{t-1}-\bar{X})(X_{t-1}-\bar{X})^{\top}.$
Then under Assumptions \ref{assu:tail}-\ref{assu:asym_n},
\begin{equation}
\frac{\kappa_{\mathbf{I}}(\ddot{\Sigma}^{(x)},L,s)}{n}\geq\dfrac{c_{\kappa}}{L^{2}s\log p}\label{eq:RE-unit-1-1}
\end{equation}
\label{lem:RE=000020LUR}holds w.p.a.1.~for any $L\geq1$.\footnote{Here we use a generic $L\geq1$ is useful for deducing the lower bound
of $\hat{\kappa}_{D}$ using $\hat{\kappa}_{\mathbf{I}}$.}
\end{lem}
\begin{proof}[Proof of Lemma \ref{lem:RE=000020LUR=000020generic=000020L}]
To simplify the proof, we assume that the values of all time series at $t=0$ are zeros.

\textbf{(a) We first impose the normality assumption} $\varepsilon_{t}\sim i.i.d.\ \mathcal{N}(0,\mathbf{I}_{p})$.
It implies $e_{t}\sim i.i.d.\ \mathcal{N}(0,\Omega_{e})$ with $\Omega_{e}=\Phi_{e}\Phi_{e}^{\top}.$
Note that for the LUR cases,
\[
X_{t}-X_{t-1}=\dfrac{\mathbf{C}}{n}X_{t-1}+e_{t}
\]
for any $t\geq1$, where $\mathbf{C}={\rm diag}(c_{1}^{*},c_{2}^{*},\dots,c_{p}^{*})$.
Define 
\begin{equation}
e_{t}^{\Delta}=\begin{cases}
\dfrac{\mathbf{C}}{n}X_{t-1}+e_{t}, & t\geq1,\\
0, & t=0,
\end{cases}\label{eq:def=000020eDelta}
\end{equation}
and note that $X_{t}=\sum_{s=1}^{t}e_{s}^{\Delta}.$ Let $R$ be an
$n\times n$ lower triangular matrix of ones on and below the diagonal.
Define $X=(X_{0},X_{1},\dots,X_{n-1})^{\top}$, $e=(e_{0},e_{1},\dots,e_{n-1})^{\top}$
and $e^{\Delta}=(e_{0}^{\Delta},e_{1}^{\Delta},\dots,e_{n-1}^{\Delta})^{\top}$.
Note that $\underset{(n\times p)}{X}=\underset{(n\times n)}{R}\underset{(n\times p)}{e^{\Delta}}$.
Therefore, the Gram matrix of demeaned regressors, $\ddot{\Sigma}^{(x)}$,
can be written as 
\[
\ddot{\Sigma}^{(x)}=n^{-1}e^{\Delta}{}^{\top}R^{\top}(\mathbf{I}_{n}-\mathbf{J}_{n})Re^{\Delta}.
\]
Define $\mathbf{J}_{n}=n^{-1}1_{n}1_{n}^{\top}.$ Let $\lambda_{1}\geq\lambda_{2}\geq\cdots\geq\lambda_{n}\geq0$
and $\widetilde{\lambda}_{1}\geq\widetilde{\lambda}_{2}\geq\cdots\geq\widetilde{\lambda}_{n}\geq0$
be the eigenvalues of $R^{\top}(\mathbf{I}_{n}-\mathbf{J}_{n})R$
and $R^{\top}R$, respectively, ordered from large to small. Let $\mu_{\ell}$
be the $\ell$th largest singular value of the idempotent matrix $\mathbf{I}_{n}-\mathbf{J}_{n}$.
Recall $\boldsymbol{1}\left(\cdot\right)$ is the indicator function,
and obviously $\mu_{\ell}=\boldsymbol{1}(1\leq\ell\leq n-1)$ for
$\ell\in[n]$. When $\ell\in[n-1]$, the first inequality of Eq.(15)
in \citet[Theorem 9]{merikoski2004inequalities} gives $\lambda_{\ell}\geq\widetilde{\lambda}_{\ell+1}\mu_{n-1}=\widetilde{\lambda}_{\ell+1}$. 

Following the technique used to prove Remark 3.5 in \citet{zhang2019identifying},
which is also used for Theorem B.2 in \citet{smeekes2021automated},
we diagonalize $R(\mathbf{I}_{n}-\mathbf{J}_{n})R^{\top}=V{\rm diag}(\lambda_{1},\lambda_{2},\cdots,\lambda_{n})V^{\top}$,
where $V$ is an orthonormal matrix. For any $\delta\in\mathbb{R}^{p}$,
$\delta\neq0$, the quadratic form 
\begin{align}
\delta^{\top}\ddot{\Sigma}^{(x)}\delta & =\dfrac{1}{n}\delta^{\top}e^{\Delta\top}R^{\top}(\mathbf{I}_{n}-\mathbf{J}_{n})Re^{\Delta}\delta=\frac{1}{n}\delta^{\top}e^{\Delta\top}V{\rm diag}(\lambda_{1},\lambda_{2},\cdots,\lambda_{n})V^{\top}e^{\Delta}\delta\nonumber \\
 & \geq\frac{1}{n}\delta^{\top}e^{\Delta\top}V_{\cdot[\ell]}{\rm diag}(\lambda_{1},\cdots,\lambda_{\ell})V_{\cdot[\ell]}^{\top}e^{\Delta}\delta\geq\frac{\lambda_{\ell}}{n}\delta^{\top}e^{\Delta\top}V_{\cdot[\ell]}V_{\cdot[\ell]}^{\top}e^{\Delta}\delta\nonumber \\
 & \geq\frac{\ell\tilde{\lambda}_{\ell+1}}{n}\cdot\delta^{\top}\Gamma_{\ell}^{\Delta}\delta\label{eq:LBeigen}
\end{align}
for any $\ell\in[n-1]$, where $V_{\cdot[\ell]}$ is the submatrix
composed of the first $\ell$ columns of $V$ and $\Gamma_{\ell}^{\Delta}=\ell^{-1}e^{\Delta\top}V_{\cdot[\ell]}V_{\cdot[\ell]}^{\top}e^{\Delta}$. 

We first work with the first factor $\ell\tilde\lambda_{\ell+1}/n$ in (\ref{eq:LBeigen}).
\citet{smeekes2021automated} provide the exact formula of $\tilde{\lambda}_{\ell+1}$:
\begin{equation}
\tilde{\lambda}_{\ell+1}=\left[2\left(1-\cos\left(\dfrac{(2\ell+1)\pi}{2n+1}\right)\right)\right]^{-1}\text{ for all }\ell\in[n].\label{eq:lambda_exact}
\end{equation}
 A Taylor expansion of $\cos\left(x\pi\right)$ around $x=0$ yields
\[
\tilde{\lambda}_{\ell+1}^{-1}=\left(\dfrac{(2\ell+1)\pi}{2n+1}\right)^{2}\left(1+o\left(\frac{\ell+1}{n}\right)\right)=\left(\dfrac{\ell\pi}{n}\right)^{2}\left(1+o\left(\frac{\ell}{n}\right)\right)
\]
whenever $\ell=o\left(n\right)$. This implies 
\begin{equation}
\dfrac{\ell\tilde{\lambda}_{\ell+1}}{n}=\dfrac{n}{\pi^{2}\ell\left(1+o\left(\ell/n\right)\right)}\geq\dfrac{n}{2\pi^{2}\ell}\label{eq:eigen_LB}
\end{equation}
for $\ell=o\left(n\right)$ when $n$ is sufficiently large. 

Next, we focus on the second factor $\delta^{\top}\Gamma_{\ell}^{\Delta}\delta$
in (\ref{eq:LBeigen}). Define $X_{\mathbb{L}}:=\left(0_{p},X_{0},X_{1},\cdots,X_{n-2}\right)^{\top}.$
By definition, we have 
\[
e^{\Delta}=X_{\mathbb{L}}\dfrac{{\rm \mathbf{C}}}{n}+e.
\]
We deduce that 
\begin{align*}
\delta^{\top}\Gamma_{\ell}^{\Delta}\delta & =\dfrac{\delta^{\top}e{}^{\top}V_{\cdot[\ell]}V_{\cdot[\ell]}^{\top}e\delta}{\ell}+\dfrac{\delta^{\top}{\rm \mathbf{C}}X_{\mathbb{L}}^{\top}V_{\cdot[\ell]}V_{\cdot[\ell]}^{\top}X_{\mathbb{L}}{\rm \mathbf{C}}\delta}{n^{2}\ell}+\dfrac{2\delta^{\top}{\rm \mathbf{C}}X_{\mathbb{L}}^{\top}V_{\cdot[\ell]}V_{\cdot[\ell]}^{\top}e\delta}{n\ell}\\
 & \geq\dfrac{\delta^{\top}e{}^{\top}V_{\cdot[\ell]}V_{\cdot[\ell]}^{\top}e\delta}{\ell}-\left|\dfrac{2\delta^{\top}{\rm \mathbf{C}}X_{\mathbb{L}}^{\top}V_{\cdot[\ell]}V_{\cdot[\ell]}^{\top}e\delta}{n}\right|.
\end{align*}
Recall the generic inequality $2|a^{\top}b|\leq a^{\top}a+b^{\top}b$
for any vectors $a$ and $b$ of the same dimension. Let $a=\frac{V_{\cdot[\ell]}^{\top}e\delta}{\sqrt{2}}$
and $b=$ $\sqrt{2}n^{-1}V_{\cdot[\ell]}^{\top}X_{\mathbb{L}}{\rm \mathbf{C}}\delta$,
we have 
\[
\left|\dfrac{2\delta^{\top}{\rm \mathbf{C}}X_{\mathbb{L}}^{\top}V_{\cdot[\ell]}V_{\cdot[\ell]}^{\top}e\delta}{n}\right|\leq0.5\delta^{\top}e{}^{\top}V_{\cdot[\ell]}V_{\cdot[\ell]}^{\top}e\delta+\dfrac{2\delta^{\top}{\rm \mathbf{C}}X_{\mathbb{L}}^{\top}V_{\cdot[\ell]}V_{\cdot[\ell]}^{\top}X_{\mathbb{L}}{\rm \mathbf{C}}\delta}{n^{2}}.
\]
It implies 
\[
\delta^{\top}\Gamma_{\ell}^{\Delta}\delta\geq\dfrac{0.5\delta^{\top}e{}^{\top}V_{\cdot[\ell]}V_{\cdot[\ell]}^{\top}e\delta}{\ell}-\dfrac{2\delta^{\top}{\rm \mathbf{C}}X_{\mathbb{L}}^{\top}V_{\cdot[\ell]}V_{\cdot[\ell]}^{\top}X_{\mathbb{L}}{\rm \mathbf{C}}\delta}{n^{2}\ell}.
\]
In addition, $\lambda_{\max}(V_{\cdot[\ell]}V_{\cdot[\ell]}^{\top})\leq\|V_{\cdot[\ell]}\|_{2}^{2}\leq1$
given $V$ is a unitary matrix. Therefore,
\begin{align}
\delta^{\top}\Gamma_{\ell}^{\Delta}\delta & \geq\dfrac{0.5\delta^{\top}e{}^{\top}V_{\cdot[\ell]}V_{\cdot[\ell]}^{\top}e\delta}{\ell}-\dfrac{2\delta^{\top}{\rm \mathbf{C}}X_{\mathbb{L}}^{\top}X_{\mathbb{L}}{\rm \mathbf{C}}\delta}{n^{2}\ell}\nonumber \\
 & =\dfrac{0.5\delta^{\top}e{}^{\top}V_{\cdot[\ell]}V_{\cdot[\ell]}^{\top}e\delta}{\ell}-\dfrac{2\delta^{\top}{\rm \mathbf{C}}\sum_{t=1}^{n-1}X_{t-1}X_{t-1}^{\top}{\rm \mathbf{C}}\delta}{n^{2}\ell}\nonumber \\
 & \geq\dfrac{0.5\delta^{\top}e{}^{\top}V_{\cdot[\ell]}V_{\cdot[\ell]}^{\top}e\delta}{\ell}-\dfrac{2\delta^{\top}{\rm \mathbf{C}}\sum_{t=1}^{n}X_{t-1}X_{t-1}^{\top}{\rm \mathbf{C}}\delta}{n^{2}\ell}\nonumber \\
 & =\dfrac{0.5\delta^{\top}e{}^{\top}V_{\cdot[\ell]}V_{\cdot[\ell]}^{\top}e\delta}{\ell}-\dfrac{2\delta^{\top}{\rm \mathbf{C}}\hat{\Sigma}^{(x)}{\rm \mathbf{C}}\delta}{n\ell},\label{eq:LB=000020Gamma=000020Delta}
\end{align}
where $\hat{\Sigma}^{(x)}=n^{-1}\sum_{t=1}^{n}X_{t-1}X_{t-1}^{\top}$.
Define $\Gamma_{\ell}=\dfrac{e^{\top}V_{\cdot[\ell]}V_{\cdot[\ell]}^{\top}e}{\ell}.$
By (\ref{eq:LBeigen}), (\ref{eq:eigen_LB}), and (\ref{eq:LB=000020Gamma=000020Delta}),
\begin{align}
\delta^{\top}\ddot{\Sigma}^{(x)}\delta & \ensuremath{\geq\dfrac{n}{2\pi^{2}\ell}\left(0.5\delta^{\top}\Gamma_{\ell}\delta-2n^{-1}\ell^{-1}\delta^{\top}{\rm \mathbf{C}}\hat{\Sigma}^{(x)}{\rm \mathbf{C}}\delta\right).}\label{eq:lower=000020bound=000020quad}
\end{align}
We first lower bound the first term. Let $\ell=(16+C_{\ell})\cdot(s+m)\log p$
for some $C_{\ell}>0$ to be determined later. Following the proof
of (B.43) in MS24 utilizing the non-asymptotic bounds for Whishart
matrices, we have 
\begin{align}
\delta^{\top}\Gamma_{\ell}\delta & \geq C_{\kappa}\|\delta\|_{2}^{2},\label{eq:LB-RE}
\end{align}
w.p.a.1, where the absolute constant $C_{\kappa}$ not dependent on
$L$ or $C_{\ell}$. We then bound the second term in (\ref{eq:lower=000020bound=000020quad}).
Note that for any $\delta\in\mathcal{R}(L,s)$ such that for any $|\mathcal{M}|\leq s$
we have $\|\delta_{\mathcal{M}^{c}}\|_{1}\leq L\|\delta_{\mathcal{M}}\|_{1},$
\begin{equation}
\|\delta\|_{1}\leq\|\delta_{\mathcal{M}}\|_{1}+\|\delta_{\mathcal{M}^{c}}\|_{1}\leq(1+L)\|\delta_{\mathcal{M}}\|_{1}\leq(1+L)\sqrt{s}\|\delta\|_{2}.\label{eq:delta=000020RE}
\end{equation}
Therefore,
\[
\delta^{\top}{\rm \mathbf{C}}\hat{\Sigma}^{(x)}{\rm \mathbf{C}}\delta\leq\left(\|\delta\|_{1}\right)^{2}\cdot\|{\rm \mathbf{C}}\|_{1}^{2}\cdot\|\hat{\Sigma}^{(x)}\|_{\infty}\leq(1+L)^{2}s\|\delta\|_{2}^{2}\cdot\|{\rm \mathbf{C}}\|_{1}^{2}\cdot\|\hat{\Sigma}^{(x)}\|_{\infty}.
\]
Note that $x_{j,t}=\sum_{s=1}^{t}(\rho_{j}^{*})^{s}e_{t-s}$ is a
partial sum of a stationary time series. By Lemma B.2 of MS, we have.
\begin{equation}
\max_{j\in\mathcal{M}_{x},t\in[n]}|x_{j,t}|\lep\sqrt{n\log p}.\label{eq:XK=000020LUR=000020bound}
\end{equation}
Therefore, $\|\hat{\Sigma}^{(x)}\|_{\infty}^{2}\leq\max_{j,t}|x_{j,t-1}|^{2}\leq C_{{\rm sup}}n\log p$
w.p.a.1 for some absolute constant $C_{{\rm sup}}$, which implies
\begin{align}
\delta^{\top}{\rm \mathbf{C}}\ddot{\Sigma}^{(x)}{\rm \mathbf{C}}^{\top}\delta & \leq(1+L)^{2}s\|\delta\|_{2}^{2}\cdot\|{\rm \mathbf{C}}\|_{1}^{2}\cdot\|\hat{\Sigma}^{(x)}\|_{\infty}\nonumber \\
 & \leq(1+L)^{2}s\|\delta\|_{2}^{2}\cdot\bar{C}^{2}\cdot C_{\sup}n\log p\nonumber \\
 & \leq4\bar{C}^{2}L^{2}\cdot s\|\delta\|_{2}^{2}\cdot C_{\sup}n\log p,\label{eq:quad=000020ineq}
\end{align}
where the second inequality applies $\|{\rm \mathbf{C}}\|_{1}\leq\sup_{j\in[p]}|c_{j}^{*}|\leq\bar{C}$,
and the third inequality applies $L\geq1$ and $\|\delta_{\mathcal{S}}\|_{1}\leq\sqrt{s}\|\delta\|_{2}.$
Recall that $\ell=(16+C_{\ell})\cdot(s+m)\log p$. Let $C_{\ell}=4\cdot(1\vee(L\bar{C}^{2}C_{\sup}/C_{\kappa}))-16$,
and recall that $m=\lceil4L\tilde{C}/\tilde{c}\rceil s$ by (\ref{eq:m=000020=00003D=000020Cms}).
Then 
\begin{align}
\ell & =4\cdot(1\vee(L\bar{C}^{2}C_{\sup}/C_{\kappa}))(s+m)\log p\nonumber \\
 & \geq\frac{4\cdot L\bar{C}^{2}C_{\sup}}{C_{\kappa}}\cdot4L\tilde{C}/\tilde{c}\cdot s\log p>\frac{16\cdot L^{2}\bar{C}^{2}C_{\sup}}{C_{\kappa}}\cdot s\log p,\label{eq:ell=000020ineq}
\end{align}
where the last inequality applies the fact that $\tilde{C}>\tilde{c}$.
By (\ref{eq:quad=000020ineq}) and (\ref{eq:ell=000020ineq}), we
deduce that 
\begin{equation}
\delta^{\top}{\rm \mathbf{C}}\ddot{\Sigma}^{(x)}{\rm \mathbf{C}}^{\top}\delta\leq\|\delta\|_{2}^{2}\cdot0.25C_{\kappa}n\cdot\ell.\label{eq:bound=000020quad=0000202}
\end{equation}
Insert (\ref{eq:LB-RE}) and (\ref{eq:bound=000020quad=0000202})
into (\ref{eq:lower=000020bound=000020quad}), we have 
\begin{align*}
\ensuremath{\delta^{\top}\ddot{\Sigma}^{(x)}\delta} & \geq\dfrac{n}{2\pi^{2}\ell}\left(0.5C_{\kappa}-0.25C_{\kappa}\right)\|\delta\|_{2}^{2}=\dfrac{nC_{\kappa}}{8\pi^{2}\ell}\|\delta\|_{2}^{2}.
\end{align*}
By $\ell=(16+C_{\ell})\cdot(s+m)\log p$, $m=\lceil4L\tilde{C}/\tilde{c}\rceil s,$
and $L\geq1$
\begin{align}
\frac{\delta^{\top}\text{\ensuremath{\ddot{\Sigma}^{(x)}}}\delta}{n\|\delta\|_{2}^{2}} & \geq\dfrac{C_{\kappa}}{8\pi^{2}(16+C_{\ell})\cdot(s+m)\log p}\nonumber \\
 & \geq\dfrac{C_{\kappa}}{8\pi^{2}(16+C_{\ell})\cdot(1+\lceil4L\tilde{C}/\tilde{c}\rceil)s\log p}\label{eq:LB-RE2}\\
 & \geq\dfrac{C_{\kappa}}{8\pi^{2}\cdot16\cdot(1\vee(L\bar{C}^{2}C_{\sup}/C_{\kappa}))\cdot(8L\tilde{C}/\tilde{c})\cdot s\log p}\nonumber \\
 & \geq\dfrac{\tilde{c}_{\kappa}}{L^{2}\cdot s\log p}\nonumber 
\end{align}
w.p.a.1, where $\tilde{c}_{k}$ is an absolute constant dependent
on $C_{\kappa},C_{{\rm sup}},\tilde{C},\tilde{c},$ and $\bar{C}$.
Then (\ref{eq:RE-unit-1-1}) holds.

\textbf{(b) We then extend the result to non-normal errors.} Let $\ddot{\Upsilon}=n^{-1}\sum_{t=1}^{n}(\xi_{t-1}-\bar{\xi})(\xi_{t-1}-\bar{\xi})^{\top}$,
where $\xi_{t}=(\xi_{j,t})_{j\in\mathcal{M}_{x}}$is the vector of
LUR processes with normally distributed errors as in Lemma \ref{lem:LUR=000020Gaussian=000020Approx}.
\begin{align}
\delta^{\top}\ddot{\Sigma}^{(x)}\delta & \geq\delta^{\top}\ddot{\Upsilon}\delta-\left|\delta^{\top}(\ddot{\Sigma}^{(x)}-\ddot{\Upsilon})\delta\right|\label{eq:norm_app_1}
\end{align}
Notice that $\hat{\Upsilon}$ is the Gram matrix of the LUR processes
$\zeta_{t}$ with normally distributed errors. The procedures as in
Part (a) bounds the first term on the right-hand side of the above
expression 
\begin{equation}
\delta^{\top}\ddot{\Upsilon}\delta\geq\dfrac{c_{\kappa}^{\prime}}{L^{2}s\log p}n\|\delta\|_{2}^{2}\label{eq:boundmaxtemp5}
\end{equation}
w.p.a.1 for some absolute constant $c_{\kappa}^{\prime}$. We move
on to the second term
\begin{align}
\left|\delta^{\top}(\ddot{\Sigma}^{(x)}-\ddot{\Upsilon})\delta\right| & \leq\|\delta\|_{1}^{2}\|\ddot{\Sigma}^{(x)}-\ddot{\Upsilon}\|_{\infty}\leq(1+L)^{2}s\|\delta\|_{2}^{2}\|\ddot{\Sigma}^{(x)}-\ddot{\Upsilon}\|_{\infty}\nonumber \\
 & \leq4L^{2}s\cdot\|\delta\|_{2}^{2}\|\ddot{\Sigma}^{(x)}-\ddot{\Upsilon}\|_{\infty}\label{eq:boundmaxtemp4}
\end{align}
whenever $L\geq1$. Also, note that 
\begin{eqnarray*}
\|\ddot{\Sigma}^{(x)}-\ddot{\Upsilon}\|_{\text{\ensuremath{\infty}}} & \leq & C_{L}^{2}\left(\|n^{-1}\sum_{t=1}^{n}(X_{t-1}X_{t-1}^{\top}-\xi_{t-1}\xi_{t-1}^{\top})\|_{\infty}+\|\bar{X}\bar{X}^{\top}-\bar{\xi}\bar{\xi}^{\top}\|_{\infty}\right).
\end{eqnarray*}
Following the proof of Part (b) in Proposition B.4 of MS24, we can
show that under Lemma \ref{lem:LUR=000020Gaussian=000020Approx},
\[
\|\ddot{\Sigma}^{(x)}-\ddot{\Upsilon}\|_{\text{\ensuremath{\infty}}}=O_{p}\left(n^{3/4+\nu^{\prime}}\sqrt{\log p}\right)
\]
 for any arbitrary small absolute value $\nu^{\prime}$. Inserting
the above expression into (\ref{eq:boundmaxtemp4}), we have 
\begin{equation}
\frac{|\delta^{\top}(\ddot{\Sigma}^{(x)}-\ddot{\Upsilon})\delta|}{n\|\delta\|_{2}^{2}}\leq4L^{2}s\cdot O_{p}\left(n^{-1/4+\nu^{\prime}}\sqrt{\log p}\right)=o_{p}\left(\frac{L^{-2}}{s\log p}\right)\label{eq:boundmaxtemp6}
\end{equation}
given the condition $s^{2}L^{4}(\log p)^{3/2}=o(n^{1/4-\nu^{\prime}})$
implied by Assumption \ref{assu:asym_n}. (\ref{eq:boundmaxtemp5})
and (\ref{eq:boundmaxtemp6}) then provide 
\begin{align*}
\frac{\delta^{\top}\ddot{\Sigma}^{(x)}\delta}{n\|\delta\|_{2}^{2}} & \geq\dfrac{c_{\kappa}^{\prime}}{L^{2}s\log p}-o_{p}\left(\frac{L^{-2}}{s\log p}\right)\geq\dfrac{c_{\kappa}}{L^{2}s\log p}
\end{align*}
w.p.a.1 when $n$ is large enough, where $c_{\kappa}=0.5c_{\kappa}^{\prime}.$
\end{proof}
The following Lemma establishes the RE condition for LUR regressors
without standardization nor demeaning.
\begin{lem}
\label{lem:RE=000020LUR-1} Under Assumptions \ref{assu:tail}-\ref{assu:asym_n},
w.p.a.1
\begin{equation}
\kappa_{{\bf I}}(\hat{\Sigma}^{(x)},3,s)\geq\dfrac{c_{\kappa}^{(x)}n}{s\cdot\log p}.\label{eq:RE=000020I=000020LUR}
\end{equation}
\end{lem}
\begin{proof}[Proof of Lemma \ref{lem:RE=000020LUR-1}]
 This lemma is a direct corollary of Lemma \ref{lem:RE=000020LUR=000020generic=000020L}
by taking $L=3$ and $c_{\kappa}^{(x)}=c_{\kappa}/9$, and the fact
that $\kappa_{{\bf I}}(\hat{\Sigma}^{(x)},3,s)\geq\kappa_{{\bf I}}(\ddot{\Sigma}^{(x)},3,s)$. 
\end{proof}
The following Lemma establishes the RE condition for stationary regressors
without standardization.
\begin{lem}
\label{lem:RE=000020k} Under Assumptions \ref{assu:tail}-\ref{assu:asym_n},
\begin{equation}
\kappa_{{\bf I}}(\hat{\Sigma}^{(z)},3,s)\geq c_{\kappa}^{(z)}\label{eq:RE=000020ST=000020MI}
\end{equation}
w.p.a.1 for some absolute constant $c_{\kappa}^{(z)}$.
\end{lem}
\begin{proof}[Proof of Lemma \ref{lem:RE=000020k}]
 The proof follows standard arguments using concentration inequalities
for weakly dependent time series, like (B.30) in MS24. 
\end{proof}

The following Lemma delivers the bounds of the standard deviations
used for scaling in Slasso.
\begin{lem}
\label{lem:sd=000020bound} Under Assumptions \ref{assu:tail}-\ref{assu:asym_n},
\end{lem}
\begin{itemize}
\item[(a)] For stationary regressors, there exists some absolute constants $\sigma_{\min}<\sigma_{\max}$,
such that with probability approaching one, 
\begin{equation}
\sigma_{\min}\leq\min_{j\in\mathcal{M}_{z}}\hat{\sigma}_{j}\leq\max_{j\in\mathcal{M}_{z}}\hat{\sigma}_{j}\leq\sigma_{\max}.\label{eq:sd=000020bound=000020ST=000020MI}
\end{equation}
\item[(b)] For nonstationary regressors,
\begin{equation}
\sqrt{\dfrac{n}{\log p}}\lep\min_{j\in\mathcal{M}_{x}}\hat{\sigma}_{j}\leq\min_{j\in\mathcal{M}_{x}}\hat{\sigma}_{j}\lep\sqrt{n\log p}.\label{eq:sd=000020bound=000020LUR}
\end{equation}
\end{itemize}
\begin{proof}[Proof of Lemma \ref{lem:sd=000020bound}]
\textbf{\uline{For Part (a)}}, the proof follows (B.60) and (B.61)
in the appendix of MS24. 

\textbf{\uline{For Part (b)}}, the lower bound follows by 
\[
\min_{j\in\mathcal{M}_{x}}\hat{\sigma}_{j}^{2}\geq\kappa_{{\bf I}}(\ddot{\Sigma}^{(x)},3,1)\geq\dfrac{c_{\kappa}n}{9\log p}
\]
where the second inequality applies Lemma \ref{lem:RE=000020LUR=000020generic=000020L}.
For the upper bound, the LUR regressor $x_{j,t}=\sum_{s=0}^{t}\rho_{j}^{*t-s}e_{j,s}$
is a partial sum of a sub-exponential and $\alpha$-mixing sequence.
Therefore, 
\[
\max_{j\in\mathcal{M}_{x}}\hat{\sigma}_{j}^{2}\leq\max_{j\in\mathcal{M}_{x}}n^{-1}\sum_{t=1}^{n}x_{j,t-1}^{2}\leq\max_{j\in\mathcal{M}_{x},t\in[n]}|x_{j,t}|^{2}\lep n\log p,
\]
where the last inequality applies (\ref{eq:XK=000020LUR=000020bound}).
We complete the proof of Lemma \ref{lem:sd=000020bound}.
\end{proof}

Define $D^{(x)}={\rm diag}(\hat{\sigma}_{j})_{j\in\mathcal{M}_{x}}$
as the diagonal matrix that stores the standard deviations of the
LUR regressors. In the following Lemma \ref{lem:std=000020RE=000020LUR},
we establish a lower bound of RE for standardized LURs.
\begin{lem}
\label{lem:std=000020RE=000020LUR} Under Assumptions \ref{assu:tail}-\ref{assu:asym_n},
there exists an absolute constant $c_{\kappa}^{(1)}$ such that 
\begin{equation}
\kappa_{D^{(x)}}(\hat{\Sigma}^{(x)},3,s)\geq\frac{c_{\kappa}^{(1)}}{s(\log p)^{6}}\label{eq:UnitDB-1}
\end{equation}
 w.p.a.1 as $n\to\infty.$
\end{lem}
\begin{proof}[Proof of Lemma \ref{lem:std=000020RE=000020LUR}]
Define $\hat{\sigma}_{\max}^{(x)}=\max_{j\in\mathcal{M}_{x}}\hat{\sigma}_{j}$,
$\hat{\sigma}_{\min}^{(x)}=\min_{j\in\mathcal{M}_{x}}\hat{\sigma}_{j}$,
and $\hat{\varsigma}^{(x)}=\hat{\sigma}_{\max}^{(x)}/\hat{\sigma}_{\min}^{(x)}$.
Further define $\widetilde{\delta}^{(x)}:=(D^{(x)})^{-1}\delta=(\widehat{\sigma}_{j}^{-1}\delta_{j}){}_{j\in\mathcal{M}_{x}}.$
Note that $\widehat{\sigma}_{\min}\|\widetilde{\delta}_{\mathcal{M}^{c}}\|_{1}\leq\|\delta_{\mathcal{M}^{c}}\|_{1}$
and $\|\delta_{\mathcal{M}}\|_{1}\leq\widehat{\sigma}_{\max}\|\widetilde{\delta}_{\mathcal{M}}\|_{1}$.
Therefore, whenever $\delta\in\mathcal{R}(3,s)$ such that $\|\delta_{\mathcal{M}^{c}}\|_{1}\leq3\|\delta_{\mathcal{M}}\|_{1}$
for any $|\mathcal{M}|\leq s$, we have $\|\widetilde{\delta}_{\mathcal{M}^{c}}\|_{1}\leq3\hat{\varsigma}^{(x)}\|\widetilde{\delta}_{\mathcal{M}}\|_{1}$
and $\widetilde{\delta}\in\mathcal{R}\left(3\widehat{\varsigma}^{(x)},s\right)$.
Then 
\begin{align*}
\kappa_{D^{(x)}}(\hat{\Sigma}^{(x)},3,s) & =\inf_{\delta\in\mathcal{R}(3,s)}\dfrac{\delta^{\top}(D^{(x)})^{-1}\hat{\Sigma}^{(x)}(D^{(x)})^{-1}\delta}{\delta^{\top}\delta}=\inf_{\delta\in\mathcal{R}\left(3,s\right)}\dfrac{\widetilde{\delta}^{\top}\hat{\Sigma}^{(x)}\widetilde{\delta}}{\widetilde{\delta}^{\top}(D^{(x)})^{2}\widetilde{\delta}}\\
 & \geq\inf_{\widetilde{\delta}\in\mathcal{R}\left(3\widehat{\varsigma}^{(x)},s\right)}\dfrac{\widetilde{\delta}^{\top}\hat{\Sigma}^{(x)}\widetilde{\delta}}{\widetilde{\delta}^{\top}(D^{(x)})^{2}\widetilde{\delta}}\geq(\widehat{\sigma}_{{\rm \max}}^{(x)})^{-2}\inf_{\widetilde{\delta}\in\mathcal{R}\left(3\widehat{\varsigma}^{(x)},s\right)}\dfrac{\widetilde{\delta}^{\top}\hat{\Sigma}^{(x)}\widetilde{\delta}}{\widetilde{\delta}^{\top}\widetilde{\delta}}=\dfrac{\kappa(\hat{\Sigma}^{(x)},3\widehat{\varsigma}^{(x)},s)}{(\widehat{\sigma}_{{\rm \max}}^{(x)})^{2}}.
\end{align*}
Taking $L=3\widehat{\varsigma}^{(x)}$. By Lemma \ref{lem:RE=000020LUR=000020generic=000020L}
we have $\kappa(\hat{\Sigma}^{(x)},3\widehat{\varsigma}^{(x)},s)\geq\dfrac{cn}{9s\log p(\widehat{\sigma}_{{\rm \max}}^{(x)})^{2}\cdot(\widehat{\varsigma}^{(x)})^{2}}$
w.p.a.1 for some absolute constant $c$. By (\ref{eq:sd=000020bound=000020LUR}),
there exists some absolute constant $c^{\prime}$ such that 
\[
(\widehat{\varsigma}^{(x)})^{2}\geq c^{\prime}(\log p)^{2}\text{ and }(\widehat{\sigma}_{{\rm \max}}^{(x)})^{2}\leq c^{\prime}n\log p.
\]
Therefore, 
\[
\kappa(\hat{\Sigma}^{(x)},3\widehat{\varsigma}^{(x)},s)\geq\dfrac{cn}{9s(\log p)^{4}(c^{\prime})^{2}}
\]
w.p.a.1. Then Lemma \ref{lem:std=000020RE=000020LUR} follows with
$c_{\kappa}^{(x)}=c/(3c^{\prime})^{2}.$
\end{proof}

The following lemma gives the DB condition. Compared to RE, the DB
condition for mixed roots is more straightforward and not substantially
distinguished from that in MS. 
\begin{lem}
\label{lem:DB=000020mds} Under Assumptions \ref{assu:tail}-\ref{assu:asym_n},
we have 
\begin{equation}
\|n^{-1}\sum_{t=1}^{n}Z_{t-1}u_{t}\|_{\infty}+\dfrac{1}{\sqrt{n}}\|n^{-1}\sum_{t=1}^{n}X_{t-1}u_{t}\|_{\infty}\lep\dfrac{(\log p)^{1+\frac{1}{2r}}}{\sqrt{n}},\label{eq:DB=000020u}
\end{equation}
\begin{equation}
\dfrac{1}{\sqrt{n}}\|n^{-1}\sum_{t=1}^{n}X_{t-1}e_{t-1}^{\top}\|_{\infty}+\dfrac{1}{\sqrt{n}}\|n^{-1}\sum_{t=1}^{n}X_{t-1}Z_{t-1}^{\top}\|_{\infty}\lep\dfrac{(\log p)^{1+\frac{1}{2r}}}{\sqrt{n}}\label{eq:DB=000020e}
\end{equation}
as $n\to\infty.$
\end{lem}
\begin{proof}[Proof of Lemma \ref{lem:DB=000020mds}]
$\|n^{-1}\sum_{t=1}^{n}Z_{t-1}u_{t}\|_{\infty}$ can be bounded following
the proofs in (B.29) of MS24. $\|n^{-1}\sum_{t=1}^{n}X_{t-1}u_{t}\|_{\infty}$,
$\|n^{-1}\sum_{t=1}^{n}X_{t-1}Z_{t-1}^{\top}\|_{\infty}$, and $\|n^{-1}\sum_{t=1}^{n}X_{t-1}e_{t-1}^{\top}\|_{\infty}$
can be bounded following exactly the same procedures in the proof
of MS24's Proposition B.2 about deviation bound (DB) for unit root.
Take $\|n^{-1}\sum_{t=1}^{n}X_{t-1}e_{t-1}^{\top}\|_{\infty}$ as
an example. The essential modification is the expression of $T_{2}$
above MS24's Equation (B.12). It should be changed into
\begin{align*}
T_{2} & =\sup_{k\in\mathcal{M}_{x},j\in[p],t\in[n]}\left|\sum_{t=G+1}^{n}e_{k,t}\sum_{r=t-G+1}^{t-1}\rho_{j}^{^{*}t-1-r}e_{j,r}\right|,
\end{align*}
which can be bounded following the same procedures in MS24. 
\end{proof}

\subsubsection{Proofs of Lemmas \ref{lem:std=000020RE=000020all} and \ref{lem:SlassoError} \protect 
}\begin{proof}[Proof of Lemma \ref{lem:std=000020RE=000020all}]
\uline{Proof of (\mbox{\ref{eq:RE}})}. We have for any $\delta$
\begin{align*}
\delta^{\top}D^{-1}\hat{\Sigma}D{}^{-1}\delta & \geq\delta^{\top}D^{-1}\hat{\Delta}D^{-1}\delta-\|D^{-1}(\hat{\Sigma}-\hat{\Delta})D^{-1}\|_{\infty}\|\delta\|_{1}^{2}.
\end{align*}
Lemmas \ref{lem:std=000020RE=000020LUR} and (\ref{eq:RE=000020ST=000020MI})
suggest that for any $\delta\in\mathcal{R}(3,s)$
\[
\dfrac{\delta^{\top}D^{-1}\hat{\Delta}D^{-1}\delta}{\|\delta\|_{2}^{2}}\geq\dfrac{c}{s(\log p)^{4}}
\]
for some absolute constant $c$. By (\ref{eq:delta=000020RE}) and
Lemma \ref{lem:diagonal},
\begin{align*}
\dfrac{\|D^{-1}(\hat{\Sigma}-\hat{\Delta})D^{-1}\|_{\infty}\|\delta\|_{1}^{2}}{\|\delta\|_{2}^{2}} & =O_{p}\left(\dfrac{(\log p)^{\frac{3}{2}+\frac{1}{2r}}}{\sqrt{n}}\right)\cdot\dfrac{\|\delta\|_{1}^{2}}{\|\delta\|_{2}^{2}}\\
 & =O_{p}\left(\dfrac{s(\log p)^{\frac{3}{2}+\frac{1}{2r}}}{\sqrt{n}}\right).
\end{align*}
Therefore, 
\[
\dfrac{\delta^{\top}D^{-1}\hat{\Sigma}D{}^{-1}\delta}{\|\delta\|_{2}^{2}}\geq\dfrac{c}{s(\log p)^{4}}+O_{p}\left(\dfrac{s(\log p)^{\frac{3}{2}+\frac{1}{2r}}}{\sqrt{n}}\right)\geq\dfrac{0.5c}{s(\log p)^{4}}
\]
when $n$ is sufficiently large, where the second inequality applies
the fact that $s(\log p)^{\frac{3}{2}+\frac{1}{2r}}/\sqrt{n}\gg1/[s(\log p)^{4}]$,
which is implied by Assumption \ref{assu:asym_n}. Then (\ref{eq:RE})
follows with $c_{\kappa}=0.5c.$

\uline{Proof of (\mbox{\ref{eq:DB}})}. The DB condition follows by
\begin{align*}
n^{-1}\|\sum_{t=1}^{n}D^{-1}W_{t-1}u_{t}\|_{\infty} & =\max_{j\in[p]}n^{-1}\left|\sum_{t=1}^{n}\frac{w_{j,t-1}}{\widehat{\sigma}_{j}}u_{t}\right|\\
 & \leq\frac{1}{\min_{j\in\mathcal{M}_{z}}\hat{\sigma}_{j}}\|n^{-1}\sum_{t=1}^{n}Z_{t-1}u_{t}\|_{\infty}+\dfrac{\sqrt{n}}{\min_{j\in\mathcal{M}_{x}}\hat{\sigma}_{j}}\dfrac{1}{\sqrt{n}}\|n^{-1}\sum_{t=1}^{n}X_{t-1}u_{t}\|_{\infty}\\
 & \lep(1+\sqrt{\log p})\dfrac{(\log p)^{1+\frac{1}{2r}}}{\sqrt{n}}\\
 & =\dfrac{(\log p)^{\frac{3}{2}+\frac{1}{2r}}}{\sqrt{n}},
\end{align*}
where the third row applies Lemmas \ref{lem:sd=000020bound} and \ref{lem:DB=000020mds}.
\end{proof}
\begin{proof}[Proof of Lemma \ref{lem:SlassoError}]
By Lemma 1 of MS, our Lemma \ref{lem:std=000020RE=000020all} implies
that 
\begin{equation}
\|D(\hat{\beta}^{{\rm S}}-\beta^{*})\|_{1}=O_{p}\left(\dfrac{s(\log p)^{\frac{3}{2}+\frac{1}{2r}}/\sqrt{n}}{1/s(\log p)^{4}}\right)=O_{p}\left(\dfrac{s^{2}(\log p)^{6+\frac{1}{2r}}}{\sqrt{n}}\right).\label{eq:LASSO=000020L1-1}
\end{equation}
as $n\to\infty$, which verifies Lemma \ref{lem:SlassoError}.
\end{proof}

\subsection{Proofs for Section \ref{subsec:The-Auxiliary-Regression}}

\label{subsec:Proofs-for-Section}

\subsubsection{Local Unit Root Regressors}

We first introduce and prove several technical lemmas. 
\begin{lem}
\label{lem:MI}Under Assumptions \ref{assu:tail}--\ref{assu:asym_n},
for any fixed $j\in\mathcal{M}_{x}$ 
\begin{equation}
\sup_{t\in[n]}|\zeta_{j,t}|\lep n^{\tau/2}(\log p)^{3/2},\label{eq:sup_IV}
\end{equation}
\begin{equation}
\sup_{k\in\mathcal{M}_{x}}\left|\sum_{t=0}^{n-1}e_{k,t}\zeta_{j,t}\right|\lep n(\log p)^{1+\frac{1}{2r}},\label{eq:sup_IV_e}
\end{equation}
\begin{equation}
\sup_{k\in\mathcal{M}_{x}}\left|\sum_{t=1}^{n}x_{k,t-1}\zeta_{j,t-1}\right|\lep n^{1+\tau}(\log p)^{1+\frac{1}{2r}}.\label{eq:sup_X_IV}
\end{equation}
Furthermore, recall that $\hat{\varsigma}_{j}=\sqrt{n^{-1}\sum_{t=1}^{n}(\zeta_{j,t-1}-\bar{\zeta}_{j})^{2}}.$
Then for any $j\in\mathcal{M}_{x}$, 
\begin{equation}
\hat{\varsigma}_{j}^{2}\nep n^{\tau}.\label{eq:sigma=000020IV=000020rate}
\end{equation}
\end{lem}
\begin{proof}[Proof of Lemma \ref{lem:MI}]
 We work on these inequalities one by one. 

\textbf{Proof of (\ref{eq:sup_IV})}. Recall from \eqref{eq:rho=000020zeta}
that $\rho_{\zeta}=1-C_{\zeta}/n^{\tau}$. By (13) in \citet{phillips2009econometric},
when $j\in\mathcal{M}_{x}$ we have 
\begin{equation}
\zeta{}_{j,t}=\zeta_{j,t}^{0}+\frac{c_j^*}{n}\psi_{j,t}^{0}\text{ for }t\geq1,\label{eq:decom=000020IV}
\end{equation}
where 
\[
\zeta_{j,t}^{0}=\sum_{s=1}^{t}\rho_{\zeta}^{s}e_{j,t-s}\text{ and }\psi_{j,t}^{0}=\sum_{s=0}^{t-1}\rho_{\zeta}^{s}x_{j,t-s-1}
\]
 is a partial sum of $\alpha$-mixing sup-exponential components $e_{j,t-s}$
weighted by $\rho_{\zeta}^{s}.$ 

We first bound $\zeta_{j,t}^{0}$. Define 
\[
a_{n}:=\lfloor n^{\tau}(\log p)^{2}\rfloor.
\]
Note that $\rho_{\zeta}^{s}e_{j,t-s}$ is sub-exponential with an
exponentially decaying $\alpha$-mixing coefficient, and thus $\zeta^{0}{}_{j,t}$
is the partial sum of $t$ observations from a sub-exponential and
$\alpha$-mixing time series. MS24's Lemma B.2 yields 
\begin{equation}
\sup_{t\leq a_{n}}|\zeta_{j,t}^{0}|\lep\sqrt{a_{n}\cdot\log p}=O\left[n^{\tau/2}(\log p)^{3/2}\right].\label{eq:bound=000020zeta=0000201}
\end{equation}
In addition, when $t>a_{n}$, 
\begin{equation}
|\zeta_{j,t}^{0}|\leq\left|\sum_{s\leq a_{n}}\rho_{\zeta}^{s}e_{j,t-s}\right|+\left|\sum_{a_{n}<s\leq t}\rho_{\zeta}^{s}e_{j,t-s}\right|\leq\left|\sum_{s\leq a_{n}}\rho_{\zeta}^{s}e_{j,t-s}\right|+\rho_{\zeta}^{a_{n}}\left|\sum_{0<s\leq t-a_{n}}\rho_{\zeta}^{s}e_{j,t-s-a_{n}}\right|.\label{eq:bound=000020zeta=0000202-1}
\end{equation}
By the same arguments for (\ref{eq:bound=000020zeta=0000201}), we
bound the two sums on the right-hand side of (\ref{eq:bound=000020zeta=0000202-1})
by 
\begin{equation}
\sup_{a_{n}<t\leq n}\left|\sum_{s\leq a_{n}}\rho_{\zeta}^{s}e_{j,t-s}\right|\lep n^{\tau/2}(\log p)^{3/2},\label{eq:bound=000020zeta=0000202-2}
\end{equation}
and 
\[
\sup_{a_{n}<t\leq n}\left|\sum_{0<s\leq t-a_{n}}\rho_{\zeta}^{s}e_{j,t-s-a_{n}}\right|\lep\sqrt{(n-a_{n})\cdot\log p}\text{.}
\]
Besides, under the assumption $p\geq n^{\nu_{1}},$ the sequence
\[
\rho_{\zeta}^{a_{n}}=\left(1-C_{\zeta}/n^{\tau}\right)^{\lfloor n^{\tau}(\log p)^{2}\rfloor}=O\left(\exp\left(-C_{\zeta}(\log p)^{2}\right)\right)=O\left(p^{-C_{\zeta}\log p}\right)
\]
converges to zero faster than the reciprocal of any polynomial function
of $n.$ Thus, 
\begin{equation}
\sup_{a_{n}<t\leq n}\rho_{\zeta}^{a_{n}}\left|\sum_{0<s\leq t-a_{n}}\rho_{\zeta}^{s}e_{j,t-s-a_{n}}\right|\lep\rho_{\zeta}^{a_{n}}\sqrt{(n-a_{n})\cdot\log p}=o\left(n^{\tau/2}(\log p)^{3/2}\right).\label{eq:bound=000020zeta=0000202-3}
\end{equation}
By (\ref{eq:bound=000020zeta=0000202-1}), (\ref{eq:bound=000020zeta=0000202-2}),
and (\ref{eq:bound=000020zeta=0000202-3}), it follows that 
\begin{equation}
\sup_{a_{n}<t\leq n}|\zeta_{j,t}^{0}|\lep n^{\tau/2}(\log p)^{3/2}.\label{eq:bound=000020zeta=0000202}
\end{equation}
 By (\ref{eq:bound=000020zeta=0000201}) and (\ref{eq:bound=000020zeta=0000202}),
we have 
\begin{equation}
\sup_{1<t\leq n}|\zeta_{j,t}^{0}|\lep n^{\tau/2}(\log p)^{3/2}.\label{eq:bound=000020zeta=0000203}
\end{equation}

We then bound $\psi_{j,t}^{0}$. By (\ref{eq:XK=000020LUR=000020bound}),
\begin{align}
\sup_{t\in[n]}\left|\psi_{j,t}^{0}\right| & \leq\sum_{s=0}^{n}\rho_{\zeta}^{s}\sup_{t\in[n]}\left|x_{j,t-s-1}\right|\nonumber \\
 & \lep\sqrt{n\log p}\cdot\sum_{s=0}^{n}\rho_{\zeta}^{s}\nonumber \\
 & \leq\sqrt{n\log p}\cdot\frac{1}{1-\rho_{\zeta}}=O(n^{\frac{1}{2}+\tau}\sqrt{\log p}).\label{eq:bound=000020psi0}
\end{align}
Then by (\ref{eq:decom=000020IV}), (\ref{eq:bound=000020zeta=0000203})
and (\ref{eq:bound=000020psi0}),
\[
\sup_{t\in[n]}\left|\zeta_{j,t}\right|\lep n^{\tau/2}(\log p)^{3/2}+n^{\tau-\frac{1}{2}}\sqrt{\log p}=O(n^{\tau/2}(\log p)^{3/2}),
\]
where the second step applies the fact that $\tau\in(0,1)$, which
implies $\tau-\frac{1}{2}<\frac{\tau}{2}$. Then (\ref{eq:sup_IV})
follows. 

\textbf{Proof of (\ref{eq:sup_IV_e})}. Using the result (\ref{eq:sup_IV}),
the proof of (\ref{eq:sup_IV_e}) follows exactly the same procedures
in the proof of MS24's Proposition B.2 about deviation bound (DB)
for unit root. Essential modifications include:
\begin{enumerate}
\item Change the bound $\sup_{t}|x_{j,t}|\lep\sqrt{n\log p}$ in MS24's
Equation (B.12) into 
\[
\sup_{t}|\zeta{}_{j,t}|\lep n^{\tau/2}(\log p)^{3/2},
\]
for a fixed $j$, which has been established in (\ref{eq:sup_IV}).
With this result, we can deduce the same upper bound of $T_{1}$ in
MS24's Equation (B.13).
\item Change the expression of $T_{2}$ above MS24's Equation (B.12) into
\begin{align*}
T_{2} & =\sup_{k\in\mathcal{M}_{x},t\in[n]}\left|\sum_{t=G+1}^{n}e_{k,t}\sum_{r=t-G+1}^{t-1}\rho_{\zeta}^{t-1-r}\Delta x_{j,r}\right|\\
 & \leq\sup_{k\in\mathcal{M}_{x},t\in[n]}\left|\sum_{t=G+1}^{n}e_{k,t}\sum_{r=t-G+1}^{t-1}\rho_{\zeta}^{t-1-r}e_{j,r}\right|+\dfrac{|c_{j}|}{n}\cdot\sup_{k\in\mathcal{M}_{x},t\in[n]}\left|\sum_{t=G+1}^{n}e_{k,t}\sum_{r=t-G+1}^{t-1}\rho_{\zeta}^{t-1-r}x_{j,r-1}\right|,
\end{align*}
where the second step follows the DGP of LURs $x_{j,t}=(1-c_{j}^{*}/n)x_{j,t-1}+e_{j,t}$.
Following the same way as MS24 to bound $T_{2}$, we can bound the
first term by
\begin{equation}
\sup_{k\in\mathcal{M}_{x},t\in[n]}\left|\sum_{t=G+1}^{n}e_{k,t}\sum_{r=t-G+1}^{t-1}\rho_{\zeta}^{t-1-r}e_{j,r}\right|\lep n\log p.\label{eq:T2=000020bound=0000201}
\end{equation}
In addition, it is easy to show that 
\begin{equation}
\sup_{k\in\mathcal{M}_{x},t\in[n]}|e_{k,t}|\lep\log p\label{eq:bound=000020e=000020k}
\end{equation}
given that $e_{k,t}$ is sub-exponential. Therefore, the second term
is bounded by
\begin{align}
n^{-1}\sup_{k\in\mathcal{M}_{x},t\in[n]}\left|\sum_{t=G+1}^{n}e_{k,t}\sum_{r=t-G+1}^{t-1}\rho_{\zeta}^{t-1-r}x_{j,r-1}\right| & \leq\sup_{k\in\mathcal{M}_{x}}n^{-1}\sum_{t=G+1}^{n}\left|e_{k,t}\right|\sup_{k\in\mathcal{M}_{x},t\in[n]}\sum_{r=t-G+1}^{t-1}|x_{j,r-1}|\nonumber \\
 & \leq G\cdot\sup_{k\in\mathcal{M}_{x},t\in[n]}|e_{k,t}|\cdot\sup_{t\in[n]}|x_{j,t}|\\
 & \lep\sqrt{n}(\log p)^{3/2},\label{eq:T2=000020bound=0000202}
\end{align}
where the last step follows (\ref{eq:XK=000020LUR=000020bound}) and
(\ref{eq:bound=000020e=000020k}). By (\ref{eq:T2=000020bound=0000201})
and (\ref{eq:T2=000020bound=0000202}), for our case we also have
\[
T_{2}\lep n\log p.
\]
\item Change the definition of the event $\mathcal{X}_{t}$ below MS24's
Equation (B.16) into 
\[
\mathcal{X}_{t}=\{|\zeta{}_{j,t}|\leq C_{X}\sqrt{n\log p}\}
\]
for some absolute constant $C_{X}$. Equation (\ref{eq:sup_IV}) in
the current paper implies that $\Pr\{\bigcup_{t=1}^{n}\mathcal{X}_{t}^{c}\}\to0$
as $n\to\infty$. Therefore, the arguments below MS24's Equation (B.24)
are still valid. We can thus show that the upper bound of $T_{3}$
defined above into MS24's Equation (B.12) still holds. 
\end{enumerate}
\textbf{Proof of (\ref{eq:sup_X_IV})}. By the decomposition (\ref{eq:decom=000020IV}),
we have 
\begin{equation}
\sup_{k\in\mathcal{M}_{x}}\left|\sum_{t=1}^{n}x_{k,t-1}\zeta_{j,t-1}\right|\leq\sup_{k\in\mathcal{M}_{x}}\left|\sum_{t=1}^{n}x_{k,t-1}\zeta_{j,t-1}^{0}\right|+\dfrac{\bar C}{n}\sup_{k\in\mathcal{M}_{x}}\left|\sum_{t=1}^{n}x_{k,t-1}\psi_{j,t-1}^{0}\right|.\label{eq:x=000020zeta=000020bound}
\end{equation}
We first bound $\sup_{k\in\mathcal{M}_{x}}\left|\sum_{t=1}^{n}x_{k,t-1}\zeta_{j,t-1}^{0}\right|.$
Note that when $j\in\mathcal{M}_{x}$, we have the following AR(1)
representation 
\begin{equation}
\zeta_{j,t}^{0}=\rho_{\zeta}\zeta_{j,t-1}^{0}+e_{j,t}.\label{eq:IV=000020AR1=000020x}
\end{equation}
Thus, for any $k\in\mathcal{M}_{x},$
\[
x_{k,t-1}\zeta_{j,t}^{0}=\rho_{\zeta}x_{k,t-1}\zeta_{j,t-1}^{0}+x_{k,t-1}e_{j,t}
\]
where $\rho_{\zeta}$ is defined in (\ref{eq:rho=000020zeta}). By
$x_{k,t-1}\zeta_{j,t}^{0}=(\rho_{k}^{*})^{-1}(x_{k,t}-e_{k,t})\zeta_{j,t}^{0},$
we obtain 
\begin{equation}
x_{k,t}\zeta_{j,t}^{0}=\rho_{k}^{*}\rho_{\zeta}x_{k,t-1}\zeta_{j,t-1}^{0}+\rho_{k}^{*}x_{k,t-1}e_{j,t}+e_{k,t}\zeta_{j,t}^{0}.\label{eq:temp=000020eq=0000201}
\end{equation}
Summing up both sides of (\ref{eq:temp=000020eq=0000201}) and using
the fact that $\sum_{t=1}^{n}x_{k,t}\zeta_{j,t}^{0}=\sum_{t=1}^{n}x_{k,t-1}\zeta_{j,t-1}^{0}-x_{k,0}\zeta_{j,0}+x_{k,n}\zeta_{j,n}^{0}$,
we deduce  
\[
\sum_{t=1}^{n}x_{k,t-1}\zeta_{j,t-1}^{0}-x_{k,0}\zeta_{j,0}^{0}+x_{k,n}\zeta_{k,n}^{0}=\rho_{k}^{*}\rho_{\zeta}\sum_{t=1}^{n}x_{k,t-1}\zeta_{j,t-1}^{0}+\rho_{k}^{*}\sum_{t=1}^{n}x_{k,t-1}e_{j,t}+\sum_{t=1}^{n}e_{k,t}\zeta_{j,t}^{0}.
\]
It can be further arranged into 
\begin{equation}
(1-\rho_{k}^{*}\rho_{\zeta})\sum_{t=1}^{n}x_{k,t-1}\zeta_{j,t-1}^{0}=\left(x_{k,0}\zeta_{j,0}^{0}-x_{k,n}\zeta_{j,n}^{0}\right)+\rho_{k}^{*}\sum_{t=1}^{n}x_{k,t-1}e_{j,t}+\sum_{t=1}^{n}e_{k,t}\zeta_{j,t}^{0}.\label{eq:X=000020IV=000020decom}
\end{equation}
By (\ref{eq:XK=000020LUR=000020bound}) and (\ref{eq:sup_IV}), 
\begin{equation}
\sup_{k\in\mathcal{M}_{x}}\left|x_{k,0}\zeta_{j,0}^{0}-x_{k,n}\zeta_{j,n}^{0}\right|\lep n^{\frac{1+\tau}{2}}(\log p)^{2}.\label{eq:temp=000020bound=0000201}
\end{equation}
By (\ref{eq:DB=000020e}) and (\ref{eq:sup_IV_e}),
\begin{equation}
\sup_{k\in\mathcal{M}_{x}}\left|\rho_{k}^{*}\sum_{t=1}^{n}x_{k,t-1}e_{j,t}\right|+\sup_{k\in\mathcal{M}_{x}}\left|\sum_{t=1}^{n}e_{k,t}\zeta_{j,t}^{0}\right|\lep n(\log p)^{1+\frac{1}{2r}}.\label{eq:temp=000020bound=0000202}
\end{equation}
By (\ref{eq:X=000020IV=000020decom}), (\ref{eq:temp=000020bound=0000201}),
and (\ref{eq:temp=000020bound=0000202}), 
\begin{align*}
\sup_{k\in\mathcal{M}_{x}}\left|(1-\rho_{k}^{*}\rho_{\zeta})\sum_{t=1}^{n}x_{k,t-1}\zeta_{j,t-1}^{0}\right| & \leq\sup_{k\in\mathcal{M}_{x}}\left|x_{k,0}\zeta_{j,0}^{0}-x_{k,n}\zeta_{j,n}^{0}\right|+\sup_{k\in\mathcal{M}_{x}}\left|\sum_{t=1}^{n}x_{k,t-1}e_{j,t}\right|+\sup_{k\in\mathcal{M}_{x}}\left|\sum_{t=1}^{n}e_{k,t}\zeta_{j,t}^{0}\right|\\
 & \lep n^{\frac{1+\tau}{2}}(\log p)^{2}+n(\log p)^{1+\frac{1}{2r}}\\
 & \leq2n(\log p)^{1+\frac{1}{2r}}.
\end{align*}
Since $\sup_{k\in\mathcal{M}_{x}}\left|\frac{1}{1-\rho_{k}^{*}\rho_{\zeta}}\right|=O(n^{\tau})$,
we have 
\begin{equation}
\sup_{k\in\mathcal{M}_{x}}\left|\sum_{t=1}^{n}x_{k,t-1}\zeta_{j,t-1}^{0}\right|\lec n^{1+\tau}(\log p)^{1+\frac{1}{2r}}.\label{eq:bound=000020x=000020zeta0}
\end{equation}
In addition, 
\begin{align}
\sup_{k\in\mathcal{M}_{x}}\left|\sum_{t=1}^{n}x_{k,t-1}\psi_{j,t-1}^{0}\right| & \leq n\sup_{k\in\mathcal{M}_{x},t\in[n]}|x_{k,t-1}|\cdot\sup_{k\in\mathcal{M}_{x},t\in[n]}|\psi_{j,t-1}^{0}|\nonumber \\
 & \lep n\cdot\sqrt{n\log p}\cdot n^{\frac{1}{2}+\tau}\sqrt{\log p}=n^{2+\tau}\log p,\label{eq:bound=000020x=000020psi0}
\end{align}
where the second step applies (\ref{eq:XK=000020LUR=000020bound})
and (\ref{eq:bound=000020psi0}). By (\ref{eq:x=000020zeta=000020bound}),
(\ref{eq:bound=000020x=000020zeta0}) and (\ref{eq:bound=000020x=000020psi0}),
\[
\sup_{k\in\mathcal{M}_{x}}\left|\sum_{t=1}^{n}x_{k,t-1}\zeta_{j,t-1}\right|\lep n^{1+\tau}(\log p)^{1+\frac{1}{2r}}+\frac{n^{2+\tau}\log p}{n}=O(n^{1+\tau}(\log p)^{1+\frac{1}{2r}}).
\]
We complete the proof of (\ref{eq:sup_X_IV}). 

\textbf{Proof of (\ref{eq:sigma=000020IV=000020rate})}. We have the
following decomposition
\[
\dfrac{\hat{\varsigma}_{j}^{2}}{n^{\tau}}=\dfrac{\sum_{t=1}^{n}\zeta_{j,t-1}^{2}}{n^{1+\tau}}-\dfrac{1}{n^{\tau}}\left(\dfrac{\sum_{t=1}^{n}\zeta_{j,t-1}}{n}\right)^{2}.
\]
By the law of large number (LLN) in \citet[Lemma 3.6(ii)]{phillips2009econometric},
we have 
\begin{equation}
\dfrac{\sum_{t=1}^{n}(\zeta_{j,t-1})^{2}}{n^{1+\tau}}\convp\dfrac{{\rm lvar}(e_{j,t})}{2C_{\zeta}},\label{eq:zeta=0000202=000020lim}
\end{equation}
where ${\rm lvar}(e_{j,t})$ is the long-run variance of $e_{j,t}$.

We then bound $\sum_{t=1}^{n}\zeta_{j,t-1}.$ By (\ref{eq:decom=000020IV}),
we have
\begin{equation}
\sum_{t=1}^{n}\zeta_{j,t-1}=\sum_{t=1}^{n}\zeta_{j,t-1}^{0}+\frac{c_j^*}{n}\sum_{t=1}^{n}\psi_{j,t-1}^{0}.\label{eq:zeta=000020sum=000020decom}
\end{equation}
We first bound $\sum_{t=1}^{n}\zeta_{j,t-1}^{0}$. Without loss of
generality, assume $\zeta_{j,0}^{0}=0$. Summing up both sides of
(\ref{eq:IV=000020AR1=000020x}), we have 
\[
\sum_{t=1}^{n}e_{j,t}=\sum_{t=1}^{n}\zeta_{j,t}^{0}-\rho_{\zeta}\sum_{t=1}^{n}\zeta_{j,t-1}^{0}=\zeta_{j,n}^{0}+(1-\rho_{\zeta})\sum_{t=1}^{n}\zeta_{j,t-1}^{0}.
\]
Since $1-\rho_{\zeta}=C_{\zeta}n^{-\tau},$we have $\sum_{t=1}^{n}\zeta_{j,t-1}^{0}=n^{\tau}C_{\zeta}^{-1}(\sum_{t=1}^{n}e_{j,t}-\zeta_{j,n}^{0}).$
Note that $\sum_{t=1}^{n}e_{j,t}$ is a unit root, and thus $\sum_{t=1}^{n}e_{j,t}=O_{p}(\sqrt{n\log p})$
by (\ref{eq:XK=000020LUR=000020bound}). Also, $|\zeta_{j,n}^{0}|=O_{p}(n^{\tau/2}(\log p)^{3/2})$
by (\ref{eq:sup_IV}). Therefore,
\begin{equation}
\sum_{t=1}^{n}\zeta_{j,t-1}^{0}=O_{p}\left[n^{\tau}\cdot\left(\sqrt{n\log p}+n^{\tau/2}(\log p)^{3/2}\right)\right]=O_{p}\left(n^{\tau+1/2}\sqrt{\log p}\right).\label{eq:sum=000020zeta0=000020bound}
\end{equation}
In addition, by (\ref{eq:bound=000020psi0})
\begin{equation}
\sum_{t=1}^{n}\psi_{j,t-1}^{0}\leq n\cdot\sup_{t\in[n]}|\psi_{j,t-1}^{0}|=O_{p}(n^{3/2+\tau}\sqrt{\log p}).\label{eq:sum=000020psi0=000020bound}
\end{equation}
By (\ref{eq:zeta=000020sum=000020decom}), (\ref{eq:sum=000020zeta0=000020bound}),
and (\ref{eq:sum=000020psi0=000020bound}) we have 
\begin{equation}
\sum_{t=1}^{n}\zeta_{j,t-1}=O_{p}\left(n^{\tau+1/2}\sqrt{\log p}+\frac{n^{3/2+\tau}\sqrt{\log p}}{n}\right)=O_{p}\left(n^{\tau+1/2}\sqrt{\log p}\right),\label{eq:sum=000020IV=000020rate}
\end{equation}
and thus 
\begin{equation}
\dfrac{1}{n^{\tau}}\left(\dfrac{\sum_{t=1}^{n}\zeta_{j,t-1}}{n}\right)^{2}=O_{p}\left(\dfrac{\log p}{n^{1-\tau}}\right)\convp0.\label{eq:zeta=000020mean=000020zeta}
\end{equation}
By (\ref{eq:zeta=0000202=000020lim}) and (\ref{eq:zeta=000020mean=000020zeta}),
\begin{equation}
\dfrac{\hat{\varsigma}_{j}^{2}}{n^{\tau}}=\dfrac{\sum_{t=1}^{n}\zeta_{j,t-1}^{2}}{n^{1+\tau}}-\dfrac{1}{n^{\tau}}\left(\dfrac{\sum_{t=1}^{n}\zeta_{j,t-1}}{n}\right)^{2}\convp\dfrac{{\rm lvar}(e_{j,t})}{2C_{\zeta}}.\label{eq:sigma=000020IV=000020ntau=000020lim}
\end{equation}
This completes the proof of Lemma \ref{lem:MI}.
\end{proof}
With these preparatory lemmas, we will prove Proposition \ref{prop:DB-Aux}
for $j\in\mathcal{M}_{x}$. 
\begin{proof}[Proof of Proposition \ref{prop:DB-Aux} for $j\in\mathcal{M}_{x}$.]
 For simplicity of exposition, define $\tilde{\zeta}_{j}=(\tilde{\zeta}_{j,1},\dots,\tilde{\zeta}_{j,n-1})^{\top}$
and recall that $W_{-j,\cdot}=(W_{-j,1},\dots,W_{-j,n-1})^{\top}.$
By the definition of Slasso, we have
\[
\dfrac{1}{n}\|\tilde{\zeta}_{j}-W_{-j,\cdot}\hat{\varphi}^{(j)}\|_{2}^{2}+\mu\|D_{-j}\hat{\varphi}^{(j)}\|_{1}\leq\dfrac{1}{n}\|\tilde{\zeta}_{j}-W_{-j}\varphi\|_{2}^{2}+\mu\|D_{-j}\varphi\|_{1}
\]
for an arbitrary $(p-1)$-dimensional vector $\varphi$. We can write
the above inequality into
\begin{align}
 & n^{-1}\lVert W_{-j,\cdot}\left(\hat{\varphi}^{(j)}-\varphi\right)\rVert_{2}^{2}+\mu\|D_{-j}\hat{\varphi}^{(j)}\|_{1}\nonumber \\
\leq & \dfrac{2}{n}\|D_{-j}^{-1}W_{-j,\cdot}^{\top}(\tilde{\zeta}_{j}-W_{-j,\cdot}\varphi)\|_{\infty}\|D_{-j}(\hat{\varphi}^{(j)}-\varphi)\|_{1}+\mu\|D_{-j}\varphi\|_{1}.\label{eq:basic=000020ineq}
\end{align}
Write the coefficient vector as $\varphi=(\varphi_{k})_{k\in[p],k\neq j}$,
where $\varphi_{k}$ is the coefficient of $w_{k,t}$. Define $\varphi_{\mathcal{M}_{x}}=(\varphi_{k})_{k\in\mathcal{M}_{x},k\neq j}$
and $\varphi_{\mathcal{M}_{z}}=(\varphi_{k})_{k\in\mathcal{M}_{z}}.$
Take a vector $\varphi$ such that 
\begin{align}
\|\varphi_{\mathcal{M}_{x}}\|_{1} & =\dfrac{C_{\varphi}(\log p)^{\frac{1}{2}+\frac{1}{2r}}}{\sqrt{n^{1+\tau\wedge(1-\tau)}}}\text{ and }\|\varphi_{\mathcal{M}_{z}}\|_{1}=\dfrac{C_{\varphi}(\log p)^{\frac{1}{2r}}}{\sqrt{n^{\tau\wedge(1-\tau)}}}\label{eq:bound=000020tilde=000020phi}
\end{align}
 for a positive number $C_{\varphi}=O(1)$. Without loss of generality,
suppose that $s\geq1$. We will show that 
\begin{equation}
n^{-1}\|D_{-j}^{-1}W_{-j,\cdot}^{\top}(\tilde{\zeta}_{j}-W_{-j,\cdot}\varphi^ {})\|_{\infty}\leq\frac{\mu}{2}\left(1-\frac{1}{2s^{2}}\right).\label{eq:DB_MI}
\end{equation}
Under (\ref{eq:DB_MI}), we can deduce by (\ref{eq:basic=000020ineq})
that 
\begin{align*}
\mu\|D_{-j}\hat{\varphi}^{(j)}\|_{1} & \leq\mu\left[\left(1-\frac{1}{2s^{2}}\right)\|D_{-j}(\hat{\varphi}^{(j)}-\varphi)\|_{1}+\|D_{-j}\varphi\|_{1}\right]\\
 & \leq\mu\left[\left(1-\frac{1}{2s^{2}}\right)\|D_{-j}\hat{\varphi}^{(j)}\|_{1}+2\|D_{-j}\varphi\|_{1}\right],
\end{align*}
which implies $\|D_{-j}\hat{\varphi}^{(j)}\|_{1}\leq4s^{2}\|D_{-j}\varphi\|_{1}.$
In addition, by (\ref{eq:sd=000020bound=000020ST=000020MI}) and (\ref{eq:sd=000020bound=000020LUR})
in Lemma \ref{lem:sd=000020bound}, there exists a absolute constant
$C_{\sigma}$ such that 
\begin{equation}
\max_{k\in\mathcal{M}_{x}}\hat{\sigma}_{k}\leq C_{\sigma}\sqrt{n\log p}\text{ and }\max_{k\in\mathcal{M}_{z}}\hat{\sigma}_{k}\leq C_{\sigma}\label{eq:sigma=000020bound=000020C}
\end{equation}
 w.p.a.1. Therefore, by (\ref{eq:bound=000020tilde=000020phi}) and
(\ref{eq:sigma=000020bound=000020C}), 
\begin{align}
\|D_{-j}\hat{\varphi}^{(j)}\|_{1} & \leq4s^{2}\|D_{-j}\varphi\|_{1}\nonumber \\
 & \leq4s^{2}\left(\max_{k\in\mathcal{M}_{x}}\hat{\sigma}_{k}\cdot\|\varphi_{\mathcal{M}_{x}}\|_{1}+\max_{k\in\mathcal{M}_{z}}\hat{\sigma}_{k}\cdot\|\varphi_{\mathcal{M}_{z}}\|_{1}\right)\nonumber \\
 & \leq\dfrac{4C_{\varphi}\cdot C_{\sigma}s^{2}(\log p)^{1+\frac{1}{2r}}}{\sqrt{n^{\tau\wedge(1-\tau)}}}\label{eq:D=000020hat=000020phi=000020L1}
\end{align}
w.p.a.1, which implies (\ref{eq:local_L1}). 

It is thus sufficient to prove (\ref{eq:DB_MI}). Note that 
\[
\|n^{-1}D_{-j}^{-1}W_{-j,\cdot}^{\top}\tilde{\zeta}_{j}\|_{\infty}\leq(n\hat{\varsigma}_{j})^{-1}\left(\sup_{k\in\mathcal{M}_{x}}\left|\sum_{t=1}^{n}\hat{\sigma}_{k}^{-1}w_{k,t-1}\zeta_{j,t-1}\right|+\sup_{k\in\mathcal{M}_{z}}\left|\sum_{t=1}^{n}\hat{\sigma}_{k}^{-1}w_{k,t-1}\zeta_{j,t-1}\right|\right).
\]
Repeatedly using the bounds in (\ref{eq:sd=000020bound=000020ST=000020MI})
and (\ref{eq:sd=000020bound=000020LUR}) in Lemma \ref{lem:sd=000020bound},
and (\ref{eq:sup_IV_e}), (\ref{eq:sup_X_IV}), and (\ref{eq:sigma=000020IV=000020rate})
in Lemma \ref{lem:MI}, 
\begin{align}
\|n^{-1}D_{-j}^{-1}W_{-j,\cdot}^{\top}\tilde{\zeta}_{j}\|_{\infty} & \lep(n^{1+\tau/2})^{-1}\left(\dfrac{n^{1+\tau}(\log p)^{1+\frac{1}{2r}}}{\sqrt{n/\log p}}+n(\log p)^{1+\frac{1}{2r}}\right)\nonumber \\
 & \leq\ 2(\log p)^{2+\frac{1}{2r}}/\sqrt{n^{(1-\tau)\wedge\tau}}.\label{eq:std=000020W=000020IV}
\end{align}
In addition, 
\begin{align}
 & n^{-1}\|D_{-j}^{-1}W_{-j,\cdot}^{\top}W_{-j,\cdot}\varphi\|_{\infty}\nonumber \\
\leq & \max_{j\in[p],k\in\mathcal{M}_{x}}\left|\frac{1}{n}\sum_{t=1}^{n}\dfrac{w_{j,t}w_{k,t}}{\hat{\sigma}_{j}}\right|\cdot\|\varphi_{\mathcal{M}_{x}}\|_{1}+\max_{j\in[p],k\in\mathcal{M}_{z}}\left|\frac{1}{n}\sum_{t=1}^{n}\dfrac{w_{j,t}w_{k,t}}{\hat{\sigma}_{j}}\right|\cdot\|\varphi_{\mathcal{M}_{z}}\|_{1}\nonumber \\
\leq & \ \max_{j\in[p],k\in\mathcal{M}_{x}}\left|\frac{1}{n}\sum_{t=1}^{n}\dfrac{w_{j,t}w_{k,t}}{\hat{\sigma}_{j}}\right|\cdot\dfrac{C_{\varphi}(\log p)^{\frac{1}{2}+\frac{1}{2r}}}{\sqrt{n^{1+(1-\tau)\wedge\tau}}}+\max_{j\in[p],k\in\mathcal{M}_{z}}\left|\frac{1}{n}\sum_{t=1}^{n}\dfrac{w_{j,t}w_{k,t}}{\hat{\sigma}_{j}}\right|\cdot\dfrac{C_{\varphi}(\log p)^{\frac{1}{2r}}}{\sqrt{n^{(1-\tau)\wedge\tau}}},\label{eq:bound=000020D=000020WW=000020phi}
\end{align}
where the second inequality applies (\ref{eq:bound=000020tilde=000020phi}).
By the sup-exponential distributions of the stationary components,
we can deduce that
\begin{equation}
\sup_{k\in\mathcal{M}_{z},t\in[n]}|w_{k,t}|\lep\log p.\label{eq:sup_z}
\end{equation}
Therefore, 
\begin{align}
\max_{j\in[p],k\in\mathcal{M}_{x}}\left|\frac{1}{n}\sum_{t=1}^{n}\dfrac{w_{j,t}w_{k,t}}{\hat{\sigma}_{j}}\right| & \leq\max_{j\in[p],k\in\mathcal{M}_{x}}\max_{t\in[n]}\left|\dfrac{w_{j,t}w_{k,t}}{\hat{\sigma}_{j}}\right|\nonumber \\
 & \leq\left(\dfrac{\max_{j\in\mathcal{M}_{z},t\in[n]}\left|z_{j,t}\right|}{\inf_{j\in\mathcal{M}_{z}}\hat{\sigma}_{j}}\vee\dfrac{\max_{j\in\mathcal{M}_{x},t\in[n]}\left|x_{j,t}\right|}{\inf_{j\in\mathcal{M}_{x}}\hat{\sigma}_{j}}\right)\cdot\max_{k\in\mathcal{M}_{x},t\in[n]}\left|w_{k,t}\right|\nonumber \\
 & \lep\log p\cdot\sqrt{n\log p}=\sqrt{n(\log p)^{3}},\label{eq:WWmax}
\end{align}
where the last row applies (\ref{eq:XK=000020LUR=000020bound}), (\ref{eq:sd=000020bound=000020ST=000020MI}),
(\ref{eq:sd=000020bound=000020LUR}), and (\ref{eq:sup_z}). Following
similar arguments we can deduce 
\begin{align}
\max_{j\in[p],k\in\mathcal{M}_{z}}\left|\frac{1}{n}\sum_{t=1}^{n}\dfrac{w_{j,t}w_{k,t}}{\hat{\sigma}_{j}}\right| & \leq\max_{j\in[p],k\in\mathcal{M}_{z}}\max_{t\in[n]}\left|\dfrac{w_{j,t}w_{k,t}}{\hat{\sigma}_{j}}\right|\nonumber \\
 & \leq\left(\dfrac{\max_{j\in\mathcal{M}_{z},t\in[n]}\left|w_{j,t}\right|}{\inf_{j\in\mathcal{M}_{z}}\hat{\sigma}_{j}}\vee\dfrac{\max_{j\in\mathcal{M}_{x},t\in[n]}\left|w_{j,t}\right|}{\inf_{j\in\mathcal{M}_{x}}\hat{\sigma}_{j}}\right)\cdot\max_{k\in\mathcal{M}_{z},t\in[n]}\left|w_{k,t}\right|\nonumber \\
 & \lep\log p\cdot\log p=(\log p)^{2}.\label{eq:WWmax-1}
\end{align}
Thus, by (\ref{eq:bound=000020D=000020WW=000020phi}), (\ref{eq:WWmax}),
and (\ref{eq:WWmax-1}), 
\begin{align}
n^{-1}\|D_{-j}^{-1}W_{-j,\cdot}^{\top}W_{-j,\cdot}\varphi\|_{\infty} & \lep\sqrt{n(\log p)^{3}}\cdot\dfrac{(\log p)^{\frac{1}{2}+\frac{1}{2r}}}{\sqrt{n^{1+(1-\tau)\wedge\tau}}}+(\log p)^{2}\cdot\dfrac{(\log p)^{\frac{1}{2r}}}{\sqrt{n^{(1-\tau)\wedge\tau}}}\nonumber \\
 & \leq\dfrac{2(\log p)^{2+\frac{1}{2r}}}{\sqrt{n^{(1-\tau)\wedge\tau}}}.\label{eq:std=000020W=000020W=000020phi}
\end{align}
 (\ref{eq:std=000020W=000020IV}) and (\ref{eq:std=000020W=000020W=000020phi})
yield
\[
n^{-1}\|D_{-j}^{-1}W_{-j,\cdot}^{\top}(\tilde{\zeta}_{j}-W_{-j,\cdot}\varphi)\|_{\infty}\lep\dfrac{(\log p)^{2+\frac{1}{2r}}}{\sqrt{n^{(1-\tau)\wedge\tau}}}.
\]
Therefore, (\ref{eq:DB_MI}) holds as $\mu=C_{{\rm a}}(\log p)^{2+\frac{1}{2r}}/\sqrt{n^{(1-\tau)\wedge\tau}}$
with a sufficiently large $C_{{\rm a}}$. This completes the proof
of Proposition \ref{prop:DB-Aux} for $j\in\mathcal{M}_{x}$.
\end{proof}

\subsubsection{Stationary Regressor}

Recall $\rho_{\zeta}=1-C_{\zeta}/n^{\tau}$ as defined in (\ref{eq:rho=000020zeta}),
and by (23) in \citet{phillips2009econometric} we have the following
decomposition 
\begin{equation}
\zeta_{j,t}=w_{j,t}-\dfrac{C_{\zeta}}{n^{\tau}}\phi_{j,t},\ \phi_{j,t}:=\sum_{s=1}^{t}\rho_{\zeta}^{t-s}w_{j,s-1}.\label{eq:decom=000020zeta=000020z}
\end{equation}
In addition, we define 
\begin{equation}
\eta_{j,t}=\zeta_{j,t}-Z_{-j,t}^{\top}\varphi_{0z}^{(j)*}\label{eq:def=000020eta=000020jt}
\end{equation}
where $\varphi_{0z}^{(j)*}$ is defined in (\ref{eq:phi=000020j=0000200=000020z}).
Compare $\eta_{j,t}$ to its standardized version $\tilde{\eta}_{j,t}$
defined in (\ref{eq:zeta=000020I0=000020pseudo=000020true}), we have
\begin{equation}
\eta_{j,t}=\hat{\varsigma}_{j}\tilde{\eta}_{j,t}.\label{eq:eta=000020tilde=000020eta}
\end{equation}
Finally, define the standardized regressors
\begin{equation}
\tilde{W}_{-j,t}=D_{-j}^{-1}W_{-j,t},\label{eq:std=000020W=000020mj}
\end{equation}
and $\tilde{W}_{-j,\cdot}=(W_{-j,0},W_{-j,1},\dots,W_{-j,n-1})^{\top}.$

\begin{lem}
\label{lem:IV_I0}Under Assumptions \ref{assu:tail}--\ref{assu:phi=000020z},
for $j\in\mathcal{M}_{z}$ 

\begin{equation}
\sup_{t}|\phi_{j,t}|\lep n^{\tau/2}(\log p)^{3/2},\label{eq:=000020bound=000020phi=000020I0}
\end{equation}
\begin{equation}
\hat{\varsigma}_{j}^{2}\nep1,\label{eq:sigma=000020IV=000020rate-1}
\end{equation}
and
\begin{equation}
\lVert n^{-1}\sum_{t=1}^{n}\tilde{W}_{-j,t-1}\tilde{\eta}_{j,t}\rVert_{\infty}\lep n^{-\tau/2}(\log p)^{\frac{3}{2}+\frac{1}{2r}},\label{eq:tilde=000020W=000020tilde=000020eta}
\end{equation}
with $\tilde{\eta}_{j,t}$ defined in (\ref{eq:zeta=000020I0=000020pseudo=000020true}),
and $\tilde{W}_{-j,t}$ defined in (\ref{eq:std=000020W=000020mj}).
\end{lem}
\begin{proof}[Proof of Lemma \ref{lem:IV_I0}]
By definition of $\phi_{j,t}$ in (\ref{eq:decom=000020zeta=000020z}),
we can easily deduce the following recursive formula for $\phi_{j,t}$
that 
\[
\phi_{j,t}=\rho_{\zeta}\phi_{j,t-1}+w_{j,t-1}
\]
Note that $\phi_{j,t-1}\in\sigma(w_{j,0},\dots,w_{j,t-2}).$ Then
$\phi_{j,t}$ is an AR(1) process with coefficient $1-C_{\zeta}/n^{\tau}$
and innovation $w_{j,t-1}$. Recall that $w_{j,t}$ is $\alpha$-mixing
and sub-exponential; we obtain (\ref{eq:=000020bound=000020phi=000020I0})
following the same arguments for (\ref{eq:sup_IV}).

In addition, following the same arguments for (\ref{eq:sigma=000020IV=000020rate}),
we have 
\[
n^{-1}\sum_{t=1}^{n}\phi_{j,t}^{2}\nep n^{\tau}.
\]
By (\ref{eq:sup_IV_e}), the cross-product between a mildly integrated
$\phi_{j,t}$ and a stationary $w_{j,t}$ is bounded by 
\[
\sum_{t=1}^{n}w_{j,t}\phi_{j,t}=o_{p}(n^{1+\tau}).
\]
Thus by (\ref{eq:decom=000020zeta=000020z}), 
\begin{equation}
\dfrac{1}{n}\sum_{t=1}^{n}\zeta_{j,t}^{2}=\dfrac{1}{n}\sum_{t=1}^{n}w_{j,t}^{2}+\dfrac{C_{\zeta}^{2}}{n^{1+2\tau}}\sum_{t=1}^{n}\phi_{j,t}^{2}-\dfrac{2C_{\zeta}}{n^{1+\tau}}\sum_{t=1}^{n}w_{j,t}\phi_{j,t}=\dfrac{1}{n}\sum_{t=1}^{n}w_{j,t}^{2}+o_{p}(1).\label{eq:zeta=0000202=000020w=0000202}
\end{equation}
Furthermore, (\ref{eq:sum=000020IV=000020rate}) implies that the
mildly integrated time series $\phi_{j,t}$ satisfies 
\[
\sum_{t=1}^{n}\phi_{j,t}=O_{p}\left(n^{\tau+1/2}\sqrt{\log p}\right)=o_{p}(n^{1+\tau}).
\]
We thus have by (\ref{eq:decom=000020zeta=000020z}) that 
\begin{equation}
n^{-1}\sum_{t=1}^{n}\zeta_{j,t}=n^{-1}\sum_{t=1}^{n}w_{j,t}+o_{p}(1).\label{eq:zeta=000020w}
\end{equation}
 (\ref{eq:zeta=0000202=000020w=0000202}), (\ref{eq:zeta=000020w}),
and a standard law of large number imply 
\begin{equation}
\hat{\varsigma}_{j}^{2}=n^{-1}\sum_{t=1}^{n}\zeta_{j,t}^{2}-\left(n^{-1}\sum_{t=1}^{n}\zeta_{j,t}\right)^{2}=n^{-1}\sum_{t=1}^{n}w_{j,t}^{2}-\left(n^{-1}\sum_{t=1}^{n}w_{j,t}\right)^{2}+o_{p}(1)\convp{\rm var}\left(w_{j,t}\right),\label{eq:sigma=000020IV=000020I0=000020lim}
\end{equation}
which verifies (\ref{eq:sigma=000020IV=000020rate-1}). 

We then show (\ref{eq:tilde=000020W=000020tilde=000020eta}). Note
that 
\begin{align*}
\|n^{-1}\sum_{t=1}^{n}Z_{-j,t-1}(\zeta_{j,t}-w_{j,t})\|_{\infty} & \stackrel{\mathrm{p}}{\preccurlyeq}\|n^{-(1+\tau)}\sum_{t=1}^{n}Z_{-j,t-1}\phi_{j,t}\|_{\infty}\\
 & \stackrel{\mathrm{p}}{\preccurlyeq}n^{-(1+\tau)}\cdot n\sqrt{\log p}\cdot n^{\tau/2}(\log p)^{3/2}\\
 & =n^{-\tau/2}(\log p)^{2},
\end{align*}
where the first row applies the decomposition (\ref{eq:decom=000020zeta=000020z}),
and the second row applies the bounds (\ref{eq:sup_z}) and (\ref{eq:=000020bound=000020phi=000020I0}).
Also, 
\[
\|n^{-1}\sum_{t=1}^{n}X_{t-1}(\zeta_{j,t}-w_{j,t})\|_{\infty}\stackrel{\mathrm{p}}{\preccurlyeq}\|n^{-(1+\tau)}\sum_{t=1}^{n}X_{t-1}\phi_{j,t}\|_{\infty}\stackrel{\mathrm{p}}{\preccurlyeq}(\log p)^{1+\frac{1}{2r}},
\]
where the second inequality applies the rate in (\ref{eq:sup_X_IV})
of the cross product between a local unit root and a mildly integrated
series. Thus,
\begin{align*}
\|n^{-1}\sum_{t=1}^{n}\tilde{W}_{-j,t-1}(\zeta_{j,t}-w_{j,t})\|_{\infty} & \stackrel{\mathrm{p}}{\preccurlyeq}\dfrac{\|n^{-1}\sum_{t=1}^{n}X_{t-1}(\zeta_{j,t}-w_{j,t})\|_{\infty}}{\inf_{j\in\mathcal{M}_{x}}\hat{\sigma}_{j}}+\dfrac{\|n^{-1}\sum_{t=1}^{n}Z_{-j,t-1}(\zeta_{j,t}-w_{j,t})\|_{\infty}}{\inf_{j\in\mathcal{M}_{z}}\hat{\sigma}_{j}}\\
 & \stackrel{\mathrm{p}}{\preccurlyeq}\dfrac{(\log p)^{\frac{3}{2}+\frac{1}{2r}}}{\sqrt{n}}+\dfrac{(\log p)^{2}}{n^{\tau/2}}\leq\dfrac{2(\log p)^{2}}{n^{\tau/2}}
\end{align*}
with $\tau\in(0,1)$ and $n$ large enough, where the second inequality
applies by (\ref{eq:sd=000020bound=000020ST=000020MI}) and (\ref{eq:sd=000020bound=000020LUR}).
Further by (\ref{eq:tilde=000020W=000020tilde=000020eta}), 
\begin{equation}
\|n^{-1}\sum_{t=1}^{n}\tilde{W}_{-j,t-1}\left(\frac{\zeta_{j,t}-w_{j,t}}{\hat{\varsigma}_{j}}\right)\|_{\infty}\stackrel{\mathrm{p}}{\preccurlyeq}\dfrac{(\log p)^{2}}{n^{\tau/2}}.\label{eq:W=000020eta=000020bound=0000201}
\end{equation}
In addition, define 
\begin{equation}
\eta_{j,t}^{(1)}=w_{j,t}-Z_{-j,t-1}^{\top}\varphi_{0z}^{(j)*}.\label{eq:def=000020eta=0000201}
\end{equation}
By the definition of $\eta_{j,t}$, we have 
\begin{equation}
\eta_{j,t}^{(1)}=(\zeta_{j,t}-w_{j,t})+\eta_{j,t}.\label{eq:eta=000020(1)=000020eta}
\end{equation}
note that the time series $\eta_{j,t}^{(1)}$ is stationary and $\mathbb{E}\left[Z_{-j,t-1}\eta_{j,t}^{(1)}\right]=0$.
We then have 
\[
\|n^{-1}\sum_{t=1}^{n}Z_{-j,t-1}\eta_{j,t}^{(1)}\|_{\infty}\stackrel{\mathrm{p}}{\preccurlyeq}\sqrt{\dfrac{\log p}{n}}
\]
by standard concentration inequalities; e.g. MS24's Eq.~(B.31). Also,
following the same way to prove (\ref{eq:DB=000020e}), we have 
\[
\|n^{-1}\sum_{t=1}^{n}X_{t-1}\eta_{j,t}^{(1)}\|_{\infty}\stackrel{\mathrm{p}}{\preccurlyeq}(\log p)^{1+\frac{1}{2r}}.
\]
Thus, 
\begin{align}
\|n^{-1}\sum_{t=1}^{n}\tilde{W}_{-j,t-1}\eta_{j,t}^{(1)}\|_{\infty} & \stackrel{\mathrm{p}}{\preccurlyeq}\dfrac{\|n^{-1}\sum_{t=1}^{n}X_{t-1}\eta_{j,t}^{(1)}\|_{\infty}}{\inf_{j\in\mathcal{M}_{x}}\hat{\sigma}_{j}}+\dfrac{\|n^{-1}\sum_{t=1}^{n}Z_{-j,t-1}\eta_{j,t}^{(1)}\|_{\infty}}{\inf_{j\in\mathcal{M}_{z}}\hat{\sigma}_{j}}\nonumber \\
 & \stackrel{\mathrm{p}}{\preccurlyeq}\dfrac{(\log p)^{\frac{3}{2}+\frac{1}{2r}}}{\sqrt{n}}+\sqrt{\dfrac{\log p}{n}}\leq\dfrac{2(\log p)^{\frac{3}{2}+\frac{1}{2r}}}{\sqrt{n}}.\label{eq:W=000020tilde=000020eta=000020bound}
\end{align}
Furthermore, by (\ref{eq:sigma=000020IV=000020rate-1}) and (\ref{eq:W=000020tilde=000020eta=000020bound})
we deduce that 
\begin{align}
\|n^{-1}\sum_{t=1}^{n}\tilde{W}_{-j,t-1}\left(\frac{\eta_{j,t}^{(1)}}{\hat{\varsigma}_{j}}\right)\|_{\infty} & \stackrel{\mathrm{p}}{\preccurlyeq}\dfrac{(\log p)^{\frac{3}{2}+\frac{1}{2r}}}{\sqrt{n}}\leq\dfrac{(\log p)^{2}}{n^{\tau/2}}.\label{eq:W=000020eta=000020bound=0000202}
\end{align}
By (\ref{eq:eta=000020(1)=000020eta}) and (\ref{eq:eta=000020tilde=000020eta}),
we have the following decomposition 
\[
\tilde{\eta}_{j,t}=\frac{\zeta_{j,t}-w_{j,t}}{\hat{\varsigma}_{j}}+\frac{\eta_{j,t}^{(1)}}{\hat{\varsigma}_{j}}.
\]
 Then by the triangular inequality, we have 
\begin{align*}
\|n^{-1}\sum_{t=1}^{n}\tilde{W}_{-j,t-1}\tilde{\eta}_{j,t}\|_{\infty} & \leq\|n^{-1}\sum_{t=1}^{n}\tilde{W}_{-j,t-1}\left(\frac{\zeta_{j,t}-w_{j,t}}{\hat{\varsigma}_{j}}\right)\|_{\infty}+\|n^{-1}\sum_{t=1}^{n}\tilde{W}_{-j,t-1}\left(\frac{\eta_{j,t}^{(1)}}{\hat{\varsigma}_{j}}\right)\|_{\infty}\\
 & \stackrel{\mathrm{p}}{\preccurlyeq}\dfrac{(\log p)^{2}}{n^{\tau/2}}.
\end{align*}
where the second row applies (\ref{eq:W=000020eta=000020bound=0000201})
and (\ref{eq:W=000020eta=000020bound=0000202}). We complete the proof
of Lemma \ref{lem:IV_I0}. 
\end{proof}
\begin{proof}[Proof of Proposition \ref{prop:DB-Aux} for $j\in\mathcal{M}_{z}$.]
According to MS24's Proposition 3(c), the RE is bounded from below
by $\frac{c}{s(\log p)^{4}}$ for some absolute constant $c$ w.p.a.1.
In addition, by (\ref{eq:tilde=000020W=000020tilde=000020eta}), the
DB is $O_{p}\left(\dfrac{\log p}{n^{\tau/2}}\right)$ when $\mu$
follows the order in Proposition \ref{prop:DB-Aux}. Then by MS24's
Lemma 1
\[
\|D_{-j}(\hat{\varphi}^{(j)}-\varphi^{(j)*})\|_{1}\stackrel{\mathrm{p}}{\preccurlyeq}s\cdot\dfrac{(\log p)^{2+\frac{1}{2r}}}{\sqrt{n^{\tau\wedge(1-\tau)}}}/(s(\log p)^{4})=\dfrac{s^{2}(\log p)^{6+\frac{1}{2r}}}{\sqrt{n^{\tau\wedge(1-\tau)}}}.
\]
This completes the proof of Proposition \ref{prop:DB-Aux} for $j\in\mathcal{M}_{z}$. 
\end{proof}

\subsection{Proofs for Section \ref{subsec:Asymptotic-Normality} }

\label{subsec:Proofs-for-Section-normal}

\subsubsection{Technical Lemmas}
\begin{lem}
\label{lem:s.e.consis}Suppose Assumptions \ref{assu:tail}-\ref{assu:phi=000020z}
hold. Then for any $j\in[p]$ , 
\begin{equation}
\dfrac{\left|\sum_{t=1}^{n}\hat{r}_{j,t-1}w_{j,t-1}\right|}{\sum_{t=1}^{n}\hat{r}_{j,t-1}w_{j,t-1}}\cdot\dfrac{\sigma_{u}}{\hat{\sigma}_{u}}\convd{\rm sgn}(G_{j}^{*})\label{eq:lim=000020sgn}
\end{equation}
where 
\begin{equation}
G_{j}^{*}=\begin{cases}
\dfrac{1}{C_{\zeta}}\left({\rm lvar}\left(e_{j,t}\right)+\int_{0}^{1}\mathcal{U}_{j}(r)d\mathcal{U}_{j}(r)\right), & j\in\mathcal{M}_{x},\\
{\rm cov}\left(w_{j,t},\eta_{j,t}^{(1)}\right), & j\in\mathcal{M}_{z},
\end{cases}\label{eq:express=000020G=000020star}
\end{equation}
with ${\rm lvar}\left(e_{j,t}\right)$ being the long-run variance
of $e_{j,t}$, $\mathcal{U}_{j}(r)=\int_{0}^{1}{\rm e}^{c_{j}^{*}(r-s)}{\rm d}\mathcal{B}_{j}(s)$
being an OU process, $\mathcal{B}_{j}$ being the Brownian motion
of variance ${\rm lvar}\left(e_{j,t}\right)$, and $\eta_{j,t}^{(1)}$
defined in (\ref{eq:def=000020eta=0000201}). In addition, 
\begin{equation}
\dfrac{1}{n}\sum_{t=1}^{n}\hat{r}_{j,t-1}^{2}\convp H_{j}\label{eq:r=000020hat=0000202=000020lim}
\end{equation}
where 
\[
H_{j}=\begin{cases}
1, & j\in\mathcal{M}_{x},\\
\frac{{\rm var}(\eta_{j,t}^{(1)})}{{\rm var}\left(w_{j,t}\right)} & j\in\mathcal{M}_{z}.
\end{cases}
\]
\end{lem}
\begin{proof}[Proof of Lemma \ref{lem:s.e.consis}]
We first prove (\ref{eq:lim=000020sgn}). The first step is to show
\begin{equation}
\dfrac{\sigma_{u}}{\hat{\sigma}_{u}}\convp1.\label{eq:consistent=000020Omega=000020u}
\end{equation}
By definition of the Slasso residual, 
\[
\hat{u}_{t}=y_{t}-W_{t-1}^{\top}\hat{\theta}^{{\rm S}}=u_{t}+W_{t-1}^{\top}(\theta^{*}-\hat{\theta}^{{\rm S}}).
\]
Thus,
\begin{equation}
\hat{\sigma}_{u}^{2}=\dfrac{1}{n}\sum_{t=1}^{n}\hat{u}_{t}^{2}=\dfrac{1}{n}\sum_{t=1}^{n}u_{t}^{2}+\dfrac{1}{n}\sum_{t=1}^{n}(W_{t-1}^{\top}(\theta^{*}-\hat{\theta}^{{\rm S}}))^{2}+\dfrac{2}{n}\sum_{t=1}^{n}u_{t}W_{t-1}^{\top}(\theta^{*}-\hat{\theta}^{{\rm S}}).\label{eq:omega=000020hat=000020u=0000202=000020decom}
\end{equation}
By MS24's Theorem 3, we have 
\begin{equation}
\dfrac{1}{n}\sum_{t=1}^{n}(W_{t-1}^{\top}(\theta^{*}-\hat{\theta}^{{\rm S}}))^{2}=\dfrac{1}{n}\|W(\theta^{*}-\hat{\theta}^{{\rm S}})\|_{2}^{2}\convp0.\label{eq:forecast=000020lim}
\end{equation}
By the Cauchy-Schwartz inequality, 
\begin{equation}
\dfrac{2}{n}\sum_{t=1}^{n}u_{t}W_{t-1}^{\top}(\theta^{*}-\hat{\theta}^{{\rm S}})\leq2\sqrt{\dfrac{1}{n}\sum_{t=1}^{n}u_{t}^{2}}\sqrt{\dfrac{1}{n}\sum_{t=1}^{n}(W_{t-1}^{\top}(\theta^{*}-\hat{\theta}^{{\rm S}}))^{2}}\convp0.\label{eq:C-S=000020u2=000020forecast=000020lim}
\end{equation}
Combining (\ref{eq:omega=000020hat=000020u=0000202=000020decom}),
(\ref{eq:forecast=000020lim}), and (\ref{eq:C-S=000020u2=000020forecast=000020lim}),
we have 
\begin{equation}
\hat{\sigma}_{u}^{2}=\dfrac{1}{n}\sum_{t=1}^{n}u_{t}^{2}+o_{p}(1).\label{eq:uhat2=000020u2}
\end{equation}
Using a standard law of large number, we deduce 
\begin{equation}
\dfrac{1}{n}\sum_{t=1}^{n}u_{t}^{2}\convp\sigma_{u}^{2}.\label{eq:u=0000202=000020lim}
\end{equation}
(\ref{eq:uhat2=000020u2}) and (\ref{eq:u=0000202=000020lim}) imply
that $\hat{\sigma}_{u}^{2}\convp\sigma_{u}^{2}$. Then for (\ref{eq:lim=000020sgn})
it suffices to show 
\begin{equation*}
\left|\sum_{t=1}^{n}\hat{r}_{j,t-1}w_{j,t-1}\right|/\sum_{t=1}^{n}\hat{r}_{j,t-1}w_{j,t-1}\convd{\rm sgn}(G_{j}^{*}).
\end{equation*} 
In the following, we will show (\ref{eq:lim=000020sgn}) and (\ref{eq:r=000020hat=0000202=000020lim}) for $j\in\mathcal{M}_x$ and $j\in\mathcal{M}_z$, respectively.

\textbf{CASE I: }$j\in\mathcal{M}_{x}$. Define 
\begin{equation}
\check{r}_{j,t}=\hat{r}_{j,t}\hat{\varsigma}_{j},\label{eq:r=000020check=000020r=000020hat}
\end{equation}
 and by definition of $\hat{r}_{j,t}$ in (\ref{eq:def-rhat}) we
have 
\begin{equation}
\check{r}_{j,t}=\zeta_{j,t}-\hat{\varsigma}_{j}W_{-j,t}^{\top}\hat{\varphi}^{(j)}.\label{eq:check=000020r=000020def}
\end{equation}
Then 
\begin{equation}
\sum_{t=1}^{n}\check{r}_{j,t-1}w_{j,t-1}=\sum_{t=1}^{n}\zeta_{j,t-1}w_{j,t-1}-\hat{\varsigma}_{j}\sum_{t=1}^{n}w_{j,t-1}W_{-j,t-1}^{\top}\hat{\varphi}^{(j)}.\label{eq:check=000020r=000020w}
\end{equation}
Note that $w_{j,t}$ is unit root and $\zeta_{j,t}$ is the IV. By
the functional CLT for the case of local unit roots in \citet[Lemma 3.2]{phillips2016robust},
we have 
\begin{equation}
\dfrac{1}{n^{1+\tau}}\sum_{t=1}^{n}\zeta_{j,t-1}w_{j,t-1}\convd G_{j}^{*}:=\dfrac{1}{C_{\zeta}}\left({\rm lvar}(e_{j,t})+\int_{0}^{1}\mathcal{U}_{j}d\mathcal{U}_{j}\right).\label{eq:weak=000020lim=000020zeta=000020w}
\end{equation}
By the bound of $\min_{k\in[p]}\hat{\sigma}_{k}$ by (\ref{eq:sd=000020bound=000020ST=000020MI})
and (\ref{eq:sd=000020bound=000020LUR}), the bound of LUR processes
(\ref{eq:XK=000020LUR=000020bound}), and the bound of a stationary
component (\ref{eq:sup_z}), we have 
\begin{equation}
\sup_{t\in[n]}\|D_{-j}^{-1}W_{-j,t-1}\|_{\infty}\leq\dfrac{\sup_{k\in\mathcal{M}_{x}}|w_{k,t}|}{\min_{k\in\mathcal{M}_{x}}\hat{\sigma}_{k}}+\dfrac{\sup_{k\in\mathcal{M}_{z}}|w_{k,t}|}{\min_{k\in\mathcal{M}_{z}}\hat{\sigma}_{k}}=O_{p}(\log p).\label{eq:bound=000020D=000020W}
\end{equation}
Note that in (\ref{eq:D=000020hat=000020phi=000020L1}), we allow
for a small $C_{\varphi}$ that shrinks to zero as $n\to\infty$.
Let $C_{\varphi}=1/(4C_{\sigma}s^{2}n^{1-\tau/2})$ with the absolute
constant $C_{\sigma}$ in (\ref{eq:D=000020hat=000020phi=000020L1}).
Then 
\begin{equation}
\|D_{-j}\hat{\varphi}^{(j)}\|_{1}\leq\dfrac{4C_{\sigma}C_{\varphi}s^{2}(\log p)^{\frac{1}{2}+\frac{1}{2r}}}{\sqrt{n^{(1-\tau)\wedge\tau}}}\leq\dfrac{(\log p)^{\frac{1}{2}+\frac{1}{2r}}}{n^{1-\tau/2}\cdot\sqrt{n^{(1-\tau)\wedge\tau}}},\label{eq:D=000020hat=000020phi=000020L1=000020smaller}
\end{equation}
where the first inequality applies (\ref{eq:D=000020hat=000020phi=000020L1}).
Therefore, 
\begin{align*}
\left|n^{-(1+\tau)}\hat{\varsigma}_{j}\sum_{t=1}^{n}w_{j,t-1}W_{-j,t-1}^{\top}\hat{\varphi}^{(j)}\right| & \leq\hat{\varsigma}_{j}\lVert n^{-(1+\tau)}\sum_{t=1}^{n}D_{-j}^{-1}W_{-j,t-1}w_{j,t-1}\rVert_{\infty}\cdot\|D_{-j}\hat{\varphi}^{(j)}\|_{1}\\
 & \leq\hat{\varsigma}_{j}\lVert n^{-(1+\tau)}\sum_{t=1}^{n}D_{-j}^{-1}W_{-j,t-1}w_{j,t-1}\rVert_{\infty}\cdot\dfrac{4C_{\sigma}C_{\varphi}s^{2}(\log p)^{\frac{1}{2}+\frac{1}{2r}}}{\sqrt{n^{(1-\tau)\wedge\tau}}}\\
 & \leq\dfrac{\hat{\varsigma}_{j}\sup_{t\in[n]}\|D_{-j}^{-1}W_{-j,t-1}\|_{\infty}|w_{j,t-1}|}{n^{\tau}}\cdot\dfrac{(\log p)^{\frac{1}{2}+\frac{1}{2r}}}{n^{1-\tau/2}\cdot\sqrt{n^{(1-\tau)\wedge\tau}}}\\
 & \lep\dfrac{n^{\tau/2}\cdot\sqrt{n}(\log p)^{\frac{3}{2}}}{n^{\tau}\cdot}\cdot\dfrac{(\log p)^{\frac{1}{2}+\frac{1}{2r}}}{n^{1-\tau/2}\cdot\sqrt{n^{(1-\tau)\wedge\tau}}}\leq\dfrac{(\log p)^{\frac{3}{2}+\frac{1}{2r}}}{\sqrt{n^{(1-\tau)\wedge\tau}}}\to0,
\end{align*}
where the fourth inequality applies $\hat{\varsigma}_{j}=O_{p}(n^{\tau/2})$
in (\ref{eq:sigma=000020IV=000020rate}), and the bound of $\sup_{t\in[n]}\|D_{-j}^{-1}W_{-j,t-1}\|_{\infty}$
in (\ref{eq:bound=000020D=000020W}). Thus, 
\begin{equation}
\dfrac{1}{n^{1+\tau}}\sum_{t=1}^{n}\check{r}_{j,t-1}w_{j,t-1}\convd G_{j}^{*}.\label{eq:check=000020r=000020w=000020weak=000020lim}
\end{equation}
By $\check{r}_{j,t}=\hat{r}_{j,t}\hat{\varsigma}_{j}$ and the continuous
mapping theorem, 
\begin{equation}
\dfrac{\left|\sum_{t=1}^{n}\hat{r}_{j,t-1}w_{j,t-1}\right|}{\sum_{t=1}^{n}\hat{r}_{j,t-1}w_{j,t-1}}=\dfrac{\left|\dfrac{1}{n^{1+\tau}}\sum_{t=1}^{n}\check{r}_{j,t-1}w_{j,t-1}\right|}{\dfrac{1}{n^{1+\tau}}\sum_{t=1}^{n}\check{r}_{j,t-1}w_{j,t-1}}\convd\dfrac{\left|G_{j}^{*}\right|}{G_{j}^{*}}={\rm sgn}(G_{j}^{*}).\label{eq:lim=000020sgn=000020Gj}
\end{equation}
Then (\ref{eq:lim=000020sgn}) is implied by (\ref{eq:consistent=000020Omega=000020u})
and (\ref{eq:lim=000020sgn=000020Gj}).

For (\ref{eq:r=000020hat=0000202=000020lim}), we have $\check{r}_{j,t}^{2}=\zeta_{j,t}^{2}+\hat{\varsigma}_{j}^{2}(W_{-j,t}^{\top}\hat{\varphi}^{(j)})^{2}-2\hat{\varsigma}_{j}\zeta_{j,t}W_{-j,t}^{\top}\hat{\varphi}^{(j)}$
. When $j\in\mathcal{M}_{x}$,
\begin{align}
\left|\dfrac{1}{n}\sum_{t=1}^{n}\check{r}_{j,t-1}^{2}-\frac{1}{n}\sum_{t=1}^{n}\zeta_{j,t-1}^{2}\right| & \leq\hat{\varsigma}_{j}^{2}\|D_{-j}\hat{\varphi}^{(j)}\|_{1}^{2}\cdot\|n^{-1}\tilde{W}_{-j,\cdot}^{\top}\tilde{W}_{-j,\cdot}\|_{\infty}\nonumber \\
 & \ \ \ \ +2\hat{\varsigma}_{j}\|D_{-j}\hat{\varphi}^{(j)}\|_{1}\cdot\|\dfrac{1}{n}\sum_{t=1}^{n}D_{-j}^{-1}W_{-j,t-1}\zeta_{j,t}\|_{\infty}.\label{eq:r=000020check=0000202=000020bound=000020temp}
\end{align}
By (\ref{eq:sd=000020bound=000020ST=000020MI}), (\ref{eq:sd=000020bound=000020LUR}),
(\ref{eq:WWmax}), and (\ref{eq:WWmax-1}), 
\begin{align}
\|n^{-1}\tilde{W}_{-j,\cdot}^{\top}\tilde{W}_{-j,\cdot}\|_{\infty} & \leq\max_{j\in[p],k\in\mathcal{M}_{x}}\left|\frac{1}{n}\sum_{t=1}^{n}\dfrac{w_{j,t}w_{k,t}}{\hat{\sigma}_{j}}\right|\dfrac{1}{\min_{k\in\mathcal{M}_{z}}\hat{\sigma}_{k}}+\max_{j\in[p],k\in\mathcal{M}_{z}}\left|\frac{1}{n}\sum_{t=1}^{n}\dfrac{w_{j,t}w_{k,t}}{\hat{\sigma}_{j}}\right|\dfrac{1}{\min_{k\in\mathcal{M}_{z}}\hat{\sigma}_{k}}\nonumber \\
 & =O_{p}\left((\log p)^{2}\right).\label{eq:temp=000020check=0000201}
\end{align}
Furthermore, 
\begin{align}
\|\dfrac{1}{n}\sum_{t=1}^{n}D_{-j}^{-1}W_{-j,t-1}\zeta_{j,t}\|_{\infty} & \leq\dfrac{\|\dfrac{1}{n}\sum_{t=1}^{n}X_{-j,t-1}\zeta_{j,t}\|_{\infty}}{\inf_{j\in\mathcal{M}_{x}}\hat{\sigma}_{j}}+\dfrac{\|\dfrac{1}{n}\sum_{t=1}^{n}Z_{-j,t-1}\zeta_{j,t}\|_{\infty}}{\inf_{j\in\mathcal{M}_{z}}\hat{\sigma}_{j}}\nonumber \\
 & \lep n^{\tau-1/2}(\log p)^{\frac{3}{2}+\frac{1}{2r}}+(\log p)^{1+\frac{1}{2r}}\nonumber \\
 & =O\left(n^{\tau-1/2}(\log p)^{\frac{3}{2}+\frac{1}{2r}}\right)=o_{p}(n^{\tau/2}).\label{eq:temp=000020check=0000203}
\end{align}
where the second inequality applies (\ref{eq:sup_IV_e}), (\ref{eq:sup_X_IV}),
(\ref{eq:sd=000020bound=000020ST=000020MI}), and (\ref{eq:sd=000020bound=000020LUR}).
Combining (\ref{eq:r=000020check=0000202=000020bound=000020temp}),
(\ref{eq:temp=000020check=0000201}), (\ref{eq:temp=000020check=0000203}),
the rate of $\hat{\varsigma}_{j}$ in (\ref{eq:sigma=000020IV=000020rate}),
and Proposition \ref{prop:DB-Aux}, we have 
\begin{align}
\left|\dfrac{1}{n}\sum_{t=1}^{n}\check{r}_{j,t-1}^{2}-\frac{1}{n}\sum_{t=1}^{n}\zeta_{j,t-1}^{2}\right| & =O_{p}\left(n^{\tau-\tau\wedge(1-\tau)}(\log p)^{4+1/r}\right)+O_{p}(n^{\tau/2})\cdot o_{p}(n^{\tau/2})\nonumber \\
 & =o_{p}\left(n^{\tau}\right).\label{eq:temp=000020check=0000204}
\end{align}
In addition, (\ref{eq:zeta=0000202=000020lim}) implies $\frac{1}{n}\sum_{t=1}^{n}\zeta_{j,t-1}^{2}\nep n^{\tau}$.
Then by (\ref{eq:temp=000020check=0000204}), 
\[
\left|\dfrac{\dfrac{1}{n}\sum_{t=1}^{n}\check{r}_{j,t-1}^{2}}{\frac{1}{n}\sum_{t=1}^{n}\zeta_{j,t-1}^{2}}-1\right|\convp0
\]
as $n\to\infty$, or equivalently
\begin{equation}
\dfrac{1}{n}\sum_{t=1}^{n}\check{r}_{j,t-1}^{2}\bigg/\frac{1}{n}\sum_{t=1}^{n}\zeta_{j,t-1}^{2}\convp1.\label{eq:check=000020r2=000020zeta=0000202=0000201}
\end{equation}
Recall from (\ref{eq:zeta=0000202=000020lim}) and (\ref{eq:sigma=000020IV=000020ntau=000020lim})
that $\frac{1}{n^{1+\tau}}\sum_{t=1}^{n}\zeta_{j,t-1}^{2}$ and $\frac{\hat{\varsigma}_{j}^{2}}{n^{\tau}}$
have the same probability limit, and thus 
\begin{equation}
\frac{1}{n}\sum_{t=1}^{n}\zeta_{j,t-1}^{2}\bigg/\hat{\varsigma}_{j}^{2}\convp1.\label{eq:zeta=000020std=000020lim=0000201}
\end{equation}
Recall that $\hat{r}_{j,t}=\check{r}_{j,t-1}/\hat{\varsigma}_{j}$
as shown in (\ref{eq:r=000020check=000020r=000020hat}). Thus, (\ref{eq:check=000020r2=000020zeta=0000202=0000201})
and (\ref{eq:zeta=000020std=000020lim=0000201}) imply 
\[
\dfrac{1}{n}\sum_{t=1}^{n}\hat{r}_{j,t-1}^{2}=\dfrac{1}{n}\sum_{t=1}^{n}\check{r}_{j,t-1}^{2}/\hat{\varsigma}_{j}^{2}\convp1.
\]
Then (\ref{eq:r=000020hat=0000202=000020lim}) is verified for CASE
I.

\textbf{CASE II: }$j\in\mathcal{M}_{z}$. The definition of $\hat{r}_{j,t}$
gives 
\begin{equation}
\sum_{t=1}^{n}\hat{r}_{j,t-1}w_{j,t-1}=\sum_{t=1}^{n}\tilde{\eta}_{j,t}w_{j,t-1}+\sum_{t=1}^{n}w_{j,t-1}W_{-j,t-1}^{\top}\left(\hat{\varphi}^{(j)}-\varphi^{*(j)}\right),\label{eq:r=000020hat=000020w=000020decom}
\end{equation}
where $\tilde{\eta}_{j,t}$ is defined below (\ref{eq:zeta=000020I0=000020pseudo=000020true}).
Note that 
\begin{align*}
\dfrac{1}{n}\sum_{t=1}^{n}\tilde{\eta}_{j,t}w_{j,t-1} & =\frac{1}{n\hat{\varsigma}_{j}}\sum_{t=1}^{n}w_{j,t-1}(w_{j,t-1}-Z_{-j,t-1}^{\top}\varphi_{0z}^{(j)*})\\
 & =\frac{1}{n\hat{\varsigma}_{j}}\sum_{t=1}^{n}w_{j,t-1}\eta_{j,t-1}^{(1)}.
\end{align*}
By (\ref{eq:sigma=000020IV=000020I0=000020lim}), we have 
\begin{equation}
\dfrac{1}{n}\sum_{t=1}^{n}\tilde{\eta}_{j,t}w_{j,t-1}\convp\dfrac{{\rm cov}\left(w_{j,t},\eta_{j,t}^{(1)}\right)}{\sqrt{{\rm var}\left(w_{j,t}\right)}}\label{eq:limit=000020eta=000020w}
\end{equation}
where $\eta_{j,t}^{(1)}=w_{j,t}-Z_{-j,t}^{\top}\varphi_{0z}^{(j)*}$
was defined in (\ref{eq:def=000020eta=0000201}). In addition, we
deduce that 
\begin{align*}
\|n^{-1}\sum_{t=1}^{n}D_{-j}^{-1}W_{-j,t-1}w_{j,t-1}\|_{\infty} & \leq\dfrac{\max_{k\in\mathcal{M}_{x}}\left|n^{-1}\sum_{t=1}^{n}w_{k,t-1}w_{j,t-1}\right|}{\min_{k\in\mathcal{M}_{x}}\hat{\sigma}_{k}}+\dfrac{\max_{k\in\mathcal{M}_{z}}\left|n^{-1}\sum_{t=1}^{n}w_{k,t-1}w_{j,t-1}\right|}{\min_{k\in\mathcal{M}_{z}}\hat{\sigma}_{k}}\\
 & =O_{p}\left(\sqrt{\dfrac{\log p}{n}}\cdot(\log p)^{1+\frac{1}{2r}}\right)+O_{p}(1)=O_{p}(1),
\end{align*}
where the second step follows (\ref{eq:DB=000020e}) bounding the
cross product between LURs and a stationary component, (\ref{eq:sd=000020bound=000020ST=000020MI})
and (\ref{eq:sd=000020bound=000020ST=000020MI}) bounding the standard
deviations, and the fact that 
\[
\max_{k\in\mathcal{M}_{z}}\left|n^{-1}\sum_{t=1}^{n}w_{k,t-1}w_{j,t-1}\right|=O_{p}(1)
\]
following MS24's Eq.(B.28). Thus, by Proposition \ref{prop:DB-Aux}
in our main text,
\begin{align}
\left|\dfrac{1}{n}\sum_{t=1}^{n}w_{j,t-1}W_{-j,t-1}^{\top}\left(\hat{\varphi}^{(j)}-\varphi^{*(j)}\right)\right| & \leq\|n^{-1}\sum_{t=1}^{n}D_{-j}^{-1}W_{-j,t-1}w_{j,t-1}\|_{\infty}\|D_{-j}(\hat{\varphi}^{(j)}-\varphi^{*(j)})\|_{1}\nonumber \\
 & =o_{p}(1).\label{eq:=000020small=000020term=0000201}
\end{align}
Combining (\ref{eq:r=000020hat=000020w=000020decom}), (\ref{eq:limit=000020eta=000020w}),
and (\ref{eq:=000020small=000020term=0000201}), 
\begin{equation}
\dfrac{1}{n}\sum_{t=1}^{n}\hat{r}_{j,t-1}w_{j,t-1}\convp\dfrac{{\rm cov}\left(w_{j,t},\eta_{j,t}^{(1)}\right)}{\sqrt{{\rm var}\left(w_{j,t}\right)}}.\label{eq:=000020hat=000020r=000020w=000020lim}
\end{equation}
Thus,
\[
\dfrac{\left|\sum_{t=1}^{n}\hat{r}_{j,t-1}w_{j,t-1}\right|}{\sum_{t=1}^{n}\hat{r}_{j,t-1}w_{j,t-1}}\convp\dfrac{\left|{\rm cov}\left(w_{j,t},\eta_{j,t}^{(1)}\right)\right|}{{\rm cov}\left(w_{j,t},\eta_{j,t}^{(1)}\right)}={\rm sgn}\left({\rm cov}\left(w_{j,t},\eta_{j,t}^{(1)}\right)\right),
\]
which together with (\ref{eq:consistent=000020Omega=000020u}) implies
(\ref{eq:lim=000020sgn}) for $j\in\mathcal{M}_{z}$. 

For (\ref{eq:r=000020hat=0000202=000020lim}), note that $\hat{r}_{j,t}$
is the LASSO residual of regression (\ref{eq:def-rhat}). Recall that
$\tilde{\eta}_{j,t}$ is the error term of the pseudo-true model (\ref{eq:zeta=000020I0=000020pseudo=000020true}).
Following the ideas in the proof of (\ref{eq:uhat2=000020u2}), we
can show that $n^{-1}\sum_{t=1}^{n}\hat{r}_{j,t}^{2}$ and $n^{-1}\sum_{t=1}^{n}\tilde{\eta}_{j,t}^{2}$
share the same probability limit, which is 
\begin{align}
{\rm plim}_{n\to\infty}\dfrac{1}{n}\sum_{t=1}^{n}\hat{r}_{j,t}^{2} & ={\rm plim}_{n\to\infty}\dfrac{1}{n}\sum_{t=1}^{n}\tilde{\eta}_{j,t}^{2}.\nonumber \\
 & ={\rm plim}_{n\to\infty}\dfrac{\dfrac{1}{n}\sum_{t=1}^{n}\eta_{j,t}^{2}}{\hat{\varsigma}_{j}^{2}}=\dfrac{{\rm plim}_{n\to\infty}\dfrac{1}{n}\sum_{t=1}^{n}\eta_{j,t}^{2}}{{\rm var}(w_{j,t})},\label{eq:plim=000020hat=000020r=0000202}
\end{align}
where the last step applies (\ref{eq:sigma=000020IV=000020I0=000020lim}).
In addition, note that 
\begin{align}
\eta_{j,t} & =\eta_{j,t-1}^{(1)}+\zeta_{j,t}-w_{j,t}=\eta_{j,t-1}^{(1)}-n^{-\tau}C_{\zeta}\phi_{j,t-1},\label{eq:eta=000020eta(1)=000020phi}
\end{align}
where the first equality applies (\ref{eq:eta=000020(1)=000020eta}),
and the second equality applies (\ref{eq:decom=000020zeta=000020z}).
Therefore,
\begin{align*}
\dfrac{1}{n}\sum_{t=1}^{n}\eta_{j,t}^{2} & =\dfrac{1}{n}\sum_{t=1}^{n}(\eta_{j,t-1}^{(1)})^{2}+\dfrac{C_{\zeta}^{2}}{n^{2\tau+1}}\sum_{t=1}^{n}\phi_{j,t-1}^{2}-\dfrac{2C_{\zeta}}{n^{1+\tau}}\sum_{t=1}^{n}\phi_{j,t-1}\eta_{j,t-1}^{(1)}\\
 & =\dfrac{1}{n}\sum_{t=1}^{n}(\eta_{j,t-1}^{(1)})^{2}+o_{p}(1)\\
 & \convp{\rm var}(\eta_{j,t}^{(1)}),
\end{align*}
where the second row applies the same arguments for (\ref{eq:zeta=0000202=000020w=0000202}).
Then we have 
\[
{\rm plim}_{n\to\infty}\dfrac{1}{n}\sum_{t=1}^{n}\hat{r}_{j,t}^{2}=\dfrac{{\rm plim}_{n\to\infty}\dfrac{1}{n}\sum_{t=1}^{n}\eta_{j,t}^{2}}{{\rm var}(w_{j,t})}=\dfrac{{\rm var}(\eta_{j,t}^{(1)})}{{\rm var}(w_{j,t})}.
\]
This completes the proof of Lemma \ref{lem:s.e.consis}.
\end{proof}

\subsubsection{Proofs of main results in Section \ref{subsec:Asymptotic-Normality}}
\begin{proof}[Proof of Theorem \ref{thm:CLT}]
By the definition of the XDlasso estimator, 
\begin{align*}
\hat{\theta}_{j}^{{\rm XD}}-\theta_{j}^{*} & =\frac{\sum_{t=1}^{n}\hat{r}_{j,t-1}u_{t}}{\sum_{t=1}^{n}\hat{r}_{j,t-1}w_{j,t-1}}+\frac{\sum_{t=1}^{n}\hat{r}_{j,t-1}W_{-j,t-1}^{\top}(\theta_{-j}^{*}-\hat{\theta}_{-j})}{\sum_{t=1}^{n}\hat{r}_{j,t-1}w_{j,t-1}}.
\end{align*}
Then the $t$-statistic can be decomposed as 
\begin{align}
\dfrac{\hat{\theta}_{j}^{{\rm XD}}-\theta_{j}^{*}}{\hat{\omega}_{j}^{{\rm XD}}} & =\mathcal{Z}_{j}+\varDelta_{j},\text{ where }\nonumber \\
\mathcal{Z}_{j} & =\dfrac{\left|\sum_{t=1}^{n}\hat{r}_{j,t-1}w_{j,t}\right|}{\sum_{t=1}^{n}\hat{r}_{j,t-1}w_{j,t}}\cdot\dfrac{\sigma_{u}}{\hat{\sigma}_{u}}\cdot\dfrac{\sum_{t=1}^{n}\hat{r}_{j,t-1}u_{t}}{\text{\ensuremath{\sigma_{u}\sqrt{\sum_{t=1}^{n}\hat{r}_{j,t-1}^{2}}}}},\label{eq:Zj=000020def}\\
\varDelta_{j} & =\dfrac{\sum_{t=1}^{n}\hat{r}_{j,t-1}W_{-j,t-1}^{\top}(\theta_{-j}^{*}-\hat{\theta}_{-j})}{\ensuremath{\sqrt{\sum_{t=1}^{n}\hat{r}_{j,t-1}^{2}}}}.\label{eq:Delta=000020j=000020def}
\end{align}
We first bound $\varDelta_{j}$. By the Karush-Kuhn-Tucker condition,
we can establish
\begin{equation}
\|\sum_{t=1}^{n}D_{-j}W_{-j,t-1}\hat{r}_{j,t-1}\|_{\infty}\leq\dfrac{C_{{\rm a}}(\log p)^{2+\frac{1}{2r}}}{\sqrt{n^{(1-\tau)\wedge\tau}}}\label{eq:KKT=000020app}
\end{equation}
 w.p.a.1 as in (\ref{eq:KKT}). Thus, 

\begin{align}
\left|\varDelta_{j}\right| & \leq\sqrt{n}\dfrac{\left|n^{-1}\sum_{t=1}^{n}\hat{r}_{j,t-1}W_{-j,t-1}^{\top}D_{-j}^{-1}D_{-j}(\theta_{-j}^{*}-\hat{\theta}_{-j})\right|}{\ensuremath{\sqrt{n^{-1}\sum_{t=1}^{n}\hat{r}_{j,t-1}^{2}}}}\nonumber \\
 & \leq\dfrac{\sqrt{n}\cdot\|n^{-1}\sum_{t=1}^{n}D_{-j}^{-1}W_{-j,t-1}\hat{r}_{j,t-1}\|_{\infty}\cdot\|D_{-j}(\theta{}_{-j}^{*}-\hat{\theta}_{-j})\|_{1}}{\ensuremath{\sqrt{n^{-1}\sum_{t=1}^{n}\hat{r}_{j,t-1}^{2}}}}\nonumber \\
 & \lep\dfrac{\sqrt{n}}{\ensuremath{\sqrt{n^{-1}\sum_{t=1}^{n}\hat{r}_{j,t-1}^{2}}}}\cdot\left(\dfrac{(\log p)^{2+\frac{1}{2r}}}{\sqrt{n^{(1-\tau)\wedge\tau}}}\right)\cdot\left(\dfrac{s^{2}}{\sqrt{n}}(\log p)^{6+\frac{1}{2r}}\right)\nonumber \\
 & \lep\dfrac{s^{2}(\log p)^{8+\frac{1}{r}}}{\sqrt{n^{\tau\wedge(1-\tau)}}}\to0\label{eq:Delta=000020j=000020rate}
\end{align}
where the third inequality applies (\ref{eq:KKT=000020app}) and Lemma
\ref{lem:SlassoError}, and the fourth inequality applies
\begin{equation}
n^{-1}\sum_{t=1}^{n}\hat{r}_{j,t-1}^{2}\nep1\label{eq:r=000020hat=0000202=0000201}
\end{equation}
 implied by (\ref{eq:r=000020hat=0000202=000020lim}). 

It then suffices to show that $\mathcal{Z}_{j}\convd\mathcal{N}(0,1).$\textbf{
}Given the limit in (\ref{eq:lim=000020sgn}), when $j\in\mathcal{M}_{z}$
it suffices to show 
\begin{equation}
\mathcal{Z}_{j}^{(1)}:=\dfrac{\sum_{t=1}^{n}\hat{r}_{j,t-1}u_{t}}{\text{\ensuremath{\sigma_{u}\sqrt{\sum_{t=1}^{n}\hat{r}_{j,t-1}^{2}}}}}\convd\mathcal{N}(0,1).\label{eq:proof=000020normal}
\end{equation}
When $j\in\mathcal{M}_{x}$, we need to additionally show the asymptotic
distribution in (\ref{eq:proof=000020normal}) is independent of the
Brownian motion in the $G_{j}^{*}$ of limit (\ref{eq:lim=000020sgn}),
so that 
\begin{align}
\mathcal{Z}_{j} & =\dfrac{\left|\sum_{t=1}^{n}\hat{r}_{j,t-1}w_{j,t}\right|}{\sum_{t=1}^{n}\hat{r}_{j,t-1}w_{j,t}}\cdot\dfrac{\sigma_{u}}{\hat{\sigma}_{u}}\cdot\mathcal{Z}_{j}^{(1)}\label{eq:z=000020z1}\\
 & \convd{\rm sgn}(G_{j}^{*})\cdot\mathcal{N}(0,1)=\mathcal{MN}\left(0,{\rm sgn}(G_{j}^{*})^{2}\right)=\mathcal{N}(0,1),\nonumber 
\end{align}
where $\mathcal{MN}$ denotes a mixed normal distribution. 

\textbf{CASE I}. When $j\in\mathcal{M}_{x}$, 
\[
\mathcal{Z}_{j}^{(1)}=-\dfrac{n^{-1/2}\sum_{t=1}^{n}u_{t}W_{-j,t-1}^{\top}\hat{\varphi}^{(j)}}{\text{\ensuremath{\sqrt{\dfrac{1}{n}\sum_{t=1}^{n}\hat{r}_{j,t-1}^{2}}}}}+\dfrac{n^{-1/2}\sum_{t=1}^{n}\tilde{\zeta}_{j,t-1}u_{t}}{\text{\ensuremath{\sqrt{\dfrac{1}{n}\sum_{t=1}^{n}\hat{r}_{j,t-1}^{2}}}}}\cdot
\]
We bound the first term by 
\begin{align}
\left|n^{-1/2}\sum_{t=1}^{n}u_{t}W_{-j,t-1}^{\top}\hat{\varphi}^{(j)}\right| & \lep\|n^{-1/2}\sum_{t=1}^{n}D_{-j}^{-1}W_{-j,t-1}u_{t}\|_{\infty}\cdot\|D_{-j}\hat{\varphi}^{(j)}\|_{1}\nonumber \\
 & \lep(\log p)^{3/2+1/(2r)}\cdot\dfrac{s^{2}(\log p)^{6+\frac{1}{2r}}}{\sqrt{n^{\tau\wedge(1-\tau)}}}\label{eq:first=000020term}
\end{align}
where the second step applies (\ref{eq:DB}) and Proposition \ref{prop:DB-Aux}.
Therefore, by (\ref{eq:r=000020hat=0000202=000020lim})
\[
\dfrac{\left|n^{-1/2}\sum_{t=1}^{n}u_{t}W_{-j,t-1}^{\top}\hat{\varphi}^{(j)}\right|}{\text{\ensuremath{\sqrt{\dfrac{1}{n}\sum_{t=1}^{n}\hat{r}_{j,t-1}^{2}}}}}\lep\dfrac{s^{2}(\log p)^{\frac{15}{2}+\frac{1}{r}}}{\sqrt{n^{\tau\wedge(1-\tau)}}}\to0,
\]
which implies 
\begin{equation}
\mathcal{Z}_{j}^{(1)}=\dfrac{n^{-1/2}\sum_{t=1}^{n}\tilde{\zeta}_{j,t-1}u_{t}}{\text{\ensuremath{\sqrt{\dfrac{1}{n}\sum_{t=1}^{n}\hat{r}_{j,t-1}^{2}}}}}+o_{p}(1)\cdot\label{eq:Z1=000020MX}
\end{equation}
In addition, by the central limit theorem in Lemma B4(ii) of \citet{kostakis2015robust}
and Eq. (28) in their appendix, the law of large numbers Eq. (13)
and (21) in the appendix of the same reference, and the Slutsky's
Theorem, we have
\begin{equation}
\frac{n^{-1/2}\sum_{t=1}^{n}\zeta_{j,t-1}u_{t}}{\sigma_{u}\sqrt{\frac{1}{n}\sum_{t=1}^{n}\zeta_{j,t-1}^{2}}}\convd\mathcal{N}(0,1).\label{eq:generic=000020CLT}
\end{equation}
Besides, \citet{phillips2009econometric}'s Lemma 3.2 shows that the
asymptotic distribution in (\ref{eq:generic=000020CLT}) is independent
of the Brownian motion in the expression (\ref{eq:express=000020G=000020star})
of $G_{j}^{*}$. Also, recall that $\tilde{\zeta}_{j,t-1}=\zeta_{j,t-1}/\hat{\varsigma}_{j}$
and $\check{r}_{j,t-1}=\hat{r}_{j,t-1}\hat{\varsigma}_{j}.$ By (\ref{eq:check=000020r2=000020zeta=0000202=0000201})
and the Slutsky's Theorem, we have 
\[
\dfrac{n^{-1/2}\sum_{t=1}^{n}\tilde{\zeta}_{j,t-1}u_{t}}{\sigma_{u}\text{\ensuremath{\sqrt{\dfrac{1}{n}\sum_{t=1}^{n}\hat{r}_{j,t-1}^{2}}}}}=\dfrac{n^{-1/2}\sum_{t=1}^{n} {\zeta}_{j,t-1}u_{t}}{\sigma_{u}\text{\ensuremath{\sqrt{\dfrac{1}{n}\sum_{t=1}^{n}\check{r}_{j,t-1}^{2}}}}}=\sqrt{\dfrac{\frac{1}{n}\sum_{t=1}^{n}\zeta_{j,t-1}^{2}}{\dfrac{1}{n}\sum_{t=1}^{n}\check{r}_{j,t-1}^{2}}}\cdot\dfrac{n^{-1/2}\sum_{t=1}^{n}\zeta_{j,t-1}u_{t}}{\sigma_{u}\sqrt{\frac{1}{n}\sum_{t=1}^{n}\zeta_{j,t-1}^{2}}}\convd\mathcal{N}(0,1).
\]
This completes the proof of (\ref{eq:proof=000020normal}) when $j\in\mathcal{M}_{x}.$

\textbf{CASE II}. When $j\in\mathcal{M}_{z}$, recall that we have
defined $\eta_{j,t}=\zeta_{j,t}-Z_{-j,t}^{\top}\varphi_{0,z}^{(j)*}$
in (\ref{eq:def=000020eta=000020jt}), with $\varphi_{0,z}^{(j)*}$
defined as (\ref{eq:phi=000020j=0000200=000020z}). Then 
\begin{align}
\hat{r}_{j,t} & =\tilde{\zeta}_{j,t}-W_{-j,t}^{\top}\hat{\varphi}^{(j)}\nonumber \\
 & =W_{-j,t}^{\top}\left(\varphi^{*(j)}-\hat{\varphi}^{(j)}\right)+\dfrac{\eta_{j,t}}{\hat{\varsigma}_{j}}\nonumber \\
 & =W_{-j,t}^{\top}\left(\varphi^{*(j)}-\hat{\varphi}^{(j)}\right)+\dfrac{\eta_{j,t-1}^{(1)}-n^{-\tau}C_{\zeta}\phi_{j,t-1}}{\hat{\varsigma}_{j}},\label{eq:hat=000020r=000020Lasso}
\end{align}
where the first equality is by the definition of $\hat{r}_{j,t}$
in (\ref{eq:def-rhat}), the second row applies the pseudo-true regression
model (\ref{eq:zeta=000020I0=000020pseudo=000020true}) and the equality
(\ref{eq:eta=000020tilde=000020eta}), and the third row applies (\ref{eq:eta=000020eta(1)=000020phi}).
Then by the definition of $\mathcal{Z}_{0,j}$ in (\ref{eq:proof=000020normal}),
we have the following decomposition 
\begin{equation}
\mathcal{Z}_{j}^{(1)}=\dfrac{n^{-1/2}\sum_{t=1}^{n}u_{t}W_{-j,t-1}^{\top}\left(\varphi^{*(j)}-\hat{\varphi}^{(j)}\right)}{\text{\ensuremath{\sqrt{\dfrac{1}{n}\sum_{t=1}^{n}\hat{r}_{j,t-1}^{2}}}}}+\dfrac{n^{-1/2}\sum_{t=1}^{n}\left(\eta_{j,t-1}^{(1)}+n^{-\tau}\phi_{j,t-1}\right)u_{t}}{\hat{\varsigma}_{j}\text{\ensuremath{\sqrt{\dfrac{1}{n}\sum_{t=1}^{n}\hat{r}_{j,t-1}^{2}}}}}.\label{eq:Z=000020j=0000203=000020terms}
\end{equation}
We first bound the first term by 

\begin{align}
\left|n^{-1/2}\sum_{t=1}^{n}u_{t}W_{-j,t-1}^{\top}\left(\varphi^{*(j)}-\hat{\varphi}^{(j)}\right)\right| & \lep\|n^{-1/2}\sum_{t=1}^{n}D_{-j}^{-1}W_{-h,t-1}u_{t}\|_{\infty}\cdot\|D_{-j}\left(\varphi^{*(j)}-\hat{\varphi}^{(j)}\right)\|_{1}\nonumber \\
 & \lep(\log p)^{3/2+1/(2r)}\cdot\dfrac{s^{2}(\log p)^{6+\frac{1}{2r}}}{\sqrt{n^{\tau\wedge(1-\tau)}}}\to0,\label{eq:term=0000201=000020lim=0000200-1}
\end{align}
where the second step applies MS24's (B.63), and Proposition \ref{prop:DB-Aux}
in this current paper. By (\ref{eq:r=000020hat=0000202=000020lim}),
we have $\dfrac{1}{n}\sum_{t=1}^{n}\hat{r}_{j,t-1}^{2}\nep1$ and
thus 
\begin{equation}
\left|\dfrac{n^{-1/2}\sum_{t=1}^{n}u_{t}W_{-j,t-1}^{\top}\left(\varphi^{*(j)}-\hat{\varphi}^{(j)}\right)}{\text{\ensuremath{\sqrt{\dfrac{1}{n}\sum_{t=1}^{n}\hat{r}_{j,t-1}^{2}}}}}\right|=o_{p}(1).\label{eq:term=0000201=000020lim=0000200}
\end{equation}

We then show the central limit theorem for the second term. Recall
that $u_{t}$ is m.d.s.~by Assumption \ref{assu:tail}, and $\eta_{j,t}^{(1)}$
is stationary and strong mixing. By a standard martingale central
limit theorem we have $\frac{n^{-1/2}\sum_{t=1}^{n}\eta_{j,t-1}^{(1)}u_{t}}{\sigma_{u}\sqrt{{\rm var}(\eta_{j,t}^{(1)})}}\convd\mathcal{N}(0,1).$
By (\ref{eq:r=000020hat=0000202=000020lim}) and (\ref{eq:sigma=000020IV=000020I0=000020lim}),
we have $\hat{\varsigma}_{j}^{2}\cdot\dfrac{1}{n}\sum_{t=1}^{n}\hat{r}_{j,t-1}^{2}\convp{\rm var}(\eta_{j,t}^{(1)})$.
Thus, by the Slutsky's Theorem we have 
\begin{equation}
\dfrac{n^{-1/2}\sum_{t=1}^{n}\eta_{j,t-1}^{(1)}u_{t}}{\text{\ensuremath{\sigma_{u}\hat{\varsigma}_{j}\cdot\sqrt{\dfrac{1}{n}\sum_{t=1}^{n}\hat{r}_{j,t-1}^{2}}}}}=\sqrt{\dfrac{{\rm var}(\eta_{j,t}^{(1)})}{\hat{\varsigma}_{j}^{2}\cdot\dfrac{1}{n}\sum_{t=1}^{n}\hat{r}_{j,t-1}^{2}}}\cdot\frac{n^{-1/2}\sum_{t=1}^{n}\eta_{j,t-1}^{(1)}u_{t}}{\sigma_{u}\sqrt{{\rm var}(\eta_{j,t}^{(1)})}}\convd\mathcal{N}(0,1).\label{eq:term=0000202}
\end{equation}
Finally, note that $\phi_{j,t-1}$ is mildly integrated. Again by
\citet[Lemma B4(ii)]{kostakis2015robust} we have $\sum_{t=1}^{n}\phi_{j,t-1}u_{t}=O_{p}\left(n^{(\tau+1)/2}\right),$
and thus 
\begin{equation}
n^{-1/2}\sum_{t=1}^{n}n^{-\tau}\phi_{j,t-1}u_{t}\convp0.\label{eq:phi=000020u=000020lim=0000200}
\end{equation}
By (\ref{eq:r=000020hat=0000202=000020lim}) and (\ref{eq:sigma=000020IV=000020I0=000020lim}),
we further have $\hat{\varsigma}_{j}\text{\ensuremath{\sqrt{\dfrac{1}{n}\sum_{t=1}^{n}\hat{r}_{j,t-1}^{2}}}}\nep1$
and thus 
\begin{equation}
\dfrac{n^{-1/2}\sum_{t=1}^{n}n^{-\tau}\phi_{j,t-1}u_{t}}{\hat{\varsigma}_{j}\text{\ensuremath{\sqrt{\dfrac{1}{n}\sum_{t=1}^{n}\hat{r}_{j,t-1}^{2}}}}}\convp0.\label{eq:term=0000203=000020lim=0000200}
\end{equation}
By (\ref{eq:Z=000020j=0000203=000020terms}), (\ref{eq:term=0000201=000020lim=0000200-1}),
(\ref{eq:term=0000202}), and (\ref{eq:term=0000203=000020lim=0000200}),
we have 
\begin{align}
\mathcal{Z}_{j}^{(1)} & =o_{p}(1)+\dfrac{n^{-1/2}\sum_{t=1}^{n}\eta_{j,t-1}^{(1)}u_{t}}{\text{\ensuremath{\sigma_{u}\hat{\varsigma}_{j}\cdot\sqrt{\dfrac{1}{n}\sum_{t=1}^{n}\hat{r}_{j,t-1}^{2}}}}}+o_{p}(1)\label{eq:Z1=000020Mz}\\
 & \convd \mathcal{N}(0,1),
\end{align}
which verifies (\ref{eq:proof=000020normal}) when $j\in\mathcal{M}_{z}.$
This completes the proof of Theorem \ref{thm:CLT}. 
\end{proof}
\begin{proof}[Proof of Theorem \ref{thm:length}]
Recall that 
\begin{equation}
\hat{\omega}_{j}^{{\rm XD}}=\frac{\text{\ensuremath{\sqrt{\sum_{t=1}^{n}\hat{r}_{j,t-1}^{2}}}}}{\left|\sum_{t=1}^{n}\hat{r}_{j,t-1}w_{j,t-1}\right|}=\frac{\text{\ensuremath{\sqrt{\sum_{t=1}^{n}\check{r}_{j,t-1}^{2}}}}}{\left|\sum_{t=1}^{n}\check{r}_{j,t-1}w_{j,t-1}\right|}\label{eq:robustse}
\end{equation}
where $\check{r}_{j,t}=\hat{r}_{j,t}\cdot\hat{\varsigma}_{j}$ as
defined in (\ref{eq:r=000020check=000020r=000020hat}). 

\textbf{CASE I}. When $j\in\mathcal{M}_{x}$, (\ref{eq:check=000020r2=000020zeta=0000202=0000201})
and (\ref{eq:zeta=0000202=000020lim}) implies
\begin{equation}
\dfrac{1}{n^{1+\tau}}\sum_{t=1}^{n}\check{r}_{j,t-1}^{2}=\dfrac{\sum_{t=1}^{n}\check{r}_{j,t-1}^{2}}{\sum_{t=1}^{n}\zeta_{j,t-1}^{2}}\cdot\dfrac{1}{n^{1+\tau}}\sum_{t=1}^{n}\zeta_{j,t-1}^{2}\convp\dfrac{{\rm lvar}(e_{j,t})}{2C_{\zeta}},\label{eq:r=000020check=000020limit}
\end{equation}
which implies $\sum_{t=1}^{n}\hat{r}_{j,t-1}^{2}=O_{p}(n^{1+\tau}).$
In addition, the weak convergence (\ref{eq:check=000020r=000020w=000020weak=000020lim})
implies that $\frac{1}{\left|\sum_{t=1}^{n}\check{r}_{j,t-1}w_{j,t-1}\right|}=O_{p}\left(\frac{1}{n^{1+\tau}}\right)$.
Then 
\[
\hat{\omega}_{j}^{{\rm XD}}=O_{p}\left(\dfrac{\sqrt{n^{1+\tau}}}{n^{1+\tau}}\right)=O_{p}\left(\dfrac{1}{n^{(1+\tau)/2}}\right).
\]

\textbf{CASE II}. When $j\in\mathcal{M}_{z}$, by (\ref{eq:=000020hat=000020r=000020w=000020lim})
and (\ref{eq:r=000020hat=0000202=000020lim}), 
\[
\hat{\omega}_{j}^{{\rm XD}}=O_{p}\left(\dfrac{\sqrt{n}}{n}\right)=O_{p}\left(\dfrac{1}{\sqrt{n}}\right).
\]
We complete the proof of Theorem \ref{thm:length}. 
\end{proof}
\begin{proof}[Proof of Theorem \ref{thm:wald}]
Define $\hat{\Pi}_{j}=\sum_{t=1}^{n}\check{r}_{j,t-1}\left[u_{t}+W_{-j,t-1}^{\top}(\theta_{-j}^{*}-\hat{\theta}_{-j})\right]$,
and $\hat{\Pi}_{\mathcal{A}}=(\hat{\Pi}_{j})_{j\in\mathcal{A}}$for
any subset $\mathcal{A}\in[p]$. Note that 
\begin{align}
\hat{\theta}_{j}^{{\rm XD}}-\theta_{j}^{*} & =\dfrac{\sum_{t=1}^{n}\hat{r}_{j,t-1}\left[u_{t}+W_{-j,t-1}^{\top}(\theta_{-j}^{*}-\hat{\theta}_{-j})\right]}{\sum_{t=1}^{n}\hat{r}_{j,t-1}w_{j,t-1}}=\dfrac{\hat{\Pi}_{j}}{\sum_{t=1}^{n}\check{r}_{j,t-1}w_{j,t-1}},\label{eq:theta=000020XD=000020theta}
\end{align}
where the second row applies the fact the equality $\check{r}_{j,t}=\hat{r}_{j,t}\hat{\varsigma}_{j}$
in (\ref{eq:r=000020check=000020r=000020hat}). Furthermore, define
the matrix $\hat{\Theta}_{\mathcal{J}}=(\sum_{t=1}^{n}\check{r}_{j,t-1}\check{r}_{k,t-1})_{j\in\mathcal{J},k\in\mathcal{J}}.$
Also, note that 
\begin{align}
\hat{\Omega}_{\mathcal{J}}^{{\rm XD}} & =\hat{\sigma}_{u}^{2}\left(\dfrac{\sum_{t=1}^{n}\hat{r}_{j,t-1}\hat{r}_{k,t-1}}{\sum_{t=1}^{n}\hat{r}_{j,t-1}w_{j,t-1}\sum_{t=1}^{n}\hat{r}_{k,t-1}w_{k,t-1}}\right)_{j,k\in\mathcal{J}}\nonumber \\
 & =\hat{\sigma}_{u}^{2}\left(\dfrac{\sum_{t=1}^{n}\check{r}_{j,t-1}\check{r}_{k,t-1}}{\sum_{t=1}^{n}\check{r}_{j,t-1}w_{j,t-1}\sum_{t=1}^{n}\check{r}_{k,t-1}w_{k,t-1}}\right)_{j,k\in\mathcal{J}}\nonumber \\
 & =\hat{\sigma}_{u}^{2}\left[{\rm diag}(\sum_{t=1}^{n}\check{r}_{j,t-1}w_{j,t-1})_{j\in\mathcal{J}}\right]^{-1}\hat{\Theta}_{\mathcal{J}}\left[{\rm diag}(\sum_{t=1}^{n}\check{r}_{j,t-1}w_{j,t-1})_{j\in\mathcal{J}}\right]^{-1}.\label{eq:Omega=000020XD=000020Theta}
\end{align}
By (\ref{eq:theta=000020XD=000020theta}) and (\ref{eq:Omega=000020XD=000020Theta}),
some fundamental calculation yields that under $\mathbb{H}_{0}:\theta_{\mathcal{J}}^{*}=\theta_{0,\mathcal{J}},$
\[
{\rm Wald}_{\mathcal{J}}^{{\rm XD}}=\dfrac{1}{\hat{\sigma}_{u}^{2}}\hat{\Pi}_{\mathcal{J}}^{\top}\hat{\Theta}_{\mathcal{J}}^{-1}\hat{\Pi}_{\mathcal{J}}.
\]
The proof will consist of the following essential steps:

1. Show that 
\begin{equation}
\dfrac{1}{\sqrt{n^{1+\tau}}}\left|\hat{\Pi}_{j}-\sum_{t=1}^{n}\zeta_{j,t-1}u_{t}\right|=o_{p}(1)\text{ for any fixed }j\in\mathcal{M}_{x}.\label{eq:pi=000020x=000020rate}
\end{equation}

2. Show that 
\begin{equation}
\dfrac{1}{\sqrt{n}}\left|\hat{\Pi}_{j}-\sum_{t=1}^{n}\eta_{j,t-1}^{(1)}u_{t}\right|=o_{p}(1)\text{ for any fixed }j\in\mathcal{M}_{z}.\label{eq:pi=000020z=000020rate}
\end{equation}

3. Show that 
\begin{equation}
\left(\begin{array}{c}
\dfrac{1}{\sqrt{n^{1+\tau}}}\sum_{t=1}^{n}\zeta_{\mathcal{J}_{x},t-1}u_{t}\\
\dfrac{1}{\sqrt{n}}\sum_{t=1}^{n}\eta_{\mathcal{J}_{z},t-1}^{(1)}u_{t}
\end{array}\right)\convd\mathcal{N}\left(0,\sigma_{u}^{2}\Theta_{\mathcal{J}}\right),\label{eq:joint=000020normal}
\end{equation}
where $\zeta_{\mathcal{J}_{x},t-1}=(\zeta_{j,t-1})_{j\in\mathcal{J}_{x}}$,
$\eta_{\mathcal{J}_{z},t-1}^{(1)}=(\eta_{j,t-1}^{(1)})_{j\in\mathcal{M}_{z}}$,
and $\Theta_{\mathcal{J}}$ is a nonrandom positive definite matrix.

4. Show that 
\begin{equation}
\tilde{\Theta}_{\mathcal{J}}:=\left(\begin{array}{cc}
\sqrt{n^{1+\tau}}\mathbf{I}_{|\mathcal{J}_{x}|}\\
 & \sqrt{n}\mathbf{I}_{|\mathcal{J}_{z}|}
\end{array}\right)^{-1}\hat{\Theta}_{\mathcal{J}}\left(\begin{array}{cc}
\sqrt{n^{1+\tau}}\mathbf{I}_{|\mathcal{J}_{x}|}\\
 & \sqrt{n}\mathbf{I}_{|\mathcal{J}_{z}|}
\end{array}\right)^{-1}\convp\Theta_{\mathcal{J}}.\label{eq:cov=000020mat=000020convp}
\end{equation}
Equations (\ref{eq:pi=000020x=000020rate}), (\ref{eq:pi=000020z=000020rate}),
and (\ref{eq:joint=000020normal}) imply that 
\begin{equation}
\left(\begin{array}{c}
\dfrac{1}{\sqrt{n^{1+\tau}}}\hat{\Pi}_{\mathcal{J}_{x}}\\
\dfrac{1}{\sqrt{n}}\hat{\Pi}_{\mathcal{J}_{z}}
\end{array}\right)\convd\mathcal{N}\left(0,\sigma_{u}^{2}\Theta_{\mathcal{J}}\right).\label{eq:joint=000020normal=000020pi}
\end{equation}
Recall that we have shown $\hat{\sigma}_{u}^{2}/\sigma_{u}^{2}\convp1$
in (\ref{eq:consistent=000020Omega=000020u}). By (\ref{eq:cov=000020mat=000020convp}),
(\ref{eq:joint=000020normal=000020pi}), and the Slutsky's Theorem,
we have 
\[
\hat{\sigma}_{u}^{-1}\tilde{\Theta}_{\mathcal{J}}^{-1/2}\left(\begin{array}{c}
\dfrac{1}{\sqrt{n^{1+\tau}}}\hat{\Pi}_{\mathcal{J}_{x}}\\
\dfrac{1}{\sqrt{n}}\hat{\Pi}_{\mathcal{J}_{z}}
\end{array}\right)\convd\mathcal{N}\left(0,\mathbf{I}_{|\mathcal{J}|}\right),
\]
and thus 
\[
{\rm Wald}_{\mathcal{J}}^{{\rm XD}}=\dfrac{1}{\hat{\sigma}_{u}^{2}}\left(\begin{array}{c}
\dfrac{1}{\sqrt{n^{1+\tau}}}\hat{\Pi}_{\mathcal{J}_{x}}\\
\dfrac{1}{\sqrt{n}}\hat{\Pi}_{\mathcal{J}_{z}}
\end{array}\right)^{\top}\tilde{\Theta}_{\mathcal{J}}^{-1}\left(\begin{array}{c}
\dfrac{1}{\sqrt{n^{1+\tau}}}\hat{\Pi}_{\mathcal{J}_{x}}\\
\dfrac{1}{\sqrt{n}}\hat{\Pi}_{\mathcal{J}_{z}}
\end{array}\right)\convd\chi_{|\mathcal{J}|}^{2},
\]
which verifies Theorem \ref{thm:wald}.

\textbf{Proof of (\ref{eq:pi=000020x=000020rate})}. By (\ref{eq:check=000020r=000020def}),
we have $\check{r}_{j,t}=\zeta_{j,t}-\hat{\varsigma}_{j}W_{-j,t}^{\top}\hat{\varphi}^{(j)}$.
Therefore, 
\begin{align*}
\hat{\Pi}_{j} & =\sum_{t=1}^{n}\check{r}_{j,t-1}u_{t}+\sum_{t=1}^{n}\check{r}_{j,t-1}W_{-j,t-1}^{\top}(\theta_{-j}^{*}-\hat{\theta}_{-j})\\
 & =\sum_{t=1}^{n}\zeta_{j,t-1}u_{t}-\hat{\varsigma}_{j}\sum_{t=1}^{n}u_{t}W_{-j,t}^{\top}\hat{\varphi}^{(j)}+\hat{\varsigma}_{j}\sqrt{\sum_{t=1}^{n}\hat{r}_{j,t-1}^{2}}\varDelta_{j}\\
 & =\sum_{t=1}^{n}\zeta_{j,t-1}u_{t}+O_{p}(\sqrt{n^{\tau}})o_{p}(\sqrt{n})+O_{p}(\sqrt{n^{\tau}})O_{p}(\sqrt{n})o_{p}(1)\\
 & =\sum_{t=1}^{n}\zeta_{j,t-1}u_{t}+o_{p}(\sqrt{n^{1+\tau}}),
\end{align*}
where the second row applies the definition of $\varDelta_{j}$ in
(\ref{eq:Delta=000020j=000020def}), and third row applies $\hat{\varsigma}_{j}=O_{p}(\sqrt{n^{\tau}})$
by (\ref{eq:sigma=000020IV=000020rate}), $\sum_{t=1}^{n}\hat{r}_{j,t-1}^{2}=O_{p}(n)$
by (\ref{eq:plim=000020hat=000020r=0000202}), the rate of $\varDelta_{j}$
by (\ref{eq:Delta=000020j=000020rate}), and the rate of $\sum_{t=1}^{n}u_{t}W_{-j,t}^{\top}\hat{\varphi}^{(j)}$
by (\ref{eq:first=000020term}). This verifies (\ref{eq:pi=000020x=000020rate}).

\textbf{Proof of (\ref{eq:pi=000020z=000020rate})}. By (\ref{eq:hat=000020r=000020Lasso})
and the definition $\check{r}_{j,t}=\hat{\varsigma}_{j}\hat{r}_{j,t}$,
we have
\begin{equation}
\check{r}_{j,t}=\hat{\varsigma}_{j}W_{-j,t}^{\top}\left(\varphi^{*(j)}-\hat{\varphi}^{(j)}\right)+\eta_{j,t-1}^{(1)}-n^{-\tau}C_{\zeta}\phi_{j,t-1}.\label{eq:check=000020r=000020decom=000020z}
\end{equation}
Similar to the proof of (\ref{eq:pi=000020x=000020rate}), we have
\begin{align*}
\hat{\Pi}_{j} & =\sum_{t=1}^{n}\check{r}_{j,t-1}u_{t}+\sum_{t=1}^{n}\check{r}_{j,t-1}W_{-j,t-1}^{\top}(\theta_{-j}^{*}-\hat{\theta}_{-j})\\
 & =\sum_{t=1}^{n}\eta_{j,t-1}^{(1)}u_{t}+\hat{\varsigma}_{j}\sum_{t=1}^{n}u_{t}W_{-j,t}^{\top}\left(\varphi^{*(j)}-\hat{\varphi}^{(j)}\right)-\dfrac{C_{\zeta}}{n^{\tau}}\sum_{t=1}^{n}\phi_{j,t-1}u_{t}+\hat{\varsigma}_{j}\sqrt{\sum_{t=1}^{n}\hat{r}_{j,t-1}^{2}}\varDelta_{j}\\
 & =\sum_{t=1}^{n}\eta_{j,t-1}^{(1)}u_{t}+o_{p}(\sqrt{n}),
\end{align*}
where the third row applies $\hat{\varsigma}_{j}=O_{p}(1)$ by (\ref{eq:sigma=000020IV=000020rate-1}),
the rate of $\varDelta_{j}$ by (\ref{eq:Delta=000020j=000020rate}),
$\sum_{t=1}^{n}u_{t}W_{-j,t}^{\top}\left(\varphi^{*(j)}-\hat{\varphi}^{(j)}\right)=o_{p}(\sqrt{n})$
by (\ref{eq:term=0000201=000020lim=0000200-1}), and the rate of $n^{-\tau}\sum_{t=1}^{n}\phi_{j,t-1}u_{t}$
by (\ref{eq:phi=000020u=000020lim=0000200}). This verifies (\ref{eq:pi=000020z=000020rate}). 

\textbf{Proof of (\ref{eq:joint=000020normal})}. Following the proof
of (31) in \citet{phillips2009econometric}, we can show the following
Lindeberg condition for the IVs of the LURs:
\[
\lim_{n\to\infty}\mathbb{E}\left(\|n^{-\frac{1+\tau}{2}}\zeta_{\mathcal{J}_{x},t}\|_{2}\cdot\textbf{1}\{\|n^{-\frac{1+\tau}{2}}\zeta_{\mathcal{J}_{x},t}\|_{2}>\epsilon\}\right)=0
\]
for any fixed $\epsilon>0$. In addition, by standard argument it
can be shown that parallel Lindeberg condition holds for $n^{-1/2}\eta_{\mathcal{J}_{z},t}^{(1)}$,
since $\eta_{\mathcal{J}_{z},t}^{(1)}$ is a vector of stationary
and weakly dependent components.

Let ${\rm var}_{t-1}(\cdot)$ denote the conditional covariance matrix
given the information up to time $t-1$. According to the martingale
central limit theorem \citet[Corollary 3.1]{hall1980martingale},
it suffices to show that 
\begin{align}
\sum_{t=1}^{n}{\rm var}_{t-1}\left(\begin{array}{c}
\dfrac{1}{\sqrt{n^{1+\tau}}}\zeta_{\mathcal{J}_{x},t-1}u_{t}\\
\dfrac{1}{\sqrt{n}}\eta_{\mathcal{J}_{z},t-1}^{(1)}u_{t}
\end{array}\right) & \convp\sigma_{u}^{2}\Theta_{\mathcal{J}},\nonumber \\
\text{ where }\Theta_{\mathcal{J}} & =\left(\begin{array}{cc}
\dfrac{1}{2C_{\zeta}}{\rm lvar}(e_{\mathcal{J}_{x},t})\\
 & {\rm var}(\eta_{\mathcal{J}_{z},t}^{(1)})
\end{array}\right).\label{eq:Theta=000020J=000020def}
\end{align}
By Lemma 3.1 (iii) and Equation (14) of \citet{phillips2009econometric},
we have $\dfrac{1}{n^{1+\tau}}\sum_{t=1}^{n}\zeta_{\mathcal{J}_{x},t-1}\zeta_{\mathcal{J}_{x},t-1}^{\top}\convp\dfrac{1}{2C_{\zeta}}{\rm lvar}(e_{\mathcal{J}_{x},t}).$
By standard LLN, we can show that $\dfrac{1}{n}\sum_{t=1}^{n}\eta_{\mathcal{J}_{z},t-1}^{(1)}\eta_{\mathcal{J}_{z},t-1}^{(1)\top}\convp{\rm var}(\eta_{\mathcal{J}_{z},t}^{(1)})$.
Since $\sum_{t=1}^{n}\zeta_{\mathcal{J}_{x},t-1}\eta_{\mathcal{J}_{z},t-1}^{(1)\top}$
is the cross-product of the mildly integrated IVs $\zeta_{\mathcal{J}_{x},t-1}$
and the stationary components $\eta_{\mathcal{J}_{z},t-1}^{(1)},$
we have $\sum_{t=1}^{n}\zeta_{\mathcal{J}_{x},t-1}\eta_{\mathcal{J}_{z},t-1}^{(1)\top}=O_{p}(n)$
by Lemma B2 (i) of \citet{kostakis2015robust}. Therefore, 
\begin{align}
\sum_{t=1}^{n}{\rm var}_{t-1}\left(\begin{array}{c}
\dfrac{1}{\sqrt{n^{1+\tau}}}\zeta_{\mathcal{J}_{x},t-1}u_{t}\\
\dfrac{1}{\sqrt{n}}\sum_{t=1}^{n}\eta_{\mathcal{J}_{z},t-1}^{(1)}u_{t}
\end{array}\right) & =\sigma_{u}^{2}\left(\begin{array}{cc}
\dfrac{1}{n^{1+\tau}}\sum_{t=1}^{n}\zeta_{\mathcal{J}_{x},t-1}\zeta_{\mathcal{J}_{x},t-1}^{\top} & \dfrac{1}{n^{1+\tau/2}}\sum_{t=1}^{n}\zeta_{\mathcal{J}_{x},t-1}\eta_{\mathcal{J}_{z},t-1}^{(1)\top}\\
\dfrac{1}{n^{1+\tau/2}}\sum_{t=1}^{n}\eta_{\mathcal{J}_{z},t-1}^{(1)}\zeta_{\mathcal{J}_{x},t-1}^{\top} & \dfrac{1}{n}\sum_{t=1}^{n}\eta_{\mathcal{J}_{z},t-1}^{(1)}\eta_{\mathcal{J}_{z},t-1}^{(1)\top}
\end{array}\right)\nonumber \\
 & \convp\sigma_{u}^{2}\Theta_{\mathcal{J}},\label{eq:limit=000020conditional=000020var}
\end{align}
where $\Theta_{\mathcal{J}}$is defined in (\ref{eq:Theta=000020J=000020def}).
We complete the proof of (\ref{eq:joint=000020normal}).

\textbf{Proof of (\ref{eq:cov=000020mat=000020convp})}. By (\ref{eq:check=000020r=000020def}),
we have for any $j\in\mathcal{M}_{x}$
\begin{align*}
\sup_{t\in[n]}\left|\check{r}_{j,t}-\zeta_{j,t}\right| & =\hat{\varsigma}_{j}\sup_{t\in[n]}\left|W_{-j,t}^{\top}\hat{\varphi}^{(j)}\right|\\
 & \lep\sqrt{n^{\tau}}\cdot\sup_{t\in[n]}\|D_{-j}^{-1}W_{-j,t}\|_{\infty}\|D_{-j}\hat{\varphi}^{(j)}\|_{1}\\
 & \leq\sqrt{n^{\tau}}\cdot O_{p}(\log p)\cdot o_{p}(n^{-1+\tau/2})=o_{p}(n^{\tau-1/2}).
\end{align*}
where second row applies (\ref{eq:sigma=000020IV=000020rate}) and
(\ref{eq:D=000020hat=000020phi=000020L1}). In addition, for any $j_{1},j_{2}\in\mathcal{M}_{x}$,
\[\check{r}_{j_{1},t-1}\check{r}_{j_{2},t-1}-\zeta_{j_{1},t-1}\zeta_{j_{2},t-1} = (\check{r}_{j_{1},t-1}-\zeta_{j_{1},t-1})\zeta_{j_{2},t-1} + (\check{r}_{j_{2},t-1}-\zeta_{j_{2},t-1})\zeta_{j_{1},t-1} + (\check{r}_{j_{1},t-1}-\zeta_{j_{1},t-1})(\check{r}_{j_{2},t-1}-\zeta_{j_{2},t-1}).\]
Therefore, 
\begin{align}
 & \dfrac{1}{n^{1+\tau}}\sum_{t=1}^{n}\check{r}_{j_{1},t-1}\check{r}_{j_{2},t-1}-\dfrac{1}{n^{1+\tau}}\sum_{t=1}^{n}\zeta_{j_{1},t-1}\zeta_{j_{2},t-1}\nonumber \\
= & \dfrac{o_{p}(n^{\tau-1/2})}{n^{1+\tau}}\sum_{t=1}^{n}|\zeta_{j_{1},t-1}|+\dfrac{o_{p}(n^{\tau-1/2})}{n^{1+\tau}}\sum_{t=1}^{n}|\zeta_{j_{2},t-1}|+\dfrac{1}{n^{1+\tau}}\sum_{t=1}^{n}o_{p}(n^{2\tau-1})\nonumber \\
= & o_{p}(\sqrt{n^{(\tau-1)}(\log p)^{3}})+o_{p}(n^{\tau-1})=o_{p}(1),\label{eq:check=000020r=000020j1=000020j2}
\end{align}
where the second step applies $\sup_{t\in[n]}|\zeta_{j_{1},t}|\lep\sqrt{n^{\tau}(\log p)^{3}}$
by (\ref{eq:sup_IV}). In addition, by (\ref{eq:check=000020r=000020decom=000020z}),
we have for any $k\in\mathcal{M}_{z}$, 
\begin{align*}
\sup_{t\in[n]}\left|\check{r}_{k,t}-\eta_{k,t-1}^{(1)}\right| & \leq\hat{\varsigma}_{k}\cdot\sup_{t\in[n]}\left|W_{-k,t}^{\top}\left(\varphi^{*(k)}-\hat{\varphi}^{(k)}\right)+n^{-\tau}C_{\zeta}\phi_{k,t-1}\right|\\
 & \leq O_{p}(1)\cdot\sup_{t\in[n]}\|D_{-k}^{-1}W_{-k,t}\|_{\infty}\|D_{-j}(\hat{\varphi}^{(k)}-\varphi^{*(k)})\|_{1}\\
 & =O_{p}(1)\cdot O_{p}(\log p)\cdot O_{p}\left(\dfrac{s^{2}(\log p)^{6+\frac{1}{2r}}}{\sqrt{n^{\tau\wedge(1-\tau)}}}\right)=o_{p}(1/(\log p)^{3}),
\end{align*}
where the second row applies the rate $\hat{\varsigma}_{k}=O_{p}(1)$
for stationary regressors by (\ref{eq:sigma=000020IV=000020rate-1}),
and the third row applies $(\ref{eq:bound=000020D=000020W})$ and
Proposition \ref{prop:DB-Aux}. Therefore, following the arguments for (\ref{eq:check=000020r=000020j1=000020j2}) we have for any $k_{1},k_{2}\in\mathcal{M}_{z}$,
\begin{align}
 & \dfrac{1}{n}\sum_{t=1}^{n}\check{r}_{k_{1},t-1}\check{r}_{k_{2},t-1}-\dfrac{1}{n}\sum_{t=1}^{n}\eta_{k_{1},t-1}^{(1)}\eta_{k_{2},t-1}^{(1)}\nonumber \\
= & o_{p}\left(\frac{1}{n(\log p)^{3}}\right)\sum_{t=1}^{n}|\eta_{k_{1},t-1}^{(1)}|+o_{p}\left(\frac{1}{n(\log p)^{3}}\right)\sum_{t=1}^{n}|\eta_{k_{2},t-1}^{(1)}|+\dfrac{1}{n}\sum_{t=1}^{n}o_{p}\left(\frac{1}{(\log p)^{3}}\right)\label{eq:check=000020r=000020k1=000020k2}\\
= & o_{p}(1).\nonumber 
\end{align}
 In addition, similar arguments yield that for any $j\in\mathcal{M}_{x}$ and $k\in\mathcal{M}_{z}$,
\begin{align}
 & \dfrac{1}{n^{1+\tau/2}}\sum_{t=1}^{n}\check{r}_{j,t-1}\check{r}_{k,t-1}-\dfrac{1}{n^{1+\tau/2}}\sum_{t=1}^{n}\zeta_{j,t-1}\eta_{k,t-1}^{(1)}\nonumber \\
= & \dfrac{o_{p}(n^{\tau-1/2})}{n^{1+\tau/2}}\sum_{t=1}^{n}|\eta_{k,t-1}^{(1)}|+o_{p}\left(\frac{1}{n^{1+\tau/2}(\log p)^{3}}\right)\sum_{t=1}^{n}|\zeta_{j,t-1}|+\dfrac{1}{n^{1+\tau/2}}\sum_{t=1}^{n}o_{p}\left(\frac{n^{\tau-1/2}}{(\log p)^{3}}\right)\label{eq:check=000020r=000020j=000020k}\\
= & o_{p}(n^{(\tau-1)/2})+o_{p}\left(\frac{n^{1+\tau/2}(\log p)^{3/2}}{n^{1+\tau/2}(\log p)^{3}}\right)+o_{p}(1),
\end{align}
where the second row applies $\sup_{t\in[n]}|\zeta_{j_{1},t}|\lep\sqrt{n^{\tau}(\log p)^{3}}$
by (\ref{eq:sup_IV}). Therefore, 
\begin{align*}
\tilde{\Theta}_{\mathcal{J}} & =\left(\begin{array}{cc}
\dfrac{1}{n^{1+\tau}}\sum_{t=1}^{n}\check{r}_{\mathcal{J}_{x},t-1}\check{r}_{\mathcal{J}_{x},t-1}^{\top} & \dfrac{1}{n^{1+\tau/2}}\sum_{t=1}^{n}\check{r}_{\mathcal{J}_{x},t-1}\check{r}_{\mathcal{J}_{z},t-1}^{\top}\\
\dfrac{1}{n^{1+\tau/2}}\sum_{t=1}^{n}\check{r}_{\mathcal{J}_{z},t-1}\check{r}_{\mathcal{J}_{x},t-1}^{\top} & \dfrac{1}{n}\sum_{t=1}^{n}\check{r}_{\mathcal{J}_{z},t-1}\check{r}_{\mathcal{J}_{z},t-1}^{\top}
\end{array}\right)\\
 & =\left(\begin{array}{cc}
\dfrac{1}{n^{1+\tau}}\sum_{t=1}^{n}\zeta_{\mathcal{J}_{x},t-1}\zeta_{\mathcal{J}_{x},t-1}^{\top} & \dfrac{1}{n^{1+\tau/2}}\sum_{t=1}^{n}\zeta_{\mathcal{J}_{x},t-1}\eta_{\mathcal{J}_{z},t-1}^{(1)\top}\\
\dfrac{1}{n^{1+\tau/2}}\sum_{t=1}^{n}\eta_{\mathcal{J}_{z},t-1}^{(1)}\zeta_{\mathcal{J}_{x},t-1}^{\top} & \dfrac{1}{n}\sum_{t=1}^{n}\eta_{\mathcal{J}_{z},t-1}^{(1)}\eta_{\mathcal{J}_{z},t-1}^{(1)\top}
\end{array}\right)+o_{p}(1)\\
 & \convp\Theta_{\mathcal{J}},
\end{align*}
where the $o_{p}(1)$ in the second row applies (\ref{eq:check=000020r=000020j1=000020j2}),
(\ref{eq:check=000020r=000020k1=000020k2}), and (\ref{eq:check=000020r=000020j=000020k}),
and the limit follows (\ref{eq:limit=000020conditional=000020var}).
With the essential equations (\ref{eq:pi=000020x=000020rate}), (\ref{eq:pi=000020z=000020rate}),
(\ref{eq:joint=000020normal}) , and (\ref{eq:cov=000020mat=000020convp})
verified, we complete the proof of Theorem \ref{thm:wald}.
\end{proof}
\clearpage

\section{Additional Simulation Results\protect\label{sec:Additional-Simulation-Results}}

\subsection{Simulation Results with More Nonzero Coefficients\protect\label{sec:mc_more_nonzero_coef}}

We follow the same setup in Section \ref{subsec:Setup}, with a modification
in \eqref{eq:true_coef} to have 
\begin{equation}
\gamma^{*}=(\gamma_{1}^{*},0.5\times1_{2}^{\top},0.25\times1_{2}^{\top},\frac{0.25}{6^{2}},\ldots,\frac{0.25}{10^{2}},0_{p_{z}-10}^{\top})^{\top}.\label{eq:add_simulation_new_gamma}
\end{equation}
The empirical sizes are reported in Table \ref{tab:mc_size_comp_iid}
and \ref{tab:mc_size_comp_ar1}, and the empirical power is depicted
in Figure \ref{fig:power_all_comp}. The results mirror those 
of the benchmark setup in Section \ref{sec:Simulation}, which demonstrates
the robust performance of XDlasso in finite sample with more control
variables associated with nonzero coefficients. 
\begin{table}[hp]
\begin{centering}
\caption{Empirical size and length of confidence interval: IID innovations\protect\label{tab:mc_size_comp_iid}}
{\small\begin{tabular}{c|cccc|cccc|cccc}
\hline 
\multirow{3}{*}{$n$ } & \multicolumn{4}{c|}{Oracle} & \multicolumn{4}{c|}{Calibrated} & \multicolumn{4}{c}{CV}\tabularnewline
 & \multicolumn{2}{c}{IVX Oracle} & \multicolumn{2}{c|}{OLS Oracle} & \multicolumn{2}{c}{XDlasso} & \multicolumn{2}{c|}{Dlasso} & \multicolumn{2}{c}{XDlasso} & \multicolumn{2}{c}{Dlasso}\tabularnewline
\cline{2-13}
 & Size  & Len.  & Size  & Len.  & Size  & Len.  & Size  & Len.  & Size  & Len.  & Size  & Len. \tabularnewline
\hline 
\multicolumn{13}{c}{$\mathbb{H}_{0}:\beta_{1}^{*}=0$ for nonstationary regressor}\tabularnewline
\hline 
\multicolumn{13}{c}{$\left(p_{x},p_{z}\right)=\left(50,100\right)$}\tabularnewline
\hline 
200  & 0.037  & 0.217  & 0.143  & 0.099  & 0.047  & 0.223  & 0.377  & 0.104  & 0.060  & 0.230  & 0.436  & 0.156 \tabularnewline
300  & 0.047  & 0.155  & 0.142  & 0.066  & 0.047  & 0.164  & 0.430  & 0.078  & 0.064  & 0.169  & 0.524  & 0.120 \tabularnewline
400  & 0.046  & 0.122  & 0.140  & 0.050  & 0.052  & 0.133  & 0.479  & 0.064  & 0.067  & 0.135  & 0.547  & 0.096 \tabularnewline
500  & 0.047  & 0.101  & 0.143  & 0.040  & 0.054  & 0.112  & 0.498  & 0.054  & 0.072  & 0.115  & 0.577  & 0.077 \tabularnewline
600  & 0.044  & 0.087  & 0.135  & 0.033  & 0.045  & 0.097  & 0.509  & 0.047  & 0.057  & 0.099  & 0.579  & 0.065 \tabularnewline
\hline 
\multicolumn{13}{c}{$\left(p_{x},p_{z}\right)=\left(100,150\right)$}\tabularnewline
\hline 
200  & 0.046  & 0.215  & 0.147  & 0.099  & 0.046  & 0.220  & 0.371  & 0.101  & 0.057  & 0.228  & 0.498  & 0.159 \tabularnewline
300  & 0.033  & 0.154  & 0.145  & 0.066  & 0.044  & 0.162  & 0.452  & 0.076  & 0.052  & 0.169  & 0.620  & 0.130 \tabularnewline
400  & 0.039  & 0.122  & 0.142  & 0.050  & 0.048  & 0.129  & 0.517  & 0.062  & 0.064  & 0.134  & 0.689  & 0.110 \tabularnewline
500  & 0.046  & 0.101  & 0.141  & 0.040  & 0.048  & 0.111  & 0.557  & 0.053  & 0.069  & 0.114  & 0.704  & 0.089 \tabularnewline
600  & 0.039  & 0.088  & 0.148  & 0.033  & 0.049  & 0.095  & 0.605  & 0.046  & 0.070  & 0.100  & 0.738  & 0.077 \tabularnewline
\hline 
\multicolumn{13}{c}{$\left(p_{x},p_{z}\right)=\left(150,300\right)$}\tabularnewline
\hline 
200  & 0.042  & 0.218  & 0.141  & 0.100  & 0.041  & 0.215  & 0.361  & 0.096  & 0.047  & 0.222  & 0.495  & 0.140 \tabularnewline
300  & 0.051  & 0.155  & 0.134  & 0.066  & 0.045  & 0.157  & 0.435  & 0.072  & 0.055  & 0.166  & 0.594  & 0.114 \tabularnewline
400  & 0.045  & 0.122  & 0.146  & 0.049  & 0.047  & 0.127  & 0.485  & 0.059  & 0.062  & 0.135  & 0.649  & 0.095 \tabularnewline
500  & 0.040  & 0.101  & 0.146  & 0.039  & 0.048  & 0.108  & 0.532  & 0.050  & 0.060  & 0.114  & 0.690  & 0.084 \tabularnewline
600  & 0.037  & 0.087  & 0.153  & 0.033  & 0.050  & 0.092  & 0.581  & 0.044  & 0.056  & 0.099  & 0.736  & 0.073 \tabularnewline
\hline 
\multicolumn{13}{c}{$\mathbb{H}_{0}:\gamma_{1}^{*}=0$ for stationary regressor}\tabularnewline
\hline 
\multicolumn{13}{c}{$\left(p_{x},p_{z}\right)=\left(50,100\right)$}\tabularnewline
\hline 
200  & 0.044  & 0.379  & 0.054  & 0.327  & 0.066  & 0.325  & 0.078  & 0.288  & 0.069  & 0.323  & 0.080  & 0.287 \tabularnewline
300  & 0.047  & 0.298  & 0.055  & 0.265  & 0.067  & 0.265  & 0.062  & 0.240  & 0.070  & 0.264  & 0.065  & 0.240 \tabularnewline
400  & 0.048  & 0.253  & 0.054  & 0.229  & 0.059  & 0.229  & 0.064  & 0.210  & 0.062  & 0.229  & 0.062  & 0.210 \tabularnewline
500  & 0.044  & 0.223  & 0.050  & 0.204  & 0.054  & 0.204  & 0.057  & 0.189  & 0.056  & 0.205  & 0.063  & 0.190 \tabularnewline
600  & 0.045  & 0.201  & 0.050  & 0.186  & 0.054  & 0.186  & 0.056  & 0.174  & 0.053  & 0.187  & 0.057  & 0.174 \tabularnewline
\hline 
\multicolumn{13}{c}{$\left(p_{x},p_{z}\right)=\left(100,150\right)$}\tabularnewline
\hline 
200  & 0.050  & 0.377  & 0.062  & 0.326  & 0.069  & 0.325  & 0.065  & 0.289  & 0.074  & 0.320  & 0.068  & 0.284 \tabularnewline
300  & 0.047  & 0.297  & 0.060  & 0.265  & 0.063  & 0.264  & 0.064  & 0.239  & 0.067  & 0.261  & 0.063  & 0.237 \tabularnewline
400  & 0.044  & 0.252  & 0.057  & 0.229  & 0.053  & 0.228  & 0.056  & 0.209  & 0.054  & 0.227  & 0.060  & 0.209 \tabularnewline
500  & 0.052  & 0.223  & 0.051  & 0.204  & 0.059  & 0.203  & 0.057  & 0.188  & 0.060  & 0.203  & 0.061  & 0.188 \tabularnewline
600  & 0.045  & 0.202  & 0.051  & 0.186  & 0.062  & 0.185  & 0.060  & 0.173  & 0.059  & 0.186  & 0.060  & 0.173 \tabularnewline
\hline 
\multicolumn{13}{c}{$\left(p_{x},p_{z}\right)=\left(150,300\right)$}\tabularnewline
\hline 
200  & 0.042  & 0.379  & 0.053  & 0.328  & 0.062  & 0.328  & 0.055  & 0.291  & 0.069  & 0.315  & 0.057  & 0.281 \tabularnewline
300  & 0.041  & 0.299  & 0.050  & 0.266  & 0.062  & 0.265  & 0.060  & 0.241  & 0.059  & 0.260  & 0.060  & 0.237 \tabularnewline
400  & 0.040  & 0.253  & 0.052  & 0.229  & 0.058  & 0.228  & 0.066  & 0.210  & 0.056  & 0.226  & 0.065  & 0.208 \tabularnewline
500  & 0.048  & 0.223  & 0.062  & 0.205  & 0.064  & 0.203  & 0.069  & 0.189  & 0.067  & 0.202  & 0.071  & 0.188 \tabularnewline
600  & 0.050  & 0.202  & 0.061  & 0.187  & 0.064  & 0.185  & 0.064  & 0.173  & 0.061  & 0.185  & 0.067  & 0.173 \tabularnewline
\hline 
\end{tabular}

}
\par\end{centering}
{

{\footnotesize\textit{Notes}}{\footnotesize : The data generating process
corresponds to \eqref{eq:dgp_iid_inno}. The coefficients are specified
in \eqref{eq:add_simulation_new_gamma}. The upper and lower panels
report the empirical size of testing the null hypotheses $\mathbb{H}_{0}:\beta_{1}^{*}=0$
and $\mathbb{H}_{0}:\gamma_{1}^{*}=0$, respectively, at a 5\% nominal
significance level. ``Size'' is calculated as $R^{-1}\sum_{r=1}^{R}\mathbf{1}\left[|t^{\text{XD}(r)}|>\mathrm{\Phi}_{0.975}\right]$
across $R=2,000$ replications, where $t^{\text{XD}(r)}$ is computed
based on \eqref{eq:XD=000020t=000020stat} for the $r$-th replication,
and the critical value $\mathrm{\Phi}_{0.975}\left(\approx1.96\right)$
is the 97.5-th percentile of the standard normal distribution. ``Len.''
refers to the median length of the 95\% confidence intervals across
replications. The IVX oracle and OLS oracle are infeasible estimators.
The ``Calibrated'' and ``CV'' columns refer to the methods used
for choosing the tuning parameters through calibration and cross-validation,
respectively. }}{\footnotesize\par}
\end{table}
 
\begin{table}[hp]
\begin{centering}
\caption{Empirical size and length of confidence: AR(1) innovations\protect\label{tab:mc_size_comp_ar1}}
{\small\begin{tabular}{c|cccc|cccc|cccc}
\hline 
\multirow{3}{*}{$n$ } & \multicolumn{4}{c|}{Oracle} & \multicolumn{4}{c|}{Calibrated } & \multicolumn{4}{c}{CV}\tabularnewline
 & \multicolumn{2}{c}{IVX Oracle} & \multicolumn{2}{c|}{OLS Oracle} & \multicolumn{2}{c}{XDlasso} & \multicolumn{2}{c|}{Dlasso} & \multicolumn{2}{c}{XDlasso} & \multicolumn{2}{c}{Dlasso}\tabularnewline
\cline{2-13}
 & Size  & Len.  & Size  & Len.  & Size  & Len.  & Size  & Len.  & Size  & Len.  & Size  & Len. \tabularnewline
\hline 
\multicolumn{13}{c}{$\mathbb{H}_{0}:\beta_{1}^{*}=0$ for nonstationary regressor}\tabularnewline
\hline 
\multicolumn{13}{c}{$\left(p_{x},p_{z}\right)=\left(50,100\right)$}\tabularnewline
\hline 
200  & 0.046  & 0.164  & 0.151  & 0.074  & 0.048  & 0.168  & 0.419  & 0.079  & 0.073  & 0.173  & 0.484  & 0.134 \tabularnewline
300  & 0.046  & 0.112  & 0.140  & 0.048  & 0.051  & 0.122  & 0.460  & 0.059  & 0.078  & 0.125  & 0.563  & 0.096 \tabularnewline
400  & 0.047  & 0.088  & 0.150  & 0.036  & 0.051  & 0.098  & 0.521  & 0.048  & 0.081  & 0.100  & 0.590  & 0.072 \tabularnewline
500  & 0.040  & 0.073  & 0.151  & 0.028  & 0.048  & 0.083  & 0.553  & 0.041  & 0.071  & 0.083  & 0.609  & 0.056 \tabularnewline
600  & 0.049  & 0.062  & 0.141  & 0.024  & 0.049  & 0.071  & 0.561  & 0.035  & 0.071  & 0.072  & 0.606  & 0.046 \tabularnewline
\hline 
\multicolumn{13}{c}{$\left(p_{x},p_{z}\right)=\left(100,150\right)$}\tabularnewline
\hline 
200  & 0.044  & 0.159  & 0.144  & 0.073  & 0.055  & 0.165  & 0.393  & 0.076  & 0.090  & 0.174  & 0.544  & 0.131 \tabularnewline
300  & 0.039  & 0.113  & 0.147  & 0.048  & 0.052  & 0.120  & 0.494  & 0.057  & 0.084  & 0.125  & 0.650  & 0.106 \tabularnewline
400  & 0.035  & 0.088  & 0.140  & 0.036  & 0.055  & 0.096  & 0.556  & 0.047  & 0.085  & 0.098  & 0.698  & 0.087 \tabularnewline
500  & 0.043  & 0.073  & 0.153  & 0.029  & 0.060  & 0.081  & 0.605  & 0.040  & 0.089  & 0.084  & 0.739  & 0.069 \tabularnewline
600  & 0.036  & 0.063  & 0.147  & 0.023  & 0.055  & 0.070  & 0.631  & 0.035  & 0.079  & 0.072  & 0.753  & 0.057 \tabularnewline
\hline 
\multicolumn{13}{c}{$\left(p_{x},p_{z}\right)=\left(150,300\right)$}\tabularnewline
\hline 
200  & 0.045  & 0.162  & 0.144  & 0.073  & 0.053  & 0.162  & 0.388  & 0.072  & 0.088  & 0.169  & 0.543  & 0.113 \tabularnewline
300  & 0.051  & 0.112  & 0.149  & 0.048  & 0.049  & 0.117  & 0.477  & 0.054  & 0.079  & 0.126  & 0.605  & 0.087 \tabularnewline
400  & 0.034  & 0.088  & 0.150  & 0.035  & 0.044  & 0.093  & 0.547  & 0.044  & 0.066  & 0.100  & 0.676  & 0.074 \tabularnewline
500  & 0.044  & 0.074  & 0.149  & 0.028  & 0.052  & 0.077  & 0.577  & 0.037  & 0.072  & 0.083  & 0.717  & 0.065 \tabularnewline
600  & 0.044  & 0.063  & 0.155  & 0.023  & 0.049  & 0.067  & 0.631  & 0.033  & 0.067  & 0.072  & 0.746  & 0.057 \tabularnewline
\hline 
\multicolumn{13}{c}{$\mathbb{H}_{0}:\gamma_{1}^{*}=0$ for AR(1) regressor}\tabularnewline
\hline 
\multicolumn{13}{c}{$\left(p_{x},p_{z}\right)=\left(50,100\right)$}\tabularnewline
\hline 
200  & 0.040  & 0.385  & 0.062  & 0.316  & 0.067  & 0.334  & 0.075  & 0.275  & 0.072  & 0.331  & 0.081  & 0.273 \tabularnewline
300  & 0.048  & 0.301  & 0.054  & 0.255  & 0.065  & 0.269  & 0.068  & 0.229  & 0.069  & 0.268  & 0.074  & 0.227 \tabularnewline
400  & 0.052  & 0.253  & 0.050  & 0.220  & 0.066  & 0.231  & 0.066  & 0.200  & 0.068  & 0.231  & 0.071  & 0.200 \tabularnewline
500  & 0.048  & 0.223  & 0.050  & 0.196  & 0.054  & 0.205  & 0.060  & 0.181  & 0.058  & 0.205  & 0.063  & 0.180 \tabularnewline
600  & 0.046  & 0.200  & 0.053  & 0.179  & 0.054  & 0.186  & 0.064  & 0.166  & 0.055  & 0.187  & 0.068  & 0.165 \tabularnewline
\hline 
\multicolumn{13}{c}{$\left(p_{x},p_{z}\right)=\left(100,150\right)$}\tabularnewline
\hline 
200  & 0.041  & 0.386  & 0.056  & 0.318  & 0.067  & 0.334  & 0.069  & 0.275  & 0.071  & 0.328  & 0.076  & 0.271 \tabularnewline
300  & 0.044  & 0.301  & 0.049  & 0.256  & 0.060  & 0.268  & 0.067  & 0.228  & 0.065  & 0.265  & 0.067  & 0.225 \tabularnewline
400  & 0.039  & 0.254  & 0.047  & 0.220  & 0.058  & 0.230  & 0.057  & 0.199  & 0.057  & 0.228  & 0.060  & 0.198 \tabularnewline
500  & 0.041  & 0.223  & 0.045  & 0.196  & 0.051  & 0.204  & 0.060  & 0.179  & 0.052  & 0.204  & 0.060  & 0.179 \tabularnewline
600  & 0.045  & 0.201  & 0.049  & 0.178  & 0.053  & 0.185  & 0.056  & 0.164  & 0.056  & 0.185  & 0.057  & 0.164 \tabularnewline
\hline 
\multicolumn{13}{c}{$\left(p_{x},p_{z}\right)=\left(150,300\right)$}\tabularnewline
\hline 
200  & 0.035  & 0.387  & 0.044  & 0.319  & 0.064  & 0.337  & 0.057  & 0.278  & 0.061  & 0.326  & 0.067  & 0.269 \tabularnewline
300  & 0.042  & 0.302  & 0.059  & 0.257  & 0.068  & 0.269  & 0.067  & 0.229  & 0.071  & 0.265  & 0.067  & 0.225 \tabularnewline
400  & 0.047  & 0.254  & 0.057  & 0.220  & 0.065  & 0.230  & 0.072  & 0.199  & 0.064  & 0.228  & 0.070  & 0.197 \tabularnewline
500  & 0.047  & 0.223  & 0.061  & 0.196  & 0.060  & 0.204  & 0.066  & 0.179  & 0.057  & 0.203  & 0.065  & 0.178 \tabularnewline
600  & 0.044  & 0.201  & 0.058  & 0.178  & 0.061  & 0.184  & 0.065  & 0.164  & 0.058  & 0.184  & 0.066  & 0.164 \tabularnewline
\hline 
\end{tabular}
}
\par\end{centering}
{

{\footnotesize\textit{Notes}}{\footnotesize : The data generating process
corresponds to \eqref{eq:dgp_iid_inno}. The coefficients are specified
in \eqref{eq:add_simulation_new_gamma}. The upper and lower panels
report the empirical size of testing the null hypotheses $\mathbb{H}_{0}:\beta_{1}^{*}=0$
and $\mathbb{H}_{0}:\gamma_{1}^{*}=0$ at a 5\% nominal significance
level, respectively. ``Size'' is calculated as $R^{-1}\sum_{r=1}^{R}\mathbf{1}\left[|t^{\text{XD}(r)}|>\mathrm{\Phi}_{0.975}\right]$
across $R=2,000$ replications, where $t^{\text{XD}(r)}$ is computed
based on \eqref{eq:XD=000020t=000020stat} for the $r$-th replication,
and the critical value $\mathrm{\Phi}_{0.975}\left(\approx1.96\right)$
is the 97.5-th percentile of the standard normal distribution. ``Len.''
refers to the median length of the 95\% confidence intervals across
replications. The IVX oracle and OLS oracle are infeasible estimators.
The ``Calibrated'' and ``CV'' columns refer to the methods used
for choosing the tuning parameters through calibration and cross-validation,
respectively. }}{\footnotesize\par}
\end{table}

\begin{figure}[h]
\begin{centering}
\begin{subfigure}[b]{0.48\textwidth} \centering 
\includegraphics[height=0.36\textheight]{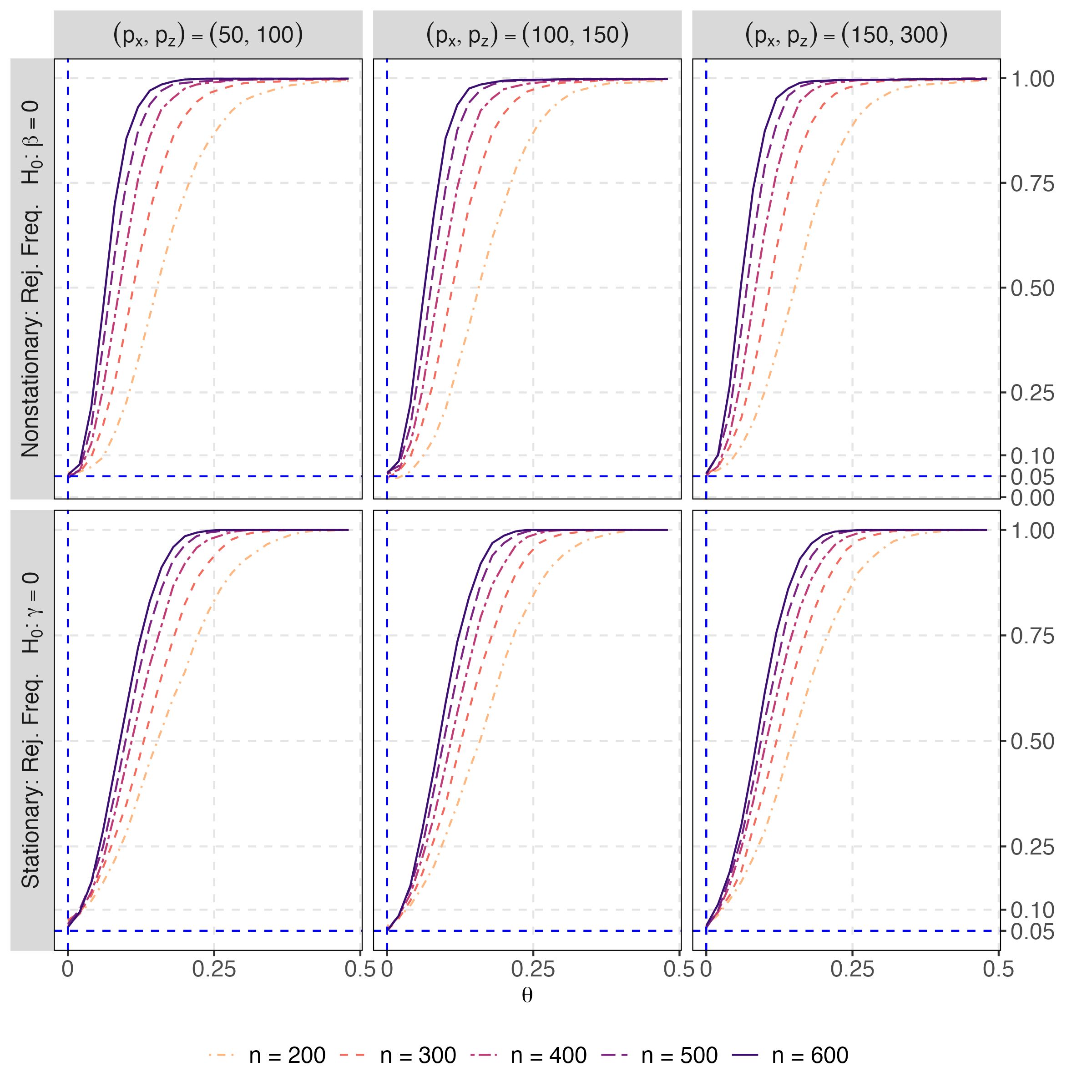} \label{fig:iid_power_comp}
\caption{IID Innovations}
\end{subfigure} \begin{subfigure}[b]{0.48\textwidth} \centering
\includegraphics[height=0.36\textheight]{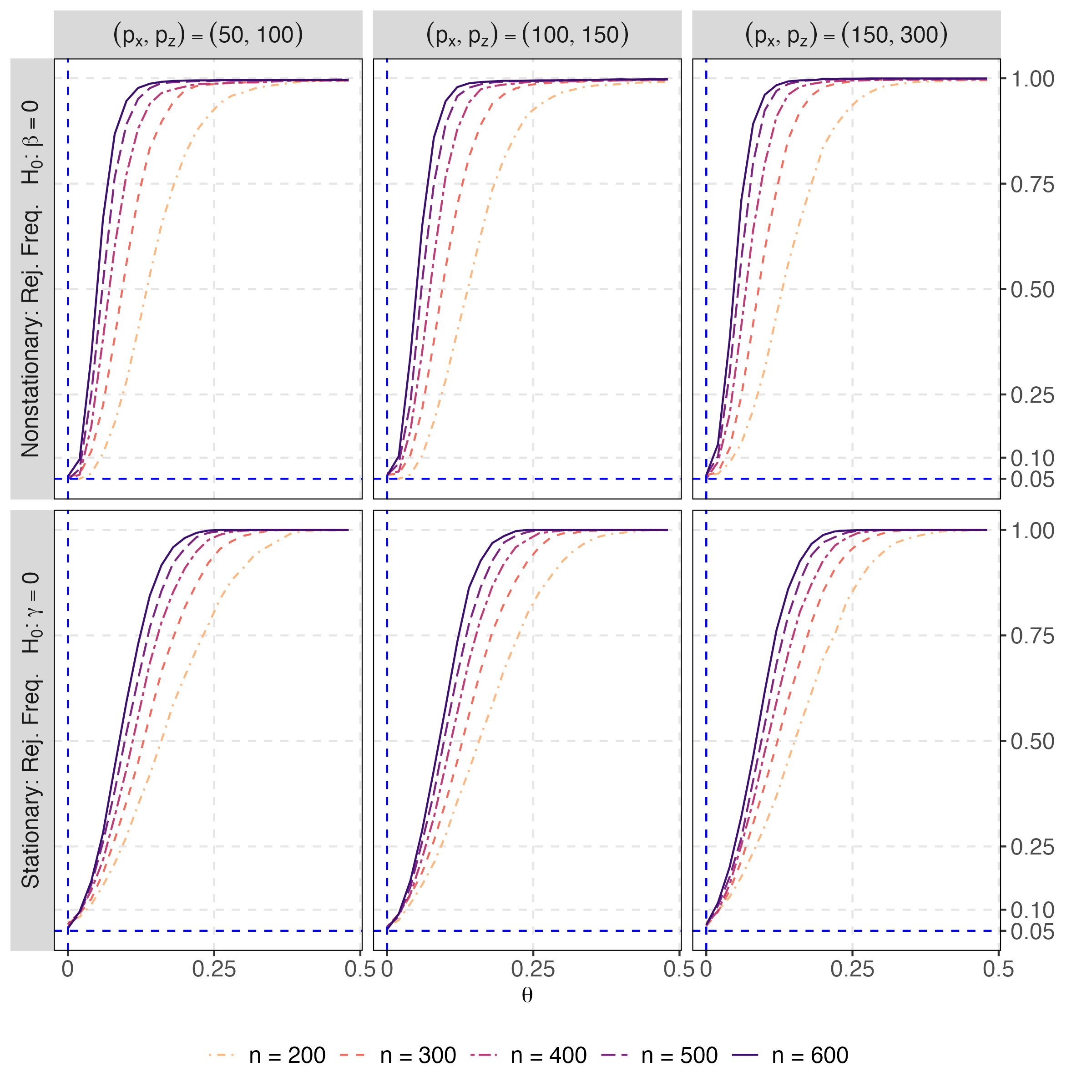} \label{fig:ar1_power_comp}
\caption{AR(1) Innovations}
\end{subfigure}
\caption{Power curves of XDlasso inference\protect\label{fig:power_all_comp}}
\par\end{centering}
\smallskip
{

{\footnotesize\textit{Notes}}{\footnotesize : The left and right panels
correspond to DGPs \eqref{eq:dgp_iid_inno} and \eqref{eq:dgp_ar1_inno},
respectively. The coefficients are specified in \eqref{eq:add_simulation_new_gamma}.
In each subplot, the first row depicts the empirical power function
for $\beta_{1}^{*}$, associated with a nonstationary regressor, across
various $\left(p_{x},p_{z}\right)$ configurations, while the second
row pertains to $\gamma_{1}^{*}$, associated with a stationary regressor.
The empirical power is calculated as $R^{-1}\sum_{r=1}^{R}\mathbf{1}\left[|t^{\text{XD}(r)}|>\mathrm{\Phi}_{0.975}\right]$
across $R=2,000$ replications, where $t^{\text{XD}(r)}$ is computed
based on \eqref{eq:XD=000020t=000020stat} for the $r$-th replication,
and the critical value $\mathrm{\Phi}_{0.975}\left(\approx1.96\right)$
is the 97.5-th percentile of the standard normal distribution.}}{\footnotesize\par}
\end{figure}

\subsection{Simulation Results with Cointegrated Regressors\protect\label{sec:mc_cointegration}}

In this section, we follow the data generating process in Section
\ref{subsec:Setup} with the same innovation processes. The LUR regressors
are generated by $x_{j,t}=\rho_{j}^{*}x_{j,t-1}+e_{j,t},$ for $j=1,2,\cdots,p_{x}-3,p_{x}-1$
with $\rho^{*}=(1,1-1/n,1+1/n,1,1-1/n,1+1/n,\cdots)^{\top}\in\mathbb{R}^{p_{x}-2}$,
and $x_{j,t}=x_{j-1,t}-e_{j,t},$ for $j=p_{x}-2$ and $p_{x}$, so
that the last four LURs are cointegrated. The true coefficient vectors
are: 
\begin{equation}
\beta^{*}=\left(\beta_{1}^{*},\dfrac{0.5}{\sqrt{n}}\times1_{4}^{\top},0_{p_{x}-7}^{\top},0.5,-0.5\right)^{\top},\ \ \gamma^{*}=(\gamma_{1}^{*},0.5\times1_{2}^{\top},0.25\times1_{2}^{\top},0_{p_{z}-5}^{\top})^{\top},\label{eq:true_coef-coint}
\end{equation}
so that we include one cointegration residual with nonzero coefficients,
while the other is treated as redundant control. The empirical sizes
are reported in Table \ref{tab:mc_size_coint_iid} and \ref{tab:mc_size_coint_ar1}.
In this setting, XDlasso continues to exhibit good size control in
finite sample with similar performance as in the benchmark setup.
\begin{table}[hp]
\begin{centering}
\caption{Empirical size and length of confidence interval with cointegrated
regressors: IID innovations\protect\label{tab:mc_size_coint_iid}}
{\small\begin{tabular}{c|cccc|cccc|cccc}
\hline 
\multirow{3}{*}{$n$ } & \multicolumn{4}{c|}{Oracle} & \multicolumn{4}{c|}{Calibrated } & \multicolumn{4}{c}{CV}\tabularnewline
 & \multicolumn{2}{c}{IVX Oracle} & \multicolumn{2}{c|}{OLS Oracle} & \multicolumn{2}{c}{XDlasso} & \multicolumn{2}{c|}{Dlasso} & \multicolumn{2}{c}{XDlasso} & \multicolumn{2}{c}{Dlasso}\tabularnewline
\cline{2-13}
 & Size  & Len.  & Size  & Len.  & Size  & Len.  & Size  & Len.  & Size  & Len.  & Size  & Len. \tabularnewline
\hline 
\multicolumn{13}{c}{$\mathbb{H}_{0}:\beta_{1}^{*}=0$ for I(1) regressor (raw data)}\tabularnewline
\hline 
\multicolumn{13}{c}{$\left(p_{x},p_{z}\right)=\left(50,100\right)$}\tabularnewline
\hline 
200  & 0.051  & 0.217  & 0.146  & 0.099  & 0.040  & 0.224  & 0.366  & 0.106  & 0.050  & 0.230  & 0.408  & 0.154 \tabularnewline
300  & 0.053  & 0.154  & 0.150  & 0.066  & 0.046  & 0.167  & 0.449  & 0.080  & 0.060  & 0.170  & 0.528  & 0.121 \tabularnewline
400  & 0.056  & 0.121  & 0.140  & 0.049  & 0.050  & 0.134  & 0.490  & 0.065  & 0.060  & 0.137  & 0.575  & 0.095 \tabularnewline
500  & 0.054  & 0.101  & 0.142  & 0.039  & 0.043  & 0.113  & 0.523  & 0.055  & 0.059  & 0.114  & 0.593  & 0.078 \tabularnewline
600  & 0.046  & 0.087  & 0.137  & 0.033  & 0.049  & 0.098  & 0.540  & 0.048  & 0.065  & 0.099  & 0.594  & 0.064 \tabularnewline
\hline 
\multicolumn{13}{c}{$\left(p_{x},p_{z}\right)=\left(100,150\right)$}\tabularnewline
\hline 
200  & 0.051  & 0.217  & 0.149  & 0.100  & 0.062  & 0.224  & 0.379  & 0.103  & 0.067  & 0.225  & 0.497  & 0.160 \tabularnewline
300  & 0.043  & 0.155  & 0.139  & 0.066  & 0.054  & 0.165  & 0.484  & 0.077  & 0.073  & 0.170  & 0.596  & 0.130 \tabularnewline
400  & 0.041  & 0.121  & 0.140  & 0.050  & 0.057  & 0.132  & 0.535  & 0.063  & 0.076  & 0.138  & 0.675  & 0.107 \tabularnewline
500  & 0.044  & 0.101  & 0.130  & 0.040  & 0.053  & 0.111  & 0.560  & 0.054  & 0.074  & 0.114  & 0.702  & 0.089 \tabularnewline
600  & 0.044  & 0.086  & 0.125  & 0.033  & 0.046  & 0.096  & 0.609  & 0.047  & 0.066  & 0.099  & 0.747  & 0.078 \tabularnewline
\hline 
\multicolumn{13}{c}{$\left(p_{x},p_{z}\right)=\left(150,300\right)$}\tabularnewline
\hline 
200  & 0.041  & 0.215  & 0.150  & 0.099  & 0.057  & 0.219  & 0.346  & 0.097  & 0.067  & 0.223  & 0.491  & 0.142 \tabularnewline
300  & 0.036  & 0.154  & 0.132  & 0.067  & 0.054  & 0.158  & 0.441  & 0.073  & 0.059  & 0.166  & 0.585  & 0.113 \tabularnewline
400  & 0.043  & 0.121  & 0.136  & 0.050  & 0.056  & 0.129  & 0.505  & 0.060  & 0.062  & 0.134  & 0.647  & 0.096 \tabularnewline
500  & 0.048  & 0.101  & 0.130  & 0.040  & 0.054  & 0.107  & 0.545  & 0.051  & 0.067  & 0.112  & 0.679  & 0.085 \tabularnewline
600  & 0.044  & 0.087  & 0.129  & 0.033  & 0.059  & 0.092  & 0.600  & 0.044  & 0.072  & 0.098  & 0.729  & 0.074 \tabularnewline
\hline 
\multicolumn{13}{c}{$\mathbb{H}_{0}:\gamma_{1}^{*}=0$ for stationary regressor}\tabularnewline
\hline 
\multicolumn{13}{c}{$\left(p_{x},p_{z}\right)=\left(50,100\right)$}\tabularnewline
\hline 
200  & 0.045  & 0.371  & 0.052  & 0.322  & 0.066  & 0.325  & 0.066  & 0.288  & 0.072  & 0.322  & 0.072  & 0.287 \tabularnewline
300  & 0.045  & 0.295  & 0.047  & 0.263  & 0.066  & 0.265  & 0.061  & 0.240  & 0.069  & 0.264  & 0.063  & 0.239 \tabularnewline
400  & 0.037  & 0.251  & 0.054  & 0.228  & 0.054  & 0.229  & 0.061  & 0.210  & 0.056  & 0.229  & 0.061  & 0.210 \tabularnewline
500  & 0.039  & 0.222  & 0.045  & 0.203  & 0.054  & 0.204  & 0.058  & 0.189  & 0.056  & 0.205  & 0.060  & 0.190 \tabularnewline
600  & 0.046  & 0.201  & 0.051  & 0.185  & 0.058  & 0.187  & 0.060  & 0.174  & 0.059  & 0.187  & 0.061  & 0.174 \tabularnewline
\hline 
\multicolumn{13}{c}{$\left(p_{x},p_{z}\right)=\left(100,150\right)$}\tabularnewline
\hline 
200  & 0.037  & 0.372  & 0.044  & 0.323  & 0.059  & 0.326  & 0.062  & 0.289  & 0.065  & 0.321  & 0.067  & 0.285 \tabularnewline
300  & 0.047  & 0.294  & 0.054  & 0.263  & 0.064  & 0.264  & 0.058  & 0.240  & 0.064  & 0.261  & 0.064  & 0.237 \tabularnewline
400  & 0.039  & 0.251  & 0.047  & 0.228  & 0.056  & 0.228  & 0.056  & 0.210  & 0.057  & 0.227  & 0.058  & 0.209 \tabularnewline
500  & 0.045  & 0.222  & 0.051  & 0.204  & 0.060  & 0.204  & 0.058  & 0.189  & 0.059  & 0.204  & 0.058  & 0.189 \tabularnewline
600  & 0.040  & 0.201  & 0.044  & 0.186  & 0.057  & 0.186  & 0.057  & 0.173  & 0.056  & 0.186  & 0.057  & 0.173 \tabularnewline
\hline 
\multicolumn{13}{c}{$\left(p_{x},p_{z}\right)=\left(150,300\right)$}\tabularnewline
\hline 
200  & 0.035  & 0.372  & 0.042  & 0.323  & 0.056  & 0.329  & 0.054  & 0.292  & 0.058  & 0.317  & 0.057  & 0.282 \tabularnewline
300  & 0.038  & 0.296  & 0.056  & 0.263  & 0.058  & 0.265  & 0.057  & 0.241  & 0.061  & 0.260  & 0.061  & 0.236 \tabularnewline
400  & 0.043  & 0.251  & 0.053  & 0.228  & 0.065  & 0.228  & 0.060  & 0.210  & 0.061  & 0.226  & 0.059  & 0.208 \tabularnewline
500  & 0.039  & 0.221  & 0.044  & 0.203  & 0.057  & 0.203  & 0.052  & 0.188  & 0.056  & 0.202  & 0.053  & 0.188 \tabularnewline
600  & 0.039  & 0.201  & 0.043  & 0.185  & 0.051  & 0.185  & 0.050  & 0.173  & 0.053  & 0.185  & 0.051  & 0.172 \tabularnewline
\hline 
\end{tabular}
}
\par\end{centering}
{

{\footnotesize\textit{Notes}}{\footnotesize : The data generating process
corresponds to \eqref{eq:dgp_iid_inno} with cointegrated regressors
described in Section \ref{sec:mc_cointegration}. The upper and lower
panels report the empirical size of testing the null hypotheses $\mathbb{H}_{0}:\beta_{1}^{*}=0$
and $\mathbb{H}_{0}:\gamma_{1}^{*}=0$, respectively, at a 5\% nominal
significance level. ``Size'' is calculated as $R^{-1}\sum_{r=1}^{R}\mathbf{1}\left[|t^{\text{XD}(r)}|>\mathrm{\Phi}_{0.975}\right]$
across $R=2,000$ replications, where $t^{\text{XD}(r)}$ is computed
based on \eqref{eq:XD=000020t=000020stat} for the $r$-th replication,
and the critical value $\mathrm{\Phi}_{0.975}\left(\approx1.96\right)$
is the 97.5-th percentile of the standard normal distribution. ``Len.''
refers to the median length of the 95\% confidence intervals across
replications. The IVX oracle and OLS oracle are infeasible estimators.
The ``Calibrated'' and ``CV'' columns refer to the methods used
for choosing the tuning parameters through calibration and cross-validation,
respectively. }}{\footnotesize\par}
\end{table}
 
\begin{table}[hp]
\begin{centering}
\caption{Empirical size and length of confidence interval with cointegrated
regressors: AR(1) innovations\protect\label{tab:mc_size_coint_ar1}}
{\small\begin{tabular}{c|cccc|cccc|cccc}
\hline 
\multirow{3}{*}{$n$ } & \multicolumn{4}{c|}{Oracle} & \multicolumn{4}{c|}{Calibrated } & \multicolumn{4}{c}{CV}\tabularnewline
 & \multicolumn{2}{c}{IVX Oracle} & \multicolumn{2}{c|}{OLS Oracle} & \multicolumn{2}{c}{XDlasso} & \multicolumn{2}{c|}{Dlasso} & \multicolumn{2}{c}{XDlasso} & \multicolumn{2}{c}{Dlasso}\tabularnewline
\cline{2-13}
 & Size  & Len.  & Size  & Len.  & Size  & Len.  & Size  & Len.  & Size  & Len.  & Size  & Len. \tabularnewline
\hline 
\multicolumn{13}{c}{$\mathbb{H}_{0}:\beta_{1}^{*}=0$ for I(1) regressor (raw data)}\tabularnewline
\hline 
\multicolumn{13}{c}{$\left(p_{x},p_{z}\right)=\left(50,100\right)$}\tabularnewline
\hline 
200  & 0.048  & 0.161  & 0.157  & 0.072  & 0.065  & 0.169  & 0.411  & 0.079  & 0.092  & 0.175  & 0.500  & 0.138 \tabularnewline
300  & 0.047  & 0.113  & 0.154  & 0.048  & 0.061  & 0.123  & 0.461  & 0.059  & 0.079  & 0.127  & 0.549  & 0.095 \tabularnewline
400  & 0.044  & 0.087  & 0.153  & 0.035  & 0.052  & 0.096  & 0.518  & 0.048  & 0.076  & 0.100  & 0.598  & 0.073 \tabularnewline
500  & 0.045  & 0.073  & 0.144  & 0.028  & 0.053  & 0.082  & 0.551  & 0.040  & 0.079  & 0.083  & 0.610  & 0.058 \tabularnewline
600  & 0.044  & 0.062  & 0.139  & 0.023  & 0.053  & 0.070  & 0.565  & 0.035  & 0.071  & 0.072  & 0.619  & 0.047 \tabularnewline
\hline 
\multicolumn{13}{c}{$\left(p_{x},p_{z}\right)=\left(100,150\right)$}\tabularnewline
\hline 
200  & 0.046  & 0.158  & 0.150  & 0.072  & 0.058  & 0.166  & 0.419  & 0.076  & 0.100  & 0.173  & 0.550  & 0.134 \tabularnewline
300  & 0.046  & 0.112  & 0.144  & 0.047  & 0.054  & 0.120  & 0.512  & 0.057  & 0.089  & 0.126  & 0.643  & 0.105 \tabularnewline
400  & 0.051  & 0.088  & 0.142  & 0.035  & 0.052  & 0.095  & 0.562  & 0.047  & 0.086  & 0.100  & 0.681  & 0.083 \tabularnewline
500  & 0.043  & 0.073  & 0.134  & 0.028  & 0.042  & 0.080  & 0.606  & 0.040  & 0.076  & 0.083  & 0.737  & 0.068 \tabularnewline
600  & 0.043  & 0.063  & 0.124  & 0.023  & 0.056  & 0.069  & 0.639  & 0.035  & 0.084  & 0.072  & 0.736  & 0.056 \tabularnewline
\hline 
\multicolumn{13}{c}{$\left(p_{x},p_{z}\right)=\left(150,300\right)$}\tabularnewline
\hline 
200  & 0.050  & 0.162  & 0.153  & 0.072  & 0.057  & 0.160  & 0.385  & 0.071  & 0.090  & 0.170  & 0.530  & 0.112 \tabularnewline
300  & 0.048  & 0.113  & 0.144  & 0.047  & 0.061  & 0.116  & 0.454  & 0.054  & 0.088  & 0.126  & 0.615  & 0.087 \tabularnewline
400  & 0.039  & 0.088  & 0.142  & 0.035  & 0.057  & 0.093  & 0.522  & 0.044  & 0.079  & 0.099  & 0.667  & 0.071 \tabularnewline
500  & 0.049  & 0.073  & 0.133  & 0.029  & 0.058  & 0.077  & 0.580  & 0.038  & 0.082  & 0.084  & 0.723  & 0.065 \tabularnewline
600  & 0.044  & 0.063  & 0.143  & 0.023  & 0.050  & 0.066  & 0.622  & 0.033  & 0.071  & 0.071  & 0.753  & 0.057 \tabularnewline
\hline 
\multicolumn{13}{c}{$\mathbb{H}_{0}:\gamma_{1}^{*}=0$ for stationary regressor (raw data)}\tabularnewline
\hline 
\multicolumn{13}{c}{$\left(p_{x},p_{z}\right)=\left(50,100\right)$}\tabularnewline
\hline 
200  & 0.041  & 0.380  & 0.051  & 0.311  & 0.069  & 0.334  & 0.073  & 0.275  & 0.076  & 0.330  & 0.079  & 0.272 \tabularnewline
300  & 0.045  & 0.297  & 0.049  & 0.253  & 0.058  & 0.269  & 0.061  & 0.229  & 0.059  & 0.268  & 0.065  & 0.227 \tabularnewline
400  & 0.037  & 0.251  & 0.044  & 0.219  & 0.052  & 0.231  & 0.053  & 0.200  & 0.055  & 0.231  & 0.056  & 0.200 \tabularnewline
500  & 0.047  & 0.221  & 0.048  & 0.195  & 0.058  & 0.206  & 0.061  & 0.181  & 0.060  & 0.206  & 0.061  & 0.180 \tabularnewline
600  & 0.051  & 0.200  & 0.053  & 0.177  & 0.057  & 0.187  & 0.060  & 0.166  & 0.057  & 0.187  & 0.062  & 0.166 \tabularnewline
\hline 
\multicolumn{13}{c}{$\left(p_{x},p_{z}\right)=\left(100,150\right)$}\tabularnewline
\hline 
200  & 0.038  & 0.379  & 0.046  & 0.314  & 0.058  & 0.334  & 0.052  & 0.276  & 0.068  & 0.329  & 0.059  & 0.272 \tabularnewline
300  & 0.041  & 0.297  & 0.051  & 0.253  & 0.057  & 0.268  & 0.062  & 0.228  & 0.066  & 0.266  & 0.068  & 0.225 \tabularnewline
400  & 0.040  & 0.252  & 0.050  & 0.219  & 0.064  & 0.230  & 0.058  & 0.200  & 0.067  & 0.230  & 0.067  & 0.198 \tabularnewline
500  & 0.049  & 0.222  & 0.056  & 0.195  & 0.066  & 0.204  & 0.066  & 0.180  & 0.066  & 0.205  & 0.070  & 0.180 \tabularnewline
600  & 0.041  & 0.200  & 0.055  & 0.178  & 0.067  & 0.186  & 0.063  & 0.165  & 0.067  & 0.186  & 0.068  & 0.165 \tabularnewline
\hline 
\multicolumn{13}{c}{$\left(p_{x},p_{z}\right)=\left(150,300\right)$}\tabularnewline
\hline 
200  & 0.035  & 0.380  & 0.062  & 0.312  & 0.067  & 0.335  & 0.071  & 0.277  & 0.069  & 0.325  & 0.077  & 0.267 \tabularnewline
300  & 0.042  & 0.297  & 0.059  & 0.253  & 0.061  & 0.268  & 0.067  & 0.228  & 0.067  & 0.263  & 0.069  & 0.224 \tabularnewline
400  & 0.039  & 0.251  & 0.054  & 0.219  & 0.067  & 0.230  & 0.061  & 0.199  & 0.065  & 0.227  & 0.062  & 0.197 \tabularnewline
500  & 0.038  & 0.221  & 0.043  & 0.195  & 0.056  & 0.203  & 0.058  & 0.179  & 0.060  & 0.202  & 0.053  & 0.178 \tabularnewline
600  & 0.043  & 0.200  & 0.047  & 0.178  & 0.063  & 0.184  & 0.054  & 0.164  & 0.057  & 0.184  & 0.055  & 0.163 \tabularnewline
\hline 

\end{tabular}

}
\par\end{centering}
{

{\footnotesize\textit{Notes}}{\footnotesize : The data generating process
corresponds to \eqref{eq:dgp_ar1_inno} with cointegrated regressors
described in Section \ref{sec:mc_cointegration}. The upper and lower
panels report the empirical size of testing the null hypotheses $\mathbb{H}_{0}:\beta_{1}^{*}=0$
and $\mathbb{H}_{0}:\gamma_{1}^{*}=0$ at a 5\% nominal significance
level, respectively. ``Size'' is calculated as $R^{-1}\sum_{r=1}^{R}\mathbf{1}\left[|t^{\text{XD}(r)}|>\mathrm{\Phi}_{0.975}\right]$
across $R=2,000$ replications, where $t^{\text{XD}(r)}$ is computed
based on \eqref{eq:XD=000020t=000020stat} for the $r$-th replication,
and the critical value $\mathrm{\Phi}_{0.975}\left(\approx1.96\right)$
is the 97.5-th percentile of the standard normal distribution. ``Len.''
refers to the median length of the 95\% confidence intervals across
replications. The IVX oracle and OLS oracle are infeasible estimators.
The ``Calibrated'' and ``CV'' columns refer to the methods used
for choosing the tuning parameters through calibration and cross-validation,
respectively. }}{\footnotesize\par}
\end{table}

\subsection{Simulation Results on Conditional Heteroskedasticity\protect\label{sec:mc_robustse}}

In this section, we conduct simulation experiments to investigate
the finite sample properties of XDlasso with conditional heteroskedasticity
and heteroskedastic-robust standard error. In the experiment, we adapt
the data generating process in Section \ref{subsec:Setup} to incorporate
possibly heteroskedastic error terms. The innovations $v_{t}=\left(u_{0,t},e_{t}^{\top},Z_{t}^{\top}\right)^{\top}$
are generated following \eqref{eq:dgp_iid_inno} and \eqref{eq:dgp_ar1_inno}.
We examine two cases for the error term $u_{t}$:
\begin{align}
\text{IID Error Term:\,}u_{t} & =u_{0,t},\label{eq:iid_error_dgp}\\
\text{GARCH(1,1) Error Term:}\:u_{t} & =\sqrt{h_{t}}u_{0,t},\,h_{t}=\alpha_{0}+\alpha_{u}u_{t-1}^{2}+\alpha_{h}h_{t-1},\label{eq:garch_error}
\end{align}
where we specify $\alpha_{0}=0.6,\alpha_{u}=\alpha_{h}=0.2$, and
initialize $h_{1}=1$. 

We consider both the homoskedasticity-only standard error as in \eqref{eq:s.e.}
and the heteroskedasticity-robust standard error given as 
\begin{align}
\hat{\omega}_{j}^{{\rm XD,Robust}} & =\dfrac{\sqrt{\sum_{t=1}^{n}\hat{r}_{j,t-1}^{2}\hat{u}_{t}^{2}}}{\left|\sum_{t=1}^{n}\hat{r}_{j,t-1}w_{j,t-1}\right|},\label{eq:heterose}
\end{align}
for the construction of the test statistic in \eqref{eq:XD=000020t=000020stat}. 

The empirical sizes based on homoskedastic standard errors in DGPs
with GARCH(1,1) error term $u_{t}$ are reported in Table \ref{tab:mc_size_robustse_iid_garch_error_homo}
and \ref{tab:mc_size_robustse_ar1_garch_error_homo}. The results
echo our conjecture in Remark \ref{rem:hetero} that the homoskedastic
standard error \eqref{eq:s.e.} is robust to conditional heteroskedasticity
as in \citet{kostakis2015robust}.

The empirical sizes based on the heteroskedastic-robust standard error
\eqref{eq:heterose} are reported in Table \ref{tab:mc_size_robustse_iid}
and \ref{tab:mc_size_robustse_ar1} when the error term $u_{t}$ is
IID, and in Table \ref{tab:mc_size_robustse_iid_garch_error} and
\ref{tab:mc_size_robustse_ar1_garch_error} when $u_{t}$ follows
a GARCH model. The finite sample performance of XDlasso with robust
standard error demonstrates good size control, and motivates our practice
in the empirical analysis carried out in Section \ref{sec:robustse}.

\begin{table}[hp]
\begin{centering}
\caption{Empirical size and length of confidence interval with homoskedastic
S.E.: IID innovations and GARCH error terms\protect\label{tab:mc_size_robustse_iid_garch_error_homo}}
{\small\begin{tabular}{c|cccc|cccc|cccc}
\hline 
\multirow{3}{*}{$n$ } & \multicolumn{4}{c|}{Oracle} & \multicolumn{4}{c|}{Calibrated } & \multicolumn{4}{c}{CV}\tabularnewline
 & \multicolumn{2}{c}{IVX Oracle} & \multicolumn{2}{c|}{OLS Oracle} & \multicolumn{2}{c}{XDlasso} & \multicolumn{2}{c|}{Dlasso} & \multicolumn{2}{c}{XDlasso} & \multicolumn{2}{c}{Dlasso}\tabularnewline
\cline{2-13}
 & Size  & Len.  & Size  & Len.  & Size  & Len.  & Size  & Len.  & Size  & Len.  & Size  & Len. \tabularnewline
\hline
\multicolumn{13}{c}{$\mathbb{H}_{0}:\beta_{1}^{*}=0$ for nonstationary regressor, \ $\left(p_{x},p_{z}\right)=\left(50,100\right)$}\tabularnewline
\hline
200  & 0.043  & 0.216  & 0.141  & 0.097  & 0.057  & 0.226  & 0.367  & 0.107  & 0.064  & 0.226  & 0.422  & 0.154 \tabularnewline
300  & 0.048  & 0.153  & 0.148  & 0.066  & 0.054  & 0.167  & 0.447  & 0.080  & 0.062  & 0.168  & 0.524  & 0.117 \tabularnewline
400  & 0.040  & 0.122  & 0.146  & 0.050  & 0.056  & 0.135  & 0.478  & 0.065  & 0.070  & 0.134  & 0.541  & 0.093 \tabularnewline
500  & 0.040  & 0.102  & 0.153  & 0.040  & 0.054  & 0.114  & 0.522  & 0.055  & 0.067  & 0.114  & 0.597  & 0.078 \tabularnewline
600  & 0.051  & 0.088  & 0.143  & 0.033  & 0.060  & 0.099  & 0.541  & 0.048  & 0.072  & 0.099  & 0.604  & 0.065 \tabularnewline
\hline
\multicolumn{13}{c}{$\left(p_{x},p_{z}\right)=\left(100,150\right)$}\tabularnewline
\hline
200  & 0.052  & 0.216  & 0.148  & 0.099  & 0.050  & 0.222  & 0.384  & 0.104  & 0.055  & 0.225  & 0.497  & 0.162 \tabularnewline
300  & 0.039  & 0.153  & 0.137  & 0.065  & 0.051  & 0.164  & 0.483  & 0.078  & 0.061  & 0.169  & 0.598  & 0.128 \tabularnewline
400  & 0.048  & 0.121  & 0.148  & 0.049  & 0.042  & 0.133  & 0.529  & 0.064  & 0.060  & 0.137  & 0.672  & 0.108 \tabularnewline
500  & 0.046  & 0.101  & 0.147  & 0.039  & 0.048  & 0.112  & 0.581  & 0.054  & 0.064  & 0.115  & 0.714  & 0.092 \tabularnewline
600  & 0.049  & 0.086  & 0.138  & 0.033  & 0.047  & 0.097  & 0.616  & 0.047  & 0.059  & 0.099  & 0.720  & 0.076 \tabularnewline
\hline
\multicolumn{13}{c}{$\left(p_{x},p_{z}\right)=\left(150,300\right)$}\tabularnewline
\hline
200  & 0.041  & 0.214  & 0.162  & 0.098  & 0.057  & 0.221  & 0.351  & 0.099  & 0.063  & 0.225  & 0.474  & 0.141 \tabularnewline
300  & 0.046  & 0.153  & 0.148  & 0.066  & 0.053  & 0.159  & 0.455  & 0.074  & 0.065  & 0.165  & 0.582  & 0.111 \tabularnewline
400  & 0.049  & 0.121  & 0.147  & 0.049  & 0.056  & 0.128  & 0.516  & 0.060  & 0.069  & 0.133  & 0.638  & 0.098 \tabularnewline
500  & 0.048  & 0.101  & 0.153  & 0.039  & 0.058  & 0.108  & 0.567  & 0.051  & 0.068  & 0.113  & 0.703  & 0.083 \tabularnewline
600  & 0.054  & 0.087  & 0.151  & 0.033  & 0.066  & 0.094  & 0.600  & 0.045  & 0.074  & 0.099  & 0.736  & 0.078 \tabularnewline
\hline
\multicolumn{13}{c}{$\mathbb{H}_{0}:\gamma_{1}^{*}=0$ for stationary regressor, \ $\left(p_{x},p_{z}\right)=\left(50,100\right)$}\tabularnewline
\hline
200  & 0.036  & 0.372  & 0.056  & 0.323  & 0.064  & 0.325  & 0.061  & 0.289  & 0.065  & 0.323  & 0.066  & 0.287 \tabularnewline
300  & 0.042  & 0.294  & 0.057  & 0.262  & 0.055  & 0.265  & 0.062  & 0.240  & 0.056  & 0.264  & 0.063  & 0.240 \tabularnewline
400  & 0.047  & 0.251  & 0.052  & 0.227  & 0.061  & 0.229  & 0.059  & 0.210  & 0.063  & 0.228  & 0.064  & 0.210 \tabularnewline
500  & 0.043  & 0.221  & 0.042  & 0.203  & 0.056  & 0.204  & 0.058  & 0.189  & 0.060  & 0.204  & 0.057  & 0.189 \tabularnewline
600  & 0.041  & 0.200  & 0.049  & 0.185  & 0.054  & 0.186  & 0.056  & 0.174  & 0.056  & 0.186  & 0.056  & 0.174 \tabularnewline
\hline
\multicolumn{13}{c}{$\left(p_{x},p_{z}\right)=\left(100,150\right)$}\tabularnewline
\hline
200  & 0.045  & 0.371  & 0.052  & 0.323  & 0.062  & 0.326  & 0.059  & 0.289  & 0.071  & 0.319  & 0.066  & 0.284 \tabularnewline
300  & 0.048  & 0.294  & 0.059  & 0.263  & 0.065  & 0.264  & 0.067  & 0.240  & 0.066  & 0.262  & 0.071  & 0.238 \tabularnewline
400  & 0.051  & 0.251  & 0.058  & 0.228  & 0.065  & 0.228  & 0.064  & 0.209  & 0.066  & 0.227  & 0.068  & 0.209 \tabularnewline
500  & 0.043  & 0.222  & 0.056  & 0.203  & 0.059  & 0.203  & 0.067  & 0.189  & 0.062  & 0.203  & 0.067  & 0.189 \tabularnewline
600  & 0.050  & 0.200  & 0.055  & 0.185  & 0.060  & 0.185  & 0.060  & 0.173  & 0.061  & 0.186  & 0.059  & 0.173 \tabularnewline
\hline
\multicolumn{13}{c}{$\left(p_{x},p_{z}\right)=\left(150,300\right)$}\tabularnewline
\hline
200  & 0.033  & 0.372  & 0.050  & 0.323  & 0.062  & 0.329  & 0.059  & 0.292  & 0.066  & 0.316  & 0.066  & 0.281 \tabularnewline
300  & 0.039  & 0.294  & 0.043  & 0.264  & 0.063  & 0.266  & 0.057  & 0.241  & 0.060  & 0.260  & 0.056  & 0.236 \tabularnewline
400  & 0.041  & 0.251  & 0.041  & 0.228  & 0.060  & 0.228  & 0.050  & 0.210  & 0.057  & 0.226  & 0.051  & 0.208 \tabularnewline
500  & 0.036  & 0.222  & 0.037  & 0.204  & 0.055  & 0.203  & 0.049  & 0.189  & 0.053  & 0.202  & 0.048  & 0.187 \tabularnewline
600  & 0.042  & 0.201  & 0.050  & 0.186  & 0.058  & 0.185  & 0.051  & 0.173  & 0.054  & 0.185  & 0.050  & 0.173 \tabularnewline
\hline
\end{tabular}
}
\par\end{centering}
{

{\footnotesize\textit{Notes}}{\footnotesize : The data generating process
corresponds to \eqref{eq:dgp_iid_inno} and \eqref{eq:garch_error}.
The upper and lower panels report the empirical size of testing the
null hypotheses $\mathbb{H}_{0}:\beta_{1}^{*}=0$ and $\mathbb{H}_{0}:\gamma_{1}^{*}=0$,
respectively, at a 5\% nominal significance level. ``Size'' is calculated
as $R^{-1}\sum_{r=1}^{R}\mathbf{1}\left[|t^{\text{XD}(r)}|>\mathrm{\Phi}_{0.975}\right]$
across $R=2,000$ replications, where $t^{\text{XD}(r)}$ is computed
based on \eqref{eq:XD=000020t=000020stat} for the $r$-th replication,
and the critical value $\mathrm{\Phi}_{0.975}\left(\approx1.96\right)$
is the 97.5-th percentile of the standard normal distribution. ``Len.''
refers to the median length of the 95\% confidence intervals across
replications. The IVX oracle and OLS oracle are infeasible estimators.
The ``Calibrated'' and ``CV'' columns refer to the methods used
for choosing the tuning parameters through calibration and cross-validation,
respectively. }}{\footnotesize\par}
\end{table}
 
\begin{table}[hp]
\begin{centering}
\caption{Empirical size and length of confidence with homoskedastic S.E.: AR(1)
innovations and GARCH error terms\protect\label{tab:mc_size_robustse_ar1_garch_error_homo}}
{\small\begin{tabular}{c|cccc|cccc|cccc}
\hline 
\multirow{3}{*}{$n$ } & \multicolumn{4}{c|}{Oracle} & \multicolumn{4}{c|}{Calibrated } & \multicolumn{4}{c}{CV}\tabularnewline
 & \multicolumn{2}{c}{IVX Oracle} & \multicolumn{2}{c|}{OLS Oracle} & \multicolumn{2}{c}{XDlasso} & \multicolumn{2}{c|}{Dlasso} & \multicolumn{2}{c}{XDlasso} & \multicolumn{2}{c}{Dlasso}\tabularnewline
\cline{2-13}
 & Size  & Len.  & Size  & Len.  & Size  & Len.  & Size  & Len.  & Size  & Len.  & Size  & Len. \tabularnewline
\hline
\multicolumn{13}{c}{$\mathbb{H}_{0}:\beta_{1}^{*}=0$ for nonstationary regressor, \ $\left(p_{x},p_{z}\right)=\left(50,100\right)$}\tabularnewline
\hline 
200  & 0.055  & 0.160  & 0.167  & 0.072  & 0.067  & 0.168  & 0.403  & 0.078  & 0.084  & 0.172  & 0.493  & 0.134 \tabularnewline
300  & 0.045  & 0.113  & 0.162  & 0.047  & 0.048  & 0.121  & 0.472  & 0.058  & 0.076  & 0.124  & 0.576  & 0.092 \tabularnewline
400  & 0.050  & 0.088  & 0.163  & 0.036  & 0.056  & 0.096  & 0.517  & 0.048  & 0.085  & 0.098  & 0.593  & 0.071 \tabularnewline
500  & 0.047  & 0.073  & 0.151  & 0.028  & 0.054  & 0.081  & 0.550  & 0.040  & 0.084  & 0.083  & 0.625  & 0.055 \tabularnewline
600  & 0.048  & 0.063  & 0.151  & 0.024  & 0.058  & 0.071  & 0.568  & 0.035  & 0.080  & 0.072  & 0.607  & 0.046 \tabularnewline
\hline 
\multicolumn{13}{c}{$\left(p_{x},p_{z}\right)=\left(100,150\right)$}\tabularnewline
\hline 
200  & 0.044  & 0.159  & 0.162  & 0.072  & 0.060  & 0.166  & 0.433  & 0.075  & 0.093  & 0.174  & 0.542  & 0.132 \tabularnewline
300  & 0.044  & 0.111  & 0.153  & 0.047  & 0.048  & 0.121  & 0.510  & 0.057  & 0.080  & 0.127  & 0.646  & 0.104 \tabularnewline
400  & 0.048  & 0.087  & 0.149  & 0.035  & 0.048  & 0.096  & 0.543  & 0.046  & 0.082  & 0.100  & 0.702  & 0.084 \tabularnewline
500  & 0.048  & 0.073  & 0.149  & 0.028  & 0.053  & 0.080  & 0.600  & 0.039  & 0.074  & 0.083  & 0.737  & 0.067 \tabularnewline
600  & 0.049  & 0.063  & 0.141  & 0.023  & 0.049  & 0.069  & 0.639  & 0.034  & 0.072  & 0.072  & 0.747  & 0.056 \tabularnewline
\hline 
\multicolumn{13}{c}{$\left(p_{x},p_{z}\right)=\left(150,300\right)$}\tabularnewline
\hline 
200  & 0.052  & 0.160  & 0.167  & 0.071  & 0.060  & 0.158  & 0.361  & 0.071  & 0.097  & 0.171  & 0.518  & 0.110 \tabularnewline
300  & 0.057  & 0.113  & 0.157  & 0.047  & 0.055  & 0.115  & 0.460  & 0.054  & 0.072  & 0.124  & 0.622  & 0.087 \tabularnewline
400  & 0.048  & 0.088  & 0.142  & 0.035  & 0.054  & 0.091  & 0.528  & 0.044  & 0.085  & 0.099  & 0.689  & 0.077 \tabularnewline
500  & 0.055  & 0.074  & 0.151  & 0.028  & 0.061  & 0.076  & 0.584  & 0.037  & 0.089  & 0.082  & 0.729  & 0.065 \tabularnewline
600  & 0.049  & 0.063  & 0.153  & 0.023  & 0.061  & 0.065  & 0.623  & 0.032  & 0.073  & 0.071  & 0.750  & 0.054 \tabularnewline
\hline
\multicolumn{13}{c}{$\mathbb{H}_{0}:\gamma_{1}^{*}=0$ for stationary regressor, \ $\left(p_{x},p_{z}\right)=\left(50,100\right)$}\tabularnewline
\hline 
200  & 0.040  & 0.379  & 0.054  & 0.311  & 0.062  & 0.331  & 0.073  & 0.273  & 0.066  & 0.330  & 0.075  & 0.272 \tabularnewline
300  & 0.043  & 0.296  & 0.050  & 0.252  & 0.058  & 0.267  & 0.062  & 0.228  & 0.059  & 0.268  & 0.062  & 0.227 \tabularnewline
400  & 0.043  & 0.251  & 0.048  & 0.218  & 0.051  & 0.230  & 0.067  & 0.200  & 0.051  & 0.230  & 0.068  & 0.199 \tabularnewline
500  & 0.042  & 0.221  & 0.053  & 0.194  & 0.059  & 0.204  & 0.058  & 0.180  & 0.065  & 0.205  & 0.061  & 0.179 \tabularnewline
600  & 0.040  & 0.199  & 0.047  & 0.177  & 0.055  & 0.186  & 0.057  & 0.165  & 0.056  & 0.187  & 0.057  & 0.165 \tabularnewline
\hline 
\multicolumn{13}{c}{$\left(p_{x},p_{z}\right)=\left(100,150\right)$}\tabularnewline
\hline 
200  & 0.046  & 0.381  & 0.057  & 0.311  & 0.069  & 0.332  & 0.071  & 0.273  & 0.072  & 0.328  & 0.074  & 0.270 \tabularnewline
300  & 0.043  & 0.297  & 0.049  & 0.254  & 0.062  & 0.267  & 0.059  & 0.227  & 0.065  & 0.266  & 0.061  & 0.225 \tabularnewline
400  & 0.045  & 0.252  & 0.059  & 0.218  & 0.062  & 0.229  & 0.065  & 0.198  & 0.065  & 0.229  & 0.070  & 0.198 \tabularnewline
500  & 0.049  & 0.221  & 0.058  & 0.195  & 0.062  & 0.204  & 0.069  & 0.179  & 0.062  & 0.204  & 0.069  & 0.179 \tabularnewline
600  & 0.045  & 0.199  & 0.055  & 0.178  & 0.053  & 0.185  & 0.063  & 0.164  & 0.056  & 0.186  & 0.064  & 0.164 \tabularnewline
\hline 
\multicolumn{13}{c}{$\left(p_{x},p_{z}\right)=\left(150,300\right)$}\tabularnewline
\hline 
200  & 0.038  & 0.380  & 0.046  & 0.314  & 0.064  & 0.332  & 0.059  & 0.275  & 0.063  & 0.325  & 0.062  & 0.268 \tabularnewline
300  & 0.037  & 0.298  & 0.040  & 0.253  & 0.065  & 0.267  & 0.062  & 0.227  & 0.065  & 0.264  & 0.057  & 0.224 \tabularnewline
400  & 0.040  & 0.253  & 0.040  & 0.219  & 0.056  & 0.228  & 0.054  & 0.198  & 0.054  & 0.228  & 0.052  & 0.197 \tabularnewline
500  & 0.039  & 0.222  & 0.042  & 0.195  & 0.058  & 0.203  & 0.055  & 0.178  & 0.057  & 0.203  & 0.053  & 0.178 \tabularnewline
600  & 0.046  & 0.200  & 0.052  & 0.178  & 0.063  & 0.184  & 0.056  & 0.164  & 0.060  & 0.184  & 0.057  & 0.164 \tabularnewline
\hline
\end{tabular}

}
\par\end{centering}
{

{\footnotesize\textit{Notes}}{\footnotesize : The data generating process
corresponds to \eqref{eq:dgp_ar1_inno} and \eqref{eq:garch_error}.
The upper and lower panels report the empirical size of testing the
null hypotheses $\mathbb{H}_{0}:\beta_{1}^{*}=0$ and $\mathbb{H}_{0}:\gamma_{1}^{*}=0$
at a 5\% nominal significance level, respectively. ``Size'' is calculated
as $R^{-1}\sum_{r=1}^{R}\mathbf{1}\left[|t^{\text{XD}(r)}|>\mathrm{\Phi}_{0.975}\right]$
across $R=2,000$ replications, where $t^{\text{XD}(r)}$ is computed
based on \eqref{eq:XD=000020t=000020stat} for the $r$-th replication,
and the critical value $\mathrm{\Phi}_{0.975}\left(\approx1.96\right)$
is the 97.5-th percentile of the standard normal distribution. ``Len.''
refers to the median length of the 95\% confidence intervals across
replications. The IVX oracle and OLS oracle are infeasible estimators.
The ``Calibrated'' and ``CV'' columns refer to the methods used
for choosing the tuning parameters through calibration and cross-validation,
respectively. }}{\footnotesize\par}
\end{table}
\begin{table}[hp]
\begin{centering}
\caption{Empirical size and length of confidence interval with robust S.E.:
IID innovations and IID error terms\protect\label{tab:mc_size_robustse_iid}}
{\small\begin{tabular}{c|cccc|cccc|cccc}
\hline 
\multirow{3}{*}{$n$ } & \multicolumn{4}{c|}{Oracle} & \multicolumn{4}{c|}{Calibrated } & \multicolumn{4}{c}{CV}\tabularnewline
 & \multicolumn{2}{c}{IVX Oracle} & \multicolumn{2}{c|}{OLS Oracle} & \multicolumn{2}{c}{XDlasso} & \multicolumn{2}{c|}{Dlasso} & \multicolumn{2}{c}{XDlasso} & \multicolumn{2}{c}{Dlasso}\tabularnewline
\cline{2-13}
 & Size  & Len.  & Size  & Len.  & Size  & Len.  & Size  & Len.  & Size  & Len.  & Size  & Len. \tabularnewline
\hline 
\multicolumn{13}{c}{$\mathbb{H}_{0}:\beta_{1}^{*}=0$ for nonstationary regressor}\tabularnewline
\hline 
\multicolumn{13}{c}{$\left(p_{x},p_{z}\right)=\left(50,100\right)$}\tabularnewline
\hline 
200  & 0.040  & 0.216  & 0.145  & 0.099  & 0.056  & 0.220  & 0.373  & 0.105  & 0.057  & 0.227  & 0.434  & 0.159 \tabularnewline
300  & 0.041  & 0.155  & 0.140  & 0.066  & 0.046  & 0.166  & 0.450  & 0.079  & 0.060  & 0.171  & 0.533  & 0.120 \tabularnewline
400  & 0.043  & 0.121  & 0.141  & 0.050  & 0.052  & 0.132  & 0.472  & 0.064  & 0.063  & 0.136  & 0.551  & 0.093 \tabularnewline
500  & 0.051  & 0.102  & 0.159  & 0.040  & 0.049  & 0.113  & 0.511  & 0.054  & 0.057  & 0.115  & 0.597  & 0.077 \tabularnewline
600  & 0.043  & 0.088  & 0.154  & 0.033  & 0.049  & 0.098  & 0.547  & 0.047  & 0.066  & 0.100  & 0.608  & 0.064 \tabularnewline
\hline 
\multicolumn{13}{c}{$\left(p_{x},p_{z}\right)=\left(100,150\right)$}\tabularnewline
\hline 
200  & 0.048  & 0.219  & 0.138  & 0.098  & 0.050  & 0.219  & 0.381  & 0.101  & 0.058  & 0.227  & 0.504  & 0.158 \tabularnewline
300  & 0.042  & 0.154  & 0.155  & 0.066  & 0.040  & 0.160  & 0.460  & 0.076  & 0.059  & 0.168  & 0.596  & 0.126 \tabularnewline
400  & 0.039  & 0.121  & 0.150  & 0.050  & 0.046  & 0.129  & 0.519  & 0.062  & 0.063  & 0.135  & 0.653  & 0.108 \tabularnewline
500  & 0.051  & 0.101  & 0.151  & 0.039  & 0.046  & 0.110  & 0.564  & 0.053  & 0.063  & 0.114  & 0.711  & 0.090 \tabularnewline
600  & 0.047  & 0.088  & 0.149  & 0.033  & 0.044  & 0.096  & 0.607  & 0.046  & 0.065  & 0.100  & 0.746  & 0.077 \tabularnewline
\hline 
\multicolumn{13}{c}{$\left(p_{x},p_{z}\right)=\left(150,300\right)$}\tabularnewline
\hline 
200  & 0.044  & 0.219  & 0.158  & 0.100  & 0.051  & 0.213  & 0.339  & 0.097  & 0.056  & 0.223  & 0.475  & 0.140 \tabularnewline
300  & 0.048  & 0.156  & 0.159  & 0.067  & 0.055  & 0.158  & 0.430  & 0.073  & 0.063  & 0.167  & 0.571  & 0.113 \tabularnewline
400  & 0.055  & 0.123  & 0.158  & 0.050  & 0.057  & 0.125  & 0.494  & 0.059  & 0.067  & 0.133  & 0.651  & 0.096 \tabularnewline
500  & 0.045  & 0.102  & 0.155  & 0.040  & 0.057  & 0.105  & 0.545  & 0.050  & 0.068  & 0.113  & 0.702  & 0.086 \tabularnewline
600  & 0.045  & 0.087  & 0.138  & 0.033  & 0.060  & 0.090  & 0.573  & 0.044  & 0.070  & 0.098  & 0.737  & 0.077 \tabularnewline
\hline 
\multicolumn{13}{c}{$\mathbb{H}_{0}:\gamma_{1}^{*}=0$ for stationary regressor}\tabularnewline
\hline 
\multicolumn{13}{c}{$\left(p_{x},p_{z}\right)=\left(50,100\right)$}\tabularnewline
\hline 
200  & 0.037  & 0.374  & 0.046  & 0.324  & 0.055  & 0.325  & 0.068  & 0.288  & 0.059  & 0.324  & 0.064  & 0.287 \tabularnewline
300  & 0.038  & 0.295  & 0.043  & 0.263  & 0.053  & 0.265  & 0.056  & 0.240  & 0.055  & 0.265  & 0.058  & 0.240 \tabularnewline
400  & 0.038  & 0.251  & 0.045  & 0.227  & 0.057  & 0.228  & 0.062  & 0.210  & 0.060  & 0.229  & 0.065  & 0.210 \tabularnewline
500  & 0.038  & 0.222  & 0.051  & 0.203  & 0.052  & 0.204  & 0.055  & 0.189  & 0.055  & 0.204  & 0.059  & 0.189 \tabularnewline
600  & 0.040  & 0.201  & 0.046  & 0.185  & 0.057  & 0.186  & 0.057  & 0.174  & 0.059  & 0.187  & 0.058  & 0.174 \tabularnewline
\hline 
\multicolumn{13}{c}{$\left(p_{x},p_{z}\right)=\left(100,150\right)$}\tabularnewline
\hline 
200  & 0.054  & 0.372  & 0.062  & 0.324  & 0.074  & 0.326  & 0.078  & 0.289  & 0.079  & 0.320  & 0.085  & 0.285 \tabularnewline
300  & 0.052  & 0.294  & 0.059  & 0.264  & 0.070  & 0.264  & 0.071  & 0.239  & 0.074  & 0.262  & 0.070  & 0.238 \tabularnewline
400  & 0.048  & 0.250  & 0.056  & 0.228  & 0.071  & 0.228  & 0.066  & 0.210  & 0.074  & 0.228  & 0.069  & 0.209 \tabularnewline
500  & 0.045  & 0.222  & 0.058  & 0.204  & 0.061  & 0.204  & 0.066  & 0.189  & 0.064  & 0.204  & 0.068  & 0.189 \tabularnewline
600  & 0.045  & 0.200  & 0.054  & 0.185  & 0.055  & 0.185  & 0.056  & 0.173  & 0.056  & 0.186  & 0.056  & 0.173 \tabularnewline
\hline 
\multicolumn{13}{c}{$\left(p_{x},p_{z}\right)=\left(150,300\right)$}\tabularnewline
\hline 
200  & 0.039  & 0.375  & 0.054  & 0.324  & 0.062  & 0.328  & 0.061  & 0.292  & 0.071  & 0.316  & 0.074  & 0.281 \tabularnewline
300  & 0.036  & 0.295  & 0.044  & 0.263  & 0.053  & 0.265  & 0.056  & 0.240  & 0.055  & 0.260  & 0.059  & 0.237 \tabularnewline
400  & 0.046  & 0.251  & 0.051  & 0.227  & 0.060  & 0.228  & 0.064  & 0.210  & 0.055  & 0.226  & 0.059  & 0.208 \tabularnewline
500  & 0.043  & 0.222  & 0.049  & 0.203  & 0.050  & 0.203  & 0.054  & 0.188  & 0.049  & 0.202  & 0.052  & 0.188 \tabularnewline
600  & 0.038  & 0.200  & 0.049  & 0.185  & 0.054  & 0.185  & 0.061  & 0.173  & 0.050  & 0.185  & 0.059  & 0.173 \tabularnewline
\hline 
\end{tabular}}
\par\end{centering}
{

{\footnotesize\textit{Notes}}{\footnotesize : The data generating process
corresponds to \eqref{eq:dgp_iid_inno}. The upper and lower panels
report the empirical size of testing the null hypotheses $\mathbb{H}_{0}:\beta_{1}^{*}=0$
and $\mathbb{H}_{0}:\gamma_{1}^{*}=0$, respectively, at a 5\% nominal
significance level. ``Size'' is calculated as $R^{-1}\sum_{r=1}^{R}\mathbf{1}\left[|t^{\text{XD}(r)}|>\mathrm{\Phi}_{0.975}\right]$
across $R=2,000$ replications, where $t^{\text{XD}(r)}$ is computed
based on \eqref{eq:heterose} for the $r$-th replication, and the
critical value $\mathrm{\Phi}_{0.975}\left(\approx1.96\right)$ is
the 97.5-th percentile of the standard normal distribution. ``Len.''
refers to the median length of the 95\% confidence intervals across
replications. The IVX oracle and OLS oracle are infeasible estimators.
The ``Calibrated'' and ``CV'' columns refer to the methods used
for choosing the tuning parameters through calibration and cross-validation,
respectively. }}{\footnotesize\par}
\end{table}
 
\begin{table}[hp]
\begin{centering}
\caption{Empirical size and length of confidence with robust S.E.: AR(1) innovations
and IID error terms\protect\label{tab:mc_size_robustse_ar1}}
{\small\begin{tabular}{c|cccc|cccc|cccc}
\hline 
\multirow{3}{*}{$n$ } & \multicolumn{4}{c|}{Oracle} & \multicolumn{4}{c|}{Calibrated } & \multicolumn{4}{c}{CV}\tabularnewline
 & \multicolumn{2}{c}{IVX Oracle} & \multicolumn{2}{c|}{OLS Oracle} & \multicolumn{2}{c}{XDlasso} & \multicolumn{2}{c|}{Dlasso} & \multicolumn{2}{c}{XDlasso} & \multicolumn{2}{c}{Dlasso}\tabularnewline
\cline{2-13}
 & Size  & Len.  & Size  & Len.  & Size  & Len.  & Size  & Len.  & Size  & Len.  & Size  & Len. \tabularnewline
\hline 
\multicolumn{13}{c}{$\mathbb{H}_{0}:\beta_{1}^{*}=0$ for nonstationary regressor}\tabularnewline
\hline 
\multicolumn{13}{c}{$\left(p_{x},p_{z}\right)=\left(50,100\right)$}\tabularnewline
\hline 
200  & 0.044  & 0.161  & 0.146  & 0.071  & 0.057  & 0.169  & 0.412  & 0.078  & 0.090  & 0.173  & 0.487  & 0.135 \tabularnewline
300  & 0.044  & 0.112  & 0.156  & 0.047  & 0.058  & 0.123  & 0.459  & 0.059  & 0.075  & 0.126  & 0.570  & 0.096 \tabularnewline
400  & 0.052  & 0.089  & 0.149  & 0.036  & 0.054  & 0.098  & 0.521  & 0.048  & 0.076  & 0.100  & 0.604  & 0.072 \tabularnewline
500  & 0.046  & 0.074  & 0.156  & 0.029  & 0.050  & 0.083  & 0.549  & 0.041  & 0.071  & 0.084  & 0.605  & 0.057 \tabularnewline
600  & 0.053  & 0.063  & 0.150  & 0.024  & 0.052  & 0.072  & 0.584  & 0.035  & 0.073  & 0.073  & 0.615  & 0.047 \tabularnewline
\hline 
\multicolumn{13}{c}{$\left(p_{x},p_{z}\right)=\left(100,150\right)$}\tabularnewline
\hline 
200  & 0.051  & 0.163  & 0.164  & 0.072  & 0.053  & 0.169  & 0.431  & 0.076  & 0.094  & 0.175  & 0.548  & 0.136 \tabularnewline
300  & 0.045  & 0.112  & 0.150  & 0.048  & 0.049  & 0.119  & 0.501  & 0.057  & 0.084  & 0.124  & 0.636  & 0.102 \tabularnewline
400  & 0.047  & 0.088  & 0.156  & 0.036  & 0.051  & 0.095  & 0.555  & 0.047  & 0.086  & 0.100  & 0.709  & 0.084 \tabularnewline
500  & 0.057  & 0.073  & 0.153  & 0.028  & 0.049  & 0.080  & 0.617  & 0.040  & 0.080  & 0.084  & 0.746  & 0.067 \tabularnewline
600  & 0.057  & 0.063  & 0.145  & 0.024  & 0.052  & 0.069  & 0.655  & 0.035  & 0.075  & 0.072  & 0.746  & 0.057 \tabularnewline
\hline 
\multicolumn{13}{c}{$\left(p_{x},p_{z}\right)=\left(150,300\right)$}\tabularnewline
\hline 
200  & 0.049  & 0.162  & 0.168  & 0.073  & 0.059  & 0.160  & 0.372  & 0.072  & 0.079  & 0.169  & 0.512  & 0.108 \tabularnewline
300  & 0.050  & 0.113  & 0.156  & 0.048  & 0.062  & 0.117  & 0.458  & 0.054  & 0.080  & 0.124  & 0.603  & 0.089 \tabularnewline
400  & 0.052  & 0.088  & 0.145  & 0.036  & 0.058  & 0.091  & 0.525  & 0.044  & 0.078  & 0.097  & 0.672  & 0.075 \tabularnewline
500  & 0.048  & 0.073  & 0.152  & 0.028  & 0.066  & 0.077  & 0.592  & 0.037  & 0.084  & 0.082  & 0.716  & 0.065 \tabularnewline
600  & 0.051  & 0.062  & 0.155  & 0.024  & 0.057  & 0.066  & 0.622  & 0.033  & 0.071  & 0.071  & 0.765  & 0.057 \tabularnewline
\hline 
\multicolumn{13}{c}{$\mathbb{H}_{0}:\gamma_{1}^{*}=0$ for stationary regressor}\tabularnewline
\hline 
\multicolumn{13}{c}{$\left(p_{x},p_{z}\right)=\left(50,100\right)$}\tabularnewline
\hline 
200  & 0.037  & 0.380  & 0.053  & 0.313  & 0.066  & 0.332  & 0.067  & 0.274  & 0.072  & 0.332  & 0.071  & 0.273 \tabularnewline
300  & 0.047  & 0.297  & 0.049  & 0.253  & 0.063  & 0.268  & 0.064  & 0.228  & 0.066  & 0.268  & 0.068  & 0.228 \tabularnewline
400  & 0.041  & 0.251  & 0.052  & 0.218  & 0.061  & 0.230  & 0.071  & 0.200  & 0.065  & 0.230  & 0.074  & 0.199 \tabularnewline
500  & 0.046  & 0.221  & 0.052  & 0.195  & 0.062  & 0.205  & 0.066  & 0.180  & 0.062  & 0.205  & 0.065  & 0.179 \tabularnewline
600  & 0.048  & 0.199  & 0.051  & 0.178  & 0.056  & 0.186  & 0.056  & 0.165  & 0.056  & 0.186  & 0.056  & 0.165 \tabularnewline
\hline 
\multicolumn{13}{c}{$\left(p_{x},p_{z}\right)=\left(100,150\right)$}\tabularnewline
\hline 
200  & 0.045  & 0.380  & 0.057  & 0.312  & 0.078  & 0.332  & 0.063  & 0.274  & 0.089  & 0.328  & 0.076  & 0.271 \tabularnewline
300  & 0.047  & 0.297  & 0.049  & 0.253  & 0.068  & 0.267  & 0.064  & 0.227  & 0.073  & 0.266  & 0.066  & 0.226 \tabularnewline
400  & 0.050  & 0.251  & 0.050  & 0.218  & 0.066  & 0.229  & 0.057  & 0.199  & 0.064  & 0.229  & 0.059  & 0.198 \tabularnewline
500  & 0.050  & 0.221  & 0.052  & 0.195  & 0.060  & 0.203  & 0.059  & 0.179  & 0.057  & 0.204  & 0.059  & 0.179 \tabularnewline
600  & 0.049  & 0.199  & 0.053  & 0.178  & 0.060  & 0.185  & 0.063  & 0.164  & 0.060  & 0.186  & 0.063  & 0.164 \tabularnewline
\hline 
\multicolumn{13}{c}{$\left(p_{x},p_{z}\right)=\left(150,300\right)$}\tabularnewline
\hline 
200  & 0.034  & 0.380  & 0.053  & 0.311  & 0.060  & 0.334  & 0.065  & 0.275  & 0.061  & 0.323  & 0.073  & 0.267 \tabularnewline
300  & 0.039  & 0.297  & 0.044  & 0.252  & 0.057  & 0.267  & 0.059  & 0.227  & 0.054  & 0.263  & 0.058  & 0.224 \tabularnewline
400  & 0.035  & 0.251  & 0.049  & 0.218  & 0.059  & 0.228  & 0.056  & 0.198  & 0.059  & 0.227  & 0.055  & 0.196 \tabularnewline
500  & 0.036  & 0.221  & 0.044  & 0.194  & 0.052  & 0.202  & 0.059  & 0.178  & 0.054  & 0.203  & 0.060  & 0.178 \tabularnewline
600  & 0.038  & 0.199  & 0.045  & 0.177  & 0.051  & 0.184  & 0.054  & 0.164  & 0.051  & 0.185  & 0.056  & 0.164 \tabularnewline
\hline 
\end{tabular}

}
\par\end{centering}
{

{\footnotesize\textit{Notes}}{\footnotesize : The data generating process
corresponds to \eqref{eq:dgp_ar1_inno}. The upper and lower panels
report the empirical size of testing the null hypotheses $\mathbb{H}_{0}:\beta_{1}^{*}=0$
and $\mathbb{H}_{0}:\gamma_{1}^{*}=0$ at a 5\% nominal significance
level, respectively. ``Size'' is calculated as $R^{-1}\sum_{r=1}^{R}\mathbf{1}\left[|t^{\text{XD}(r)}|>\mathrm{\Phi}_{0.975}\right]$
across $R=2,000$ replications, where $t^{\text{XD}(r)}$ is computed
based on \eqref{eq:heterose} for the $r$-th replication, and the
critical value $\mathrm{\Phi}_{0.975}\left(\approx1.96\right)$ is
the 97.5-th percentile of the standard normal distribution. ``Len.''
refers to the median length of the 95\% confidence intervals across
replications. The IVX oracle and OLS oracle are infeasible estimators.
The ``Calibrated'' and ``CV'' columns refer to the methods used
for choosing the tuning parameters through calibration and cross-validation,
respectively. }}{\footnotesize\par}
\end{table}
\begin{table}[hp]
\begin{centering}
\caption{Empirical size and length of confidence interval with robust S.E.:
IID innovations and GARCH error terms\protect\label{tab:mc_size_robustse_iid_garch_error}}
{\small\begin{tabular}{c|cccc|cccc|cccc}
\hline 
\multirow{3}{*}{$n$ } & \multicolumn{4}{c|}{Oracle} & \multicolumn{4}{c|}{Calibrated } & \multicolumn{4}{c}{CV}\tabularnewline
 & \multicolumn{2}{c}{IVX Oracle} & \multicolumn{2}{c|}{OLS Oracle} & \multicolumn{2}{c}{XDlasso} & \multicolumn{2}{c|}{Dlasso} & \multicolumn{2}{c}{XDlasso} & \multicolumn{2}{c}{Dlasso}\tabularnewline
\cline{2-13}
 & Size  & Len.  & Size  & Len.  & Size  & Len.  & Size  & Len.  & Size  & Len.  & Size  & Len. \tabularnewline

 \hline 
 \multicolumn{13}{c}{$\mathbb{H}_{1}:\gamma_{1}^{*}=0.1$ for nonstationary regressor}\tabularnewline
 \hline 
 \multicolumn{13}{c}{$\left(p_{x},p_{z}\right)=\left(50,100\right)$}\tabularnewline
 \hline 
 200  & 0.042  & 0.218  & 0.135  & 0.100  & 0.045  & 0.222  & 0.340  & 0.105  & 0.052  & 0.227  & 0.415  & 0.153 \tabularnewline
 300  & 0.039  & 0.154  & 0.155  & 0.066  & 0.051  & 0.164  & 0.434  & 0.079  & 0.059  & 0.167  & 0.517  & 0.118 \tabularnewline
 400  & 0.045  & 0.122  & 0.157  & 0.050  & 0.050  & 0.132  & 0.462  & 0.064  & 0.067  & 0.135  & 0.555  & 0.095 \tabularnewline
 500  & 0.045  & 0.102  & 0.147  & 0.040  & 0.054  & 0.112  & 0.512  & 0.054  & 0.060  & 0.115  & 0.591  & 0.077 \tabularnewline
 600  & 0.049  & 0.088  & 0.142  & 0.033  & 0.053  & 0.098  & 0.531  & 0.047  & 0.064  & 0.099  & 0.610  & 0.065 \tabularnewline
 \hline 
 \multicolumn{13}{c}{$\left(p_{x},p_{z}\right)=\left(100,150\right)$}\tabularnewline
 \hline 
 200  & 0.042  & 0.220  & 0.162  & 0.101  & 0.043  & 0.224  & 0.390  & 0.104  & 0.055  & 0.229  & 0.516  & 0.160 \tabularnewline
 300  & 0.048  & 0.154  & 0.148  & 0.067  & 0.045  & 0.164  & 0.471  & 0.077  & 0.060  & 0.170  & 0.602  & 0.131 \tabularnewline
 400  & 0.045  & 0.123  & 0.159  & 0.051  & 0.046  & 0.130  & 0.531  & 0.063  & 0.064  & 0.135  & 0.675  & 0.106 \tabularnewline
 500  & 0.050  & 0.103  & 0.154  & 0.040  & 0.055  & 0.110  & 0.589  & 0.053  & 0.077  & 0.114  & 0.720  & 0.091 \tabularnewline
 600  & 0.051  & 0.089  & 0.151  & 0.034  & 0.056  & 0.097  & 0.598  & 0.047  & 0.072  & 0.100  & 0.734  & 0.077 \tabularnewline
 \hline 
 \multicolumn{13}{c}{$\left(p_{x},p_{z}\right)=\left(150,300\right)$}\tabularnewline
 \hline 
 200  & 0.044  & 0.214  & 0.138  & 0.098  & 0.054  & 0.219  & 0.335  & 0.097  & 0.064  & 0.226  & 0.498  & 0.139 \tabularnewline
 300  & 0.040  & 0.154  & 0.157  & 0.066  & 0.056  & 0.158  & 0.423  & 0.073  & 0.055  & 0.166  & 0.594  & 0.112 \tabularnewline
 400  & 0.044  & 0.123  & 0.153  & 0.050  & 0.055  & 0.127  & 0.500  & 0.060  & 0.061  & 0.135  & 0.672  & 0.095 \tabularnewline
 500  & 0.049  & 0.101  & 0.150  & 0.040  & 0.052  & 0.106  & 0.546  & 0.051  & 0.056  & 0.113  & 0.703  & 0.084 \tabularnewline
 600  & 0.050  & 0.088  & 0.152  & 0.033  & 0.059  & 0.093  & 0.589  & 0.044  & 0.061  & 0.099  & 0.726  & 0.076 \tabularnewline
 \hline 
 \multicolumn{13}{c}{$\mathbb{H}_{0}:\gamma_{1}^{*}=0$ for nonstationary regressor}\tabularnewline
 \hline 
 \multicolumn{13}{c}{$\left(p_{x},p_{z}\right)=\left(50,100\right)$}\tabularnewline
 \hline 
 200  & 0.044  & 0.370  & 0.057  & 0.323  & 0.067  & 0.325  & 0.065  & 0.288  & 0.065  & 0.323  & 0.070  & 0.287 \tabularnewline
 300  & 0.038  & 0.295  & 0.054  & 0.263  & 0.055  & 0.265  & 0.059  & 0.241  & 0.055  & 0.264  & 0.058  & 0.239 \tabularnewline
 400  & 0.050  & 0.251  & 0.057  & 0.228  & 0.055  & 0.229  & 0.056  & 0.211  & 0.059  & 0.229  & 0.058  & 0.211 \tabularnewline
 500  & 0.044  & 0.222  & 0.056  & 0.203  & 0.054  & 0.205  & 0.064  & 0.190  & 0.058  & 0.205  & 0.064  & 0.190 \tabularnewline
 600  & 0.045  & 0.201  & 0.053  & 0.185  & 0.054  & 0.187  & 0.058  & 0.174  & 0.054  & 0.187  & 0.061  & 0.174 \tabularnewline
 \hline 
 \multicolumn{13}{c}{$\left(p_{x},p_{z}\right)=\left(100,150\right)$}\tabularnewline
 \hline 
 200  & 0.043  & 0.374  & 0.061  & 0.324  & 0.061  & 0.327  & 0.064  & 0.290  & 0.065  & 0.322  & 0.066  & 0.286 \tabularnewline
 300  & 0.043  & 0.295  & 0.056  & 0.263  & 0.059  & 0.265  & 0.057  & 0.240  & 0.067  & 0.263  & 0.062  & 0.238 \tabularnewline
 400  & 0.044  & 0.251  & 0.053  & 0.227  & 0.053  & 0.229  & 0.057  & 0.210  & 0.053  & 0.228  & 0.058  & 0.209 \tabularnewline
 500  & 0.048  & 0.221  & 0.055  & 0.203  & 0.057  & 0.204  & 0.058  & 0.189  & 0.058  & 0.204  & 0.059  & 0.189 \tabularnewline
 600  & 0.047  & 0.201  & 0.047  & 0.185  & 0.058  & 0.186  & 0.058  & 0.173  & 0.061  & 0.186  & 0.057  & 0.173 \tabularnewline
 \hline 
 \multicolumn{13}{c}{$\left(p_{x},p_{z}\right)=\left(150,300\right)$}\tabularnewline
 \hline 
 200  & 0.040  & 0.371  & 0.053  & 0.325  & 0.063  & 0.329  & 0.064  & 0.292  & 0.066  & 0.316  & 0.076  & 0.282 \tabularnewline
 300  & 0.042  & 0.295  & 0.052  & 0.263  & 0.057  & 0.266  & 0.062  & 0.241  & 0.054  & 0.260  & 0.064  & 0.236 \tabularnewline
 400  & 0.038  & 0.251  & 0.049  & 0.227  & 0.058  & 0.228  & 0.062  & 0.210  & 0.055  & 0.226  & 0.062  & 0.208 \tabularnewline
 500  & 0.042  & 0.222  & 0.049  & 0.203  & 0.056  & 0.203  & 0.063  & 0.189  & 0.050  & 0.202  & 0.060  & 0.188 \tabularnewline
 600  & 0.042  & 0.200  & 0.048  & 0.185  & 0.058  & 0.185  & 0.060  & 0.173  & 0.055  & 0.185  & 0.059  & 0.173 \tabularnewline
 \hline
\end{tabular}

}
\par\end{centering}
{

{\footnotesize\textit{Notes}}{\footnotesize : The data generating process
corresponds to \eqref{eq:dgp_iid_inno} and \eqref{eq:garch_error}.
The upper and lower panels report the empirical size of testing the
null hypotheses $\mathbb{H}_{0}:\beta_{1}^{*}=0$ and $\mathbb{H}_{0}:\gamma_{1}^{*}=0$,
respectively, at a 5\% nominal significance level. ``Size'' is calculated
as $R^{-1}\sum_{r=1}^{R}\mathbf{1}\left[|t^{\text{XD}(r)}|>\mathrm{\Phi}_{0.975}\right]$
across $R=2,000$ replications, where $t^{\text{XD}(r)}$ is computed
based on \eqref{eq:heterose} for the $r$-th replication, and the
critical value $\mathrm{\Phi}_{0.975}\left(\approx1.96\right)$ is
the 97.5-th percentile of the standard normal distribution. ``Len.''
refers to the median length of the 95\% confidence intervals across
replications. The IVX oracle and OLS oracle are infeasible estimators.
The ``Calibrated'' and ``CV'' columns refer to the methods used
for choosing the tuning parameters through calibration and cross-validation,
respectively. }}{\footnotesize\par}
\end{table}
 
\begin{table}[hp]
\begin{centering}
\caption{Empirical size and length of confidence with robust S.E.: AR(1) innovations
and GARCH error terms\protect\label{tab:mc_size_robustse_ar1_garch_error}}
{\small\begin{tabular}{c|cccc|cccc|cccc}
\hline 
\multirow{3}{*}{$n$ } & \multicolumn{4}{c|}{Oracle} & \multicolumn{4}{c|}{Calibrated } & \multicolumn{4}{c}{CV}\tabularnewline
 & \multicolumn{2}{c}{IVX Oracle} & \multicolumn{2}{c|}{OLS Oracle} & \multicolumn{2}{c}{XDlasso} & \multicolumn{2}{c|}{Dlasso} & \multicolumn{2}{c}{XDlasso} & \multicolumn{2}{c}{Dlasso}\tabularnewline
\cline{2-13}
 & Size  & Len.  & Size  & Len.  & Size  & Len.  & Size  & Len.  & Size  & Len.  & Size  & Len. \tabularnewline
\hline 
\multicolumn{13}{c}{$\mathbb{H}_{0}:\beta_{1}^{*}=0$ for nonstationary regressor}\tabularnewline
\hline 
\multicolumn{13}{c}{$\left(p_{x},p_{z}\right)=\left(50,100\right)$}\tabularnewline
\hline 
200  & 0.055  & 0.161  & 0.157  & 0.072  & 0.063  & 0.170  & 0.403  & 0.079  & 0.090  & 0.174  & 0.497  & 0.137 \tabularnewline
300  & 0.050  & 0.113  & 0.157  & 0.047  & 0.056  & 0.121  & 0.478  & 0.059  & 0.085  & 0.123  & 0.573  & 0.093 \tabularnewline
400  & 0.061  & 0.088  & 0.155  & 0.035  & 0.059  & 0.097  & 0.526  & 0.048  & 0.085  & 0.099  & 0.601  & 0.072 \tabularnewline
500  & 0.055  & 0.074  & 0.153  & 0.028  & 0.052  & 0.083  & 0.552  & 0.041  & 0.062  & 0.083  & 0.613  & 0.057 \tabularnewline
600  & 0.045  & 0.063  & 0.141  & 0.024  & 0.051  & 0.072  & 0.559  & 0.036  & 0.072  & 0.072  & 0.604  & 0.046 \tabularnewline
\hline 
\multicolumn{13}{c}{$\left(p_{x},p_{z}\right)=\left(100,150\right)$}\tabularnewline
\hline 
200  & 0.039  & 0.162  & 0.158  & 0.073  & 0.057  & 0.172  & 0.411  & 0.077  & 0.095  & 0.178  & 0.545  & 0.136 \tabularnewline
300  & 0.050  & 0.114  & 0.161  & 0.048  & 0.055  & 0.122  & 0.512  & 0.058  & 0.090  & 0.127  & 0.677  & 0.115 \tabularnewline
400  & 0.048  & 0.090  & 0.165  & 0.036  & 0.064  & 0.096  & 0.581  & 0.047  & 0.096  & 0.099  & 0.706  & 0.083 \tabularnewline
500  & 0.052  & 0.074  & 0.157  & 0.028  & 0.059  & 0.080  & 0.616  & 0.040  & 0.089  & 0.084  & 0.737  & 0.072 \tabularnewline
600  & 0.054  & 0.063  & 0.154  & 0.024  & 0.053  & 0.070  & 0.648  & 0.035  & 0.080  & 0.072  & 0.751  & 0.054 \tabularnewline
\hline 
\multicolumn{13}{c}{$\left(p_{x},p_{z}\right)=\left(150,300\right)$}\tabularnewline
\hline 
200  & 0.041  & 0.161  & 0.147  & 0.072  & 0.057  & 0.165  & 0.363  & 0.073  & 0.087  & 0.172  & 0.556  & 0.109 \tabularnewline
300  & 0.044  & 0.115  & 0.161  & 0.048  & 0.060  & 0.118  & 0.458  & 0.055  & 0.083  & 0.124  & 0.628  & 0.087 \tabularnewline
400  & 0.047  & 0.089  & 0.159  & 0.036  & 0.054  & 0.092  & 0.544  & 0.045  & 0.079  & 0.100  & 0.666  & 0.072 \tabularnewline
500  & 0.041  & 0.074  & 0.155  & 0.028  & 0.051  & 0.077  & 0.587  & 0.038  & 0.075  & 0.083  & 0.715  & 0.065 \tabularnewline
600  & 0.050  & 0.064  & 0.162  & 0.024  & 0.054  & 0.067  & 0.628  & 0.033  & 0.071  & 0.072  & 0.744  & 0.058 \tabularnewline
\hline 
\multicolumn{13}{c}{$\mathbb{H}_{0}:\gamma_{1}^{*}=0$ for stationary regressor}\tabularnewline
\hline 
\multicolumn{13}{c}{$\left(p_{x},p_{z}\right)=\left(50,100\right)$}\tabularnewline
\hline 
200  & 0.041  & 0.379  & 0.055  & 0.312  & 0.060  & 0.334  & 0.064  & 0.274  & 0.062  & 0.330  & 0.068  & 0.271 \tabularnewline
300  & 0.043  & 0.298  & 0.054  & 0.252  & 0.056  & 0.269  & 0.056  & 0.228  & 0.057  & 0.268  & 0.060  & 0.227 \tabularnewline
400  & 0.045  & 0.251  & 0.053  & 0.218  & 0.053  & 0.231  & 0.061  & 0.200  & 0.058  & 0.231  & 0.062  & 0.200 \tabularnewline
500  & 0.044  & 0.221  & 0.053  & 0.195  & 0.057  & 0.205  & 0.065  & 0.180  & 0.060  & 0.205  & 0.067  & 0.180 \tabularnewline
600  & 0.046  & 0.199  & 0.052  & 0.178  & 0.059  & 0.187  & 0.059  & 0.165  & 0.060  & 0.187  & 0.064  & 0.165 \tabularnewline
\hline 
\multicolumn{13}{c}{$\left(p_{x},p_{z}\right)=\left(100,150\right)$}\tabularnewline
\hline 
200  & 0.046  & 0.382  & 0.059  & 0.313  & 0.061  & 0.334  & 0.071  & 0.275  & 0.072  & 0.330  & 0.078  & 0.271 \tabularnewline
300  & 0.046  & 0.299  & 0.053  & 0.254  & 0.064  & 0.269  & 0.063  & 0.228  & 0.064  & 0.267  & 0.063  & 0.226 \tabularnewline
400  & 0.039  & 0.252  & 0.054  & 0.219  & 0.051  & 0.230  & 0.055  & 0.200  & 0.054  & 0.230  & 0.056  & 0.198 \tabularnewline
500  & 0.041  & 0.221  & 0.049  & 0.195  & 0.052  & 0.204  & 0.054  & 0.179  & 0.051  & 0.204  & 0.058  & 0.178 \tabularnewline
600  & 0.040  & 0.200  & 0.045  & 0.178  & 0.048  & 0.186  & 0.052  & 0.165  & 0.049  & 0.186  & 0.058  & 0.165 \tabularnewline
\hline 
\multicolumn{13}{c}{$\left(p_{x},p_{z}\right)=\left(150,300\right)$}\tabularnewline
\hline 
200  & 0.036  & 0.380  & 0.054  & 0.314  & 0.060  & 0.335  & 0.064  & 0.276  & 0.059  & 0.323  & 0.067  & 0.267 \tabularnewline
300  & 0.039  & 0.298  & 0.054  & 0.253  & 0.060  & 0.268  & 0.063  & 0.228  & 0.061  & 0.263  & 0.064  & 0.224 \tabularnewline
400  & 0.040  & 0.251  & 0.053  & 0.218  & 0.060  & 0.229  & 0.072  & 0.198  & 0.056  & 0.227  & 0.065  & 0.197 \tabularnewline
500  & 0.038  & 0.221  & 0.055  & 0.195  & 0.062  & 0.203  & 0.065  & 0.179  & 0.058  & 0.203  & 0.063  & 0.178 \tabularnewline
600  & 0.047  & 0.199  & 0.049  & 0.178  & 0.059  & 0.184  & 0.058  & 0.164  & 0.053  & 0.184  & 0.057  & 0.164 \tabularnewline
\hline 
\end{tabular}

}
\par\end{centering}
{

{\footnotesize\textit{Notes}}{\footnotesize : The data generating process
corresponds to \eqref{eq:dgp_ar1_inno} and \eqref{eq:garch_error}.
The upper and lower panels report the empirical size of testing the
null hypotheses $\mathbb{H}_{0}:\beta_{1}^{*}=0$ and $\mathbb{H}_{0}:\gamma_{1}^{*}=0$
at a 5\% nominal significance level, respectively. ``Size'' is calculated
as $R^{-1}\sum_{r=1}^{R}\mathbf{1}\left[|t^{\text{XD}(r)}|>\mathrm{\Phi}_{0.975}\right]$
across $R=2,000$ replications, where $t^{\text{XD}(r)}$ is computed
based on \eqref{eq:heterose} for the $r$-th replication, and the
critical value $\mathrm{\Phi}_{0.975}\left(\approx1.96\right)$ is
the 97.5-th percentile of the standard normal distribution. ``Len.''
refers to the median length of the 95\% confidence intervals across
replications. The IVX oracle and OLS oracle are infeasible estimators.
The ``Calibrated'' and ``CV'' columns refer to the methods used
for choosing the tuning parameters through calibration and cross-validation,
respectively. }}{\footnotesize\par}
\end{table}

\newpage

\subsection{Variance Ratio Test on Slasso Residual\protect\label{sec:mc_vrtest}}

In this section, we examine the finite sample rejection rate of performing
the automatic variance ratio test based on wild bootstrap \citep{kim2009automatic}
on the first step Slasso residuals. We follow the data generating
process in Section \ref{subsec:Setup} with AR(1) innovations, where
$u_{t}\sim i.i.d.N\left(0.1\right)$, and focus on $\left(p_{x},p_{z}\right)=\left(50,100\right)$
and $n\in\left\{ 200,400,600\right\} $. Table \ref{tab:vr_test_mc}
reports the proportion of rejection at nominal significance levels
1\%, 5\% and 10\% across $2,000$ replications. Even though Slasso
is consistent, we still observe severe over-rejection of performing
the variance ratio test on $\hat{u}_{t}$ in finite sample. 

\begin{table}
\caption{Rej. Rate of the Wild Bootstrapped Automatic Variance Ratio Test on
Slasso Residual\protect\label{tab:vr_test_mc}}

\medskip{}

\begin{centering}
\begin{tabular}{cccc}
\toprule 
$n$ & 1\% & 5\% & 10\%\tabularnewline
\midrule 
$200$ & 4.4\% & 14.0\% & 22.0\%\tabularnewline
\midrule 
$400$ & 3.4\% & 10.7\% & 17.2\%\tabularnewline
\midrule 
$600$ & 2.2\% & 10.0\% & 15.9\%\tabularnewline
\bottomrule
\end{tabular}\medskip{}
\par\end{centering}
{

{\footnotesize\textit{Notes}}{\footnotesize : We follow the data generating
process in Section 4.1 with AR(1) innovations, where $u_{t}\sim i.i.d.N\left(0.1\right)$,
and focus on $\left(p_{x},p_{z}\right)=\left(50,100\right)$ and $n\in\left\{ 200,400,600\right\} $.
This table reports the rejection rates of the automatic variance ratio
test based on the wild bootstrap \citep{kim2009automatic}. The test
is applied to the first-step Slasso residuals. Rejection rates are
shown at nominal significance levels of 1\%, 5\%, and 10\%, based
on 2,000 replications.}}{\footnotesize\par}
\end{table}

\setcounter{table}{0} 
\setcounter{figure}{0} 
\setcounter{equation}{0} 

\section{Additional Empirical Results\protect\label{sec:Additional-Empirical-Results}}

\subsection{Sensitivity to the Classification of I(2) Time Series and Logarithmic
Transformation\protect\label{sec:app_sens}}

To further evaluate the robustness of our findings, we consider two
alternative specifications in addition to the main analysis in Section
\ref{sec:Empirical-Application}. First, we exclude nonstationary
variables based on their integration orders as determined by the bootstrap
sequential testing procedure of \citet{smeekes2015bootstrap}, following
the summary in \citet{smeekes2020unit}. Second, we apply only logarithmic
transformations as indicated by TCODE, without differencing, and remove
highly nonstationary time series according to both TCODE and the classifications
in \citet{smeekes2020unit}. Table \ref{tab:app_est_sens_return}
shows that, across all specifications, XDlasso consistently finds
no evidence that the log earnings-price ratio
has predictive power for stock returns, which confirms
our main findings in Section \ref{sec:emp-return-ep}. Similarly,
the regression results for inflation predictability in Table \ref{tab:app_est_sens_pc}
are largely in line with those in Section \ref{subsec:emp_inflation_unrate}.
The only exception occurs in the full-sample estimation using untransformed
data and the classification of \citet{smeekes2020unit} when labor
variables are excluded. In these exceptions, diagnostic tests suggest
a violation of the m.d.s.~condition, which undermines the validity
of the inference.

\begin{table}
\begin{centering}
\caption{Test $\mathbb{H}_{0}:\theta_{1}^{*}=0$ for stock return predictability:
Alternative set of I(2) variables and transformation \protect\label{tab:app_est_sens_return}}
\par\end{centering}
{\small 
\begin{subtable}[t]{1\textwidth}
\centering
\caption{Untransformed Data: Excluding I(2) Variables Based on  \citet{smeekes2020unit}}

\begin{tabular}{c|ccc|ccc}
  \hline
  & \multicolumn{3}{c|}{Without \(\text{Return}_{t-1}\)} & \multicolumn{3}{c}{Include \(\text{Return}_{t-1}\)} \tabularnewline
  \hline
Sample Period & Dlasso & XDlasso & VR Test  & Dlasso & XDlasso & VR Test \tabularnewline
  \hline
Full Sample & 0.020 & -0.008 & \multirow{2}{*}{0.000} & 0.025$^{*}$ & 0.012 & \multirow{2}{*}{0.831} \tabularnewline
{\footnotesize \textit{(Jan. 1960 - Apr. 2025)}} & \textit{(0.016)} & \textit{(0.011)} & & \textit{(0.015)} & \textit{(0.010)} & \tabularnewline
\hline
Pre-1994 & 0.035 & -0.208 & \multirow{2}{*}{0.049} & 0.044 & 0.168 & \multirow{2}{*}{0.494} \tabularnewline
{\footnotesize \textit{(Jan. 1960 - Dec. 1993)}} & \textit{(0.055)} & \textit{(0.215)} & & \textit{(0.042)} & \textit{(0.154)} & \tabularnewline
Post-1994 & -0.003 & -0.022 & \multirow{2}{*}{0.011} & -0.000 & -0.009 & \multirow{2}{*}{0.163} \tabularnewline
{\footnotesize \textit{(Jan. 1994 - Apr. 2025)}} & \textit{(0.009)} & \textit{(0.017)} & & \textit{(0.009)} & \textit{(0.017)} & \tabularnewline
  \hline
\end{tabular}

\end{subtable}
\hfill
\medskip

\begin{subtable}[t]{1\textwidth}
\centering
\caption{Log Transformation: Excluding I(2) Variables Based on \texttt{TCODE}}

\begin{tabular}{c|ccc|ccc}
  \hline
  & \multicolumn{3}{c|}{Without \(\text{Return}_{t-1}\)} & \multicolumn{3}{c}{Include \(\text{Return}_{t-1}\)} \tabularnewline
  \hline
Sample Period & Dlasso & XDlasso & VR Test  & Dlasso & XDlasso & VR Test \tabularnewline
  \hline
Full Sample & 0.011 & -0.005 & \multirow{2}{*}{0.003} & 0.021 & 0.012 & \multirow{2}{*}{0.576} \tabularnewline
{\footnotesize \textit{(Jan. 1960 - Apr. 2025)}} & \textit{(0.014)} & \textit{(0.011)} & & \textit{(0.014)} & \textit{(0.010)} & \tabularnewline
\hline
Pre-1994 & 0.039 & -0.360 & \multirow{2}{*}{0.033} & 0.033 & 0.087 & \multirow{2}{*}{0.489} \tabularnewline
{\footnotesize \textit{(Jan. 1960 - Dec. 1993)}} & \textit{(0.049)} & \textit{(0.291)} & & \textit{(0.048)} & \textit{(0.213)} & \tabularnewline
Post-1994 & 0.041$^{**}$ & 0.020 & \multirow{2}{*}{0.010} & 0.045$^{**}$ & 0.033 & \multirow{2}{*}{0.258} \tabularnewline
{\footnotesize \textit{(Jan. 1994 - Apr. 2025)}} & \textit{(0.020)} & \textit{(0.030)} & & \textit{(0.020)} & \textit{(0.029)} & \tabularnewline
  \hline
\end{tabular}

\end{subtable}
\hfill
\medskip

\begin{subtable}[t]{1\textwidth}
\centering
\caption{Log Transformation: Excluding I(2) Variables Based on  \citet{smeekes2020unit}}

\begin{tabular}{c|ccc|ccc}
  \hline
  & \multicolumn{3}{c|}{Without \(\text{Return}_{t-1}\)} & \multicolumn{3}{c}{Include \(\text{Return}_{t-1}\)} \tabularnewline
  \hline
Sample Period & Dlasso & XDlasso & VR Test  & Dlasso & XDlasso & VR Test \tabularnewline
  \hline
Full Sample & 0.016 & -0.009 & \multirow{2}{*}{0.001} & 0.022$^{*}$ & 0.013 & \multirow{2}{*}{0.559} \tabularnewline
{\footnotesize \textit{(Jan. 1960 - Apr. 2025)}} & \textit{(0.012)} & \textit{(0.016)} &  & \textit{(0.012)} & \textit{(0.010)} & \tabularnewline
\hline
Pre-1994 & 0.106$^{*}$ & -0.394 & \multirow{2}{*}{0.032} & 0.108$^{**}$ & 0.101 & \multirow{2}{*}{0.476} \tabularnewline
{\footnotesize \textit{(Jan. 1960 - Dec. 1993)}} & \textit{(0.055)} & \textit{(0.317)} &  & \textit{(0.054)} & \textit{(0.292)} & \tabularnewline
Post-1994 & 0.018 & -0.001 & \multirow{2}{*}{0.010} & 0.021 & 0.016 & \multirow{2}{*}{0.216} \tabularnewline
{\footnotesize \textit{(Jan. 1994 - Apr. 2025)}} & \textit{(0.017)} & \textit{(0.037)} &  & \textit{(0.017)} & \textit{(0.027)} & \tabularnewline
  \hline
\end{tabular}

\end{subtable}
}
\medskip

{ {\footnotesize\textit{{}Notes:}}{\footnotesize{} We report estimates
and the standard error (in parentheses below the estimates) across
methods and setups. The symbols {*}, {*}{*}, and {*}{*}{*} indicate
significance levels at 10\%, 5\%, and 1\%, respectively. ``VR Test''
represents the $p$-value of the variance ratio test \citep{kim2009automatic}
on the LASSO residual. The tuning parameter for LASSO estimation is
selected through 10-fold block cross-validation. In XDlasso, instruments
are generated based on \eqref{eq:IVX_original} and \eqref{eq:rho=000020zeta}
with $C_{\zeta}=5$ and $\tau=0.5$.}}
\end{table}

\begin{table}
\begin{centering}
\caption{Test $\mathbb{H}_{0}:\theta_{1}^{*}=0$ for inflation predictability: Alternative
set of I(2) variables and transformation\protect\label{tab:app_est_sens_pc}}
\par\end{centering}
{\small 
\begin{subtable}[t]{1\textwidth}
\centering
\caption{Untransformed Data: Excluding I(2) Variables Based on \citet{smeekes2020unit}}

\begin{tabular}{c|ccc|ccc}
  \hline
  & \multicolumn{3}{c|}{Include Labor Variables} & \multicolumn{3}{c}{Exclude Labor Variables}   \tabularnewline 
  \hline
Sample Period & Dlasso & XDlasso & VR Test & Dlasso & XDlasso & VR Test \tabularnewline 
  \hline
Full Sample & -0.172** & 0.118 & \multirow{2}{*}{0.104} & 0.069** & 0.116*** & \multirow{2}{*}{0.015} \tabularnewline
 {\footnotesize \textit{(Jan. 1960 - Apr. 2025)}} & \textit{(0.070)} & \textit{(0.231)} & & \textit{(0.029)} & \textit{(0.041)} & \tabularnewline 
 \hline
Pre-Volcker & -0.169 & -0.022 & \multirow{2}{*}{0.516} & 0.018 & 0.106 & \multirow{2}{*}{0.555} \tabularnewline
 {\footnotesize \textit{(Jan. 1960 - Jul. 1979)}} & \textit{(0.150)} & \textit{(0.285)} & & \textit{(0.053)} & \textit{(0.108)} & \tabularnewline
Volcker-Greenspan & -0.041 & -0.200 & \multirow{2}{*}{0.669} & -0.009 & -0.189 & \multirow{2}{*}{0.641} \tabularnewline
 {\footnotesize \textit{(Aug. 1979 - Jan. 2006)}} & \textit{(0.207)} & \textit{(0.270)} & & \textit{(0.058)} & \textit{(0.123)} & \tabularnewline
Bernanke-Yellen-Powell & 0.550* & 0.294 & \multirow{2}{*}{0.220} & 0.081 & 0.162** & \multirow{2}{*}{0.094} \tabularnewline
 {\footnotesize \textit{(Feb. 2006 - Apr. 2025)}} & \textit{(0.309)} & \textit{(0.510)} & & \textit{(0.050)} & \textit{(0.076)} & \tabularnewline
  \hline
\end{tabular}

\end{subtable}
\hfill
\medskip

\begin{subtable}[t]{1\textwidth}
\centering
\caption{Log Transformation: Excluding I(2) Variables Based on \texttt{TCODE}}

\begin{tabular}{c|ccc|ccc}
  \hline
  & \multicolumn{3}{c|}{Include Labor Variables} & \multicolumn{3}{c}{Exclude Labor Variables}   \tabularnewline 
  \hline
Sample Period & Dlasso & XDlasso & VR Test & Dlasso & XDlasso & VR Test \tabularnewline 
  \hline
Full Sample & 0.014 & -0.103 & \multirow{2}{*}{0.118} & -0.063*** & 0.044 & \multirow{2}{*}{0.005} \tabularnewline
 {\footnotesize \textit{(Jan. 1960 - Apr. 2025)}} & \textit{(0.048)} & \textit{(0.103)} & & \textit{(0.016)} & \textit{(0.056)} & \tabularnewline 
 \hline
Pre-Volcker & -0.155 & -0.056 & \multirow{2}{*}{0.516} & -0.029 & 0.131 & \multirow{2}{*}{0.827} \tabularnewline
 {\footnotesize \textit{(Jan. 1960 - Jul. 1979)}} & \textit{(0.164)} & \textit{(0.266)} & & \textit{(0.048)} & \textit{(0.117)} & \tabularnewline
Volcker-Greenspan & 0.064 & -0.068 & \multirow{2}{*}{0.448} & -0.082 & -0.057 & \multirow{2}{*}{0.427} \tabularnewline
 {\footnotesize \textit{(Aug. 1979 - Jan. 2006)}} & \textit{(0.141)} & \textit{(0.198)} & & \textit{(0.060)} & \textit{(0.332)} & \tabularnewline
Bernanke-Yellen-Powell & 0.121 & 1.231 & \multirow{2}{*}{0.398} & -0.040 & 0.142 & \multirow{2}{*}{0.119} \tabularnewline
 {\footnotesize \textit{(Feb. 2006 - Apr. 2025)}} & \textit{(0.182)} & \textit{(0.754)} & & \textit{(0.057)} & \textit{(0.102)} & \tabularnewline
  \hline
\end{tabular}

\end{subtable}
\hfill
\medskip

\begin{subtable}[t]{1\textwidth}
\centering
\caption{Log Transformation: Excluding I(2) Variables Based on  \citet{smeekes2020unit}}

\begin{tabular}{c|ccc|ccc}
  \hline
  & \multicolumn{3}{c|}{Include Labor Variables} & \multicolumn{3}{c}{Exclude Labor Variables}   \tabularnewline 
  \hline
Sample Period & Dlasso & XDlasso & VR Test & Dlasso & XDlasso & VR Test \tabularnewline 
  \hline
Full Sample & -0.016 & -0.101 & \multirow{2}{*}{0.110} & -0.030** & 0.016 & \multirow{2}{*}{0.020} \tabularnewline
 {\footnotesize \textit{(Jan. 1960 - Apr. 2025)}} & \textit{(0.049)} & \textit{(0.102)} & & \textit{(0.013)} & \textit{(0.028)} & \tabularnewline 
 \hline
Pre-Volcker & -0.137 & -0.068 & \multirow{2}{*}{0.430} & -0.017 & 0.052 & \multirow{2}{*}{0.556} \tabularnewline
 {\footnotesize \textit{(Jan. 1960 - Jul. 1979)}} & \textit{(0.166)} & \textit{(0.250)} & & \textit{(0.052)} & \textit{(0.110)} & \tabularnewline
Volcker-Greenspan & -0.054 & -0.136 & \multirow{2}{*}{0.871} & -0.014 & -0.172 & \multirow{2}{*}{0.702} \tabularnewline
 {\footnotesize \textit{(Aug. 1979 - Jan. 2006)}} & \textit{(0.135)} & \textit{(0.304)} & & \textit{(0.064)} & \textit{(0.124)} & \tabularnewline
Bernanke-Yellen-Powell & 0.018 & 0.555 & \multirow{2}{*}{0.541} & 0.087* & 0.098* & \multirow{2}{*}{0.180} \tabularnewline
 {\footnotesize \textit{(Feb. 2006 - Apr. 2025)}} & \textit{(0.133)} & \textit{(0.439)} & & \textit{(0.047)} & \textit{(0.055)} & \tabularnewline
  \hline
\end{tabular}

\end{subtable}
}
\medskip

{ {\footnotesize\textit{{}Notes:}}{\footnotesize{} We report estimates
and the standard error (in parentheses below the estimates) across
methods and setups. The symbols {*}, {*}{*}, and {*}{*}{*} indicate
significance levels at 10\%, 5\%, and 1\%, respectively. ``VR Test''
represents the $p$-value of the variance ratio test \citep{kim2009automatic}
on the LASSO residual. The tuning parameter for LASSO estimation is
selected through 10-fold block cross-validation. In XDlasso, instruments
are generated based on \eqref{eq:IVX_original} and \eqref{eq:rho=000020zeta}
with $C_{\zeta}=5$ and $\tau=0.5$.}}
\end{table}

\clearpage

\subsection{Heteroskedasticity-Robust Standard Errors\protect\label{sec:robustse}}

To assess the robustness of our results to conditional heteroskedasticity,
we recompute the heteroskedasticity-robust standard errors in \eqref{eq:heterose}.
Tables \ref{tab:app_est_return_robust} and \ref{tab:app_est_pc_robust}
present the robust results for stock return and inflation predictability
parallel to those in Section \ref{sec:Empirical-Application} but
with heteroskedasticity-robust standard errors. Tables \ref{tab:app_est_sens_robust_return}
and \ref{tab:app_est_sens_robust_pc} report the sensitivity analyses
like Section \ref{sec:app_sens} with robust standard errors. Across
all cases, our results remain consistent with those in the baseline
analysis.

\begin{table}
\begin{centering}
\caption{Test $\mathbb{H}_{0}:\theta_{1}^{*}=0$ for stock return predictability
with heteroskedasticity-robust standard errors\protect\label{tab:app_est_return_robust}}
\par\end{centering}
{\small 
\begin{subtable}[t]{1\textwidth}
\centering
\caption{\texttt{TCODE} Transformed Data}

\begin{tabular}{c|ccc|ccc}
    \hline
    & \multicolumn{3}{c|}{Without \(\text{Return}_{t-1}\)} & \multicolumn{3}{c}{Include \(\text{Return}_{t-1}\)}   \tabularnewline 
    \hline
  Sample Period & Dlasso & XDlasso & VR Test  & Dlasso & XDlasso & VR Test \tabularnewline 
    \hline
  Full Sample & 0.009 & 0.003 & \multirow{2}{*}{0.074} & 0.009 & 0.005 & \multirow{2}{*}{0.216} \tabularnewline 
   {\footnotesize \textit{(Jan. 1960 - Apr. 2025)}} & \textit{(0.006)} & \textit{(0.014)} & \multirow{2}{*}{} & \textit{(0.006)} & \textit{(0.014)} & \multirow{2}{*}{} \tabularnewline 
   \hline
  Pre-1994 & 0.025*** & 0.059** & \multirow{2}{*}{0.227} & 0.024** & 0.062** & \multirow{2}{*}{0.296} \tabularnewline 
   {\footnotesize \textit{(Jan. 1960 - Dec. 1993)}} & \textit{(0.010)} & \textit{(0.029)} & \multirow{2}{*}{} & \textit{(0.009)} & \textit{(0.029)} & \multirow{2}{*}{} \tabularnewline 
  Post-1994 & 0.002 & -0.001 & \multirow{2}{*}{0.053} & 0.002 & -0.001 & \multirow{2}{*}{0.049} \tabularnewline 
   {\footnotesize \textit{(Jan. 1994 - Apr. 2025)}} & \textit{(0.006)} & \textit{(0.015)} & \multirow{2}{*}{} & \textit{(0.006)} & \textit{(0.015)} & \multirow{2}{*}{} \tabularnewline 
    \hline
  \end{tabular}

\end{subtable}
\hfill
\medskip

\begin{subtable}[t]{1\textwidth}
\centering
\caption{Untransformed Data: Excluding I(2) Variables Based on \texttt{TCODE}}

  \begin{tabular}{c|ccc|ccc}
    \hline
    & \multicolumn{3}{c|}{Without \(\text{Return}_{t-1}\)} & \multicolumn{3}{c}{Include \(\text{Return}_{t-1}\)}   \tabularnewline 
    \hline
  Sample Period & Dlasso & XDlasso & VR Test  & Dlasso & XDlasso & VR Test \tabularnewline 
    \hline
  Full Sample & 0.013 & -0.008 & \multirow{2}{*}{0.001} & 0.019 & 0.012 & \multirow{2}{*}{0.811} \tabularnewline 
   {\footnotesize \textit{(Jan. 1960 - Apr. 2025)}} & \textit{(0.014)} & \textit{(0.010)} & \multirow{2}{*}{} & \textit{(0.014)} & \textit{(0.010)} & \multirow{2}{*}{} \tabularnewline 
   \hline
  Pre-1994 & 0.064** & -0.312 & \multirow{2}{*}{0.046} & 0.055* & 0.096 & \multirow{2}{*}{0.467} \tabularnewline 
   {\footnotesize \textit{(Jan. 1960 - Dec. 1993)}} & \textit{(0.030)} & \textit{(0.296)} & \multirow{2}{*}{} & \textit{(0.031)} & \textit{(0.070)} & \multirow{2}{*}{} \tabularnewline 
  Post-1994 & -0.003 & -0.022 & \multirow{2}{*}{0.016} & -0.000 & -0.004 & \multirow{2}{*}{0.280} \tabularnewline 
   {\footnotesize \textit{(Jan. 1994 - Apr. 2025)}} & \textit{(0.008)} & \textit{(0.015)} & \multirow{2}{*}{} & \textit{(0.007)} & \textit{(0.014)} & \multirow{2}{*}{} \tabularnewline 
    \hline
\end{tabular}

\end{subtable}
}
\medskip

{ {\footnotesize\textit{{}Notes:}}{\footnotesize{} We report estimates
and the standard error (in parentheses below the estimates) across
methods and setups. The symbols {*}, {*}{*}, and {*}{*}{*} indicate
significance levels at 10\%, 5\%, and 1\%, respectively. ``VR Test''
represents the $p$-value of the variance ratio test \citep{kim2009automatic}
on the LASSO residual. The tuning parameter for LASSO estimation is
selected through 10-fold block cross-validation. In XDlasso, instruments
are generated based on \eqref{eq:IVX_original} and \eqref{eq:rho=000020zeta}
with $C_{\zeta}=5$ and $\tau=0.5$.}}
\end{table}

\begin{table}
\begin{centering}
\caption{Test $\mathbb{H}_{0}:\theta_{1}^{*}=0$ for inflation predictability with
heteroskedasticity-robust standard errors\protect\label{tab:app_est_pc_robust}}
\par\end{centering}
{\small 
\begin{subtable}[t]{1\textwidth}
\centering
\caption{\texttt{TCODE} Transformed Data}

\begin{tabular}{c|ccc}
  \hline
Sample Period & Dlasso & XDlasso & VR Test \tabularnewline 
  \hline
Full Sample & 0.018*** & -0.024 & \multirow{2}{*}{0.000} \tabularnewline 
 {\footnotesize \textit{(Jan. 1960 - Apr. 2025)}} & \textit{(0.006)} & \textit{(0.074)} \tabularnewline 
 \hline
Pre-Volcker & 0.074*** & 0.013 & \multirow{2}{*}{0.125} \tabularnewline 
 {\footnotesize \textit{(Jan. 1960 - Jul. 1979)}} & \textit{(0.018)} & \textit{(0.211)} \tabularnewline 
Volcker-Greenspan & -0.020 & 0.161 & \multirow{2}{*}{0.025} \tabularnewline 
 {\footnotesize \textit{(Aug. 1979 - Jan. 2006)}} & \textit{(0.020)} & \textit{(0.137)} \tabularnewline 
Bernanke/Yellen/Powell & -0.002 & -0.054 & \multirow{2}{*}{0.002} \tabularnewline 
 {\footnotesize \textit{(Feb. 2006 - Apr. 2025)}} & \textit{(0.010)} & \textit{(0.098)} \tabularnewline 
  \hline
\end{tabular}

\end{subtable}
\hfill
\medskip

\begin{subtable}[t]{1\textwidth}
\centering
\caption{Untransformed Data: Excluding I(2) Variables Based on \texttt{TCODE}}

\begin{tabular}{c|ccc|ccc}
  \hline
  & \multicolumn{3}{c|}{Include Labor Variables} & \multicolumn{3}{c}{Exclude Labor Variables}   \tabularnewline 
  \hline
Sample Period & Dlasso & XDlasso & VR Test  & Dlasso & XDlasso & VR Test \tabularnewline 
  \hline
Full Sample & -0.077 & 0.068 & \multirow{2}{*}{0.080} & -0.050** & 0.033 & \multirow{2}{*}{0.065} \tabularnewline 
 {\footnotesize \textit{(Jan. 1960 - Apr. 2025)}} & \textit{(0.068)} & \textit{(0.200)} &  & \textit{(0.022)} & \textit{(0.031)} & \tabularnewline 
 \hline
Pre-Volcker & -0.129 & -0.014 & \multirow{2}{*}{0.522} & 0.007 & 0.113 & \multirow{2}{*}{0.548} \tabularnewline 
 {\footnotesize \textit{(Jan. 1960 - Jul. 1979)}} & \textit{(0.102)} & \textit{(0.274)} & & \textit{(0.051)} & \textit{(0.089)} & \tabularnewline 
Volcker-Greenspan & 0.094 & -0.272 & \multirow{2}{*}{0.741} & -0.092* & -0.259 & \multirow{2}{*}{0.339} \tabularnewline 
 {\footnotesize \textit{(Aug. 1979 - Jan. 2006)}} & \textit{(0.185)} & \textit{(0.279)} &  & \textit{(0.054)} & \textit{(0.224)} &  \tabularnewline 
Bernanke-Yellen-Powell & 0.550 & 0.585 & \multirow{2}{*}{0.283} & 0.001 & 0.040 & \multirow{2}{*}{0.230} \tabularnewline 
 {\footnotesize \textit{(Feb. 2006 - Apr. 2025)}} & \textit{(0.522)} & \textit{(0.796)} & & \textit{(0.076)} & \textit{(0.080)} & \tabularnewline 
  \hline
\end{tabular}

\end{subtable}
}
\medskip

{ {\footnotesize\textit{{}Notes:}}{\footnotesize{} We report estimates
and the standard error (in parentheses below the estimates) across
methods and setups. The symbols {*}, {*}{*}, and {*}{*}{*} indicate
significance levels at 10\%, 5\%, and 1\%, respectively. ``VR Test''
represents the $p$-value of the variance ratio test \citep{kim2009automatic}
on the LASSO residual. The tuning parameter for LASSO estimation is
selected through 10-fold block cross-validation. In XDlasso, instruments
are generated based on \eqref{eq:IVX_original} and \eqref{eq:rho=000020zeta}
with $C_{\zeta}=5$ and $\tau=0.5$.}}
\end{table}

\begin{table}
\begin{centering}
\caption{Test $\mathbb{H}_{0}:\theta_{1}^{*}=0$ for stock return predictability:
Alternative set of I(2) variables and transformation with heteroskedasticity-robust
standard errors\protect\label{tab:app_est_sens_robust_return}}
\par\end{centering}
{\small 
\begin{subtable}[t]{1\textwidth}
\centering
\caption{Untransformed Data: Excluding I(2) Variables Based on \citet{smeekes2020unit}}

\begin{tabular}{c|ccc|ccc}
  \hline
  & \multicolumn{3}{c|}{Without $\text{Return}_{t-1}$} & \multicolumn{3}{c}{Include $\text{Return}_{t-1}$} \tabularnewline
  \hline
Sample Period & Dlasso & XDlasso & VR Test & Dlasso & XDlasso & VR Test \tabularnewline
  \hline
Full Sample & 0.020 & -0.008 & 0.000 & 0.025 & 0.012 & 0.831 \tabularnewline
{\footnotesize \textit{(Jan. 1960 - Apr. 2025)}} & \textit{(0.018)} & \textit{(0.010)} & & \textit{(0.017)} & \textit{(0.010)} & \tabularnewline
\hline
Pre-1994 & 0.035 & -0.208 & 0.049 & 0.044 & 0.168 & 0.494 \tabularnewline
{\footnotesize \textit{(Jan. 1960 - Dec. 1993)}} & \textit{(0.051)} & \textit{(0.229)} & & \textit{(0.042)} & \textit{(0.143)} & \tabularnewline
Post-1994 & -0.003 & -0.022 & 0.011 & -0.000 & -0.009 & 0.163 \tabularnewline
{\footnotesize \textit{(Jan. 1994 - Apr. 2025)}} & \textit{(0.008)} & \textit{(0.015)} & & \textit{(0.007)} & \textit{(0.015)} & \tabularnewline
  \hline
\end{tabular}

\end{subtable}
\hfill
\medskip

\begin{subtable}[t]{1\textwidth}
\centering
\caption{Log Transformation: Excluding I(2) Variables Based on \texttt{TCODE}}

\begin{tabular}{c|ccc|ccc}
  \hline
  & \multicolumn{3}{c|}{Without $\text{Return}_{t-1}$} & \multicolumn{3}{c}{Include $\text{Return}_{t-1}$} \tabularnewline
  \hline
Sample Period & Dlasso & XDlasso & VR Test & Dlasso & XDlasso & VR Test \tabularnewline
  \hline
Full Sample & 0.011 & -0.005 & 0.003 & 0.021 & 0.012 & 0.576 \tabularnewline
{\footnotesize \textit{(Jan. 1960 - Apr. 2025)}} & \textit{(0.014)} & \textit{(0.010)} & & \textit{(0.014)} & \textit{(0.010)} & \tabularnewline
\hline
Pre-1994 & 0.039 & -0.360 & 0.033 & 0.033 & 0.087 & 0.489 \tabularnewline
{\footnotesize \textit{(Jan. 1960 - Dec. 1993)}} & \textit{(0.043)} & \textit{(0.294)} & & \textit{(0.041)} & \textit{(0.206)} & \tabularnewline
Post-1994 & 0.041* & 0.020 & 0.010 & 0.045** & 0.033 & 0.258 \tabularnewline
{\footnotesize \textit{(Jan. 1994 - Apr. 2025)}} & \textit{(0.021)} & \textit{(0.028)} & & \textit{(0.020)} & \textit{(0.028)} & \tabularnewline
  \hline
\end{tabular}

\end{subtable}
\hfill
\medskip

\begin{subtable}[t]{1\textwidth}
\centering
\caption{Log Transformation: Excluding I(2) Variables Based on \citet{smeekes2020unit}}

\begin{tabular}{c|ccc|ccc}
  \hline
  & \multicolumn{3}{c|}{Without $\text{Return}_{t-1}$} & \multicolumn{3}{c}{Include $\text{Return}_{t-1}$} \tabularnewline
  \hline
Sample Period & Dlasso & XDlasso & VR Test & Dlasso & XDlasso & VR Test \tabularnewline
  \hline
Full Sample & 0.016 & -0.009 & 0.001 & 0.022* & 0.013 & 0.559 \tabularnewline
{\footnotesize \textit{(Jan. 1960 - Apr. 2025)}} & \textit{(0.012)} & \textit{(0.013)} & & \textit{(0.012)} & \textit{(0.010)} & \tabularnewline
\hline
Pre-1994 & 0.106* & -0.394 & 0.032 & 0.108** & 0.101 & 0.476 \tabularnewline
{\footnotesize \textit{(Jan. 1960 - Dec. 1993)}} & \textit{(0.056)} & \textit{(0.329)} & & \textit{(0.055)} & \textit{(0.275)} & \tabularnewline
Post-1994 & 0.018 & -0.001 & 0.010 & 0.021 & 0.016 & 0.216 \tabularnewline
{\footnotesize \textit{(Jan. 1994 - Apr. 2025)}} & \textit{(0.016)} & \textit{(0.035)} & & \textit{(0.016)} & \textit{(0.025)} & \tabularnewline
  \hline
\end{tabular}

\end{subtable}
}
\medskip

{ {\footnotesize\textit{{}Notes:}}{\footnotesize{} We report estimates
and the standard error (in parentheses below the estimates) across
methods and setups. The symbols {*}, {*}{*}, and {*}{*}{*} indicate
significance levels at 10\%, 5\%, and 1\%, respectively. ``VR Test''
represents the $p$-value of the variance ratio test \citep{kim2009automatic}
on the LASSO residual. The tuning parameter for LASSO estimation is
selected through 10-fold block cross-validation. In XDlasso, instruments
are generated based on \eqref{eq:IVX_original} and \eqref{eq:rho=000020zeta}
with $C_{\zeta}=5$ and $\tau=0.5$.}}
\end{table}

\begin{table}
\begin{centering}
\caption{Test $\mathbb{H}_{0}:\theta_{1}^{*}=0$ for inflation predictability: Alternative
set of I(2) variables and transformation with heteroskedasticity-robust
standard errors \protect\label{tab:app_est_sens_robust_pc}}
\par\end{centering}
{\small 
\begin{subtable}[t]{1\textwidth}
\centering
\caption{Untransformed Data: Excluding I(2) Variables Based on \citet{smeekes2020unit}}

\begin{tabular}{c|ccc|ccc}
  \hline
  & \multicolumn{3}{c|}{Include Labor Variables} & \multicolumn{3}{c}{Exclude Labor Variables}   \tabularnewline 
  \hline
Sample Period & Dlasso & XDlasso & VR Test & Dlasso & XDlasso & VR Test \tabularnewline 
  \hline
Full Sample & -0.172** & 0.118 & \multirow{2}{*}{0.104} & 0.069* & 0.116*** & \multirow{2}{*}{0.015} \tabularnewline
 {\footnotesize \textit{(Jan. 1960 - Apr. 2025)}} & \textit{(0.072)} & \textit{(0.206)} & & \textit{(0.036)} & \textit{(0.042)} & \tabularnewline 
 \hline
Pre-Volcker & -0.169 & -0.022 & \multirow{2}{*}{0.516} & 0.018 & 0.106 & \multirow{2}{*}{0.555} \tabularnewline
 {\footnotesize \textit{(Jan. 1960 - Jul. 1979)}} & \textit{(0.134)} & \textit{(0.281)} & & \textit{(0.054)} & \textit{(0.097)} & \tabularnewline
Volcker-Greenspan & -0.041 & -0.200 & \multirow{2}{*}{0.669} & -0.009 & -0.189 & \multirow{2}{*}{0.641} \tabularnewline
 {\footnotesize \textit{(Aug. 1979 - Jan. 2006)}} & \textit{(0.188)} & \textit{(0.265)} & & \textit{(0.053)} & \textit{(0.120)} & \tabularnewline
Bernanke-Yellen-Powell & 0.550 & 0.294 & \multirow{2}{*}{0.220} & 0.081 & 0.162* & \multirow{2}{*}{0.094} \tabularnewline
 {\footnotesize \textit{(Feb. 2006 - Apr. 2025)}} & \textit{(0.395)} & \textit{(0.394)} & & \textit{(0.057)} & \textit{(0.083)} & \tabularnewline
  \hline
\end{tabular}

\end{subtable}
\hfill
\medskip

\begin{subtable}[t]{1\textwidth}
\centering
\caption{Log Transformation: Excluding I(2) Variables Based on \texttt{TCODE}}

\begin{tabular}{c|ccc|ccc}
  \hline
  & \multicolumn{3}{c|}{Include Labor Variables} & \multicolumn{3}{c}{Exclude Labor Variables}   \tabularnewline 
  \hline
Sample Period & Dlasso & XDlasso & VR Test & Dlasso & XDlasso & VR Test \tabularnewline 
  \hline
Full Sample & 0.014 & -0.103 & \multirow{2}{*}{0.118} & -0.063*** & 0.044 & \multirow{2}{*}{0.005} \tabularnewline
 {\footnotesize \textit{(Jan. 1960 - Apr. 2025)}} & \textit{(0.056)} & \textit{(0.102)} & & \textit{(0.021)} & \textit{(0.057)} & \tabularnewline 
 \hline
Pre-Volcker & -0.155 & -0.056 & \multirow{2}{*}{0.516} & -0.029 & 0.131 & \multirow{2}{*}{0.827} \tabularnewline
 {\footnotesize \textit{(Jan. 1960 - Jul. 1979)}} & \textit{(0.153)} & \textit{(0.256)} & & \textit{(0.048)} & \textit{(0.103)} & \tabularnewline
Volcker-Greenspan & 0.064 & -0.068 & \multirow{2}{*}{0.448} & -0.082 & -0.057 & \multirow{2}{*}{0.427} \tabularnewline
 {\footnotesize \textit{(Aug. 1979 - Jan. 2006)}} & \textit{(0.131)} & \textit{(0.197)} & & \textit{(0.065)} & \textit{(0.306)} & \tabularnewline
Bernanke-Yellen-Powell & 0.121 & 1.231* & \multirow{2}{*}{0.398} & -0.040 & 0.142 & \multirow{2}{*}{0.119} \tabularnewline
 {\footnotesize \textit{(Feb. 2006 - Apr. 2025)}} & \textit{(0.183)} & \textit{(0.729)} & & \textit{(0.084)} & \textit{(0.100)} & \tabularnewline
  \hline
\end{tabular}

\end{subtable}
\hfill
\medskip

\begin{subtable}[t]{1\textwidth}
\centering
\caption{Log Transformation: Excluding I(2) Variables Based on \citet{smeekes2020unit}}

\begin{tabular}{c|ccc|ccc}
  \hline
  & \multicolumn{3}{c|}{Include Labor Variables} & \multicolumn{3}{c}{Exclude Labor Variables}   \tabularnewline 
  \hline
Sample Period & Dlasso & XDlasso & VR Test & Dlasso & XDlasso & VR Test \tabularnewline 
  \hline
Full Sample & -0.016 & -0.101 & \multirow{2}{*}{0.110} & -0.030* & 0.016 & \multirow{2}{*}{0.020} \tabularnewline
 {\footnotesize \textit{(Jan. 1960 - Apr. 2025)}} & \textit{(0.057)} & \textit{(0.102)} & & \textit{(0.016)} & \textit{(0.023)} & \tabularnewline 
 \hline
Pre-Volcker & -0.137 & -0.068 & \multirow{2}{*}{0.430} & -0.017 & 0.052 & \multirow{2}{*}{0.556} \tabularnewline
 {\footnotesize \textit{(Jan. 1960 - Jul. 1979)}} & \textit{(0.154)} & \textit{(0.246)} & & \textit{(0.054)} & \textit{(0.097)} & \tabularnewline
Volcker-Greenspan & -0.054 & -0.136 & \multirow{2}{*}{0.871} & -0.014 & -0.172 & \multirow{2}{*}{0.702} \tabularnewline
 {\footnotesize \textit{(Aug. 1979 - Jan. 2006)}} & \textit{(0.131)} & \textit{(0.314)} & & \textit{(0.060)} & \textit{(0.118)} & \tabularnewline
Bernanke-Yellen-Powell & 0.018 & 0.555 & \multirow{2}{*}{0.541} & 0.087* & 0.098* & \multirow{2}{*}{0.180} \tabularnewline
 {\footnotesize \textit{(Feb. 2006 - Apr. 2025)}} & \textit{(0.134)} & \textit{(0.438)} & & \textit{(0.049)} & \textit{(0.055)} & \tabularnewline
  \hline
\end{tabular}

\end{subtable}
}
\medskip

{ {\footnotesize\textit{{}Notes:}}{\footnotesize{} We report estimates
and the standard error (in parentheses below the estimates) across
methods and setups. The symbols {*}, {*}{*}, and {*}{*}{*} indicate
significance levels at 10\%, 5\%, and 1\%, respectively. ``VR Test''
represents the $p$-value of the variance ratio test \citep{kim2009automatic}
on the LASSO residual. The tuning parameter for LASSO estimation is
selected through 10-fold block cross-validation. In XDlasso, instruments
are generated based on \eqref{eq:IVX_original} and \eqref{eq:rho=000020zeta}
with $C_{\zeta}=5$ and $\tau=0.5$.}}
\end{table}

\end{appendices}
\end{document}